\documentclass[pdflatex,sn-mathphys-num,fleqn]{sn-jnl}

\usepackage{graphicx}%
\usepackage{multirow}%
\usepackage{amsmath,amssymb,amsfonts}%
\usepackage{amsthm}%
\usepackage{mathrsfs}%
\usepackage[title]{appendix}%
\usepackage{xcolor}%
\usepackage{textcomp}%
\usepackage{manyfoot}%
\usepackage{booktabs}%
\usepackage{algorithm}%
\usepackage{algorithmicx}%
\usepackage{algpseudocode}%
\usepackage{listings}%

\theoremstyle{thmstyleone}%
\newtheorem{theorem}{Theorem}
\theoremstyle{thmstyletwo}%
\newtheorem{example}{Example}%
\newtheorem{remark}{Remark}%

\theoremstyle{thmstylethree}%
\newtheorem{Def}{Definition}%
\newtheorem{lem}{Lemma}

\begin{document}

\title[Integrable self-adaptive moving mesh schemes for multi-component short pulse type equations under periodic and one-end Dirichlet boundary conditions]{Integrable self-adaptive moving mesh schemes for multi-component short pulse type equations under periodic and one-end Dirichlet boundary conditions}

\author[1]{\fnm{Ayako} \sur{Hori}}\email{ayako0903@akane.waseda.jp}

\author[2]{\fnm{Ken-ichi} \sur{Maruno}}\email{kmaruno@waseda.jp}

\author[3]{\fnm{Yasuhiro} \sur{Ohta}}\email{ohta@math.kobe-u.ac.jp}

\author[4]{\fnm{Bao-Feng} \sur{Feng}}\email{baofeng.feng@utrgv.edu}

\affil[1]{Department of Pure and Applied Mathematics, School of Fundamental Science and Engineering, Waseda University, 3-4-1 Okubo, Shinjuku-ku, Tokyo 169-8555, Japan}

\affil[2]{Department of Applied Mathematics, Faculty of Science and Engineering, Waseda University, 3-4-1 Okubo, Shinjuku-ku, Tokyo 169-8555, Japan}

\affil[3]{Department of Mathematics, Kobe University, 1-1 Rokkodaicho, Nada-ku, Kobe 657-8501, Japan}

\affil[4]{School of Mathematical and Statistical Sciences, The University of Texas Rio Grande Valley, 1201 West University Dr., Edinburg, Texas 78541, USA}


\abstract{In this paper, we construct integrable self-adaptive moving mesh schemes for multi-component modified short pulse and short pulse equations under periodic and one-end Dirichlet boundary conditions with nonzero boundary values by using the consistency condition with the hodograph transformation. The essential point is that the edge point $x_{0}$ of the hodograph transformation cannot be kept fixed when the boundary value is nonzero. We derive the evolution equation for $x_{0}$ and incorporate it into the semi-discrete moving mesh scheme. Under periodic boundary conditions, the resulting schemes involve an algebraic constraint corresponding to the integral condition of the continuous periodic problem, whereas under one-end Dirichlet boundary conditions they form ordinary differential systems. We also derive multi-soliton solutions in Pfaffian form for the proposed schemes, preserving the integrable structure at the semi-discrete level. Numerical simulations under periodic and one-end Dirichlet boundary conditions demonstrate that the proposed schemes achieve high accuracy even in regions with rapid variation, including the passage of solitons through the computational boundary, and remain stable throughout the simulated time interval.}

\keywords{Self-adaptive moving mesh schemes, Multi-component short pulse type equations, Periodic boundary conditions, One-end Dirichlet boundary conditions, Pfaffian}



\maketitle

\begin{section}{Introduction}\label{sec_intro}

The study of discrete integrable systems has developed in connection with various fields of mathematics, physics, and engineering since the pioneering works of Hirota and Ablowitz--Ladik in the mid-1970s~\cite{Yuri, DS}.

Hirota proposed a method for discretizing soliton equations while preserving their integrability based on bilinear equations and $\tau$-functions, and succeeded in discretizing various soliton equations such as the KdV equation and the sine-Gordon equation~\cite{Hirota1, Hirota2, Hirota3, Hirota4, Hirota5}.
Furthermore, Ablowitz and Ladik proposed a method for discretizing soliton equations based on the Ablowitz--Kaup--Newell--Segur (AKNS) type linear eigenvalue problem for soliton equations, and succeeded in discretizing soliton equations such as the nonlinear Schr\"{o}dinger equation while preserving its integrability~\cite{Abl1, Abl2, Abl3, Abl4, Abl5}.

It is known that there are soliton equations associated with Wadati-Konno-Ichikawa (WKI) type linear eigenvalue problems~\cite{WKI1, WKI2}. Ishimori and Wadati-Sogo have shown that the soliton equations belonging to the WKI class can be transformed into the soliton equations associated with AKNS-type linear eigenvalue problems by hodograph transformations~\cite{ishimori1, ishimori2, wadachi1, rogers}.
For a long time, integrable discretizations of soliton equations in the WKI class had been considered difficult.
In our previous work, we obtained integrable discretizations of soliton equations associated with WKI-type linear eigenvalue problems by using Hirota's discretization method based on bilinear equations together with a discretization of the hodograph transformation~\cite{maruno13, maruno11, maruno12, maruno10, maruno16, maruno0, maruno14, maruno15}. We called the resulting integrable discrete equations self-adaptive moving mesh schemes, because the mesh interval is automatically refined in regions where the solution varies rapidly.

Integrable self-adaptive moving mesh schemes are difference schemes that automatically adjust the mesh interval, which are obtained by discretizing integrable nonlinear wave equations with exact solutions such as loop and cusp soliton solutions in a way that preserves the structure of the exact solutions.
The key to constructing a self-adaptive moving mesh scheme is the discretization of a hodograph transformation.
It is well known that a hodograph transformation of a nonlinear wave equation corresponds to a conservation law~\cite{rogers}.
In self-adaptive moving mesh schemes, a conserved density of a discrete conservation law corresponds to the mesh interval.
Therefore, the mesh interval is automatically adjusted where the displacement changes rapidly.

It has also been shown that self-adaptive moving mesh schemes can be derived as motions of discrete plane curves~\cite{geometry}. In terms of the geometric formulation, the Lagrangian representation of the evolution equation for the motion of a discrete curve is rewritten in the Eulerian representation by a discrete hodograph transformation, which is nothing but a self-adaptive moving mesh scheme.

Furthermore, it was reported that self-adaptive moving mesh schemes are effective and accurate in numerical computations~\cite{maruno0, maruno12, samms}.
However, the integrable self-adaptive moving mesh schemes constructed so far were mainly formulated for solutions that vanish at the computational boundary. In such a fixed-edge formulation, the edge point of the hodograph transformation is kept fixed, and hence situations with nonzero boundary values cannot be incorporated consistently. This becomes an obstruction, for example, in periodic computations when the field values near the computational boundary are nonzero. 

Numerical schemes based on hodograph transformations for the short pulse equation under periodic boundary conditions have been proposed by Sato et al.~\cite{Sato1, Sato2}. In~\cite{Sato1}, a moving mesh scheme was constructed on the basis of the coupled dispersionless formulation of the short pulse equation, in which the mesh points, including the edge point, evolve according to the discrete hodograph transformation while the spatial period is preserved. In~\cite{Sato2}, the short pulse equation on a periodic domain was transformed into the sine-Gordon equation by an arc-length-based hodograph transformation. It was shown that the base point of the hodograph transformation must evolve in time in order for a sine-Gordon solution to correspond to a solution of the short pulse equation. Their formulation also incorporates the zero-mean condition for periodic solutions, which is related to the integral constraint for the problem on the real line discussed in~\cite{Horikis}. These approaches are structure-preserving numerical methods, but they are not integrable discretizations in the sense considered here.

The aim of this paper is to construct integrable self-adaptive moving mesh schemes that can handle nonzero boundary values under periodic and one-end
Dirichlet boundary conditions. The theoretical new point of our construction is to derive, within the framework of integrable discretization, the evolution equation of the edge point $x_{0}$ from the consistency condition with the hodograph transformation and to include this equation in the integrable semi-discrete self-adaptive moving mesh scheme for multi-component systems, so that the resulting schemes still admit multi-soliton solutions in Pfaffian form. 
Here ``nonzero boundary values'' means boundary settings in which the field variables at the edge point of the computational interval need not vanish. More precisely, the edge point $x_{0}(T)=x(0,T)$ in the hodograph transformation must evolve whenever the boundary flux determined by those values is nonzero.
The difficulty in constructing self-adaptive moving mesh schemes with nonzero boundary values is caused by the fact that the edge point of the discrete hodograph transformation is no longer fixed when the boundary flux is nonzero.

In this paper, by imposing the consistency condition with the hodograph transformation, we construct integrable self-adaptive moving mesh schemes applicable to problems with nonzero boundary values for two equations. The first is the multi-component modified short pulse (MCmSP) equation~\cite{MCmSP}
\begin{eqnarray}
u^{(i)}_{xt}=u^{(i)}+\left(\sum_{1\leq j < k \leq n} c_{jk}u^{(j)}u^{(k)}u^{(i)}_{x}\right)_{x}-\left(\sum_ {1 \leq j < k \leq n} c_{jk}u^{(j)}_{x}u^{(k)}_{x}\right)u^{(i)},\nonumber\\
\hspace{90mm}  i=1,2,\cdots,n, \label{MCmSP}
\end{eqnarray}
which is a multi-component generalization of the modified short pulse equation~\cite{mSP2}. The second is the multi-component short pulse (MCSP) equation~\cite{MCSP}
\begin{eqnarray}
u^{(i)}_{xt}=u^{(i)}+\frac{1}{2}\left(\sum_{1\leq j < k \leq n} c_{jk}u^{(j)}u^{(k)}u^{(i)}_{x}\right)_{x}, \qquad i=1,2,\cdots,n,\label{MCSP}
\end{eqnarray}
which is a multi-component generalization of the short pulse equation~\cite{SP1, SP5, SP2, SP3, SP4}. 
 The coefficients $c_{jk}$ are arbitrary constants satisfying the symmetry $c_{jk}=c_{kj}(j,k=1,2,\cdots,n)$, with $c_{jj}=0$. Although the sums are written only for $j<k$ to avoid double counting, this symmetry is used later in the proof of Lemma \ref{lemma1}, where the coefficients are extended to expressions involving all pairs of component labels. We also construct multi-soliton solutions in Pfaffian form for our self-adaptive moving mesh schemes. Using the obtained self-adaptive moving mesh schemes, we perform numerical simulations to evaluate their effectiveness and accuracy. Under periodic boundary conditions, we consider the bright-soliton solutions of the two-component modified short pulse (2-mSP), two-component short pulse (2-SP), complex modified short pulse (CmSP), and complex short pulse (CSP) equations. Under one-end Dirichlet boundary conditions, we consider the dark-soliton solutions of the defocusing CmSP and
defocusing CSP equations, which have nonzero boundary values.

This paper is structured as follows.
In Section~\ref{sec_MCmSP}, we describe the bilinear form of the MCmSP equation and its multi-soliton solution. In Section~\ref{sec_disMCmSP}, we construct an integrable self-adaptive moving mesh scheme with a moving edge point for the MCmSP equation, which has long been considered difficult. In Section~\ref{sec_numexpMCmSP}, we present numerical simulations of the 2-mSP and the CmSP equations under periodic boundary conditions, and we also perform dark-soliton simulations for the defocusing CmSP equation under a one-end Dirichlet boundary condition. In Sections~\ref{sec_MCSP} and \ref{sec_disMCSP}, we discretize the MCSP equation and construct an integrable self-adaptive moving mesh scheme for the MCSP equation with nonzero boundary values; such an integrable scheme was not previously available. In Section~\ref{sec_numexpMCSP}, using the same boundary settings as in Section~\ref{sec_numexpMCmSP}, we present numerical simulations of the 2-SP and CSP equations, including dark-soliton simulations for the defocusing CSP equation. Finally, in Section~\ref{sec_conc}, we discuss the results and give conclusions.

\section{Bilinear equations and $N$-soliton solution for the MCmSP equation}\label{sec_MCmSP}

In this section, we describe the bilinear form and the hodograph transformation of the MCmSP equation (\ref{MCmSP}) and construct its $N$-soliton solution.
First, we rewrite the MCmSP equation (\ref{MCmSP}) into the conservation law form
\begin{eqnarray}
\left(\frac{1}{\rho}\right)_{t}-\left(\frac{F}{\rho}\right)_{x}=0,\label{m-conservationlaw}
\end{eqnarray}
where
\begin{eqnarray}
\rho=\left(1+\sum_{1\leq j < k\leq n}c_{jk}u^{(j)}_{x}u^{(k)}_{x}\right)^{-1},\qquad F=\sum_{1\leq j < k\leq n}c_{jk}u^{(j)}u^{(k)}.\label{m-conservationlaw2}
\end{eqnarray}
Indeed, if
$S=\sum_{1\leq j<k\leq n}c_{jk}u^{(j)}_{x}u^{(k)}_{x}$, then (\ref{MCmSP}) can be written as
\begin{eqnarray}
u^{(i)}_{xt}=(1-S)u^{(i)}+F_{x}u^{(i)}_{x}+F u^{(i)}_{xx}.
\nonumber
\end{eqnarray}
Using this expression, one obtains
$S_{t}=F_{x}(1+S)+F S_{x}=(F(1+S))_{x}$, which is equivalent to (\ref{m-conservationlaw}) because $1/\rho=1+S$.
Therefore, the hodograph transformation
\begin{eqnarray}
dX=\frac{1}{\rho}dx+\frac{F}{\rho}dt,\qquad dT=dt,\label{m-hodograph}
\end{eqnarray}
is introduced. The conservation law (\ref{m-conservationlaw}) is the closedness condition for this one-form.
From the  hodograph transformation (\ref{m-hodograph}), we obtain the derivative law
\begin{eqnarray}
\frac{\partial }{\partial X}=\rho \frac{\partial }{\partial x},\qquad\frac{\partial }{\partial T}=\frac{\partial }{\partial t}-F\frac{\partial }{\partial x}.\label{m-derivativelaw}
\end{eqnarray}
Next, transforming the MCmSP equation into
\begin{eqnarray}
\partial_{x}\left(\partial_{t}-\sum_{1\leq j < k \leq n}c_{jk}u^{(j)}u^{(k)}\partial_{x}\right)u^{(i)}=u^{(i)}\left(2-\left(1+\sum_{1 \leq j < k \leq n}c_{jk}u^{(j)}_{x}u^{(k)}_{x}\right)\right),
\label{7}
\end{eqnarray}
and applying the derivative law (\ref{m-derivativelaw}) and (\ref{m-conservationlaw2}) to (\ref{7}), we obtain
\begin{eqnarray}
u^{(i)}_{XT}=u^{(i)}\left(2\rho-1\right).
\end{eqnarray}
Furthermore, applying the derivative law (\ref{m-derivativelaw}) to the conservation law (\ref{m-conservationlaw}), we then have
\begin{eqnarray}
 \rho_{T}+\left(\displaystyle\sum_{1\leq j < k\leq n}c_{jk}u^{(j)}u^{(k)}\right)_{X}=0.\label{m-conservationlaw3}
 \end{eqnarray}
The equation (\ref{m-conservationlaw3}) is the conservation law for $X$ and $T$, where $\rho$ is the conserved density and $\sum_{1\leq j < k\leq n}c_{jk}u^{(j)}u^{(k)}$ is the flux.

Therefore,  the MCmSP equation (\ref{MCmSP}) is transformed into the multi-component coupled integrable dispersionless  (MCCID) equations
\begin{eqnarray}
\left\{
\begin{array}{ll}
u^{(i)}_{XT}=u^{(i)}\left(2\rho-1\right),\\
 \rho_{T}+\left(\displaystyle\sum_{1\leq j < k\leq n}c_{jk}u^{(j)}u^{(k)}\right)_{X}=0,
\end{array}
\right.
\label{m-dispersionless}
\end{eqnarray}
by the hodograph transformation (\ref{m-hodograph}).

Introducing the dependent variable transformation
\begin{eqnarray}
u^{(i)}=\frac{g^{(i)}}{f}\quad(i=1,2,\cdots,n),\qquad\rho=1-({\rm log}f)_{XT},\label{m-dependenttransformation}
\end{eqnarray}
the MCCID equations (\ref{m-dispersionless}) lead to
\begin{eqnarray}
\left\{
\begin{array}{ll}
D_{X}D_{T}f\cdot g^{(i)}=fg^{(i)},\\
 D_{T}^{2}f\cdot f=2\displaystyle\sum_{1\leq j < k\leq n}c_{jk}g^{(j)}g^{(k)},
\end{array}
\right.
\label{MCmSPbilinear}
\end{eqnarray}
where $D_{X}$ and $D_{T}$ are Hirota's D-operators which are defined as
\begin{align}
&D_{X}^{m}D_{T}^{n}f(X,T)\cdot g(X,T)\nonumber\\
&\quad=\left(\frac{\partial}{\partial X}-\frac{\partial}{\partial X^{\prime}}\right)^{m}
\left(\frac{\partial}{\partial T}-\frac{\partial}{\partial T^{\prime}}\right)^{n}
f(X,T) g(X^{\prime},T^{\prime})|_{X=X^{\prime},T=T^{\prime}}.\nonumber
 \end{align}


\begin{Def}
Let $A$ be a $2N \times 2N$ antisymmetric matrix defined by
\begin{eqnarray}
A=(a_{i,j})_{1\leq i,j \leq 2N},\nonumber
\end{eqnarray}
where $a_{i,j}=-a_{j,i}$ for $i,j=1,2,\cdots,2N$. The Pfaffian of $A$ is defined by
\begin{align}
{\rm Pf}(A)&={\rm Pf}(a_{i,j})_{1\leq i,j \leq 2N}={\rm Pf}(1,2,\dots,2N)\nonumber\\
&=\sum_{\begin{array}{c}
 i_{1}<i_{2}<\cdots< i_{N}\\
 i_{1}<j_{1},\ i_{2}<j_{2},\ \cdots,\ i_{N}<j_{N}
\end{array}}{\rm sgn}  (\pi) a_{i_{1}, j_{1}}a_{i_{2}, j_{2}}\cdots a_{i_{N}, j_{N}},\nonumber
\end{align}
where
$\pi=\left(
    \begin{array}{ccccc}
    1 & 2 & \cdots & 2N-1 &2N \\
    i_{1} & j_{1} & \cdots & i_{N} & j_{N}  \\
       \end{array}
    \right)
$
is a permutation of $\{1,2,\cdots,2N-1,2N\}$.

The Pfaffian is defined by the following expansion rule
\begin{align}
{\rm Pf}(1, 2, \cdots, 2N) =\sum_{i=2}^{2N}(-1)^{i}{\rm Pf}(1,i){\rm Pf}(2,3,\cdots, i-1, i+1, \cdots, 2N-1, 2N).\nonumber
\end{align}
For an antisymmetric matrix $A$, we have $[{\rm Pf}(A)]^2={\rm det}(A)$.
\end{Def}

In what follows, ${\rm Pf}(s_{1},s_{2},\ldots,s_{2M})$ denotes the Pfaffian of the antisymmetric matrix whose rows and columns are indexed by the formal symbols $s_{1},s_{2},\ldots,s_{2M}$.
Thus symbols such as $a_{j}$, $b_{j}$, $d_{l}$, $d^{l}$ and $B_{\mu}$ are labels of Pfaffian entries. A hat over an argument, for example $\hat{a}_{j}$, means that the argument is omitted. Pfaffian entries are
antisymmetric, ${\rm Pf}(\alpha, \beta) = -{\rm Pf}(\beta, \alpha)$.


\begin{lem}
\label{lemma1}
The bilinear equations (\ref{MCmSPbilinear}) have the following Pfaffian solution:
\begin{eqnarray}
f={\rm Pf}(a_{1},\cdots,a_{2N},b_{1},\cdots,b_{2N}), \qquad g^{(i)}={\rm Pf}(d_{0},B_{i},a_{1},\cdots,a_{2N},b_{1},\cdots,b_{2N}),
\label{MCmSP_f}
\end{eqnarray}
where $i=1,2,\cdots,n$ and the elements of the Pfaffians are defined as
\begin{eqnarray}
{\rm Pf}(a_{j},a_{k})=\frac{p_{j}-p_{k}}{p_{j}+p_{k}}e^{\xi_{j}+\xi_{k}},\qquad {\rm Pf}(a_{j},b_{k})=\delta_{j,k},
\label{pf1}
\end{eqnarray}
\begin{eqnarray}
{\rm Pf}(b_{j},b_{k})=\frac{c_{\mu\nu}}{p_{j}^{-2}-p_{k}^{-2}}\quad(b_{j}\in B_{\mu},b_{k}\in B_{\nu} ),
\end{eqnarray}
\begin{eqnarray}
{\rm Pf}(d_{l},a_{k})=p_{k}^{l}e^{\xi_{k}},\qquad
{\rm Pf}(b_{j},B_{\mu})=
\left\{
\begin{array}{ll}
1 & (b_{j}\in B_{\mu}) \\
0 & (b_{j}\notin B_{\mu})
\end{array}
\right..
\label{pf2}
\end{eqnarray}
Here, $j,k =1,2,\cdots,2N$,\quad$\mu,\nu=1,2,\cdots,n,\quad\xi_{j}=p_{j}X+p_{j}^{-1}T+\xi_{j0}$, $p_j$ and $\xi_{j0}$ are arbitrary constants, $\delta_{j,k}$ denotes the Kronecker delta, and $l$ is an integer. Although the solution $g^{(i)}$ involves only the symbol $d_{0}$, additional symbols  $d_{l}$ with different integer values of $l$ are introduced to express derivatives of the Pfaffian compactly in the proof. A class of set $B_{\mu}(\mu=1,2,\cdots,n)$ satisfies the following conditions
\begin{eqnarray}
B_{\mu}\cap B_{\nu}=\varnothing\,\,{\rm if}\,\, \mu\neq\nu,\qquad\cup_{\mu=1}^{n}B_{\mu}=\{b_{1},b_{2},\cdots,b_{2N}\}.
\end{eqnarray}
All other entries of the Pfaffians, except those obtained from the above by antisymmetry, are zero.
\end{lem}

\begin{proof}
First, we prove that the solutions (\ref{MCmSP_f}) satisfy the first equation of the bilinear equation (\ref{MCmSPbilinear}).
As
\begin{eqnarray}
\frac{\partial}{\partial X}{\rm Pf}(a_{j},a_{k})&=&(p_{j}-p_{k})e^{\xi_{j}+\xi_{k}}={\rm Pf}(d_{0},d_{1},a_{j},a_{k}),\nonumber\\
\frac{\partial}{\partial T}{\rm Pf}(a_{j},a_{k})&=&(p_{k}^{-1}-p_{j}^{-1})e^{\xi_{j}+\xi_{k}}={\rm Pf}(d_{-1},d_{0},a_{j},a_{k}),\nonumber\\
\frac{\partial^{2}}{\partial T^{2}}{\rm Pf}(a_{j},a_{k})&=&(p_{k}^{-2}-p_{j}^{-2})e^{\xi_{j}+\xi_{k}}={\rm Pf}(d_{-2},d_{0},a_{j},a_{k}),\nonumber\\
\frac{\partial^{2}}{\partial X \partial T}{\rm Pf}(a_{j},a_{k})&=&(p_{j}p_{k}^{-1}-p_{k}p_{j}^{-1})e^{\xi_{j}+\xi_{k}}={\rm Pf}(d_{-1},d_{1},a_{j},a_{k}),\nonumber
\end{eqnarray}

we have
\begin{eqnarray}
\frac{\partial f}{\partial X}&=&{\rm Pf}(d_{0},d_{1},\cdots),\qquad
\frac{\partial f}{\partial T}={\rm Pf}(d_{-1},d_{0},\cdots),\nonumber\\
\frac{\partial^{2} f}{\partial T^{2}}&=&{\rm Pf}(d_{-2},d_{0},\cdots),\qquad
\frac{\partial^{2} f}{\partial X\partial T}={\rm Pf}(d_{-1},d_{1},\cdots).\nonumber
\end{eqnarray}
Here, ${\rm Pf}(d_{i},d_{j},a_{1},\cdots,a_{2N},b_{1},\cdots,b_{2N})$ is denoted as ${\rm Pf}(d_{i},d_{j},\cdots )$.

Furthermore,
\begin{align}
\frac{\partial g^{(i)}}{\partial X}&=\frac{\partial}{\partial X}\left[\sum_{j=1}^{2N}(-1)^{j}{\rm Pf}(d_{0},a_{j}){\rm Pf}(B_{i},\cdots,\hat{a}_{j},\cdots)\right]\nonumber\\
&=\sum_{j=1}^{2N}(-1)^{j}(\partial_{X}{\rm Pf}(d_{0},a_{j}))
{\rm Pf}(B_{i},\cdots,\hat{a}_{j},\cdots)\nonumber\\
&\quad+\sum_{j=1}^{2N}(-1)^{j}{\rm Pf}(d_{0},a_{j})
\partial_{X}{\rm Pf}(B_{i},\cdots,\hat{a}_{j},\cdots)\nonumber\\
&=\sum_{j=1}^{2N}(-1)^{j}{\rm Pf}(d_{1},a_{j})
{\rm Pf}(B_{i},\cdots,\hat{a}_{j},\cdots)\nonumber\\
&\quad+\sum_{j=1}^{2N}(-1)^{j}{\rm Pf}(d_{0},a_{j})
{\rm Pf}(d_{0},d_{1},B_{i},\cdots,\hat{a}_{j},\cdots)\nonumber\\
&={\rm Pf}(d_{1},B_{i},\cdots)+{\rm Pf}(d_{0},d_{0},d_{1},B_{i},\cdots)\nonumber\\
&={\rm Pf}(d_{1},B_{i},\cdots).\nonumber
\end{align}
Similarly, we have
\begin{align}
\frac{\partial g^{(i)}}{\partial T}&=\frac{\partial}{\partial T}\left[\sum_{j=1}^{2N}(-1)^{j}{\rm Pf}(d_{0},a_{j}){\rm Pf}(B_{i},\cdots,\hat{a}_{j},\cdots)\right]\nonumber\\
&=\sum_{j=1}^{2N}(-1)^{j}(\partial_{T}{\rm Pf}(d_{0},a_{j}))
{\rm Pf}(B_{i},\cdots,\hat{a}_{j},\cdots)\nonumber\\
&\quad+\sum_{j=1}^{2N}(-1)^{j}{\rm Pf}(d_{0},a_{j})
\partial_{T}{\rm Pf}(B_{i},\cdots,\hat{a}_{j},\cdots)\nonumber\\
&=\sum_{j=1}^{2N}(-1)^{j}{\rm Pf}(d_{-1},a_{j})
{\rm Pf}(B_{i},\cdots,\hat{a}_{j},\cdots)\nonumber\\
&\quad+\sum_{j=1}^{2N}(-1)^{j}{\rm Pf}(d_{0},a_{j})
{\rm Pf}(d_{-1},d_{0},B_{i},\cdots,\hat{a}_{j},\cdots)\nonumber\\
&={\rm Pf}(d_{-1},B_{i},\cdots)+{\rm Pf}(d_{0},d_{-1},d_{0},B_{i},\cdots)\nonumber\\
&={\rm Pf}(d_{-1},B_{i},\cdots)\nonumber,\\
\frac{\partial^{2}g^{(i)}}{\partial X\partial T}&=\frac{\partial}{\partial X}\left[\sum_{j=1}^{2N}(-1)^{j}{\rm Pf}(d_{-1},a_{j}){\rm Pf}(B_{i},\cdots,\hat{a}_{j},\cdots)\right]\nonumber\\
&=\sum_{j=1}^{2N}(-1)^{j}(\partial_{X}{\rm Pf}(d_{-1},a_{j}))
{\rm Pf}(B_{i},\cdots,\hat{a}_{j},\cdots)\nonumber\\
&\quad+\sum_{j=1}^{2N}(-1)^{j}{\rm Pf}(d_{-1},a_{j})
\partial_{X}{\rm Pf}(B_{i},\cdots,\hat{a}_{j},\cdots)\nonumber\\
&=\sum_{j=1}^{2N}(-1)^{j}{\rm Pf}(d_{0},a_{j})
{\rm Pf}(B_{i},\cdots,\hat{a}_{j},\cdots)\nonumber\\
&\quad+\sum_{j=1}^{2N}(-1)^{j}{\rm Pf}(d_{-1},a_{j})
{\rm Pf}(d_{0},d_{1},B_{i},\cdots,\hat{a}_{j},\cdots)\nonumber\\
&={\rm Pf}(d_{0},B_{i},\cdots)+{\rm Pf}(d_{-1},d_{0},d_{1},B_{i},\cdots)\nonumber\\
&={\rm Pf}(d_{0},B_{i},\cdots)+{\rm Pf}(d_{-1},B_{i},d_{0},d_{1},\cdots).\nonumber
\end{align}

An algebraic identity of Pfaffian~\cite{Hirotabook}
\begin{align}
&{\rm Pf}(d_{-1},B_{i},d_{0},d_{1},\cdots){\rm Pf}(\cdots)\nonumber\\
&={\rm Pf}(d_{-1},d_{0},\cdots){\rm Pf}(d_{1},B_{i},\cdots)
 -{\rm Pf}(d_{-1},d_{1},\cdots){\rm Pf}(d_{0},B_{i},\cdots)\nonumber\\
&\quad+{\rm Pf}(d_{-1},B_{i},\cdots){\rm Pf}(d_{0},d_{1},\cdots)\nonumber
\end{align}
leads to
\begin{eqnarray}
(\partial_{X}\partial_{T} g^{(i)}-g^{(i)})\times f
=\partial_{T} f \times \partial_{X} g^{(i)}
-\partial_{X}\partial_{T} f \times g^{(i)}
+\partial_{T} g^{(i)} \times \partial_{X} f ,\nonumber
\end{eqnarray}
which is the first equation of the bilinear equations (\ref{MCmSPbilinear}).

Next, we prove that the solutions (\ref{MCmSP_f}) satisfy the second equation of the bilinear equations (\ref{MCmSPbilinear}).
The right-hand side of the second equation of the bilinear equations (\ref{MCmSPbilinear}) leads to
\begin{align}
&2\sum_{1\leq \mu < \nu \leq n}c_{\mu \nu}g^{(\mu)}g^{(\nu)}\nonumber\\
&=\sum_{1\leq \mu , \nu \leq n}c_{\mu \nu}{\rm Pf}(d_{0},B_{\mu},\cdots){\rm Pf}(d_{0},B_{\nu},\cdots)\nonumber\\
&=\sum_{1\leq \mu , \nu \leq n}c_{\mu \nu}\sum_{j,k=1}^{2N}(-1)^{j+k}{\rm Pf}(B_{\mu},b_{j}){\rm Pf}(d_{0},\cdots,\hat{b}_{j},\cdots){\rm Pf}(B_{\nu},b_{k}){\rm Pf}(d_{0},\cdots,\hat{b}_{k},\cdots)\nonumber\\
&=\sum_{j,k=1}^{2N}(-1)^{j+k}\sum_{1\leq \mu , \nu \leq n}c_{\mu \nu}{\rm Pf}(B_{\mu},b_{j}){\rm Pf}(B_{\nu},b_{k}){\rm Pf}(d_{0},\cdots,\hat{b}_{j},\cdots){\rm Pf}(d_{0},\cdots,\hat{b}_{k},\cdots)\nonumber\\
&=\sum_{j,k=1}^{2N}(-1)^{j+k}(p_{j}^{-2}-p_{k}^{-2}){\rm Pf}(b_{j},b_{k}){\rm Pf}(d_{0},\cdots,\hat{b}_{j},\cdots){\rm Pf}(d_{0},\cdots,\hat{b}_{k},\cdots).\label{1.3.18}
\end{align}
In the first equality above, we use $c_{\mu\mu}=0$ and the symmetry $c_{\mu\nu}=c_{\nu\mu}$.
Now, expanding ${\rm Pf}(b_{j},d_{0},\cdots)$ which is trivially zero ${\rm Pfaffian}$, we obtain

\begin{eqnarray}
\sum_{k=1}^{2N}(-1)^{j+k}{\rm Pf}(b_{j},b_{k}){\rm Pf}(d_{0},\cdots,\hat{b}_{k},\cdots)={\rm Pf}(d_{0},\cdots,\hat{a}_{j},\cdots),\nonumber
\end{eqnarray}
which leads to
\begin{align}
&\sum_{j,k=1}^{2N}(-1)^{j+k}p_{j}^{-2}{\rm Pf}(b_{j},b_{k}){\rm Pf}(d_{0},\cdots,\hat{b}_{j},\cdots){\rm Pf}(d_{0},\cdots,\hat{b}_{k},\cdots)\nonumber\\
&=\sum_{j=1}^{2N}p_{j}^{-2}{\rm Pf}(d_{0},\cdots,\hat{a}_{j},\cdots){\rm Pf}(d_{0},\cdots,\hat{b}_{j},\cdots).\label{1.3.25}
\end{align}
Similarly, we have
\begin{align}
&-\sum_{j,k=1}^{2N}(-1)^{j+k}p_{k}^{-2}{\rm Pf}(b_{j},b_{k}){\rm Pf}(d_{0},\cdots,\hat{b}_{j},\cdots){\rm Pf}(d_{0},\cdots,\hat{b}_{k},\cdots)\nonumber\\
&=\sum_{k=1}^{2N}p_{k}^{-2}{\rm Pf}(d_{0},\cdots,\hat{a}_{k},\cdots){\rm Pf}(d_{0},\cdots,\hat{b}_{k},\cdots).\label{1.3.26}
\end{align}
Substituting (\ref{1.3.25}) and (\ref{1.3.26}) into (\ref{1.3.18}), we then have
\begin{eqnarray}
\label{1.3.27}
\sum_{1\leq \mu < \nu \leq n}c_{\mu \nu}g^{(\mu)}g^{(\nu)}=\sum_{j=1}^{2N}p_{j}^{-2}{\rm Pf}(d_{0},\cdots,\hat{a}_{j},\cdots){\rm Pf}(d_{0},\cdots,\hat{b}_{j},\cdots).
\end{eqnarray}
Since ${\rm Pf}(d_0,d_0,\ldots)=0$, we have
\begin{align}
&\frac{\partial^{2}f}{\partial T^{2}}\times 0-\frac{\partial f}{\partial T}\frac{\partial f}{\partial T}\nonumber\\
&={\rm Pf}(d_{-2},d_{0},\cdots){\rm Pf}(d_{0},d_{0},\cdots)-{\rm Pf}(d_{-1},d_{0},\cdots){\rm Pf}(d_{-1},d_{0},\cdots)\nonumber\\
&=\sum_{j=1}^{2N}(-1)^{j}{\rm Pf}(d_{-2},a_{j}){\rm Pf}(d_{0},\cdots,\hat{a}_{j},\cdots)\sum_{k=1}^{2N}(-1)^{k}{\rm Pf}(d_{0},a_{k}){\rm Pf}(d_{0},\cdots,\hat{a}_{k},\cdots)\nonumber\\
&\qquad-\sum_{j=1}^{2N}(-1)^{j}{\rm Pf}(d_{-1},a_{j}){\rm Pf}(d_{0},\cdots,\hat{a}_{j},\cdots)\sum_{k=1}^{2N}(-1)^{k}{\rm Pf}(d_{-1},a_{k}){\rm Pf}(d_{0},\cdots,\hat{a}_{k},\cdots)\nonumber\\
&=\sum_{j,k=1}^{2N}(-1)^{j+k}[{\rm Pf}(d_{-2},a_{j}){\rm Pf}(d_{0},a_{k})-{\rm Pf}(d_{-1},a_{j}){\rm Pf}(d_{-1},a_{k})]\nonumber\\
&\quad\times{\rm Pf}(d_{0},\cdots,\hat{a}_{j},\cdots){\rm Pf}(d_{0},\cdots,\hat{a}_{k},\cdots)\nonumber\\
&=\sum_{j,k=1}^{2N}(-1)^{j+k+1}(p_{j}^{-2}+p_{j}^{-1}p_{k}^{-1}){\rm Pf}(a_{j},a_{k}){\rm Pf}(d_{0},\cdots,\hat{a}_{j},\cdots){\rm Pf}(d_{0},\cdots,\hat{a}_{k},\cdots).
\label{A55}
\end{align}
Since
\begin{align}
&\sum_{j,k=1}^{2N}(-1)^{j+k+1}p_{j}^{-1}p_{k}^{-1}{\rm Pf}(a_{j},a_{k}){\rm Pf}(d_{0},\cdots,\hat{a}_{j},\cdots){\rm Pf}(d_{0},\cdots,\hat{a}_{k},\cdots)\nonumber\\
&=\sum_{k,j=1}^{2N}(-1)^{k+j+1}p_{k}^{-1}p_{j}^{-1}{\rm Pf}(a_{k},a_{j}){\rm Pf}(d_{0},\cdots,\hat{a}_{k},\cdots){\rm Pf}(d_{0},\cdots,\hat{a}_{j},\cdots)\nonumber\\
&=-\sum_{j,k=1}^{2N}(-1)^{j+k+1}p_{j}^{-1}p_{k}^{-1}{\rm Pf}(a_{j},a_{k}){\rm Pf}(d_{0},\cdots,\hat{a}_{j},\cdots){\rm Pf}(d_{0},\cdots,\hat{a}_{k},\cdots),
\end{align}
we obtain
\begin{eqnarray}
\begin{aligned}
&\sum_{j,k=1}^{2N}(-1)^{j+k+1}p_{j}^{-1}p_{k}^{-1}{\rm Pf}(a_{j},a_{k})
{\rm Pf}(d_{0},\cdots,\hat{a}_{j},\cdots)\\
&\qquad\times{\rm Pf}(d_{0},\cdots,\hat{a}_{k},\cdots)=0.
\end{aligned}
\label{A66}
\end{eqnarray}
Substituting (\ref{A66}) into (\ref{A55}), we obtain
\begin{align}
-\frac{\partial f}{\partial T}\frac{\partial f}{\partial T}
&=\sum_{j,k=1}^{2N}(-1)^{j+k+1}p_{j}^{-2}{\rm Pf}(a_{j},a_{k})
{\rm Pf}(d_{0},\cdots,\hat{a}_{j},\cdots)\nonumber\\
&\quad\times{\rm Pf}(d_{0},\cdots,\hat{a}_{k},\cdots)\nonumber\\
&=\sum_{j=1}^{2N}(-1)^{j+1}p_{j}^{-2}
{\rm Pf}(d_{0},\cdots,\hat{a}_{j},\cdots)\nonumber\\
&\quad\times\left[\sum_{k=1}^{2N}(-1)^{k}{\rm Pf}(a_{j},a_{k})
{\rm Pf}(d_{0},\cdots,\hat{a}_{k},\cdots)\right].\label{A77}
\end{align}
Now, expanding ${\rm Pf}(a_{j},d_{0},\cdots)$ which is trivially zero ${\rm Pfaffian}$, we obtain
\begin{eqnarray}
\begin{aligned}
&\sum_{k=1}^{2N}(-1)^{k}{\rm Pf}(a_{j},a_{k})
{\rm Pf}(d_{0},\cdots,\hat{a}_{k},\cdots)\\
&\qquad={\rm Pf}(d_{0},a_{j}){\rm Pf}(\cdots)
        +(-1)^{j+1}{\rm Pf}(d_{0},\cdots,\hat{b}_{j},\cdots).\\
\end{aligned}
\label{A88}
\end{eqnarray}
Substituting (\ref{A88}) into (\ref{A77}), we have
\begin{align}
&\quad-\frac{\partial f}{\partial T}\frac{\partial f}{\partial T}\nonumber\\
&=\sum_{j=1}^{2N}(-1)^{j+1}p_{j}^{-2}
{\rm Pf}(d_{0},\cdots,\hat{a}_{j},\cdots)
\left[{\rm Pf}(d_{0},a_{j}){\rm Pf(\cdots)}
      +(-1)^{j+1}{\rm Pf}(d_{0},\cdots,\hat{b}_{j},\cdots)\right]\nonumber\\
&=\sum_{j=1}^{2N}(-1)^{j+1}p_{j}^{-2}{\rm Pf}(d_{0},a_{j})
{\rm Pf}(d_{0},\cdots,\hat{a}_{j},\cdots){\rm Pf(\cdots)}\nonumber\\
&\quad+\sum_{j=1}^{2N}p_{j}^{-2}
{\rm Pf}(d_{0},\cdots,\hat{a}_{j},\cdots)
{\rm Pf}(d_{0},\cdots,\hat{b}_{j},\cdots)\nonumber\\
&=-\sum_{j=1}^{2N}(-1)^{j}p_{j}^{-2}{\rm Pf}(d_{0},a_{j})
{\rm Pf}(d_{0},\cdots,\hat{a}_{j},\cdots){\rm Pf(\cdots)}\nonumber\\
&\quad+\sum_{j=1}^{2N}p_{j}^{-2}
{\rm Pf}(d_{0},\cdots,\hat{a}_{j},\cdots)
{\rm Pf}(d_{0},\cdots,\hat{b}_{j},\cdots)\nonumber\\
&=-\sum_{j=1}^{2N}(-1)^{j}{\rm Pf}(d_{-2},a_{j})
{\rm Pf}(d_{0},\cdots,\hat{a}_{j},\cdots){\rm Pf(\cdots)}\nonumber\\
&\quad+\sum_{j=1}^{2N}p_{j}^{-2}
{\rm Pf}(d_{0},\cdots,\hat{a}_{j},\cdots)
{\rm Pf}(d_{0},\cdots,\hat{b}_{j},\cdots)\nonumber\\
&=-{\rm Pf}(d_{-2},d_{0},\cdots){\rm Pf}(\cdots)
 +\sum_{j=1}^{2N}p_{j}^{-2}
{\rm Pf}(d_{0},\cdots,\hat{a}_{j},\cdots)
{\rm Pf}(d_{0},\cdots,\hat{b}_{j},\cdots)\nonumber\\
&=-\frac{\partial^{2} f}{\partial T^{2}}f+\displaystyle\sum_{1\leq \mu < \nu\leq n}c_{\mu\nu}g^{(\mu)}g^{(\nu)}.\nonumber
\end{align}
The last calculation used (\ref{1.3.27}).
Therefore, we arrive at
\begin{eqnarray}
\frac{\partial^{2}f}{\partial T^{2}}f-\frac{\partial f}{\partial T}\frac{\partial f}{\partial T}=\sum_{1\leq \mu < \nu\leq n}c_{\mu\nu}g^{(\mu)}g^{(\nu)},\nonumber
\end{eqnarray}
which is the second equation of the bilinear equations (\ref{MCmSPbilinear}).
\end{proof}
\end{section}

\begin{example}\label{ex1}
\leavevmode
\\
1-soliton: For $N=1, n=2, B_{1}=\{b_{1}\}, B_{2}=\{b_{2}\}, c_{12}=1$, we obtain the $\tau$-functions $f, g^{(1)}$ and $g^{(2)}$ as follows:
\begin{align}
f&={\rm Pf}(a_{1},a_{2},b_{1},b_{2})\nonumber\\
&={\rm Pf}(a_{1},a_{2}){\rm Pf}(b_{1},b_{2})
  -{\rm Pf}(a_{1},b_{1}){\rm Pf}(a_{2},b_{2})\nonumber\\
&\quad +{\rm Pf}(a_{1},b_{2}){\rm Pf}(a_{2},b_{1})\nonumber\\
&=\frac{p_{1}-p_{2}}{p_{1}+p_{2}}\frac{1}{p_{1}^{-2}-p_{2}^{-2}}e^{\xi_{1}+\xi_{2}}-1\nonumber\\
&=-1-b_{12}e^{\xi_{1}+\xi_{2}},\nonumber\\
g^{(1)}&={\rm Pf}(d_{0},B_{1},a_{1},a_{2},b_{1},b_{2})\nonumber\\
&= {\rm Pf}(d_{0},B_{1}){\rm Pf}(a_{1},a_{2},b_{1},b_{2})
   -{\rm Pf}(d_{0},a_{1}){\rm Pf}(B_{1},a_{2},b_{1},b_{2})\nonumber\\
&\quad +{\rm Pf}(d_{0},a_{2}){\rm Pf}(B_{1},a_{1},b_{1},b_{2})
   -{\rm Pf}(d_{0},b_{1}){\rm Pf}(B_{1},a_{1},a_{2},b_{2})\nonumber\\
&\quad +{\rm Pf}(d_{0},b_{2}){\rm Pf}(B_{1},a_{1},a_{2},b_{1})\nonumber\\
&={\rm Pf}(d_{0},B_{1})
\left({\rm Pf}(a_{1},a_{2}){\rm Pf}(b_{1},b_{2})
      -{\rm Pf}(a_{1},b_{1}){\rm Pf}(a_{2},b_{2})
      +{\rm Pf}(a_{1},b_{2}){\rm Pf}(a_{2},b_{1})\right)\nonumber\\
&\quad -{\rm Pf}(d_{0},a_{1})
\left({\rm Pf}(B_{1},a_{2}){\rm Pf}(b_{1},b_{2})
      -{\rm Pf}(B_{1},b_{1}){\rm Pf}(a_{2},b_{2})
      +{\rm Pf}(B_{1},b_{2}){\rm Pf}(a_{2},b_{1})\right)\nonumber\\
&\quad +{\rm Pf}(d_{0},a_{2})
\left({\rm Pf}(B_{1},a_{1}){\rm Pf}(b_{1},b_{2})
      -{\rm Pf}(B_{1},b_{1}){\rm Pf}(a_{1},b_{2})
      +{\rm Pf}(B_{1},b_{2}){\rm Pf}(a_{1},b_{1})\right)\nonumber\\
&\quad -{\rm Pf}(d_{0},b_{1})
\left({\rm Pf}(B_{1},a_{1}){\rm Pf}(a_{2},b_{2})
      -{\rm Pf}(B_{1},a_{2}){\rm Pf}(a_{1},b_{2})
      +{\rm Pf}(B_{1},b_{2}){\rm Pf}(a_{1},a_{2})\right)\nonumber\\
&\quad +{\rm Pf}(d_{0},b_{2})
\left({\rm Pf}(B_{1},a_{1}){\rm Pf}(a_{2},b_{1})
      -{\rm Pf}(B_{1},a_{2}){\rm Pf}(a_{1},b_{1})
      +{\rm Pf}(B_{1},b_{1}){\rm Pf}(a_{1},a_{2})\right)\nonumber\\
&=-e^{\xi_{1}},\nonumber\\
g^{(2)}&={\rm Pf}(d_{0},B_{2},a_{1},a_{2},b_{1},b_{2})\nonumber\\
&= {\rm Pf}(d_{0},B_{2}){\rm Pf}(a_{1},a_{2},b_{1},b_{2})
   -{\rm Pf}(d_{0},a_{1}){\rm Pf}(B_{2},a_{2},b_{1},b_{2})\nonumber\\
&\quad +{\rm Pf}(d_{0},a_{2}){\rm Pf}(B_{2},a_{1},b_{1},b_{2})
   -{\rm Pf}(d_{0},b_{1}){\rm Pf}(B_{2},a_{1},a_{2},b_{2})\nonumber\\
&\quad +{\rm Pf}(d_{0},b_{2}){\rm Pf}(B_{2},a_{1},a_{2},b_{1})\nonumber\\
&={\rm Pf}(d_{0},B_{2})
\left({\rm Pf}(a_{1},a_{2}){\rm Pf}(b_{1},b_{2})
      -{\rm Pf}(a_{1},b_{1}){\rm Pf}(a_{2},b_{2})
      +{\rm Pf}(a_{1},b_{2}){\rm Pf}(a_{2},b_{1})\right)\nonumber\\
&\quad -{\rm Pf}(d_{0},a_{1})
\left({\rm Pf}(B_{2},a_{2}){\rm Pf}(b_{1},b_{2})
      -{\rm Pf}(B_{2},b_{1}){\rm Pf}(a_{2},b_{2})
      +{\rm Pf}(B_{2},b_{2}){\rm Pf}(a_{2},b_{1})\right)\nonumber\\
&\quad +{\rm Pf}(d_{0},a_{2})
\left({\rm Pf}(B_{2},a_{1}){\rm Pf}(b_{1},b_{2})
      -{\rm Pf}(B_{2},b_{1}){\rm Pf}(a_{1},b_{2})
      +{\rm Pf}(B_{2},b_{2}){\rm Pf}(a_{1},b_{1})\right)\nonumber\\
&\quad -{\rm Pf}(d_{0},b_{1})
\left({\rm Pf}(B_{2},a_{1}){\rm Pf}(a_{2},b_{2})
      -{\rm Pf}(B_{2},a_{2}){\rm Pf}(a_{1},b_{2})
      +{\rm Pf}(B_{2},b_{2}){\rm Pf}(a_{1},a_{2})\right)\nonumber\\
&\quad +{\rm Pf}(d_{0},b_{2})
\left({\rm Pf}(B_{2},a_{1}){\rm Pf}(a_{2},b_{1})
      -{\rm Pf}(B_{2},a_{2}){\rm Pf}(a_{1},b_{1})
      +{\rm Pf}(B_{2},b_{1}){\rm Pf}(a_{1},a_{2})\right)\nonumber\\
&=-e^{\xi_{2}},\nonumber
\end{align}
where $b_{jk}=\left(\frac{p_{j}p_{k}}{p_{j}+p_{k}}\right)^{2}$ and $\xi_{j}=p_{j}X+\frac{1}{p_{j}}T+\xi_{j0}$.

Letting $\xi_{j0}=\xi_{j0}^{\prime}+\log{a_{j}}$, the $\tau$-functions can be rewritten as
\begin{eqnarray}
f=-1-a_{1}a_{2}b_{12}e^{\eta_{1}+\eta_{2}}, \qquad g^{(1)}=-a_{1}e^{\eta_{1}},\qquad
g^{(2)}=-a_{2}e^{\eta_{2}},\label{g21}
\end{eqnarray}
where $b_{jk}=\left(\frac{p_{j}p_{k}}{p_{j}+p_{k}}\right)^{2}$ and $\eta_{j} =p_{j}X+\frac{1}{p_{j}}T+\xi_{j0}^{\prime}$.
\\
\\
2-soliton: For $N=2, n=2, B_{1}=\{b_{1},b_{2}\}, B_{2}=\{b_{3}, b_{4}\}, c_{12}=1$.
If a term contains a Pfaffian with zero elements, we omit that term. By the expansion rule, we obtain the $\tau$-function $f$ as follows:
\begin{align}
f&={\rm Pf}(a_{1},a_{2},a_{3},a_{4},b_{1},b_{2},b_{3},b_{4})\nonumber\\
&={\rm Pf}(a_{1},a_{2}){\rm Pf}(a_{3},a_{4},b_{1},b_{2},b_{3},b_{4})
 -{\rm Pf}(a_{1},a_{3}){\rm Pf}(a_{2},a_{4},b_{1},b_{2},b_{3},b_{4})\nonumber\\
&\quad+{\rm Pf}(a_{1},a_{4}){\rm Pf}(a_{2},a_{3},b_{1},b_{2},b_{3},b_{4})
 -{\rm Pf}(a_{1},b_{1}){\rm Pf}(a_{2},a_{3},a_{4},b_{2},b_{3},b_{4})\nonumber\\
&\quad+{\rm Pf}(a_{1},b_{2}){\rm Pf}(a_{2},a_{3},a_{4},b_{1},b_{3},b_{4})
 -{\rm Pf}(a_{1},b_{3}){\rm Pf}(a_{2},a_{3},a_{4},b_{1},b_{2},b_{4})\nonumber\\
&\quad+{\rm Pf}(a_{1},b_{4}){\rm Pf}(a_{2},a_{3},a_{4},b_{1},b_{2},b_{3})\nonumber\\
&={\rm Pf}(a_{1},a_{2})
\left({\rm Pf}(a_{3},a_{4}){\rm Pf}(b_{1},b_{2},b_{3},b_{4})
      -{\rm Pf}(a_{3},b_{3}){\rm Pf}(a_{4},b_{1},b_{2},b_{4})\right)\nonumber\\
&\quad-{\rm Pf}(a_{1},a_{3})
\left({\rm Pf}(a_{2},a_{4}){\rm Pf}(b_{1},b_{2},b_{3},b_{4})
      +{\rm Pf}(a_{2},b_{2}){\rm Pf}(a_{4},b_{1},b_{3},b_{4})\right)\nonumber\\
&\quad+{\rm Pf}(a_{1},a_{4})
\left({\rm Pf}(a_{2},a_{3}){\rm Pf}(b_{1},b_{2},b_{3},b_{4})
      +{\rm Pf}(a_{2},b_{2}){\rm Pf}(a_{3},b_{1},b_{3},b_{4})\right)\nonumber\\
&\quad-{\rm Pf}(a_{1},b_{1})
\left({\rm Pf}(a_{2},a_{3}){\rm Pf}(a_{4},b_{2},b_{3},b_{4})
      -{\rm Pf}(a_{2},a_{4}){\rm Pf}(a_{3},b_{2},b_{3},b_{4})\right.\nonumber\\
&\qquad\left.
      +{\rm Pf}(a_{2},b_{2}){\rm Pf}(a_{3},a_{4},b_{3},b_{4})\right)\nonumber\\
&=-{\rm Pf}(a_{1},a_{2}){\rm Pf}(a_{3},a_{4})
   {\rm Pf}(b_{1},b_{3}){\rm Pf}(b_{2},b_{4})
  +{\rm Pf}(a_{1},a_{2}){\rm Pf}(a_{3},a_{4})
   {\rm Pf}(b_{1},b_{4}){\rm Pf}(b_{2},b_{3})\nonumber\\
&\quad+{\rm Pf}(a_{1},a_{3}){\rm Pf}(a_{2},a_{4})
   {\rm Pf}(b_{1},b_{3}){\rm Pf}(b_{2},b_{4})
  -{\rm Pf}(a_{1},a_{3}){\rm Pf}(a_{2},a_{4})
   {\rm Pf}(b_{1},b_{4}){\rm Pf}(b_{2},b_{3})\nonumber\\
&\quad-{\rm Pf}(a_{1},a_{3}){\rm Pf}(a_{2},b_{2})
   {\rm Pf}(a_{4},b_{4}){\rm Pf}(b_{1},b_{3})
  -{\rm Pf}(a_{1},a_{4}){\rm Pf}(a_{2},a_{3})
   {\rm Pf}(b_{1},b_{3}){\rm Pf}(b_{2},b_{4})\nonumber\\
&\quad+{\rm Pf}(a_{1},a_{4}){\rm Pf}(a_{2},a_{3})
   {\rm Pf}(b_{1},b_{4}){\rm Pf}(b_{2},b_{3})
  -{\rm Pf}(a_{1},a_{4}){\rm Pf}(a_{2},b_{2})
   {\rm Pf}(a_{3},b_{3}){\rm Pf}(b_{1},b_{4})\nonumber\\
&\quad-{\rm Pf}(a_{1},b_{1}){\rm Pf}(a_{2},a_{3})
   {\rm Pf}(a_{4},b_{4}){\rm Pf}(b_{2},b_{3})
  -{\rm Pf}(a_{1},b_{1}){\rm Pf}(a_{2},a_{4})
   {\rm Pf}(a_{3},b_{3}){\rm Pf}(b_{2},b_{4})\nonumber\\
&\quad+{\rm Pf}(a_{1},b_{1}){\rm Pf}(a_{2},b_{2}){\rm Pf}(a_{3},b_{3}){\rm Pf}(a_{4},b_{4})\nonumber\\
&=1+b_{13}e^{\xi_{1}+\xi_{3}}+b_{23}e^{\xi_{2}+\xi_{3}}
  +b_{14}e^{\xi_{1}+\xi_{4}}+b_{24}e^{\xi_{2}+\xi_{4}}\nonumber\\
&\quad +(p_{1}-p_{2})^{2}(p_{3}-p_{4})^{2}
   \frac{b_{13}b_{23}b_{14}b_{24}}{p_{1}^{2}p_{2}^{2}p_{3}^{2}p_{4}^{2}}
   e^{\xi_{1}+\xi_{2}+\xi_{3}+\xi_{4}}.\nonumber
\end{align}

Similarly, we obtain the $\tau$-functions $g^{(1)}$ and $g^{(2)}$ as follows:
\begin{align}
g^{(1)}&={\rm Pf}(d_{0},B_{1},a_{1},a_{2},a_{3},a_{4},b_{1},b_{2},b_{3},b_{4})\nonumber\\
&=e^{\xi_{1}}+e^{\xi_{2}}+\frac{(p_{1}-p_{2})^{2}p_{3}^{4}}{(p_{1}+p_{3})^{2}(p_{2}+p_{3})^{2}}e^{\xi_{1}+\xi_{2}+\xi_{3}}+\frac{(p_{1}-p_{2})^{2}p_{4}^{4}}{(p_{1}+p_{4})^{2}(p_{2}+p_{4})^{2}}e^{\xi_{1}+\xi_{2}+\xi_{4}},\nonumber\\
g^{(2)}&={\rm Pf}(d_{0},B_{2},a_{1},a_{2},a_{3},a_{4},b_{1},b_{2},b_{3},b_{4})\nonumber\\
&=e^{\xi_{3}}+e^{\xi_{4}}+\frac{(p_{3}-p_{4})^{2}p_{2}^{4}}{(p_{2}+p_{3})^{2}(p_{2}+p_{4})^{2}}e^{\xi_{2}+\xi_{3}+\xi_{4}}+\frac{(p_{3}-p_{4})^{2}p_{1}^{4}}{(p_{1}+p_{3})^{2}(p_{1}+p_{4})^{2}}e^{\xi_{1}+\xi_{3}+\xi_{4}},\nonumber
\end{align}
where $\xi_{j}=p_{j}X+\frac{1}{p_{j}}T+\xi_{j0}$.

Letting $\xi_{j0}=\xi_{j0}^{\prime}+\log{a_{j}}$, the $\tau$-functions can be rewritten as
\begin{align}
f&=1+a_{1}a_{3}b_{13}e^{\eta_{1}+\eta_{3}}
  +a_{2}a_{3}b_{23}e^{\eta_{2}+\eta_{3}}
  +a_{1}a_{4}b_{14}e^{\eta_{1}+\eta_{4}}\nonumber\\
&\quad+a_{2}a_{4}b_{24}e^{\eta_{2}+\eta_{4}}
  +a_{1}a_{2}a_{3}a_{4}(p_{1}-p_{2})^{2}(p_{3}-p_{4})^{2}
    \frac{b_{13}b_{23}b_{14}b_{24}}{p_{1}^{2}p_{2}^{2}p_{3}^{2}p_{4}^{2}}
    e^{\eta_{1}+\eta_{2}+\eta_{3}+\eta_{4}},\label{f2}
\end{align}
\begin{align}
g^{(1)}&=a_{1}e^{\eta_{1}}+a_{2}e^{\eta_{2}}
 +\frac{a_{1}a_{2}a_{3}(p_{1}-p_{2})^{2}p_{3}^{4}}
        {(p_{1}+p_{3})^{2}(p_{2}+p_{3})^{2}}
  e^{\eta_{1}+\eta_{2}+\eta_{3}}\nonumber\\
&\quad+\frac{a_{1}a_{2}a_{4}(p_{1}-p_{2})^{2}p_{4}^{4}}
        {(p_{1}+p_{4})^{2}(p_{2}+p_{4})^{2}}
  e^{\eta_{1}+\eta_{2}+\eta_{4}},\label{g12}
\end{align}
\begin{align}
g^{(2)}&=a_{3}e^{\eta_{3}}+a_{4}e^{\eta_{4}}
 +\frac{a_{2}a_{3}a_{4}(p_{3}-p_{4})^{2}p_{2}^{4}}
        {(p_{2}+p_{3})^{2}(p_{2}+p_{4})^{2}}
  e^{\eta_{2}+\eta_{3}+\eta_{4}}\nonumber\\
&\quad+\frac{a_{1}a_{3}a_{4}(p_{3}-p_{4})^{2}p_{1}^{4}}
        {(p_{1}+p_{3})^{2}(p_{1}+p_{4})^{2}}
  e^{\eta_{1}+\eta_{3}+\eta_{4}},\label{g22}
\end{align}
where $b_{jk}=\left(\frac{p_{j}p_{k}}{p_{j}+p_{k}}\right)^{2}$ and $\eta_{j} =p_{j}X+\frac{1}{p_{j}}T+\xi_{j0}^{\prime}$.

These solutions  (\ref{g21})--(\ref{g22}) can also be obtained from Hirota's direct method.
\end{example}


\begin{section}{A self-adaptive moving mesh scheme for the MCmSP equation}
\label{sec_disMCmSP}
A semi-discrete analogue of bilinear equations (\ref{MCmSPbilinear}) is
\begin{eqnarray}
\left\{
\begin{array}{ll}
  \displaystyle\frac{1}{a}D_{T}(f_{l+1}\cdot g_{l}^{(i)}-f_{l}\cdot g_{l+1}^{(i)})=g_{l+1}^{(i)}f_{l}+g_{l}^{(i)}f_{l+1},\qquad i=1,2,\cdots,n,\\
  D_{T}^{2}f_{l}\cdot f_{l}=2\displaystyle\sum_{1\leq j < k \leq n}c_{jk}g_{l}^{(j)}g_{l}^{(k)},
\end{array}
\right.
\label{disbilinear}
\end{eqnarray}
where $2a$ is the parameter related to a discrete interval.

\begin{lem}
The bilinear equations (\ref{disbilinear}) have the following Pfaffian solution:
\begin{eqnarray}
f_{l}={\rm Pf}(a_{1},\cdots,a_{2N},b_{1},\cdots,b_{2N})_{l},\qquad g_{l}^{(i)}={\rm Pf}(d_{0},B_{i},a_{1},\cdots,a_{2N},b_{1},\cdots,b_{2N})_{l},
\label{dis_f_MCmSP}
\end{eqnarray}
where $i=1,2,\cdots,n$ and the elements of the Pfaffians are defined as
\begin{eqnarray}
{\rm Pf}(a_{j},a_{k})_{l}=\frac{p_{j}-p_{k}}{p_{j}+p_{k}}\phi_{j}^{(0)}(l)\phi_{k}^{(0)}(l),\qquad{\rm Pf}(a_{j},b_{k})_{l}=\delta_{j,k},
\label{dmpf1}
\end{eqnarray}
\begin{eqnarray}
{\rm Pf}(b_{j},b_{k})_{l}=\frac{c_{\mu\nu}}{p_{j}^{-2}-p_{k}^{-2}}\quad(b_{j}\in B_{\mu},b_{k}\in B_{\nu} ),
\end{eqnarray}
\begin{eqnarray}
{\rm Pf}(d_{m},a_{j})_{l}=\phi_{j}^{(m)}(l),\qquad{\rm Pf}(d^{l},a_{j})_{l}=\phi_{j}^{(0)}(l+1),
\end{eqnarray}
\begin{eqnarray}
{\rm Pf}(b_{k},B_{\mu})_{l}=
\left\{
\begin{array}{ll}
 1 & (b_{k}\in B_{\mu}),\\
 0 & (b_{k}\notin B_{\mu}),
\end{array}
\right.
\end{eqnarray}
\begin{eqnarray}
{\rm Pf}(d_{0},d^{l})_{l}=1,\qquad{\rm Pf}(d_{-1},d^{l})_{l}=-a,
\label{dmpf2}
\end{eqnarray}
where
\begin{eqnarray}
\phi_{j}^{(n)}(l)=p_{j}^{n}\left(\frac{1+ap_{j}}{1-ap_{j}}\right)^{l}e^{p_{j}^{-1}T+\xi_{j0}},
\end{eqnarray}
and it satisfies the following equation
\begin{eqnarray}
\frac{\phi_{j}^{(n)}(l+1)-\phi_{j}^{(n)}(l)}{a}=\phi_{j}^{(n+1)}(l+1)+\phi_{j}^{(n+1)}(l).
\end{eqnarray}
All other entries of the Pfaffians, except those obtained from the above by antisymmetry, are zero.
\end{lem}

\begin{proof}
First, we show the solutions (\ref{dis_f_MCmSP}) satisfy the first equation of the bilinear equations (\ref{disbilinear}). Since
\begin{align}
\frac{\partial}{\partial T}{\rm Pf}(a_{j},a_{k})_{l}&=\phi_{j}^{(0)}(l)\phi_{k}^{(-1)}(l)-\phi_{j}^{(-1)}(l)\phi_{k}^{(0)}(l)\nonumber\\
&={\rm Pf}(d_{0},a_{j})_{l}{\rm Pf}(d_{-1},a_{k})_{l}-{\rm Pf}(d_{-1},a_{j})_{l}{\rm Pf}(d_{0},a_{k})_{l}\nonumber\\
&={\rm Pf}(d_{-1},d_{0})_{l}{\rm Pf}(a_{j},a_{k})_{l}
  -{\rm Pf}(d_{-1},a_{j})_{l}{\rm Pf}(d_{0},a_{k})_{l}\nonumber\\
&\quad+{\rm Pf}(d_{-1},a_{k})_{l}{\rm Pf}(d_{0},a_{j})_{l}\nonumber\\
&={\rm Pf}(d_{-1},d_{0},a_{j},a_{k})_{l},\nonumber
\end{align}
\begin{align}
{\rm Pf}(a_{j},a_{k})_{l+1}&={\rm Pf}(a_{j},a_{k})_{l}
 +\phi_{j}^{(0)}(l+1)\phi_{k}^{(0)}(l)
 -\phi_{j}^{(0)}(l)\phi_{k}^{(0)}(l+1)\nonumber\\
&={\rm Pf}(d_{0},d^{l})_{l}{\rm Pf}(a_{j},a_{k})_{l}
  -{\rm Pf}(d_{0},a_{j})_{l}{\rm Pf}(d^{l},a_{k})_{l}\nonumber\\
&\quad+{\rm Pf}(d_{0},a_{k})_{l}{\rm Pf}(d^{l},a_{j})_{l}\nonumber\\
&={\rm Pf}(d_{0},d^{l},a_{j},a_{k})_{l},\nonumber
\end{align}
\begin{align}
&(\partial_{T}-a){\rm Pf}(a_{j},a_{k})_{l+1}\nonumber\\
&=\partial_{T}{\rm Pf}(a_{j},a_{k})_{l+1}-a{\rm Pf}(a_{j},a_{k})_{l+1}\nonumber\\
&=\phi_{j}^{(0)}(l+1)\phi_{k}^{(-1)}(l+1)-\phi_{j}^{(-1)}(l+1)\phi_{k}^{(0)}(l+1)\nonumber\\
&\qquad-a\left({\rm Pf}(a_{j},a_{k})_{l}+\phi_{j}^{(0)}(l+1)\phi_{k}^{(0)}(l)-\phi_{j}^{(0)}(l)\phi_{k}^{(0)}(l+1)\right)\nonumber\\
&=-a{\rm Pf}(a_{j},a_{k})_{l}
 +\phi_{j}^{(0)}(l+1)\left(\phi_{k}^{(-1)}(l+1)-a\phi_{k}^{(0)}(l)\right)\nonumber\\
&\quad-\phi_{k}^{(0)}(l+1)\left(\phi_{j}^{(-1)}(l+1)-a\phi_{j}^{(0)}(l)\right)\nonumber\\
&=-a{\rm Pf}(a_{j},a_{k})_{l}
 +\phi_{j}^{(0)}(l+1)\left(a\phi_{k}^{(0)}(l+1)+\phi_{k}^{(-1)}(l)\right)\nonumber\\
&\quad-\phi_{k}^{(0)}(l+1)\left(a\phi_{j}^{(0)}(l+1)+\phi_{j}^{(-1)}(l)\right)\nonumber\\
&=-a{\rm Pf}(a_{j},a_{k})_{l}+\phi_{j}^{(0)}(l+1)\phi_{k}^{(-1)}(l)-\phi_{j}^{(-1)}(l)\phi_{k}^{(0)}(l+1)\nonumber\\
&={\rm Pf}(d_{-1},d^{l})_{l}{\rm Pf}(a_{j},a_{k})_{l}
  -{\rm Pf}(d_{-1},a_{j})_{l}{\rm Pf}(d^{l},a_{k})_{l}\nonumber\\
&\quad+{\rm Pf}(d_{-1},a_{k})_{l}{\rm Pf}(d^{l},a_{j})_{l}\nonumber\\
&={\rm Pf}(d_{-1},d^{l},a_{j},a_{k})_{l},\nonumber
\end{align}
we obtain
\begin{eqnarray}
\begin{aligned}
\partial_{T}f_{l}&={\rm Pf}(d_{-1},d_{0},\cdots)_{l},\qquad
f_{l+1}={\rm Pf}(d_{0},d^{l},\cdots)_{l},\\
(\partial_{T}-a)f_{l+1}&={\rm Pf}(d_{-1},d^{l},\cdots)_{l}.
\end{aligned}\nonumber
\end{eqnarray}
Furthermore, we have
\begin{align}
\partial_{T}g_{l}^{(\mu)}&=\partial_{T}\left(\sum_{j=1}^{2N}(-1)^{j}{\rm Pf}(d_{0},a_{j})_{l}{\rm Pf}(B_{\mu},\cdots,\hat{a}_{j},\cdots)_{l}\right)\nonumber\\
&=\sum_{j=1}^{2N}(-1)^{j}(\partial_{T}{\rm Pf}(d_{0},a_{j})_{l})
{\rm Pf}(B_{\mu},\cdots,\hat{a}_{j},\cdots)_{l}\nonumber\\
&\quad+\sum_{j=1}^{2N}(-1)^{j}{\rm Pf}(d_{0},a_{j})_{l}
\partial_{T}{\rm Pf}(B_{\mu},\cdots,\hat{a}_{j},\cdots)_{l}\nonumber\\
&=\sum_{j=1}^{2N}(-1)^{j}{\rm Pf}(d_{-1},a_{j})_{l}
{\rm Pf}(B_{\mu},\cdots,\hat{a}_{j},\cdots)_{l}\nonumber\\
&\quad+\sum_{j=1}^{2N}(-1)^{j}{\rm Pf}(d_{0},a_{j})_{l}
{\rm Pf}(B_{\mu},d_{-1},d_{0},\cdots,\hat{a}_{j},\cdots)_{l}\nonumber\\
&={\rm Pf}(d_{-1},B_{\mu},\cdots)_{l}+{\rm Pf}(d_{0},B_{\mu},d_{-1},d_{0},\cdots)_{l}\nonumber\\
&={\rm Pf}(d_{-1},B_{\mu},\cdots)_{l},\nonumber
\end{align}
\begin{align}
g_{l+1}^{(\mu)}&=\sum_{j=1}^{2N}(-1)^{j}{\rm Pf}(d_{0},a_{j})_{l+1}
{\rm Pf}(B_{\mu},\cdots,\hat{a}_{j},\cdots)_{l+1}\nonumber\\
&=\sum_{j=1}^{2N}(-1)^{j}{\rm Pf}(d^{l},a_{j})_{l}
{\rm Pf}(B_{\mu},d_{0},d^{l},\cdots,\hat{a}_{j},\cdots)_{l}\nonumber\\
&={\rm Pf}(d^{l},B_{\mu},\cdots)_{l},\nonumber
\end{align}
by expanding ${\rm Pf}(d^{l},B_{\mu},d_{0},d^{l},\cdots)$ which is trivially zero Pfaffian.
The simple expression of $(\partial_{T}-a)g_{l+1}^{(\mu)}$ is obtained as follows:
\begin{align}
&(\partial_{T}-a)g_{l+1}^{(\mu)}\nonumber\\
&=\sum_{j=1}^{2N}(-1)^{j}(\partial_{T}-a){\rm Pf}(d^{l},a_{j})_{l}
{\rm Pf}(B_{\mu},\cdots,\hat{a}_{j},\cdots)_{l}\nonumber\\
&\quad+\sum_{j=1}^{2N}(-1)^{j}{\rm Pf}(d^{l},a_{j})_{l}
\partial_{T}{\rm Pf}(B_{\mu},\cdots,\hat{a}_{j},\cdots)_{l}\nonumber\\
&=\sum_{j=1}^{2N}(-1)^{j}
({\rm Pf}(d_{-1},a_{j})_{l}+a{\rm Pf}(d_{0},a_{j})_{l})
{\rm Pf}(B_{\mu},\cdots,\hat{a}_{j},\cdots)_{l}\nonumber\\
&\quad+\sum_{j=1}^{2N}(-1)^{j}{\rm Pf}(d^{l},a_{j})_{l}
{\rm Pf}(d_{-1},d_{0},B_{\mu},\cdots,\hat{a}_{j},\cdots)_{l}\nonumber\\
&={\rm Pf}(d_{-1},B_{\mu},\cdots)_{l}
  +a {\rm Pf}(d_{0},B_{\mu},\cdots)_{l}\nonumber\\
&\quad+\sum_{j=1}^{2N}(-1)^{j}{\rm Pf}(d^{l},a_{j})_{l}
{\rm Pf}(d_{-1},d_{0},B_{\mu},\cdots,\hat{a}_{j},\cdots)_{l}\nonumber\\
&={\rm Pf}(d^{l},d_{-1},d_{0},B_{\mu},\cdots)_{l}={\rm Pf}(d_{-1},B_{\mu},d_{0},d^{l},\cdots)_{l}.\nonumber
\end{align}

An algebraic identity of {\rm Pfaffian}~\cite{Hirotabook}
\begin{align}
&{\rm Pf}(d_{-1},B_{\mu},d_{0},d^{l},\cdots)_{l}{\rm Pf}(\cdots)_{l}\nonumber\\
&={\rm Pf}(d_{-1},d_{0},\cdots)_{l}{\rm Pf}(d^{l},B_{\mu},\cdots)_{l}
 -{\rm Pf}(d_{-1},d^{l},\cdots)_{l}{\rm Pf}(d_{0},B_{\mu},\cdots)_{l}\nonumber\\
&\quad+{\rm Pf}(d_{-1},B_{\mu},\cdots)_{l}{\rm Pf}(d_{0},d^{l},\cdots)_{l},\nonumber
\end{align}
gives
\begin{eqnarray}
(\partial_{T}-a)g_{l+1}^{(\mu)}\times f_{l}
=g_{l+1}^{(\mu)}\times \partial_{T}f_{l}
-(\partial_{T}-a)f_{l+1}\times g_{l}^{(\mu)}
+\partial_{T}g_{l}^{(\mu)}\times f_{l+1},\nonumber
\end{eqnarray}
which is the first equation of the bilinear equations (\ref{disbilinear}). The second equation of the bilinear equations (\ref{disbilinear}) can be proved as in the continuous case.

\end{proof}


\begin{lem}
A semi-discrete analogue of the MCCID equations (\ref{m-dispersionless})
\begin{eqnarray}
\left\{
\begin{array}{ll}
  \displaystyle\frac{d}{dT}(u_{l+1}^{(i)}-u_{l}^{(i)})=a(2\rho_{l}-1)(u_{l+1}^{(i)}+u_{l}^{(i)}),\\
  \displaystyle\frac{d}{dT}\rho_{l}=-\displaystyle\frac{1}{2a}\sum_{1\leq j < k \leq n}c_{jk}(u_{l+1}^{(j)}u_{l+1}^{(k)}-u_{l}^{(j)}u_{l}^{(k)}),
\end{array}
\right.
\label{m-semidisdispersionless}
\end{eqnarray}
is obtained from the bilinear equations (\ref{disbilinear}) through the dependent variable transformation
\begin{eqnarray}
u_{l}^{(i)}=\frac{g_{l}^{(i)}}{f_{l}},\qquad
\rho_{l}=1-\frac{1}{2a}\left(\log{\frac{f_{l+1}}{f_{l}}}\right)_{T}.
\label{disdependenttransformation}
\end{eqnarray}

\end{lem}


\begin{proof}
First, dividing both sides of the first equation of the bilinear equations (\ref{disbilinear}) by $f_{l+1}f_{l}$, we have
\begin{eqnarray}
\left(\frac{g_{l+1}^{(i)}}{f_{l+1}}-\frac{g_{l}^{(i)}}{f_{l}}\right)_{T}=\left(a-\left(\log{\frac{f_{l+1}}{f_{l}}}\right)_{T}\right)\left(\frac{g_{l+1}^{(i)}}{f_{l+1}}+\frac{g_{l}^{(i)}}{f_{l}}\right).
\end{eqnarray}
Hence, we obtain
\begin{eqnarray}
\frac{d}{dT}(u_{l+1}^{(i)}-u_{l}^{(i)})=a(2\rho_{l}-1)(u_{l+1}^{(i)}+u_{l}^{(i)}),
\end{eqnarray}
by the dependent variable transformation (\ref{disdependenttransformation}).

Next, from the dependent variable transformation (\ref{disdependenttransformation}), the second equation of the bilinear equations (\ref{disbilinear}) is rewritten as
\begin{eqnarray}
(\log{f_{l}})_{TT}=\sum_{1\leq j < k\leq n}c_{jk}u_{l}^{(j)}u_{l}^{(k)}.\label{33}
\end{eqnarray}
By (\ref{33}) and the dependent variable transformation (\ref{disdependenttransformation}), we obtain
 \begin{eqnarray}
\frac{d}{dT}\rho_{l}=-\frac{1}{2a}\sum_{1\leq j < k \leq n}c_{jk}(u_{l+1}^{(j)}u_{l+1}^{(k)}-u_{l}^{(j)}u_{l}^{(k)}).
\end{eqnarray}

\end{proof}


\begin{lem}
A semi-discrete analogue of the MCmSP equation is of the form:
\begin{eqnarray}
\left\{
\begin{array}{ll}
  \displaystyle\frac{d} {dT}(u_{l+1}^{(i)}-u_{l}^{(i)})= (\delta_{l}-a)(u_{l+1}^{(i)}+u_{l}^{(i)}),\\
  \\
 \displaystyle\frac{d\delta_{l}}{dT}=-\displaystyle\sum_{1\leq j < k\leq n}c_{jk}(u_{l+1}^{(j)}u_{l+1}^{(k)}-u_{l}^{(j)}u_{l}^{(k)}),\\
  x_{l}=x_{0}+ \displaystyle\sum_{m=0}^{l-1}\delta_{m},\quad t=T.
\end{array}
\right.
\label{semidisMCmSP}
\end{eqnarray}
Here, $u_{l}^{(i)}$ is the value at $x_{l}$.
\end{lem}


\begin{proof}
From the derivative law (\ref{m-derivativelaw}), we obtain
\begin{eqnarray}
\frac{\partial x}{\partial X}=\rho.\label{6.19}
\end{eqnarray}
Integrating (\ref{6.19}) with respect to $X$, we have the integral form of the hodograph transformation (\ref{m-hodograph})
\begin{eqnarray}
x=\int\rho(X,T)dX=x_{0}+\int_{0}^{X}\rho(\bar{X},T)d\bar{X},
\label{999}
\end{eqnarray}
where $x_{0}$ is the integration constant with respect to $X$, determined by the value at the left-hand edge $(x=x_{0})$; it may depend on $T$ when the boundary flux is nonzero.
Discretizing (\ref{999}), we then have the discrete hodograph transformation
\begin{eqnarray}
x_{l}=x_{0}+\sum_{m=0}^{l-1}2a\rho_{m},
\label{dishodographMCmSP}
\end{eqnarray}
where $X_l=2al$ is the uniform lattice coordinate in the $X$ variable and $ \rho_{l}\equiv\rho(X_{l},T)$.

Introducing a mesh interval
\begin{eqnarray}
\delta_{l}:=x_{l+1}-x_{l},
\end{eqnarray}
we obtain
\begin{eqnarray}
\delta_{l}=2a\rho_{l},\label{meshintMCmSP}
\end{eqnarray}
by substituting the discrete hodograph transformation (\ref{dishodographMCmSP}).
Substituting  (\ref{meshintMCmSP}) into the semi-discrete MCCID equations (\ref{m-semidisdispersionless}), we have a semi-discrete analogue of the MCmSP equation (\ref{semidisMCmSP}).
\end{proof}

\begin{remark}
Equation (\ref{dishodographMCmSP}) is the summation form of the discrete hodograph transformation. At fixed $T$, its local form is
\begin{eqnarray}
\Delta x_{l}=\rho_{l}\Delta X_{l},\qquad \Delta X_l=2a.
\label{ddishodographMCmSP}
\end{eqnarray}
Here $\Delta x_l=x_{l+1}-x_l$. Dividing (\ref{ddishodographMCmSP}) by $\Delta X_l$ gives
\begin{eqnarray}
\frac{\Delta x_{l}}{\Delta X_{l}}=\rho_{l},
\end{eqnarray}
and summing this relation from $0$ to $l-1$ recovers (\ref{dishodographMCmSP}). The second equation of (\ref{m-semidisdispersionless}), which is the discrete version of the conservation law (\ref{m-conservationlaw3}), gives the compatible time evolution of the mesh interval,
\begin{eqnarray}
\frac{d}{dT}\Delta x_l
=-\sum_{1\leq j<k\leq n}c_{jk}
(u_{l+1}^{(j)}u_{l+1}^{(k)}-u_l^{(j)}u_l^{(k)}).
\label{compatible time evolution}
\end{eqnarray}
Since the hodograph transformation (\ref{dishodographMCmSP}) determines the mesh points only up to the additive term $x_0(T)$, while (\ref{compatible time evolution}) governs only the evolution of the mesh intervals, the evolution of $x_0(T)$ must be specified separately. This yields the consistency condition stated in Theorem~\ref{thm:MCmSP}.
\end{remark}


\begin{remark}
If the edge point $x_{0}$ is fixed, then $x_{0}$ in the hodograph transformation becomes time-independent. This is consistent only when the boundary flux vanishes. This fixed-edge formulation does not apply to periodic or one-end Dirichlet problems with nonzero boundary values, since in such cases the edge point generally evolves in time.
\end{remark}


\begin{theorem}
\label{thm:MCmSP}
Assume that the edge point is allowed to move and satisfies the consistency condition with the hodograph transformation, namely $dx_{0}/dT=-\sum_{1\leq j<k\leq n}c_{jk}u_{0}^{(j)}u_{0}^{(k)}$. Then the semi-discrete system (\ref{semidisMCmSP}) can be rewritten as the following self-adaptive moving mesh scheme for the MCmSP equation:
\begin{eqnarray}
\left\{
\begin{array}{ll}
\displaystyle\frac{d}{dT}(u_{l+1}^{(i)}-u_{l}^{(i)})=(\delta_{l}-a)(u_{l+1}^{(i)}+u_{l}^{(i)}), \\
\\
\displaystyle\frac{d x_{l}}{d T}=-\displaystyle\sum_{1\leq j < k \leq n}c_{jk}u_{l}^{(j)}u_{l}^{(k)}.
\end{array}
\right.
\label{selfMCmSP}
\end{eqnarray}
Here, $u_{l}^{(i)}$ is the value at $x_{l}$.
\end{theorem}

\begin{proof}
We first recall the corresponding continuous calculation, which motivates the consistency condition for the moving edge point.
From the derivative law (\ref{m-derivativelaw}), we have the evolution equation for $x$
\begin{eqnarray}
\frac{\partial x}{\partial T}=-\displaystyle\sum_{1\leq j < k \leq n}c_{jk}u^{(j)}u^{(k)}.
\label{6.27}
\end{eqnarray}
Equation (\ref{6.27}) means that the edge point $x_0(T)=x(0,T)$ is fixed only when the boundary flux $\sum_{1\leq j < k \leq n}c_{jk}u^{(j)}(0,T)u^{(k)}(0,T)
$ vanishes. When the boundary flux is nonzero, $x_{0}$ must evolve with $T$, and (\ref{999}) can be rewritten as
\begin{eqnarray}
x=x_{0}(T)+\int_{0}^{X} \rho(\bar{X},T)d\bar{X}.\label{semidis-hod-SP}
\end{eqnarray}
The following calculation derives the same edge-point condition from the integral representation. Differentiating the integral form of the hodograph transformation (\ref{semidis-hod-SP}) with respect to $T$, we obtain
\begin{align}
\frac{\partial x}{\partial T}&=\frac{\partial x_{0}(T)}{\partial T}+\frac{\partial}{\partial T}\int_{0}^{X}\rho(\bar{X},T)d\bar{X}\nonumber\\
&=\frac{\partial x_{0}(T)}{\partial T}-\int_{0}^{X}\frac{\partial}{\partial \bar{X}}\left(\sum_{1\leq j < k\leq n}c_{jk}u^{(j)}(\bar{X},T)u^{(k)}(\bar{X},T)\right)d\bar{X}\nonumber\\
&=\frac{\partial x_{0}(T)}{\partial T}
 -\displaystyle{\sum_{1\leq j < k \leq n}}c_{jk}u^{(j)}(X,T)u^{(k)}(X,T)\nonumber\\
&\quad
 +\displaystyle{\sum_{1\leq j < k \leq n}}c_{jk}u^{(j)}(0,T)u^{(k)}(0,T).
\label{6.28}
\end{align}
The calculation from the first to the second line uses the conservation law (\ref{m-conservationlaw3}).
Comparing (\ref{6.27}) with (\ref{6.28}), the $X$-dependent flux terms agree, and the remaining boundary term gives the evolution equation for the edge point $x_{0}$:
\begin{eqnarray}
\frac{\partial x_{0}(T)}{\partial T}=-\displaystyle{\sum_{1\leq j < k \leq n}}c_{jk}u^{(j)}(0,T)u^{(k)}(0,T).
\label{mSP_const}
\end{eqnarray}
We call (\ref{mSP_const}) the consistency condition with the hodograph transformation.

We now turn to the semi-discrete case, which gives the actual proof of the theorem.
Differentiating the summation form of the discrete hodograph transformation (\ref{dishodographMCmSP}) with respect to $T$ and using the second equation of (\ref{m-semidisdispersionless}), we obtain
\begin{align}
\frac{d x_{l}}{d T}=&\frac{d x_{0}(T)}{d T}+\frac{d}{d T}\sum_{m=0}^{l-1}2a\rho_{m}=\frac{d x_{0}(T)}{d T}+\sum_{m=0}^{l-1}2a\frac{d \rho_{m}}{d T}\nonumber\\
=&\frac{d x_{0}(T)}{d T}-\sum_{m=0}^{l-1}\sum_{1\leq j < k \leq n}c_{jk}(u_{m+1}^{(j)}u_{m+1}^{(k)}-u_{m}^{(j)}u_{m}^{(k)})\nonumber\\
=&\frac{d x_{0}(T)}{d T}+\sum_{1\leq j < k \leq n}c_{jk}(-u_{l}^{(j)}u_{l}^{(k)}+u_{0}^{(j)}u_{0}^{(k)}).
\label{43}
\end{align}
Equation (\ref{43}) corresponds to (\ref{6.28}) in the continuous case.
As in the continuous case, we obtain the consistency condition, i.e., the evolution equation for the edge point $x_{0}$:
\begin{eqnarray}
\frac{dx_{0}(T)}{d T}=-\displaystyle{\sum_{1\leq j < k \leq n}}c_{jk}u_{0}^{(j)}u_{0}^{(k)}.
\label{mSP_disconst}
\end{eqnarray}
Substituting (\ref{mSP_disconst}) into (\ref{43}) yields
\begin{eqnarray}
\frac{d x_{l}}{d T}=-\sum_{1\leq j < k \leq n}c_{jk}u_{l}^{(j)}u_{l}^{(k)}.
\end{eqnarray}
\end{proof}

\begin{remark}
 The integral form of the hodograph transformation (\ref{999}) and the summation form of the discrete hodograph transformation (\ref{dishodographMCmSP}) can be expressed by the $\tau$-function, with an appropriate choice of $x_{0}(T)$, as
 \begin{eqnarray}
 x=X-(\log{f})_{T},
 \end{eqnarray}
 \begin{eqnarray}
 x_{l}=2al-(\log{f_{l}})_{T}.
  \end{eqnarray}
 See Appendix~\ref{secA1} for details.
\end{remark}

Dividing the first equation of (\ref{selfMCmSP}) by $\delta_{l}$, one arrives at
\begin{eqnarray}
\frac{(u_{l+1}^{(i)}-u_{l}^{(i)})_{T}}{\delta_{l}}=\left(1-\frac{a}{\delta_{l}}\right)(u_{l+1}^{(i)}+u_{l}^{(i)}).
\label{6.36}
\end{eqnarray}
In taking the continuous limit $a\rightarrow0$ and $\delta_l\rightarrow0$, the left-hand side of (\ref{6.36}) is 
\begin{eqnarray}
\frac{(u_{l+1}^{(i)}-u_l^{(i)})_T}{\delta_l}
\rightarrow \partial_x(\partial_T u^{(i)}).\nonumber
\end{eqnarray}
On the right-hand side, we use
$u_{l+1}^{(i)}+u_l^{(i)}=2u^{(i)}+O(\delta_l)$. Since $\delta_l=2a\rho_l$,
\begin{eqnarray}
\frac{a}{\delta_{l}}=\frac{1}{2\rho_{l}}\rightarrow \frac{1}{2}\left(1+\sum_{1 \leq j < k \leq n}c_{jk}u^{(j)}_{x}u^{(k)}_{x}\right).
\end{eqnarray}
Thus the limit of (\ref{6.36}), before returning to the Eulerian variables, is
\begin{eqnarray}
u^{(i)}_{Tx}=\left(2-\left(1+\sum_{1 \leq j < k \leq n}c_{jk}u^{(j)}_{x}u^{(k)}_{x}\right)\right)u^{(i)}.
\label{68}
\end{eqnarray}
Finally, the derivative with respect to $T$ at fixed lattice label is transformed back to the Eulerian variables as
\begin{eqnarray}
\frac{d}{dT}&=&\frac{dx_{l}}{dT}\frac{d}{d x_{l}}+\frac{d}{dt}\rightarrow-\sum_{1\leq j < k \leq n}c_{jk}u^{(j)}u^{(k)} \partial_{x}+\partial_{t}\quad  {\rm for}\quad a\rightarrow 0.
\label{69}
\end{eqnarray}
Substituting (\ref{69}) into (\ref{68}), equation (\ref{6.36}) converges to the MCmSP equation
\begin{eqnarray}
\partial_{x}\left(\partial_{t}-\sum_{1\leq j<k \leq n}c_{jk}u^{(j)}u^{(k)}\partial_{x}\right)u^{(i)}=u^{(i)}\left(2-\left(1+\sum_{1 \leq j < k \leq n}c_{jk}u^{(j)}_{x}u^{(k)}_{x}\right)\right).\nonumber
\end{eqnarray}
\end{section}


\begin{section}{Numerical simulations of the 2-mSP and the CmSP equations}
\label{sec_numexpMCmSP}
\subsection{Numerical simulations of the 2-mSP and the CmSP equations under periodic boundary conditions}
\label{subsec:CCmSPperiodic}

In this subsection, we first describe the treatment of periodic boundary conditions for the semi-discrete MCmSP system, and then apply it to numerical simulations of the 2-mSP and the CmSP equations.

\medskip
\noindent\textit{Periodic boundary treatment.}

Here, we consider the semi-discrete system under the periodic boundary
conditions
\begin{eqnarray}
u_{l+L}^{(i)}=u_l^{(i)},\qquad
x_{l+L}=x_l+P_x,
\qquad i=1,\ldots,n,
\label{eq:periodicBC}
\end{eqnarray}
where $L$ is the number of mesh intervals and
 \begin{eqnarray}
P_x=x_L-x_0=\sum_{l=0}^{L-1}\delta_l
 \end{eqnarray}
is the period in the physical variable $x$.
Since
\begin{eqnarray}
\frac{dx_l}{dT}=-F_l,
\qquad
F_l=\sum_{1\le j<k\le n}c_{jk}u_l^{(j)}u_l^{(k)},
\end{eqnarray}
the periodicity of the field variables gives $F_L=F_0$, and hence
\begin{eqnarray}
\frac{dP_x}{dT}=\frac{d}{dT}(x_L-x_0)=-(F_L-F_0)=0.
\end{eqnarray}
Thus, $P_x$ is preserved.

Next, summing the first equation of (\ref{selfMCmSP}) over one
period gives
\begin{eqnarray}
\frac{d}{dT}\sum_{l=0}^{L-1}\left(u_{l+1}^{(i)}-u_l^{(i)}\right)=\sum_{l=0}^{L-1}(\delta_l-a)\left(
u_{l+1}^{(i)}+u_l^{(i)}\right).
\end{eqnarray}
By the periodic boundary condition
$u_L^{(i)}=u_0^{(i)}$, the left-hand side vanishes. Therefore,
every periodic solution satisfies the algebraic constraint
\begin{eqnarray}
C^{(i)}:=\sum_{l=0}^{L-1}(\delta_l-a)\left(u_{l+1}^{(i)}+u_l^{(i)}\right)=0,\qquad i=1,\ldots,n.
\label{eq:constraint}
\end{eqnarray}
Hence, the semi-discrete system under periodic boundary conditions is a differential-algebraic system.

The constraint (\ref{eq:constraint}) is the discrete counterpart of
the corresponding integral condition for the continuous periodic problem. Indeed, integrating the continuous MCmSP equation over one period in $x$ and applying the periodic boundary conditions yields
\begin{eqnarray}
\int_0^{P_x}(1-S)u^{(i)}\,dx=0,
\label{eq:con_constraint}
\end{eqnarray}
where
\begin{eqnarray}
S=\sum_{1\le j<k\le n}c_{jk}u_x^{(j)}u_x^{(k)}.
\end{eqnarray}
Using 
\begin{eqnarray}
dx=\rho\,dX,\qquad\frac{1}{\rho}=1+S,
\end{eqnarray}
(\ref{eq:con_constraint}) becomes
\begin{eqnarray}
\int (2\rho-1)u^{(i)}\,dX=0,
\end{eqnarray}
where the integral is taken over one period in $X$.
Since $\delta_l-a=a(2\rho_l-1)$, the left-hand side of (\ref{eq:constraint}) is a
discrete approximation of this integral.

For numerical simulation, we rewrite the periodic semi-discrete system
as an equivalent system of ordinary differential equations.
From
\begin{eqnarray}
\frac{d}{dT}\left(u_{l+1}^{(i)}-u_l^{(i)}\right)=(\delta_l-a)\left(u_{l+1}^{(i)}+u_l^{(i)}\right),
\end{eqnarray}
we obtain
\begin{eqnarray}
\frac{du_l^{(i)}}{dT}=\frac{du_0^{(i)}}{dT}+\sum_{k=0}^{l-1}(\delta_k-a)\left(u_{k+1}^{(i)}+u_k^{(i)}\right),\qquad l=1,\ldots,L-1.
\label{dul_dt}
\end{eqnarray}
To determine $du_0^{(i)}/dT$, we differentiate the algebraic constraint (\ref{eq:constraint}) with respect to $T$. Using
\begin{eqnarray}
\frac{d\delta_l}{dT}=-(F_{l+1}-F_l),
\end{eqnarray}
we have
\begin{eqnarray}
0=\frac{dC^{(i)}}{dT}=\sum_{l=0}^{L-1}\left[-(F_{l+1}-F_l)\left(u_{l+1}^{(i)}+u_l^{(i)}\right)+(\delta_l-a)\left(\frac{du_{l+1}^{(i)}}{dT}+\frac{du_l^{(i)}}{dT}\right)\right].
\end{eqnarray}
Substituting (\ref{dul_dt}) into this relation gives
\begin{align}
\frac{du_0^{(i)}}{dT}
&= \frac{1}{2\sum_{l=0}^{L-1}(\delta_l-a)}
\Bigg[
\sum_{l=0}^{L-1}(F_{l+1}-F_l)
\left(u_{l+1}^{(i)}+u_l^{(i)}\right)
\nonumber\\
&\quad
-\sum_{l=1}^{L-1}
(\delta_{l-1}+\delta_l-2a)
\sum_{k=0}^{l-1}(\delta_k-a)
\left(u_{k+1}^{(i)}+u_k^{(i)}\right)
\Bigg]
\label{du0_dt}
\end{align}
provided that $\sum_{l=0}^{L-1}(\delta_l-a)\neq0.$
Substituting (\ref{du0_dt}) into (\ref{dul_dt}) and combining the resulting equation with
\begin{eqnarray}
\frac{dx_l}{dT}=-F_l,
\label{dxl_dt}
\end{eqnarray}
we obtain a system of ordinary differential equations for $u_l^{(i)}$ and $x_l$.

In numerical time integration, however, the constraint (\ref{eq:constraint}) is not generally preserved exactly because of the time-discretization error. Therefore,
after each time step, we project the numerical solution onto
$C^{(i)}=0$ by
\begin{eqnarray}
u_l^{(i)}\longmapsto u_l^{(i)}-\frac{C^{(i)}}{2\displaystyle\sum_{l'=0}^{L-1}(\delta_{l'}-a)},
\qquad l=0,\ldots,L-1.
\label{projection}
\end{eqnarray}
This projection, which is a standard technique for differential-algebraic equations~\cite{HairerWanner, HairerLubichWanner}, ensures that the algebraic constraint (\ref{eq:constraint}) is satisfied.
Since $\sum_{l=0}^{L-1}(\delta_l-a)=P_x-aL$ and $P_x$ is conserved, the denominators in
(\ref{du0_dt}) and (\ref{projection}) are nonzero if $P_x\neq aL$ at $T=0$.

\medskip
\noindent\textit{Application to the 2-mSP and CmSP equations.}

We now apply periodic boundary conditions to the 2-mSP and CmSP equations.
The 2-mSP equation
\begin{eqnarray}
\left\{
\begin{array}{ll}
u_{xt}=u+\displaystyle\frac{1}{2}v(u^{2})_{xx}, \\
\\
v_{xt}=v+\displaystyle\frac{1}{2}u(v^{2})_{xx},
\end{array}
\right.
\label{2mSP}
\end{eqnarray}
is a special case of the MCmSP equation (\ref{MCmSP}) with $n=2, u^{(1)}=u, u^{(2)}=v$ and $c_{12}=1$.

Based on the results in the previous section, we have a self-adaptive moving mesh scheme for the 2-mSP equation
\begin{eqnarray}
\left\{
\begin{array}{ll}
\displaystyle\frac{d }{d T}(u_{l+1}-u_{l})=(\delta_{l}-a)(u_{l+1}+u_{l}), \\
\\
\displaystyle\frac{d}{d T}(v_{l+1}-v_{l})=(\delta_{l}-a)(v_{l+1}+v_{l}),\\
\\
\displaystyle\frac{d x_{l}}{d T}=-u_{l}v_{l}.
\end{array}
\right.
\label{semidis2mSP}
\end{eqnarray}
Here, $u_{l}$ and $v_{l}$ are values at $x_{l}$.

Equation (\ref{semidis2mSP}) admits the following $N$-soliton solution
\begin{eqnarray}
u_{l}=\frac{g_{l}^{(1)}}{f_{l}},\qquad v_{l}=\frac{g_{l}^{(2)}}{f_{l}}, \qquad x_{l}=2al-(\log{f_{l}})_{T},\qquad t=T,
\end{eqnarray}
with
\begin{eqnarray}
f_{l}&=&{\rm Pf}(a_{1},\cdots,a_{2N},b_{1},\cdots,b_{2N})_{l}, \nonumber\\
g^{(i)}_{l}&=&{\rm Pf}(d_{0},B_{i},a_{1},\cdots,a_{2N},b_{1},\cdots,b_{2N})_{l},
\end{eqnarray}
where $i=1,2$ and the elements of the Pfaffians are defined as (\ref{dmpf1})--(\ref{dmpf2}).


Here we show examples of numerical simulations of the 2-mSP equation  (\ref{2mSP}) with a self-adaptive moving mesh scheme (\ref{semidis2mSP}) under periodic boundary conditions. The initial conditions are given by
\begin{align}
&u_e=\frac{g^{(1)}}{f},\quad v_e=\frac{g^{(2)}}{f},\quad x_e=X-(\log{f})_{T},\quad t=T,\nonumber\\
&f=-1-a_{1}a_{2}b_{12}e^{\eta_{1}+\eta_{2}},\qquad g^{(1)}=-a_{1}e^{\eta_{1}}, \qquad g^{(2)}=-a_{2}e^{\eta_{2}},\nonumber\\
&\eta_{i}=p_{i}X+\frac{1}{p_{i}}T+\xi_{i}^{\prime},\quad b_{ij}=\left(\frac{p_{i}p_{j}}{p_{i}+p_{j}}\right)^{2},\quad i,j=1,2,\label{incon_mSP_1}
\end{align}
and
\begin{align}
&u_e=\frac{g^{(1)}}{f},\quad v_e=\frac{g^{(2)}}{f},\quad x_e=X-(\log{f})_{T},\quad t=T, \nonumber\\
&f=1+a_{1}a_{3}b_{13}e^{\eta_{1}+\eta_{3}}
  +a_{2}a_{3}b_{23}e^{\eta_{2}+\eta_{3}}
  +a_{1}a_{4}b_{14}e^{\eta_{1}+\eta_{4}}\nonumber\\
&\quad+a_{2}a_{4}b_{24}e^{\eta_{2}+\eta_{4}}
  +a_{1}a_{2}a_{3}a_{4}(p_{1}-p_{2})^{2}(p_{3}-p_{4})^{2}
   \frac{b_{13}b_{23}b_{14}b_{24}}{p_{1}^{2}p_{2}^{2}p_{3}^{2}p_{4}^{2}}
   e^{\eta_{1}+\eta_{2}+\eta_{3}+\eta_{4}},\nonumber\\
&g^{(1)}=a_{1}e^{\eta_{1}}+a_{2}e^{\eta_{2}}
 +\frac{a_{1}a_{2}a_{3}(p_{1}-p_{2})^{2}p_{3}^{4}}
        {(p_{1}+p_{3})^{2}(p_{2}+p_{3})^{2}}
  e^{\eta_{1}+\eta_{2}+\eta_{3}}\nonumber\\
&\quad+\frac{a_{1}a_{2}a_{4}(p_{1}-p_{2})^{2}p_{4}^{4}}
        {(p_{1}+p_{4})^{2}(p_{2}+p_{4})^{2}}
  e^{\eta_{1}+\eta_{2}+\eta_{4}},\nonumber\\
&g^{(2)}=a_{3}e^{\eta_{3}}+a_{4}e^{\eta_{4}}
 +\frac{a_{2}a_{3}a_{4}(p_{3}-p_{4})^{2}p_{2}^{4}}
        {(p_{2}+p_{3})^{2}(p_{2}+p_{4})^{2}}
  e^{\eta_{2}+\eta_{3}+\eta_{4}}\nonumber\\
&\quad+\frac{a_{1}a_{3}a_{4}(p_{3}-p_{4})^{2}p_{1}^{4}}
        {(p_{1}+p_{3})^{2}(p_{1}+p_{4})^{2}}
  e^{\eta_{1}+\eta_{3}+\eta_{4}},\nonumber\\
&\eta_{i}=p_{i}X+\frac{1}{p_{i}}T+\xi_{i}^{\prime},\quad b_{ij}=\left(\frac{p_{i}p_{j}}{p_{i}+p_{j}}\right)^{2},\quad i,j=1,2,3,4,\label{incon_mSP_2}
\end{align}
which are exact one- and two-soliton solutions of the 2-mSP equation obtained in {\it Example \ref{ex1}}; the subscript $e$ indicates the exact solution.

As a time marching method, we use the improved Euler method.
The number of mesh intervals is $L=8000$, the computational domain width is $D=80$, and the time step is $\Delta t=0.0001$. The uniform lattice spacing in the $X$ variable is chosen as $2a=D/L$. In Section~\ref{sec_disMCmSP} and Appendix~\ref{secA1}, the edge point is placed at $X=0$; since it can be any fixed point, all results hold with the edge point at $X_0$. In the computations, we take $X_0=-D/2$ and use the lattice $X_l=X_0+2al$, for $l=0,\cdots, L$. The computation is carried out up to $t=20$. The initial data obtained from the exact solution do not satisfy the algebraic constraint (\ref{eq:constraint}), so we also apply (\ref{projection}) to the initial data. After a soliton re-enters the
computational interval, the numerical solution is compared with the periodic extension
$\sum_{m=-1}^{1}u_e(X+mD, T)$ of the exact solution, which is exact up to exponentially
small terms.

\begin{figure}[h]
\centering
 \begin{tabular}{cc}
      \begin{minipage}[t]{0.47\hsize}
       \centering
        \includegraphics[width=\linewidth]{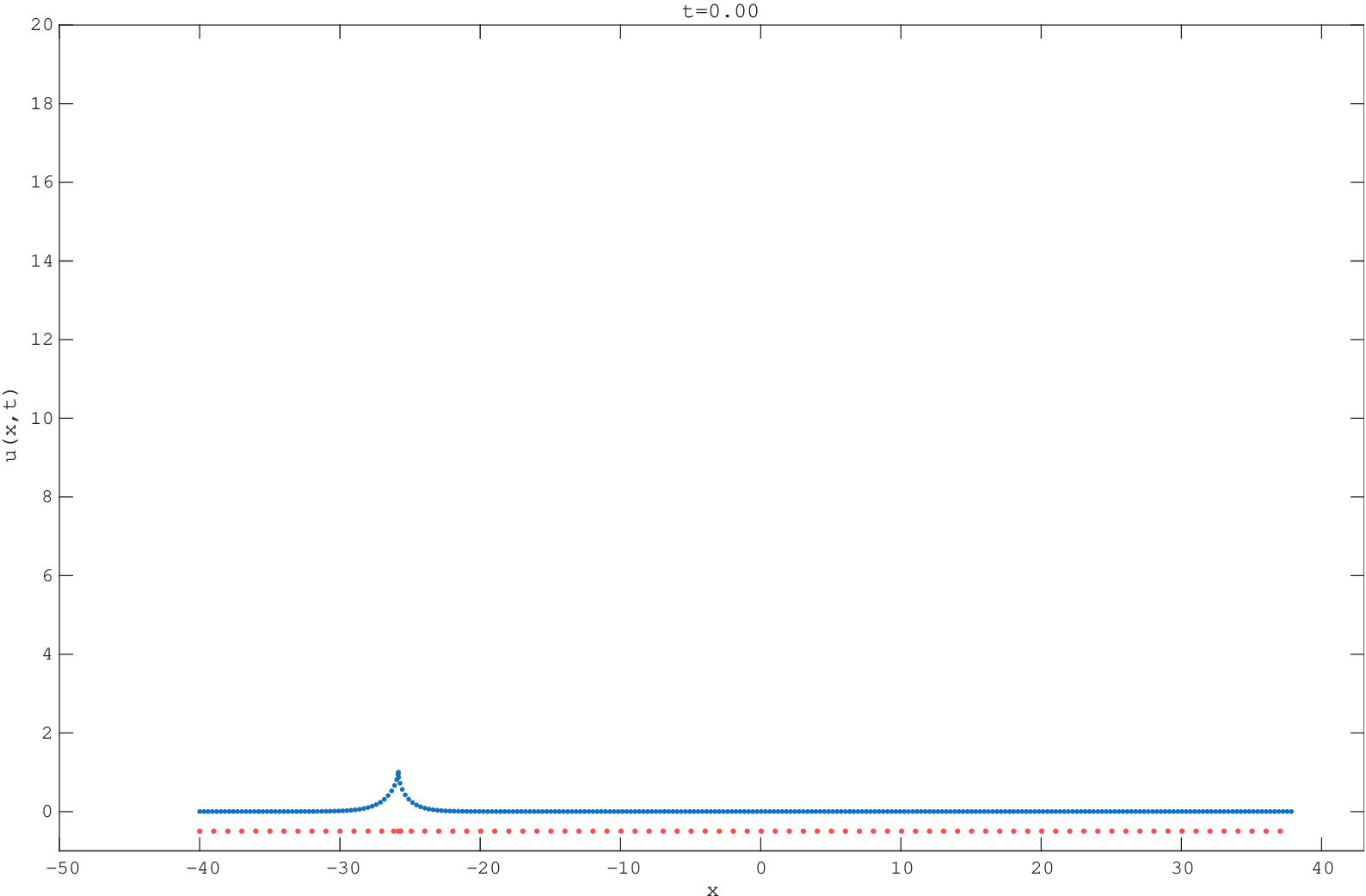}
      \end{minipage} &
      \begin{minipage}[t]{0.47\hsize}
        \centering
        \includegraphics[width=\linewidth]{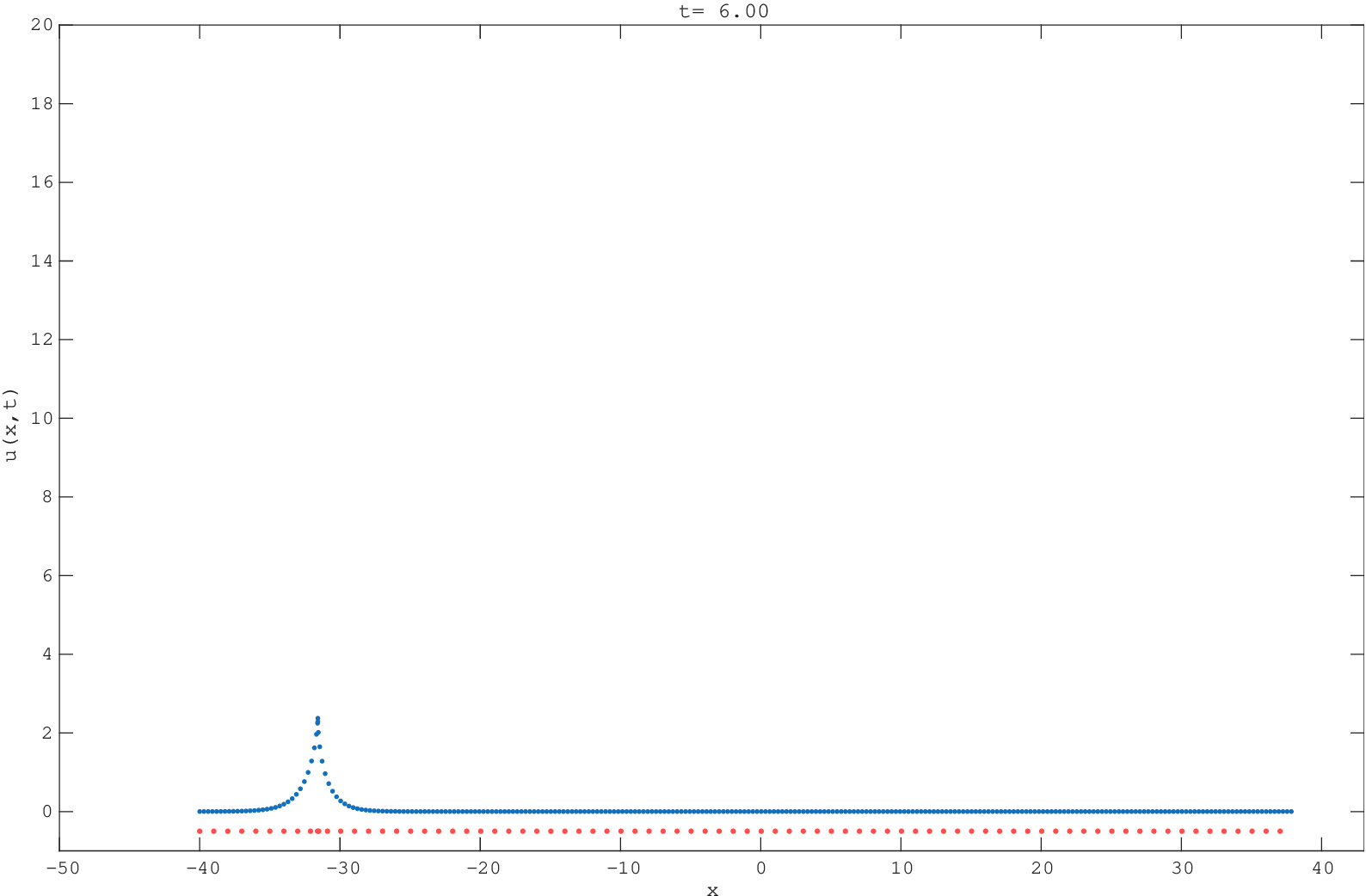}
      \end{minipage}\\

      \begin{minipage}[t]{0.47\hsize}
        \centering
        \includegraphics[width=\linewidth]{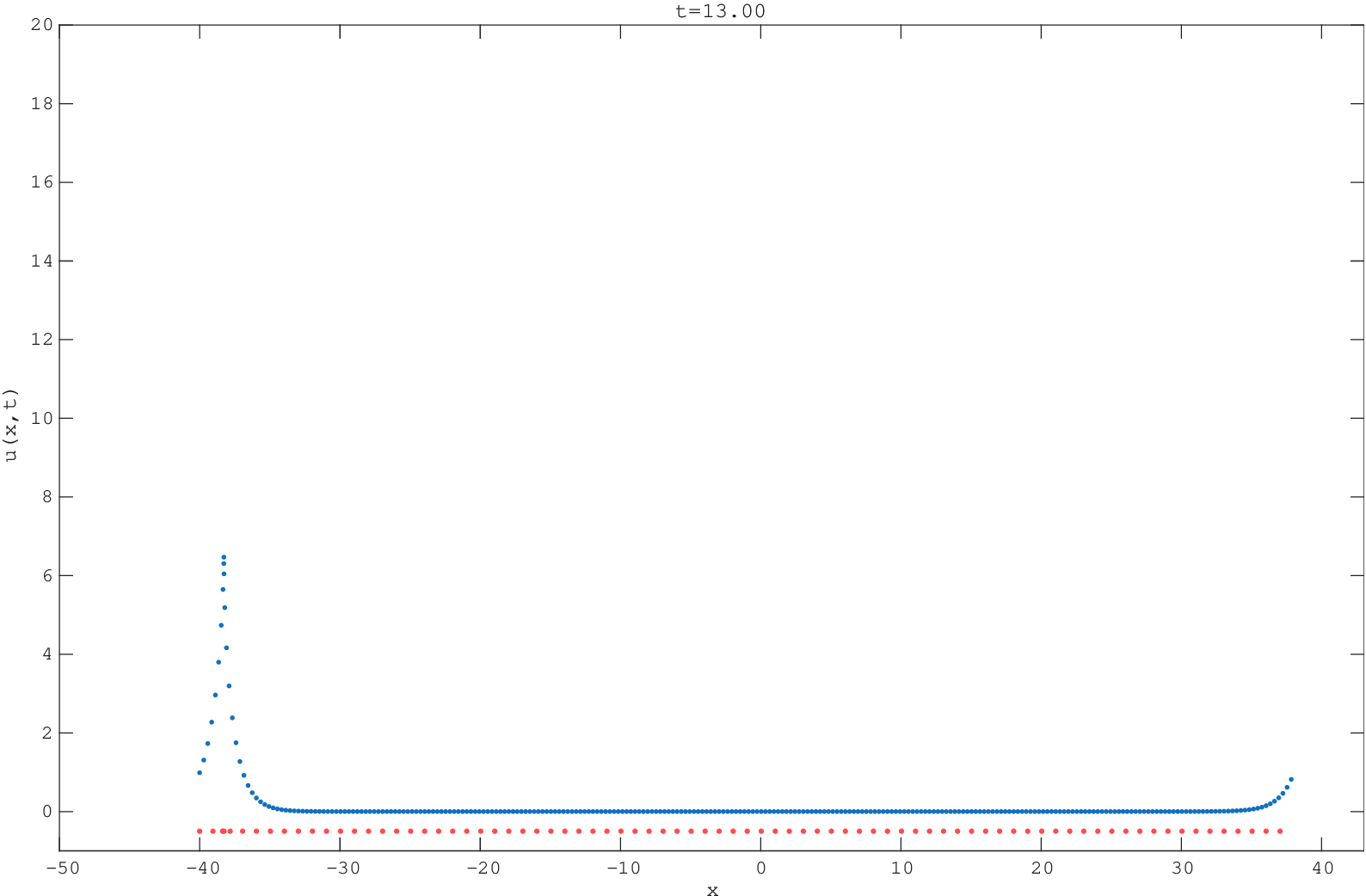}
      \end{minipage} &
      \begin{minipage}[t]{0.47\hsize}
        \centering
        \includegraphics[width=\linewidth]{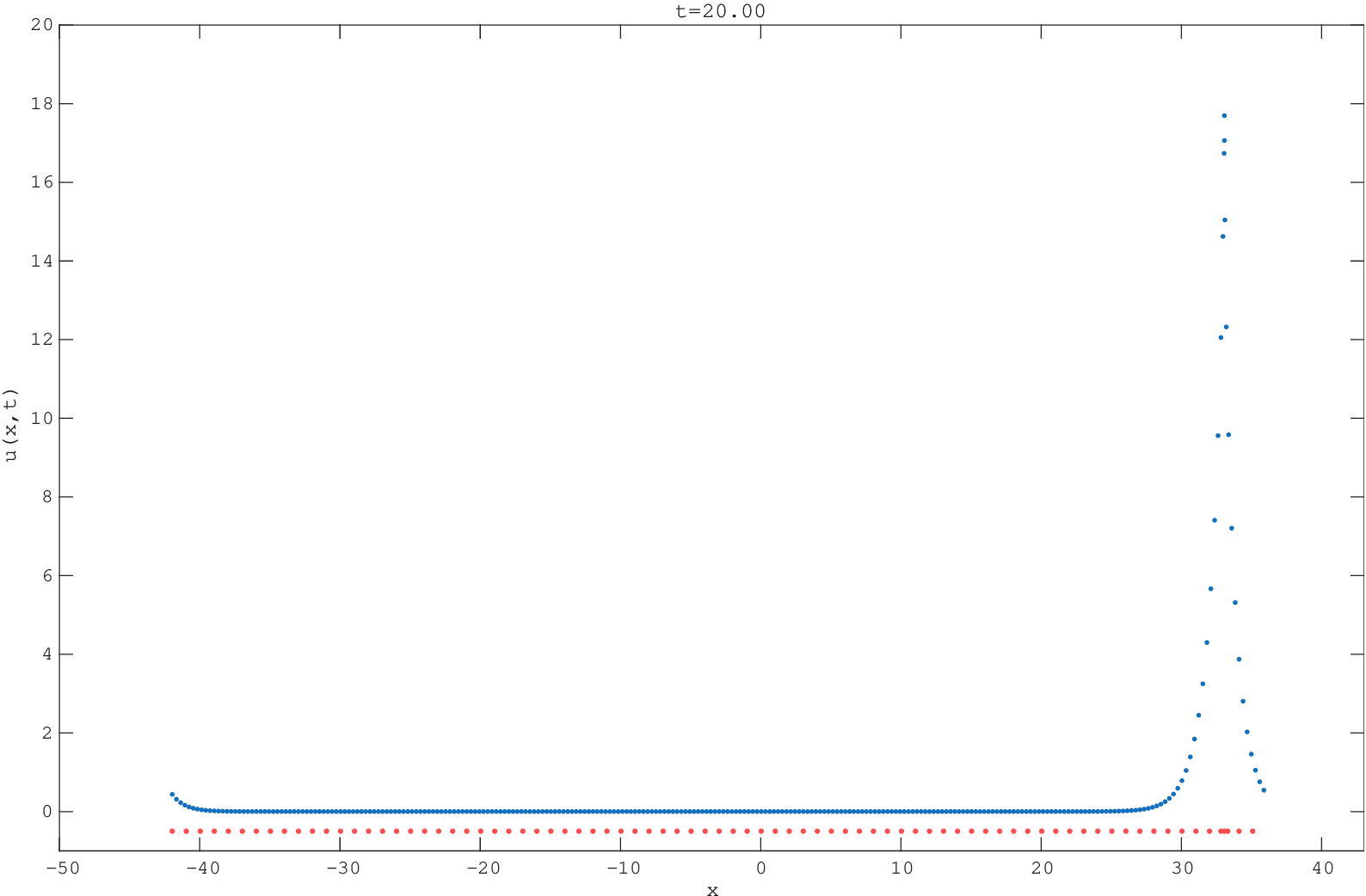}
      \end{minipage}
       \end{tabular}
     \caption{The numerical simulation of the {\it u}-profile of the one-soliton solution for the 2-mSP equation. maxerr=2.33$\times 10^{-4}$}
              \label{2mSP_u_1sol}
  \end{figure}
\begin{figure}[h]
\centering
 \begin{tabular}{cc}
      \begin{minipage}[t]{0.47\hsize}
       \centering
        \includegraphics[width=\linewidth]{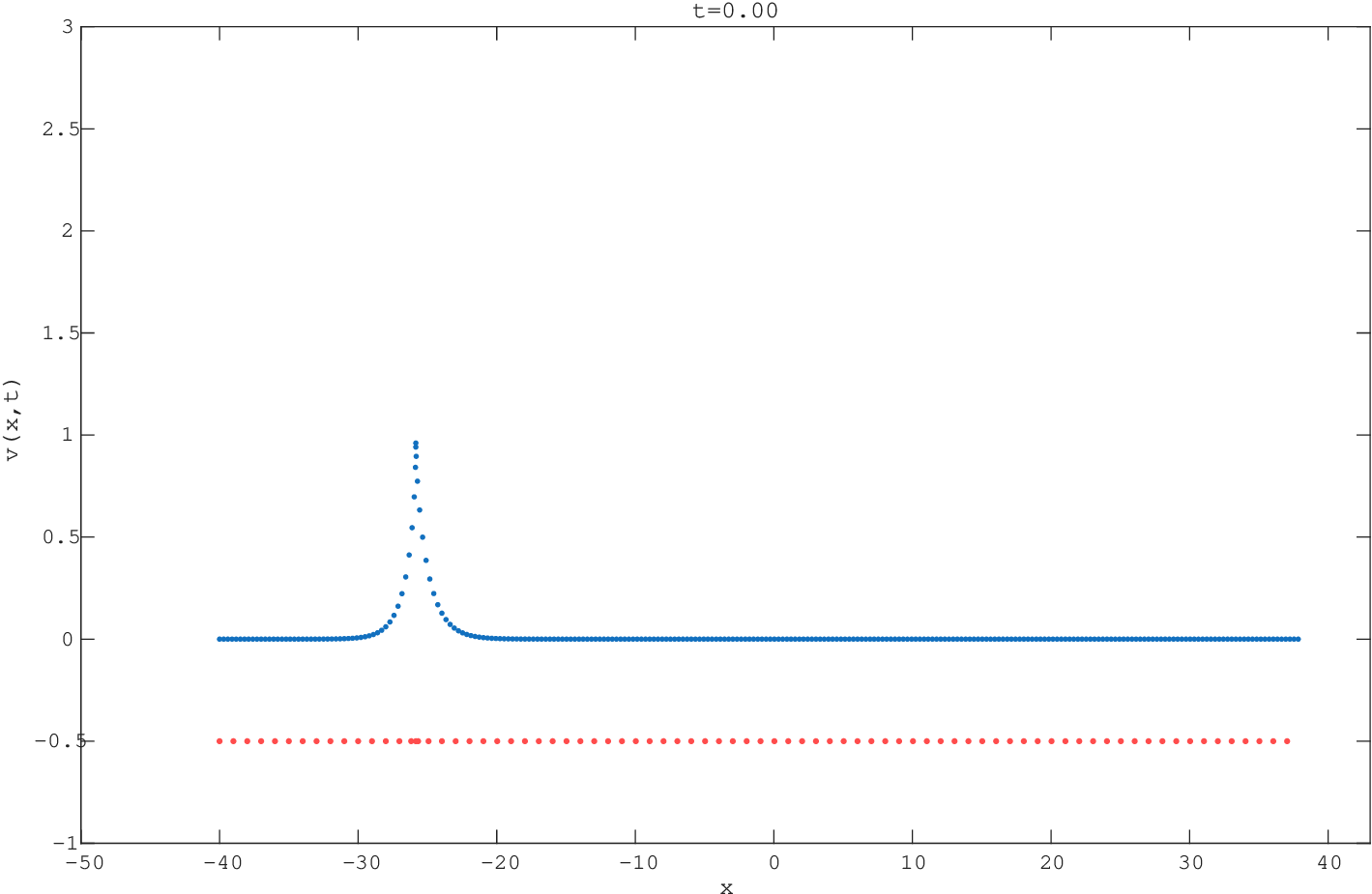}
      \end{minipage} &
      \begin{minipage}[t]{0.47\hsize}
        \centering
        \includegraphics[width=\linewidth]{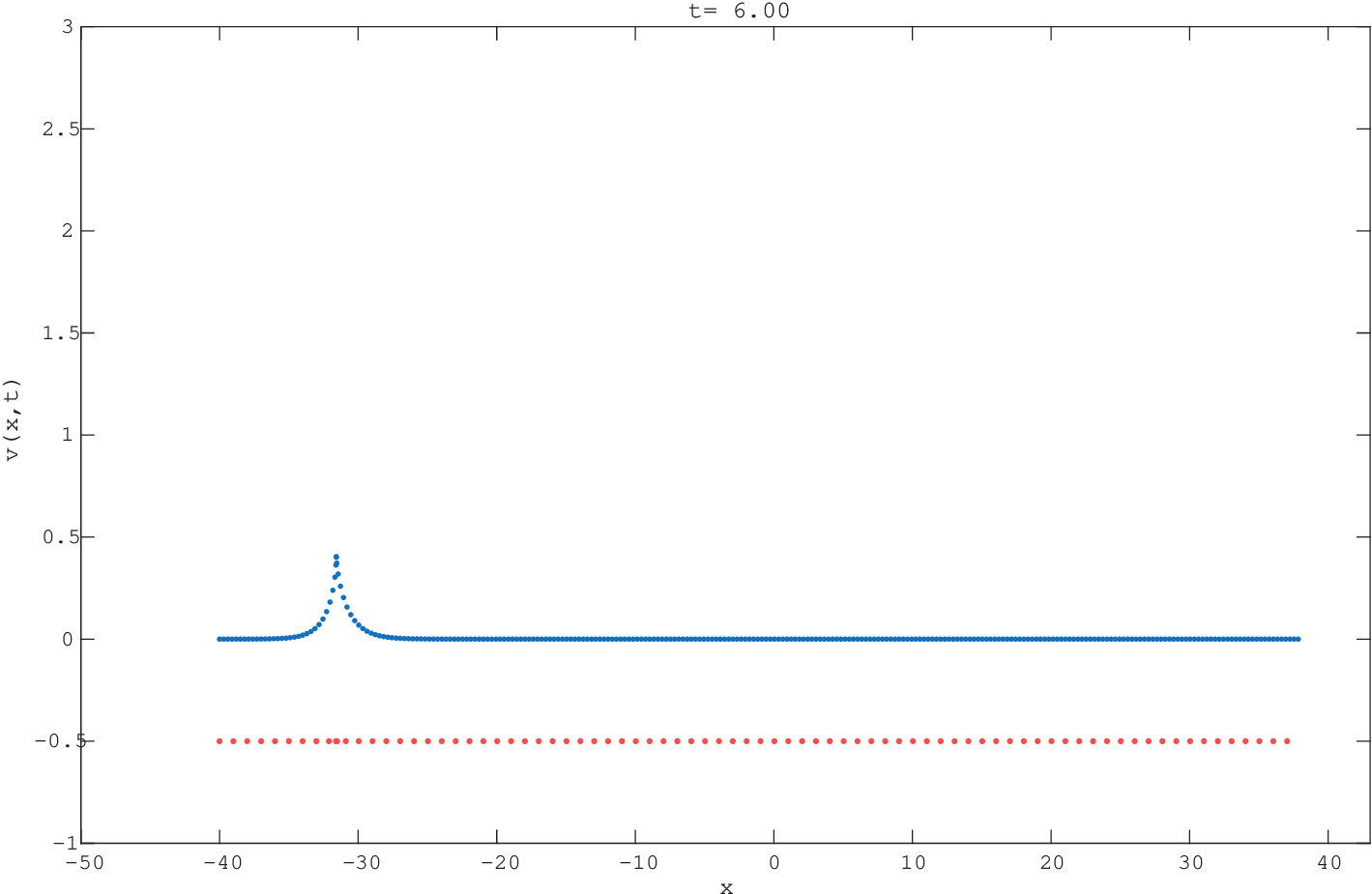}
      \end{minipage}\\

      \begin{minipage}[t]{0.47\hsize}
        \centering
        \includegraphics[width=\linewidth]{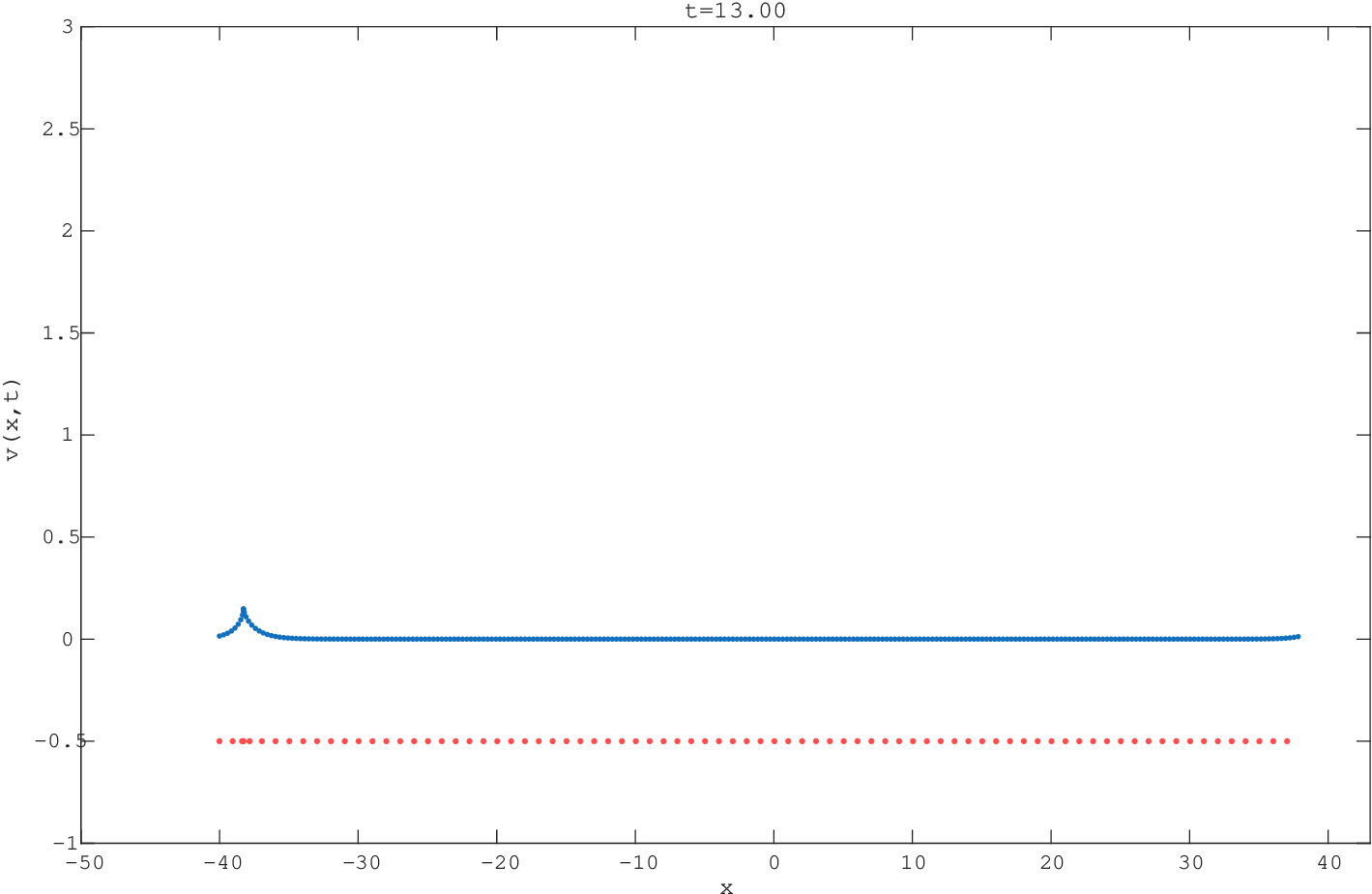}
      \end{minipage} &
      \begin{minipage}[t]{0.47\hsize}
        \centering
        \includegraphics[width=\linewidth]{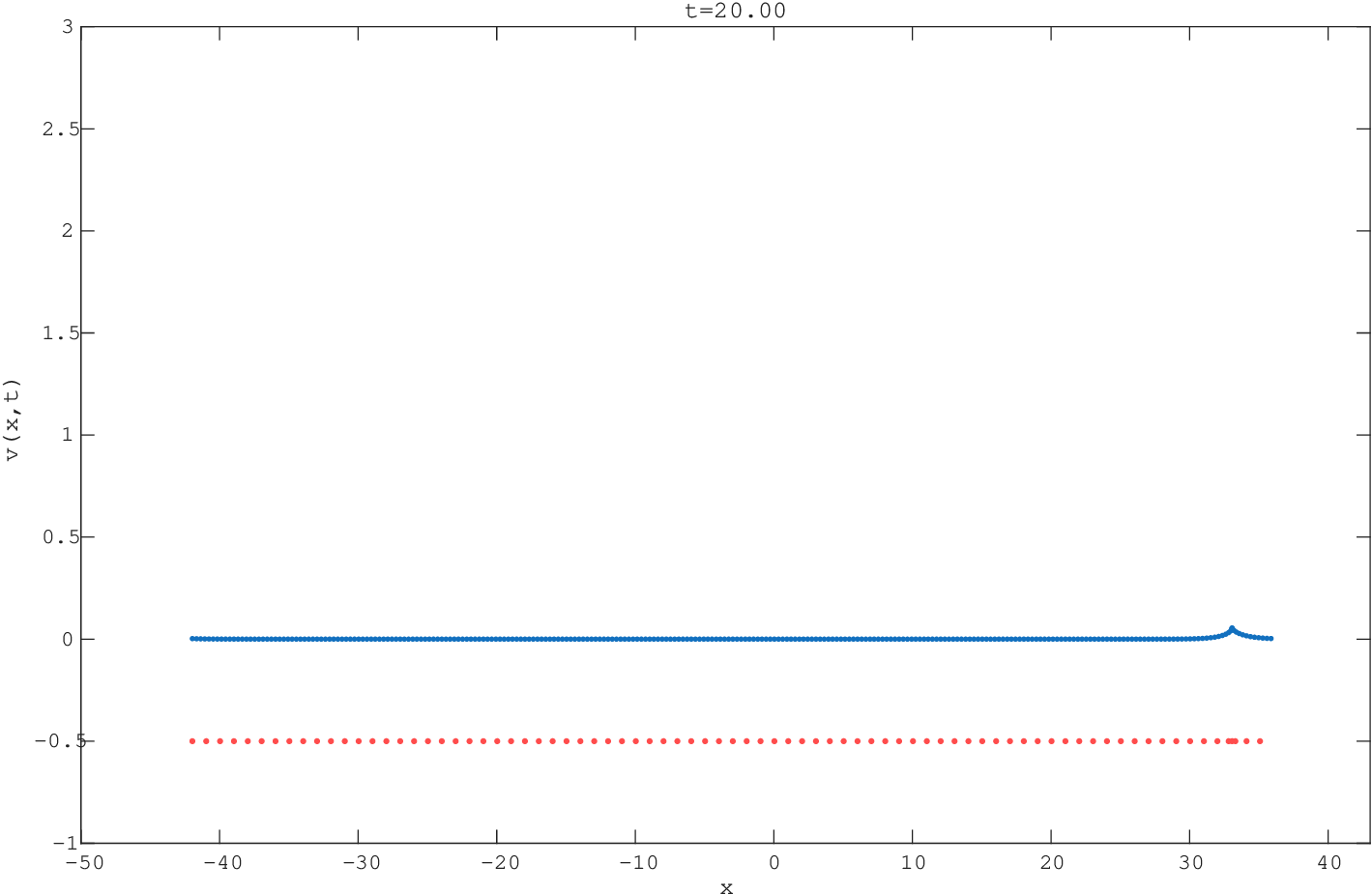}
      \end{minipage}
       \end{tabular}
     \caption{The numerical simulation of the {\it v}-profile of the one-soliton solution for the 2-mSP equation. maxerr=2.66$\times 10^{-4}$}
              \label{2mSP_v_1sol}
  \end{figure}
\begin{figure}[h]
\centering
 \begin{tabular}{cc}
      \begin{minipage}[t]{0.47\hsize}
       \centering
        \includegraphics[width=\linewidth]{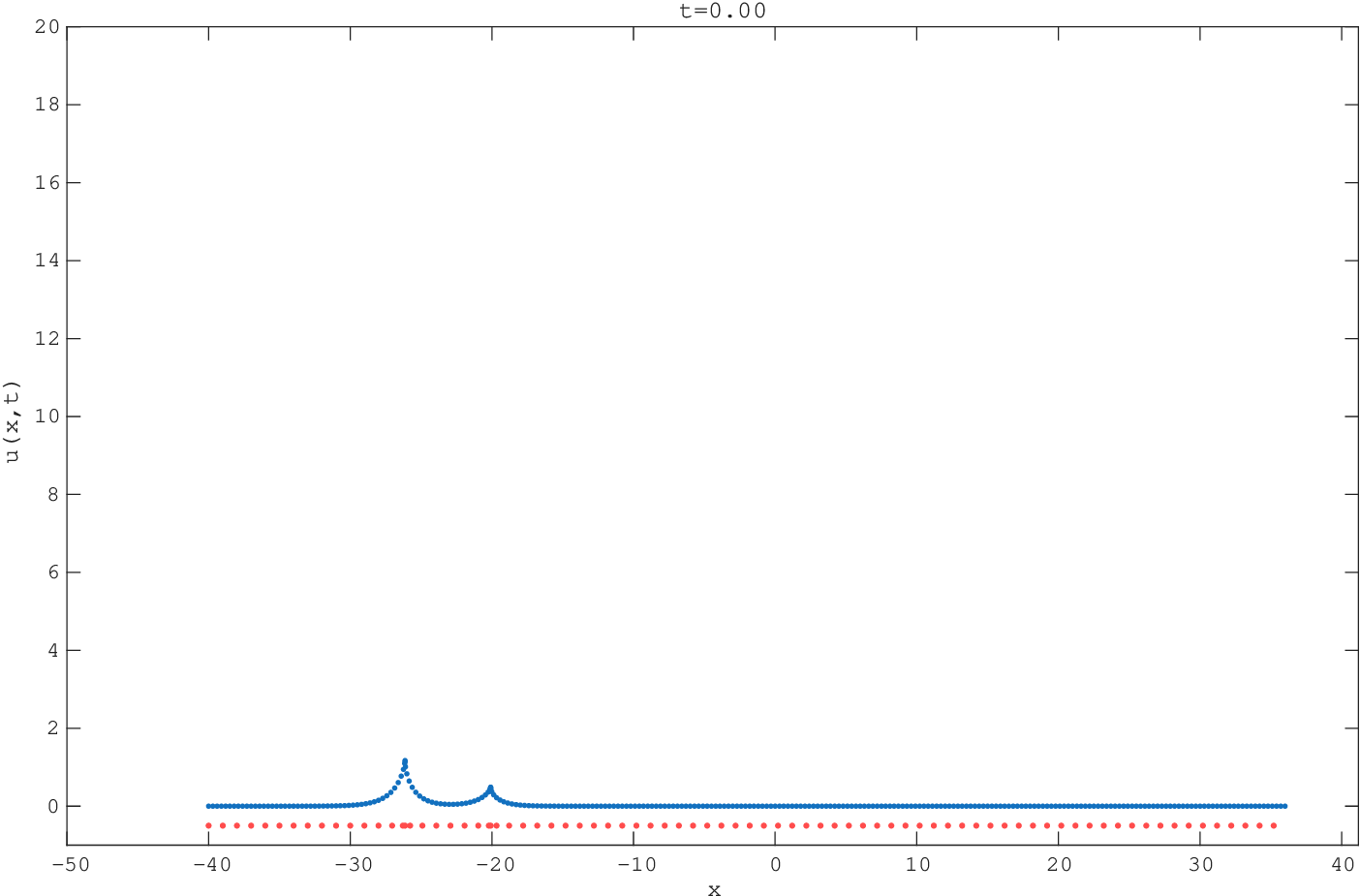}
      \end{minipage} &
      \begin{minipage}[t]{0.47\hsize}
        \centering
        \includegraphics[width=\linewidth]{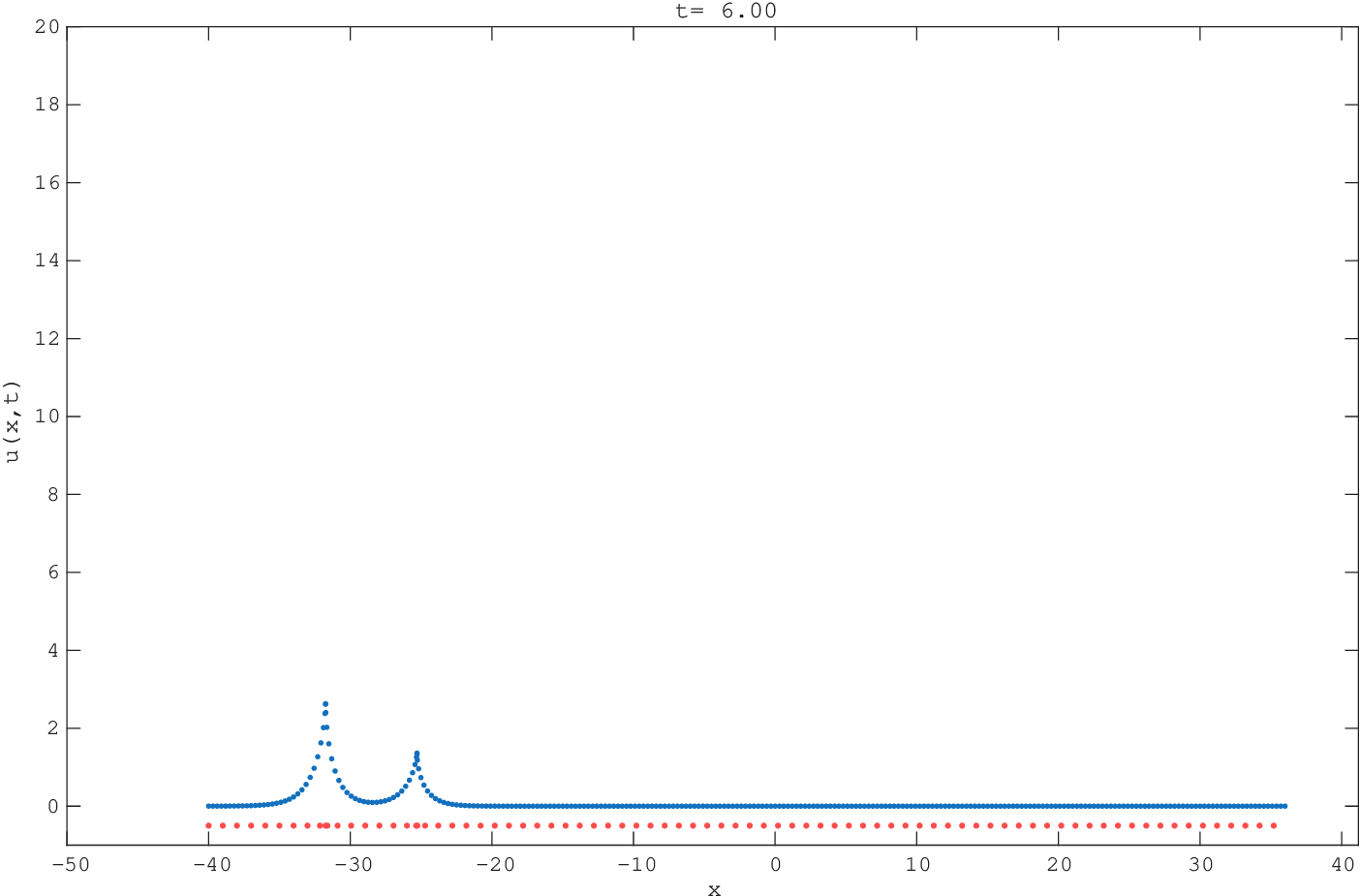}
      \end{minipage}\\

      \begin{minipage}[t]{0.47\hsize}
        \centering
        \includegraphics[width=\linewidth]{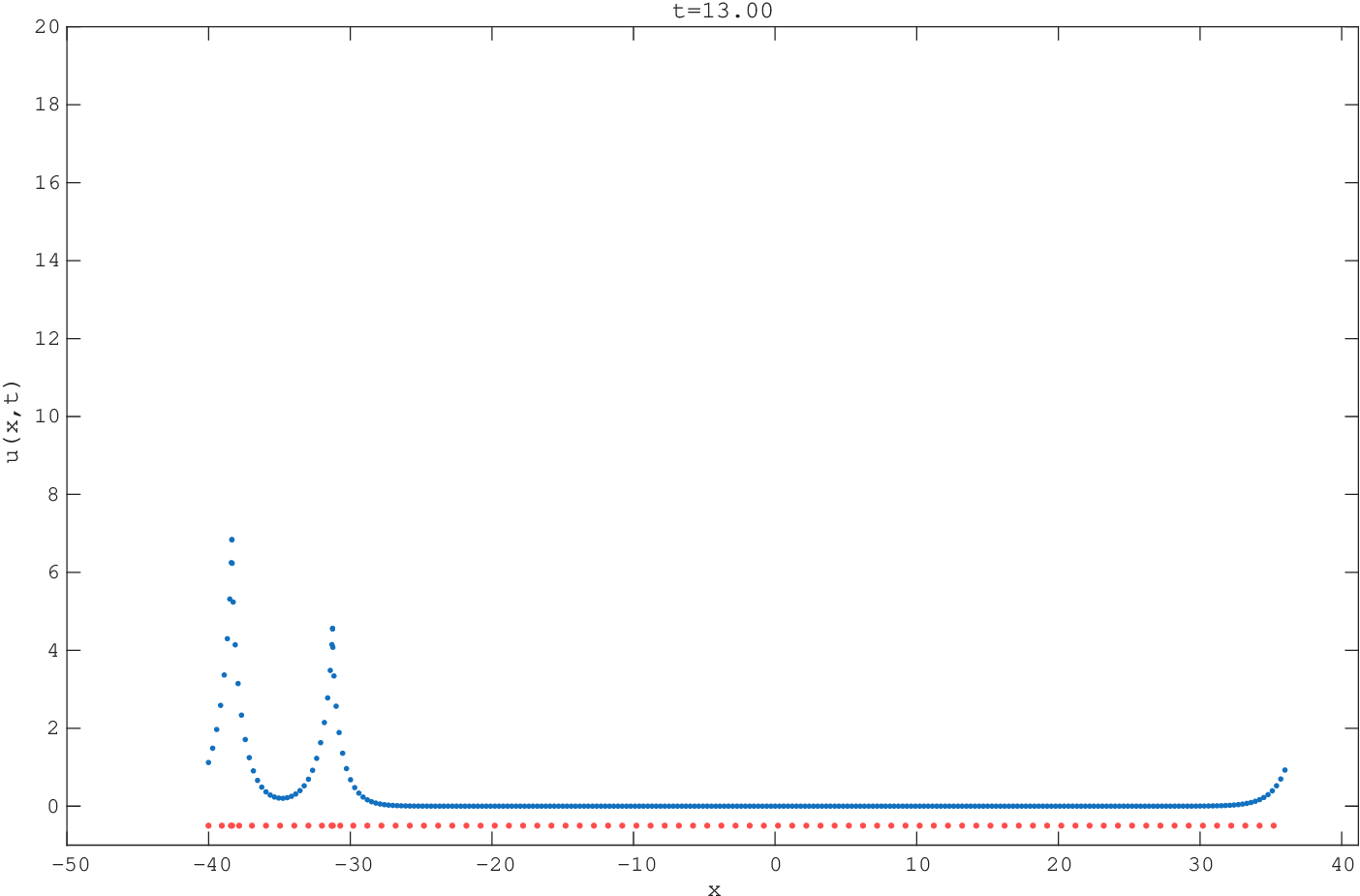}
      \end{minipage} &
      \begin{minipage}[t]{0.47\hsize}
        \centering
        \includegraphics[width=\linewidth]{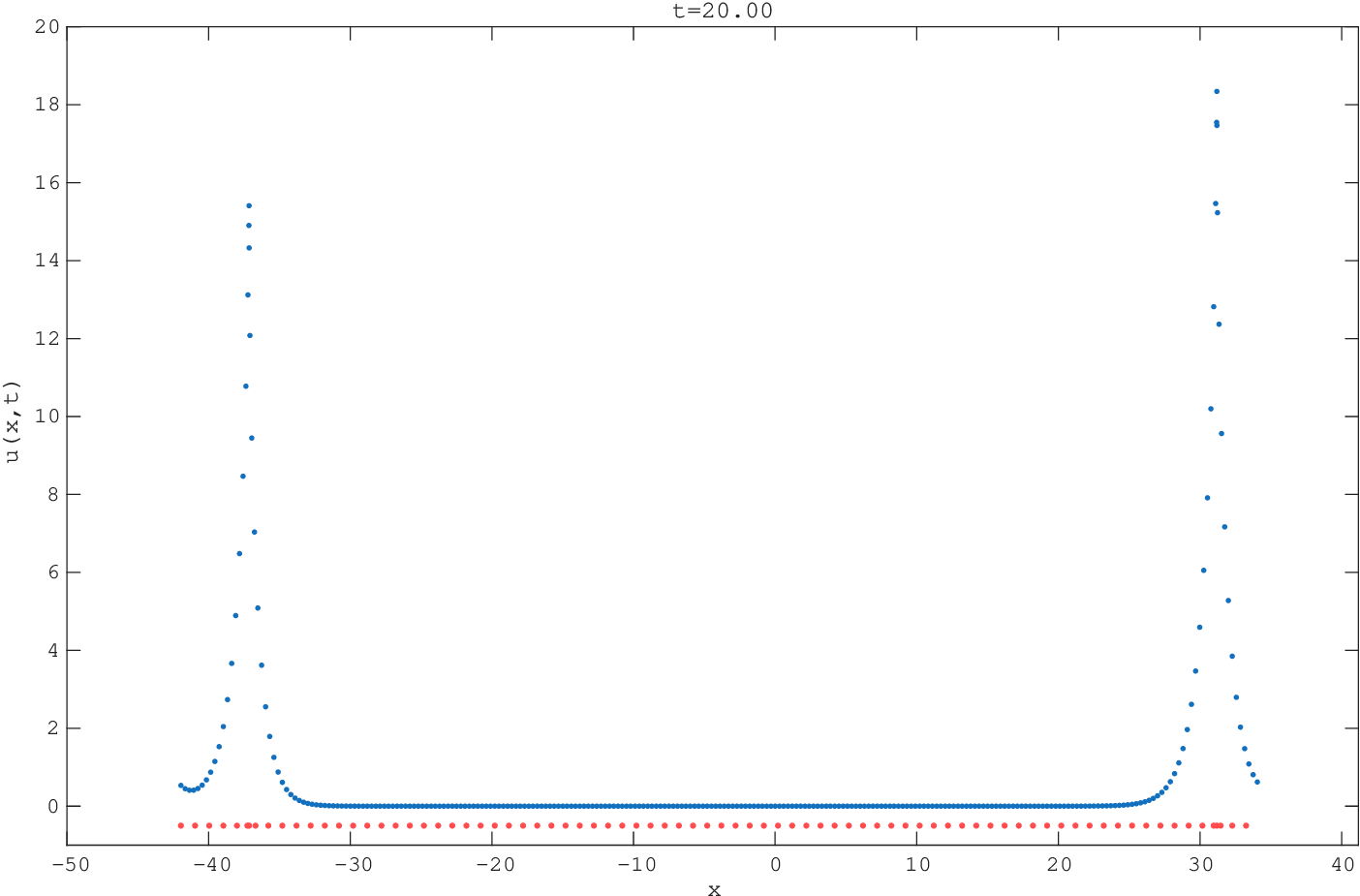}
      \end{minipage}
       \end{tabular}
     \caption{The numerical simulation of the {\it u}-profile of the two-soliton solution for the 2-mSP equation. maxerr=2.72$\times 10^{-4}$}
              \label{2mSP_u}
  \end{figure}
\begin{figure}[h]
\centering
 \begin{tabular}{cc}
      \begin{minipage}[t]{0.47\hsize}
       \centering
        \includegraphics[width=\linewidth]{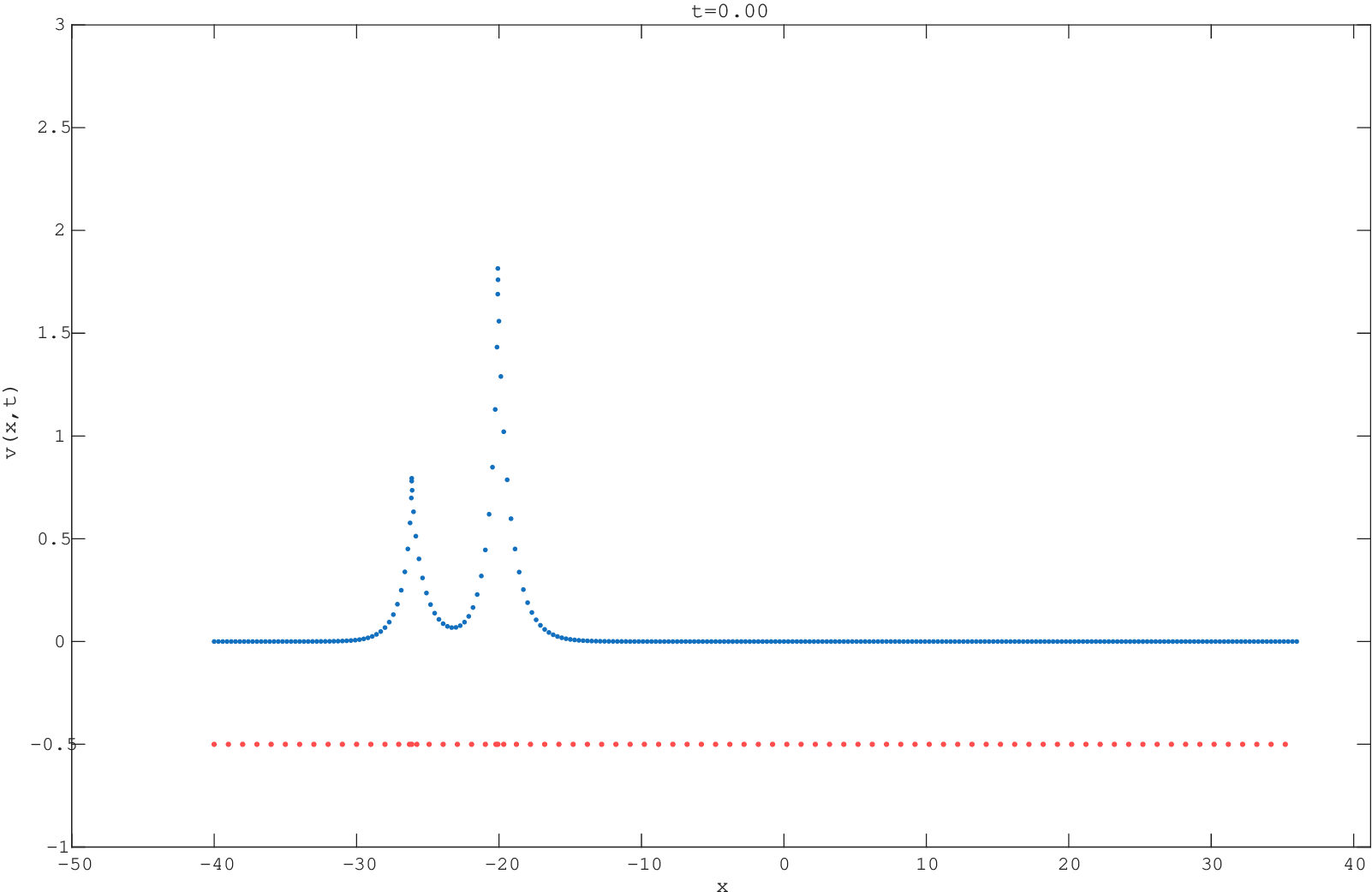}
      \end{minipage} &
      \begin{minipage}[t]{0.47\hsize}
        \centering
        \includegraphics[width=\linewidth]{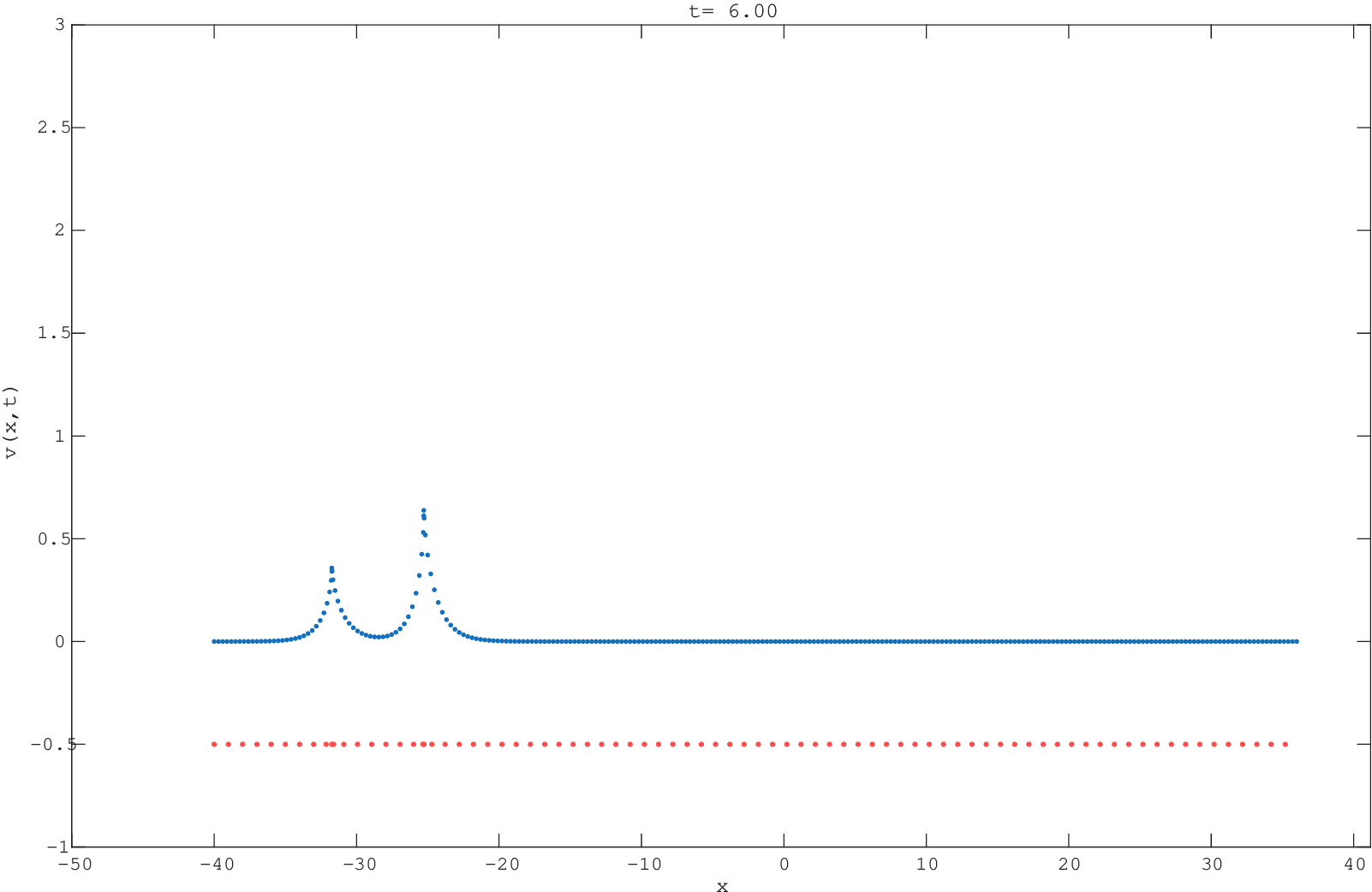}
      \end{minipage}\\

      \begin{minipage}[t]{0.47\hsize}
        \centering
        \includegraphics[width=\linewidth]{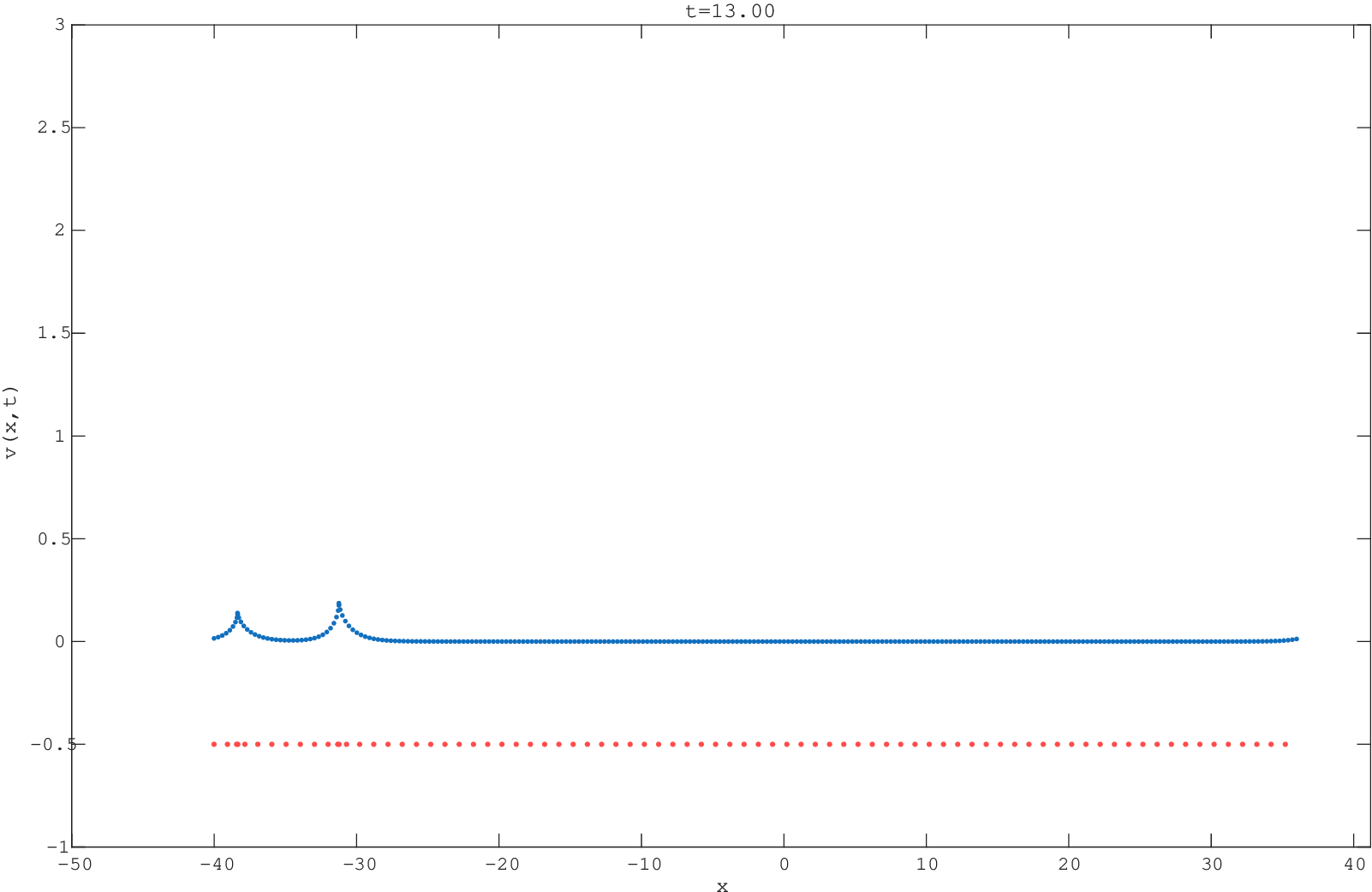}
      \end{minipage} &
      \begin{minipage}[t]{0.47\hsize}
        \centering
        \includegraphics[width=\linewidth]{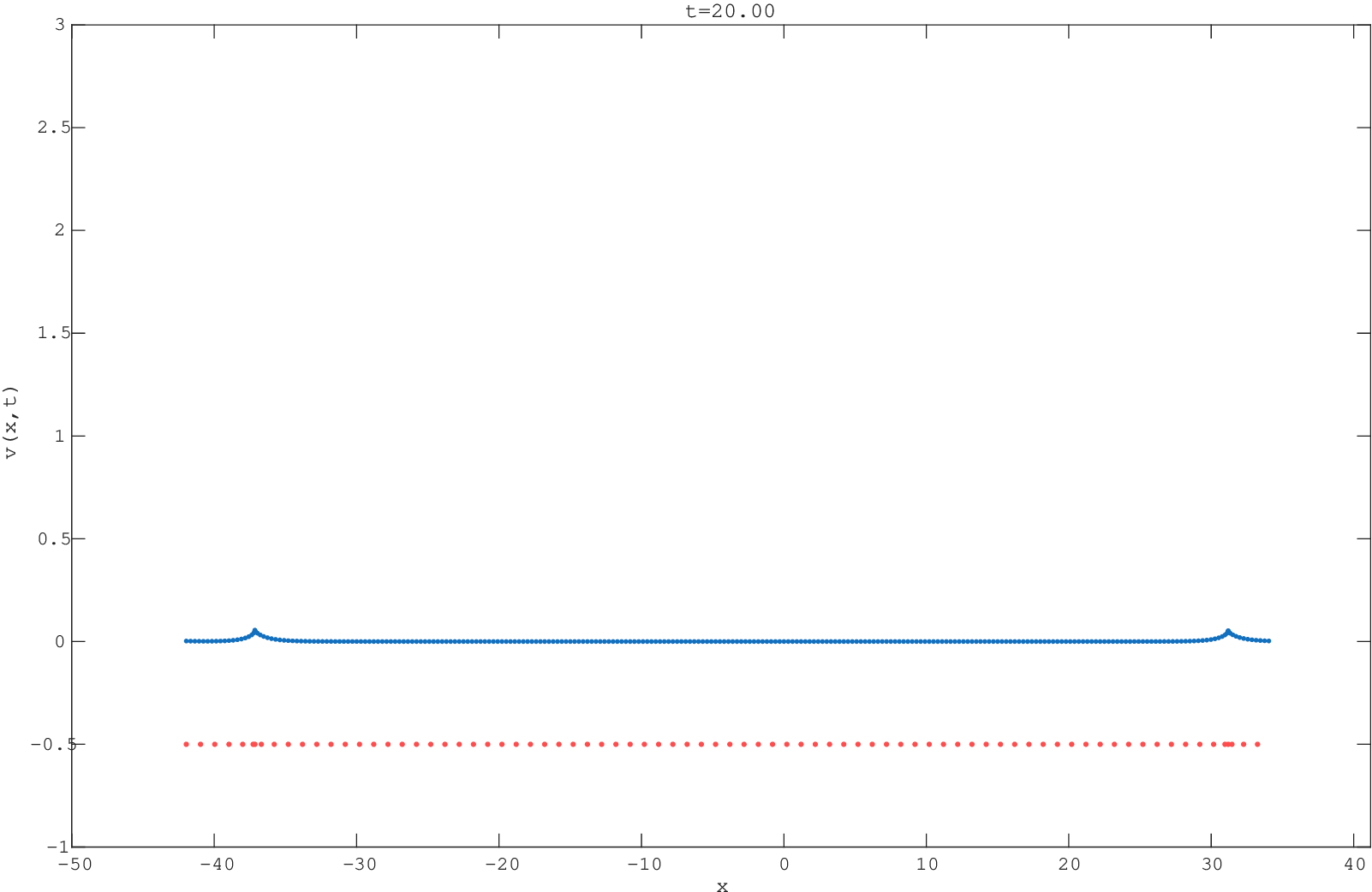}
      \end{minipage}
       \end{tabular}
     \caption{The numerical simulation of the {\it v}-profile of the two-soliton solution for the 2-mSP equation. maxerr=3.01$\times 10^{-4}$}
              \label{2mSP_v}
  \end{figure}

To examine the error in the region where the solution has large amplitude, we
evaluate the relative error only at mesh points where the exact solution is close to its
maximum amplitude at each computed time level $t_m (m=1,\cdots, M)$, where $M$ is the total number of time steps. More precisely, let
\begin{eqnarray}
I_m=\left\{\,l\in\{0,1,\cdots, L\}: \ |u_{e,l}^{\,m}|\ge0.9\max_{0\le r\le L}|u_{e,r}^{\,m}|\,\right\},\qquad m=1,\ldots,M,
\end{eqnarray}
where $u_{e,l}^{m}$ denotes the exact solution evaluated at the same lattice point $X_l$ and time $t_m$. Thus, $I_m$ is the set of mesh points at which the exact amplitude is at least 90\% of its maximum value at time $t_m$. We then define the relative error at each
time level by
\begin{eqnarray}
\mathrm{err}_m
=
\max_{l\in I_m}
\left|
\frac{u_{e,l}^{m}-u_{n,l}^{m}}
{u_{e,l}^{m}}
\right|,
\label{err}
\end{eqnarray}
where $u_{n,l}^{m}$ denotes the numerical value at the same lattice point. Finally, we define
\begin{eqnarray}
\mathrm{maxerr}
=
\max_{1\le m\le M}\mathrm{err}_m.
\label{maxerr}
\end{eqnarray}
The same definition is used for $v$ with $u$ replaced by $v$.

Figure \ref{2mSP_u_1sol} and Figure \ref{2mSP_v_1sol} show numerical simulations of the {\it u}-profile and the {\it v}-profile of the one-soliton solution for $p_{1}=0.95,p_{2}=1.1, a_{1}=0.5, a_{2}=20$ and $\xi_{1}^{\prime}=\xi_{2}^{\prime}=25$.
Figure \ref{2mSP_u} and Figure \ref{2mSP_v} show numerical simulations of the {\it u}-profile and the {\it v}-profile of the two-soliton solution for $p_{1}=0.95,p_{2}=1.0, p_{3}=1.1,p_{4}=1.2, a_{1}=0.5, a_{2}=1, a_{3}=20, a_{4}=40 $ and $ \xi_{1}^{\prime}=\xi_{2}^{\prime}=\xi_{3}^{\prime}=\xi_{4}^{\prime}=25$.

The blue dotted line represents the numerical solution and the red dotted line represents the mesh distribution. The solitons travel in a leftward direction with respect to the $x$-axis. When the solitons reach the left edge, they emerge from the right edge.

\begin{figure}[h]
\centering
 \begin{tabular}{cc}
      \begin{minipage}[t]{0.47\hsize}
       \centering
        \includegraphics[width=\linewidth]{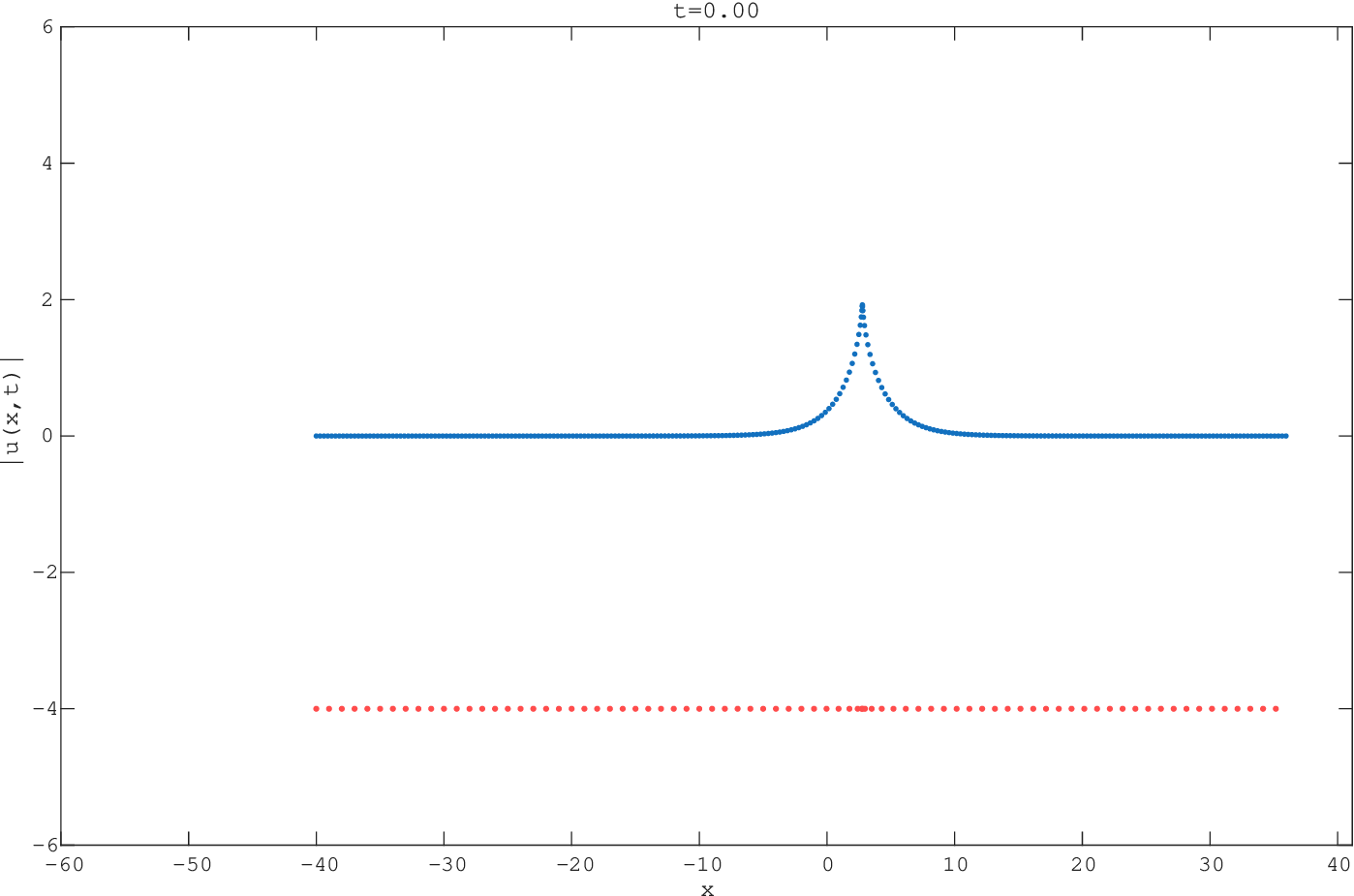}
      \end{minipage} &
      \begin{minipage}[t]{0.47\hsize}
        \centering
        \includegraphics[width=\linewidth]{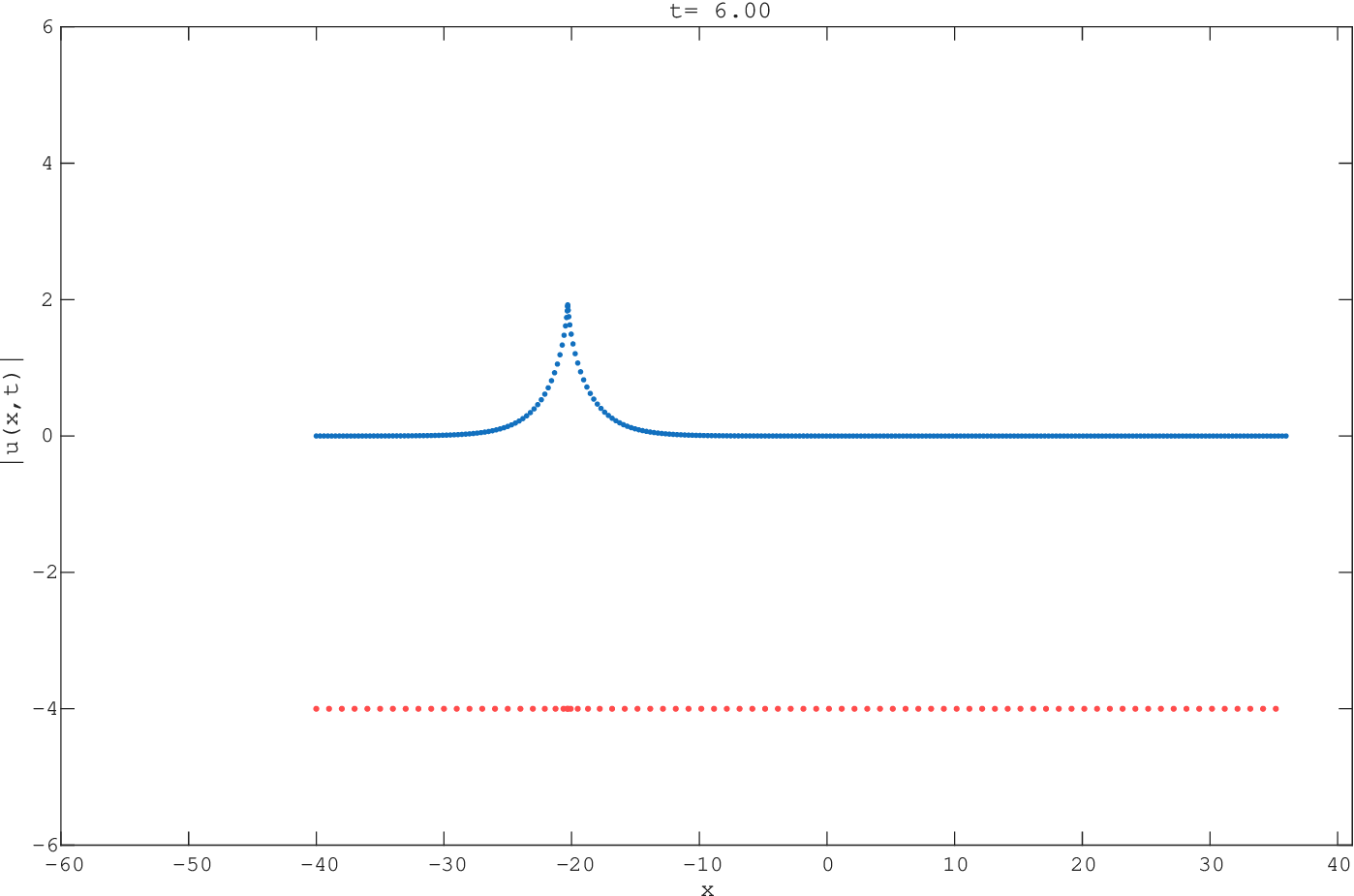}
      \end{minipage}\\

      \begin{minipage}[t]{0.47\hsize}
        \centering
        \includegraphics[width=\linewidth]{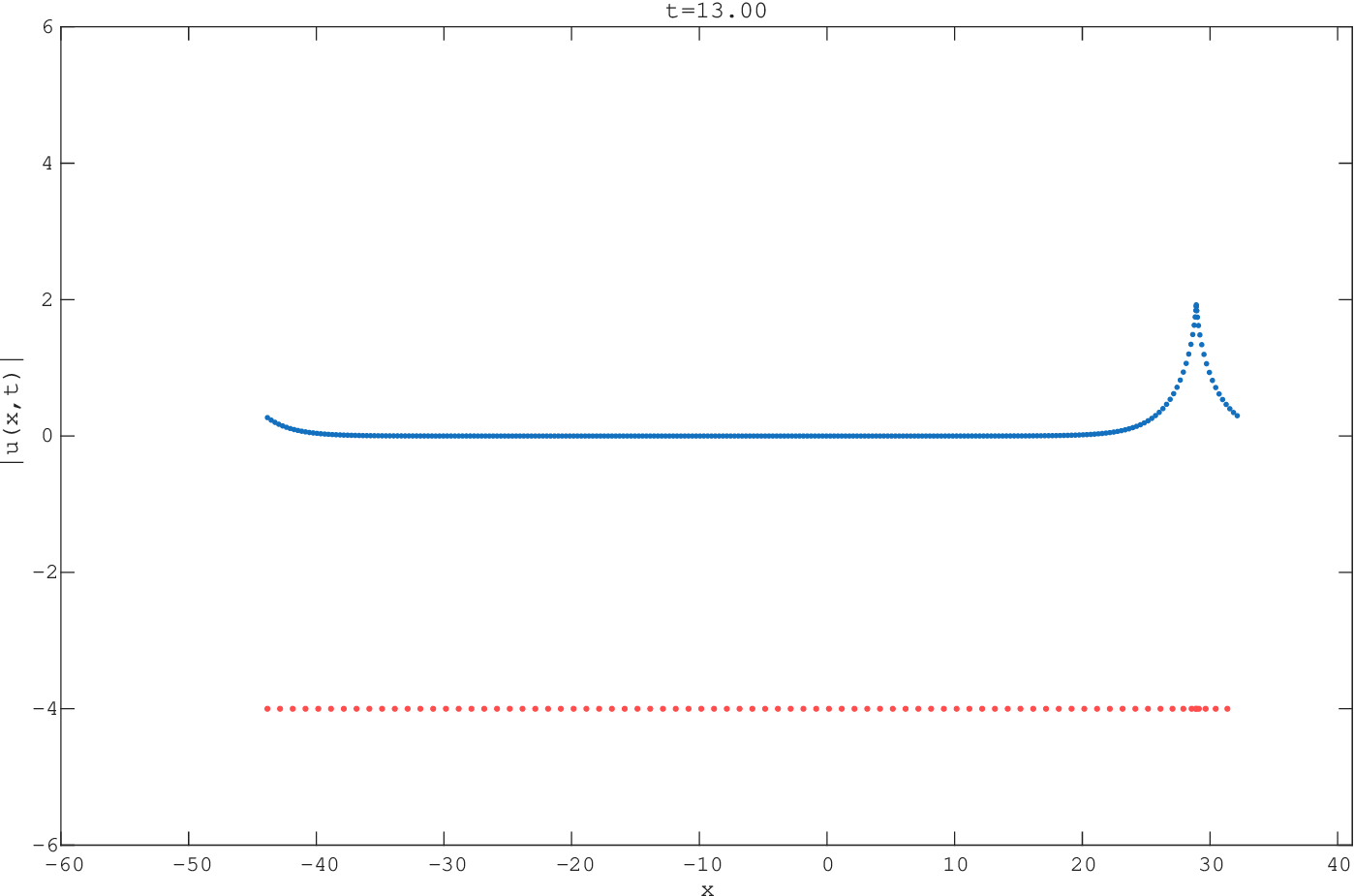}
      \end{minipage} &
      \begin{minipage}[t]{0.47\hsize}
        \centering
        \includegraphics[width=\linewidth]{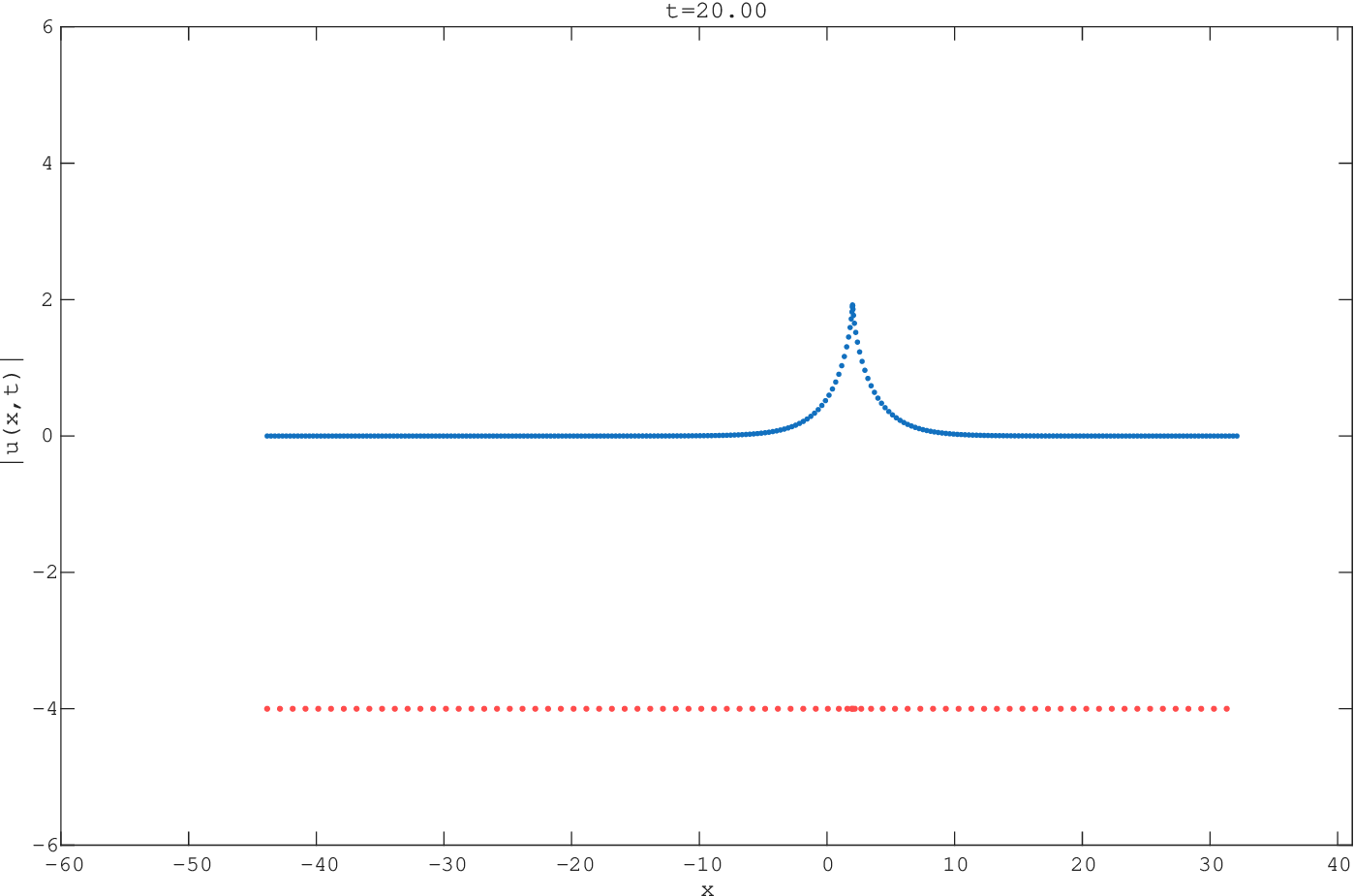}
      \end{minipage}
       \end{tabular}
     \caption{The numerical simulation of the $|u|$-profile of the one-soliton solution for the CmSP equation. maxerr=1.03$\times 10^{-4}$}
              \label{compmSP_abs_1}
  \end{figure}
\begin{figure}[h]
\centering
 \begin{tabular}{cc}
      \begin{minipage}[t]{0.47\hsize}
       \centering
        \includegraphics[width=\linewidth]{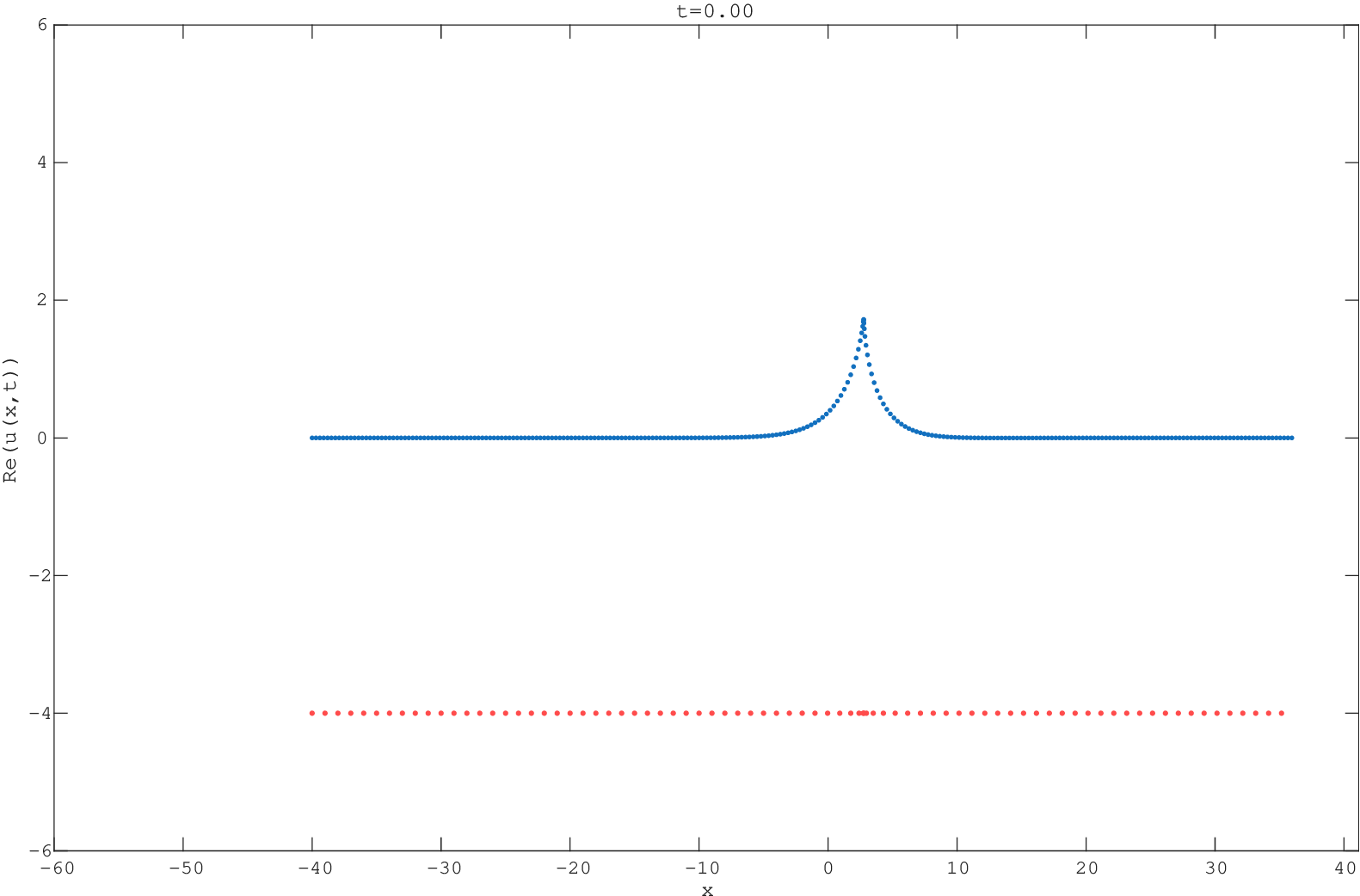}
      \end{minipage} &
      \begin{minipage}[t]{0.47\hsize}
        \centering
        \includegraphics[width=\linewidth]{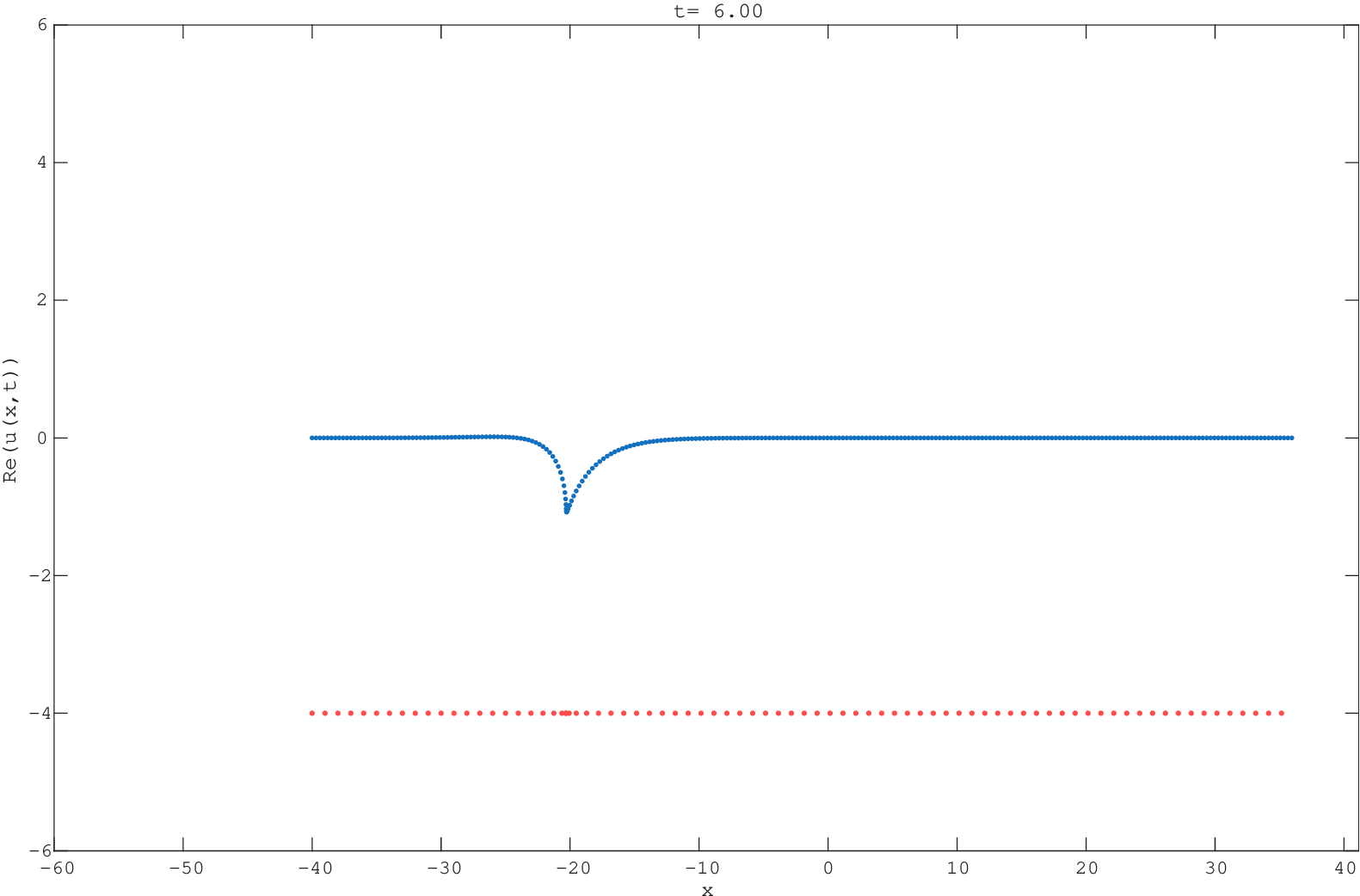}
      \end{minipage}\\

      \begin{minipage}[t]{0.47\hsize}
        \centering
        \includegraphics[width=\linewidth]{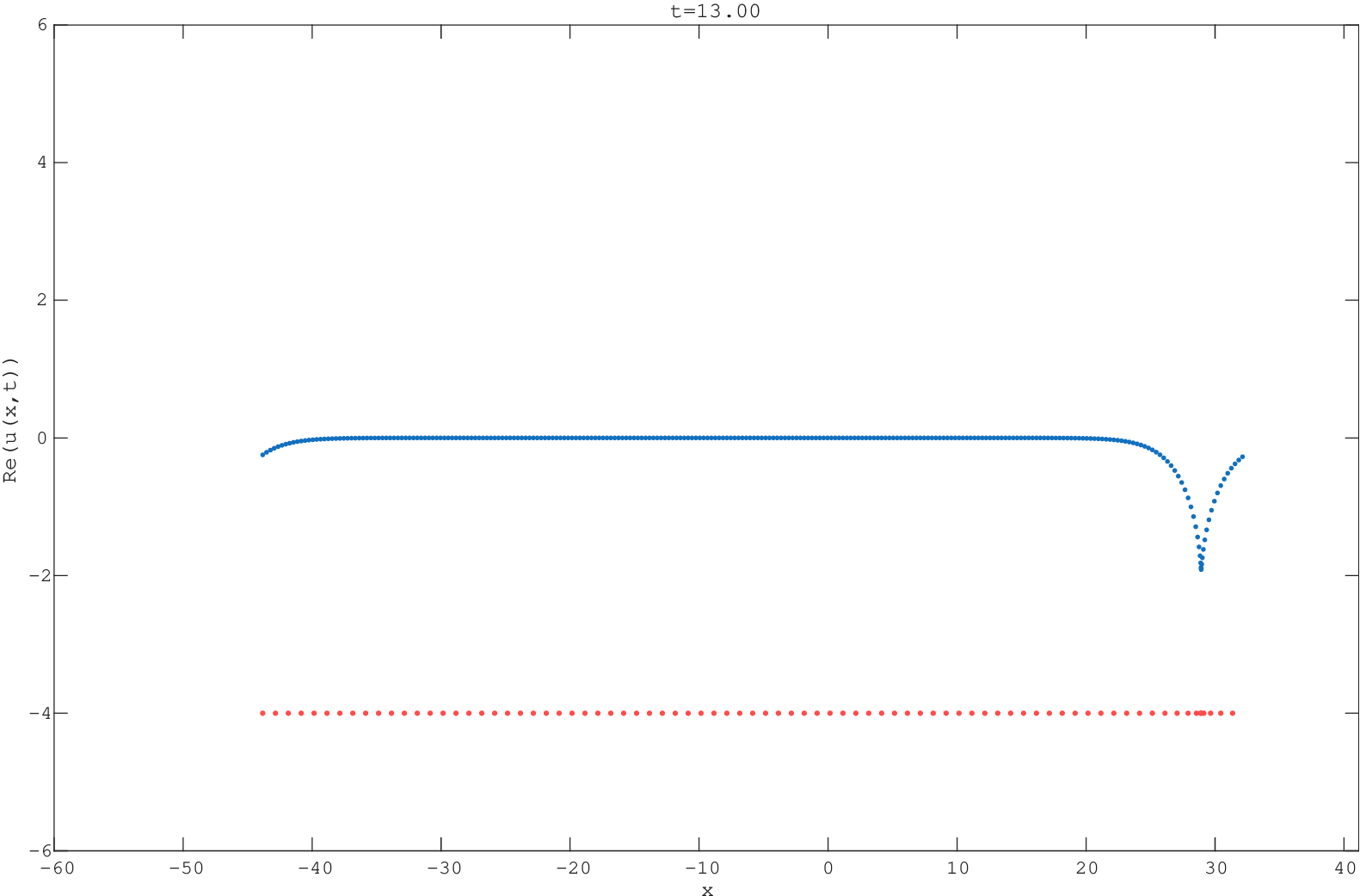}
      \end{minipage} &
      \begin{minipage}[t]{0.47\hsize}
        \centering
        \includegraphics[width=\linewidth]{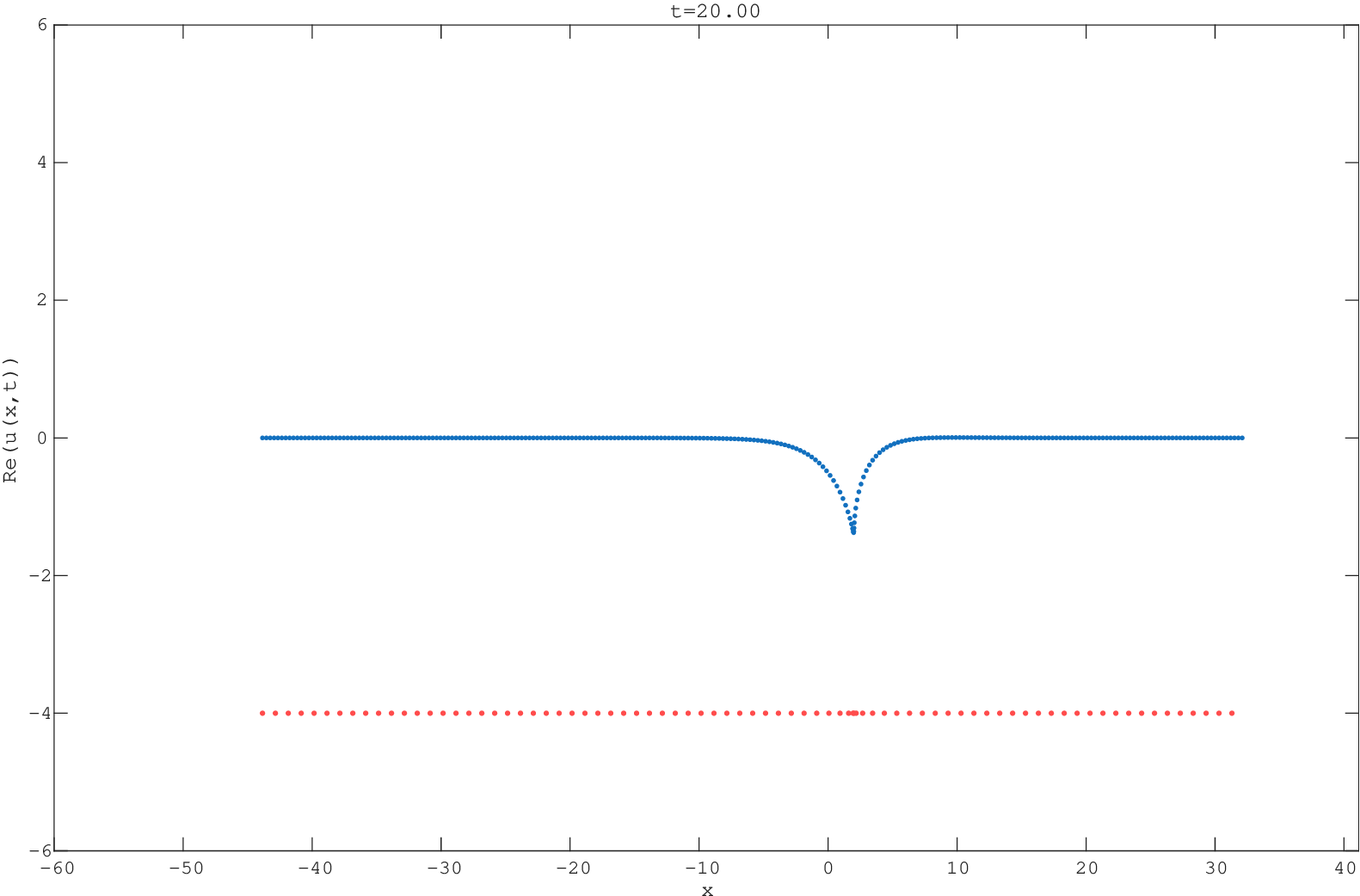}
      \end{minipage}
       \end{tabular}
     \caption{The numerical simulation of the Re($u$)-profile of the one-soliton solution for the CmSP equation. maxerr=1.38$\times 10^{-4}$}
              \label{compmSP_real_1}
  \end{figure}

\begin{figure}[h]
\centering
 \begin{tabular}{cc}
      \begin{minipage}[t]{0.47\hsize}
       \centering
        \includegraphics[width=\linewidth]{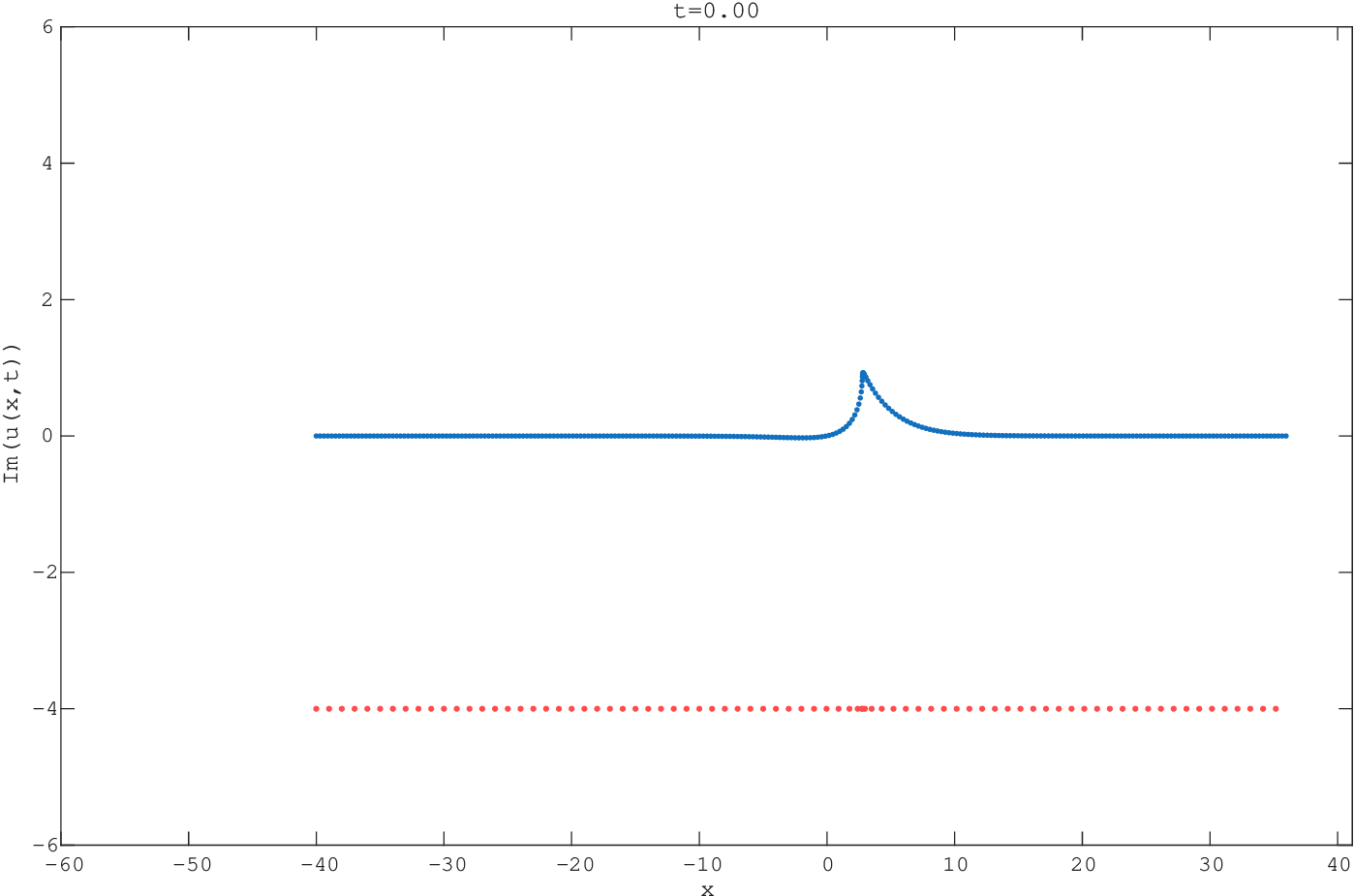}
      \end{minipage} &
      \begin{minipage}[t]{0.47\hsize}
        \centering
        \includegraphics[width=\linewidth]{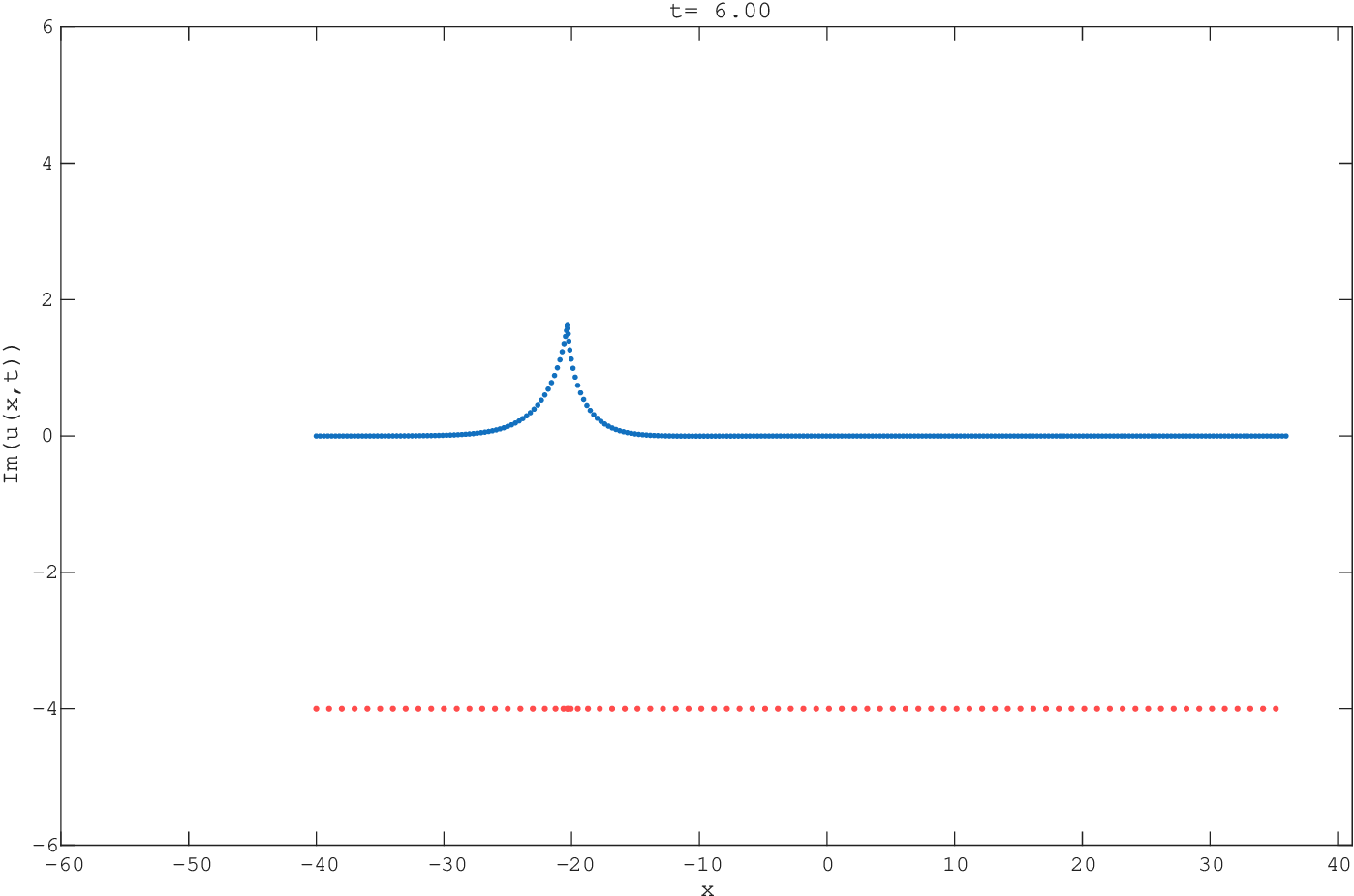}
      \end{minipage}\\

      \begin{minipage}[t]{0.47\hsize}
        \centering
        \includegraphics[width=\linewidth]{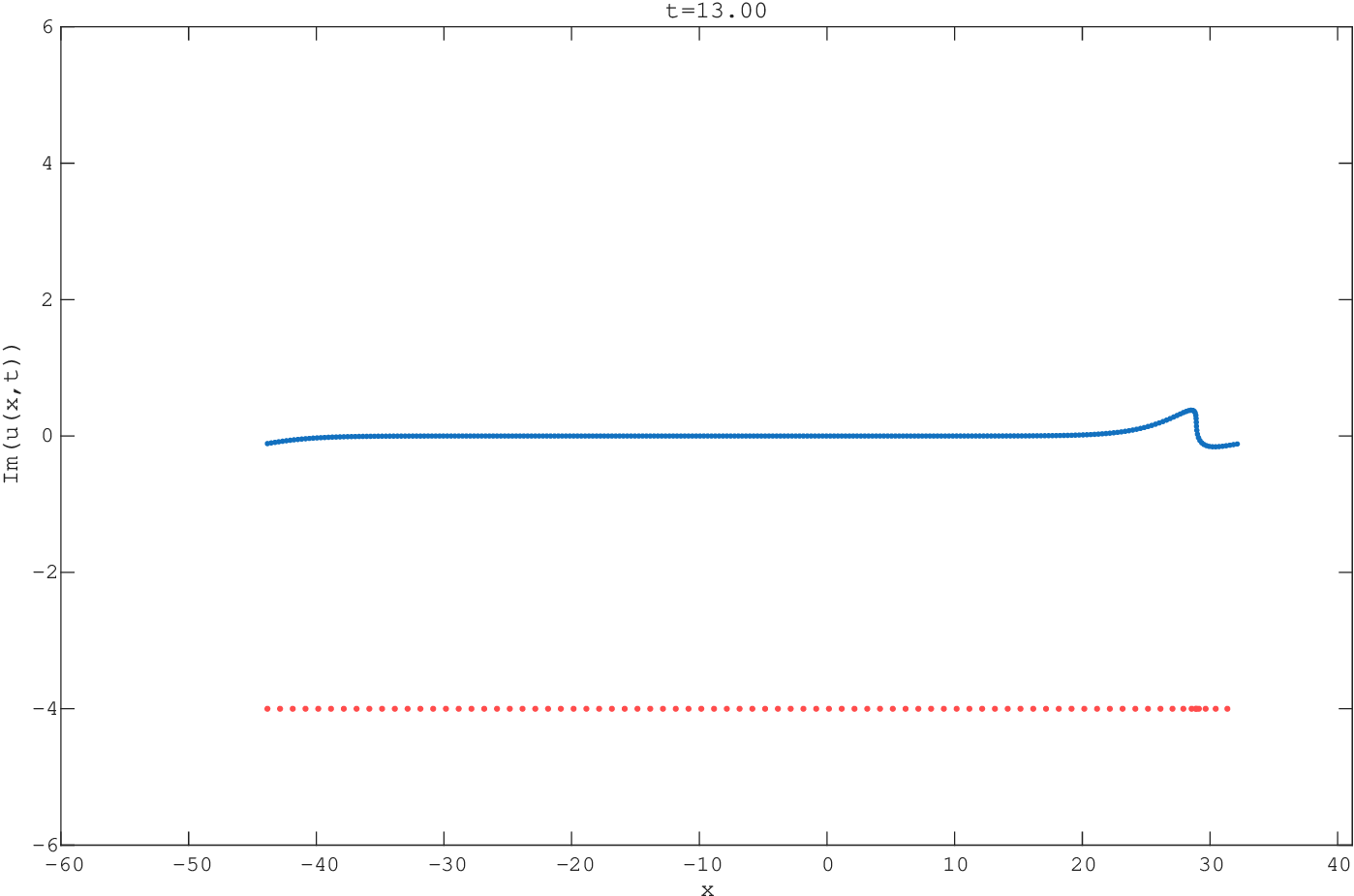}
      \end{minipage} &
      \begin{minipage}[t]{0.47\hsize}
        \centering
        \includegraphics[width=\linewidth]{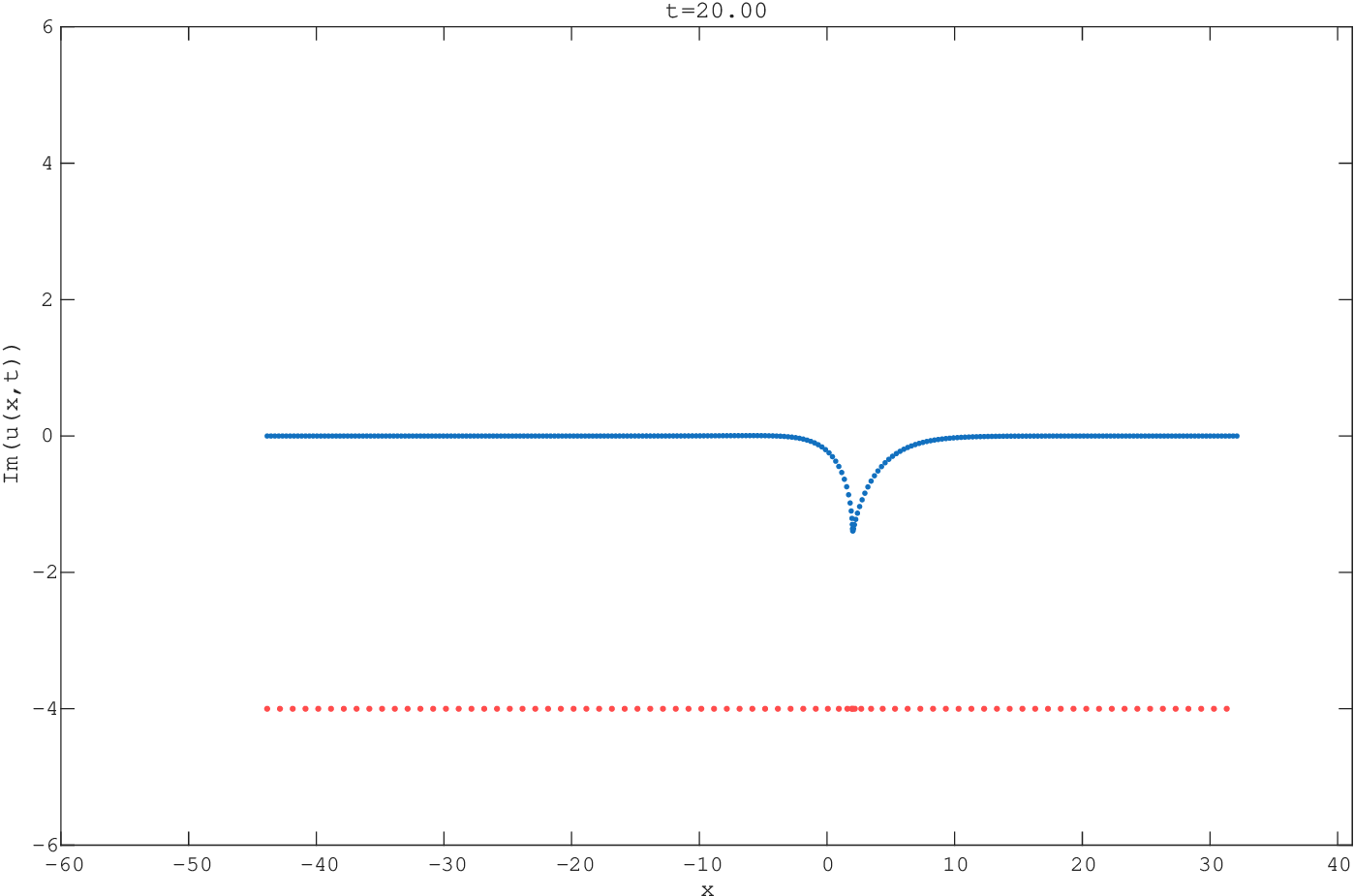}
      \end{minipage}
       \end{tabular}
     \caption{The numerical simulation of the Im($u$)-profile of the one-soliton solution for the CmSP equation. maxerr=1.21$\times 10^{-4}$}
              \label{compmSP_imag_1}
  \end{figure}

\begin{figure}[h]
\centering
 \begin{tabular}{cc}
      \begin{minipage}[t]{0.47\hsize}
       \centering
        \includegraphics[width=\linewidth]{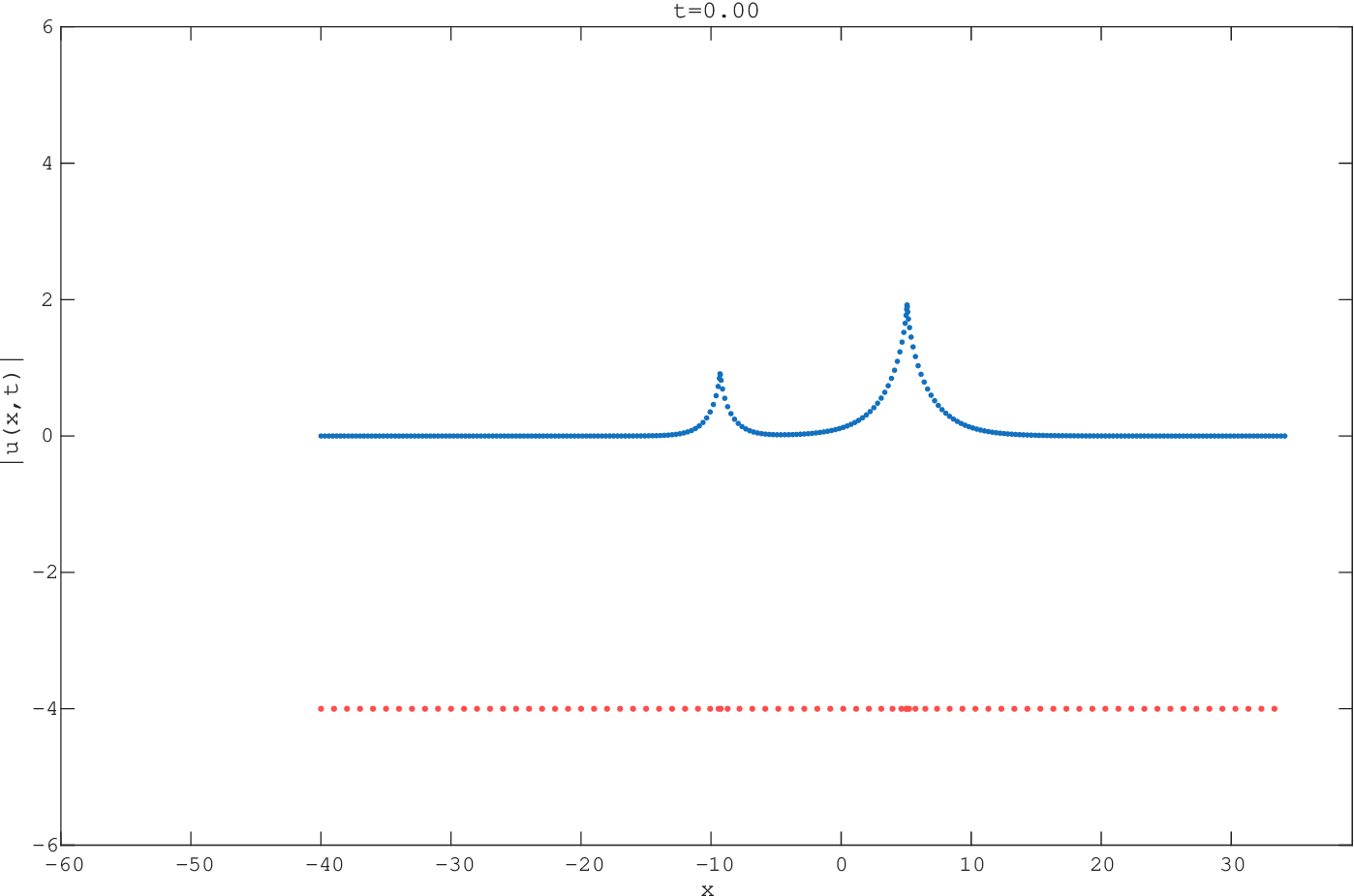}
      \end{minipage} &
      \begin{minipage}[t]{0.47\hsize}
        \centering
        \includegraphics[width=\linewidth]{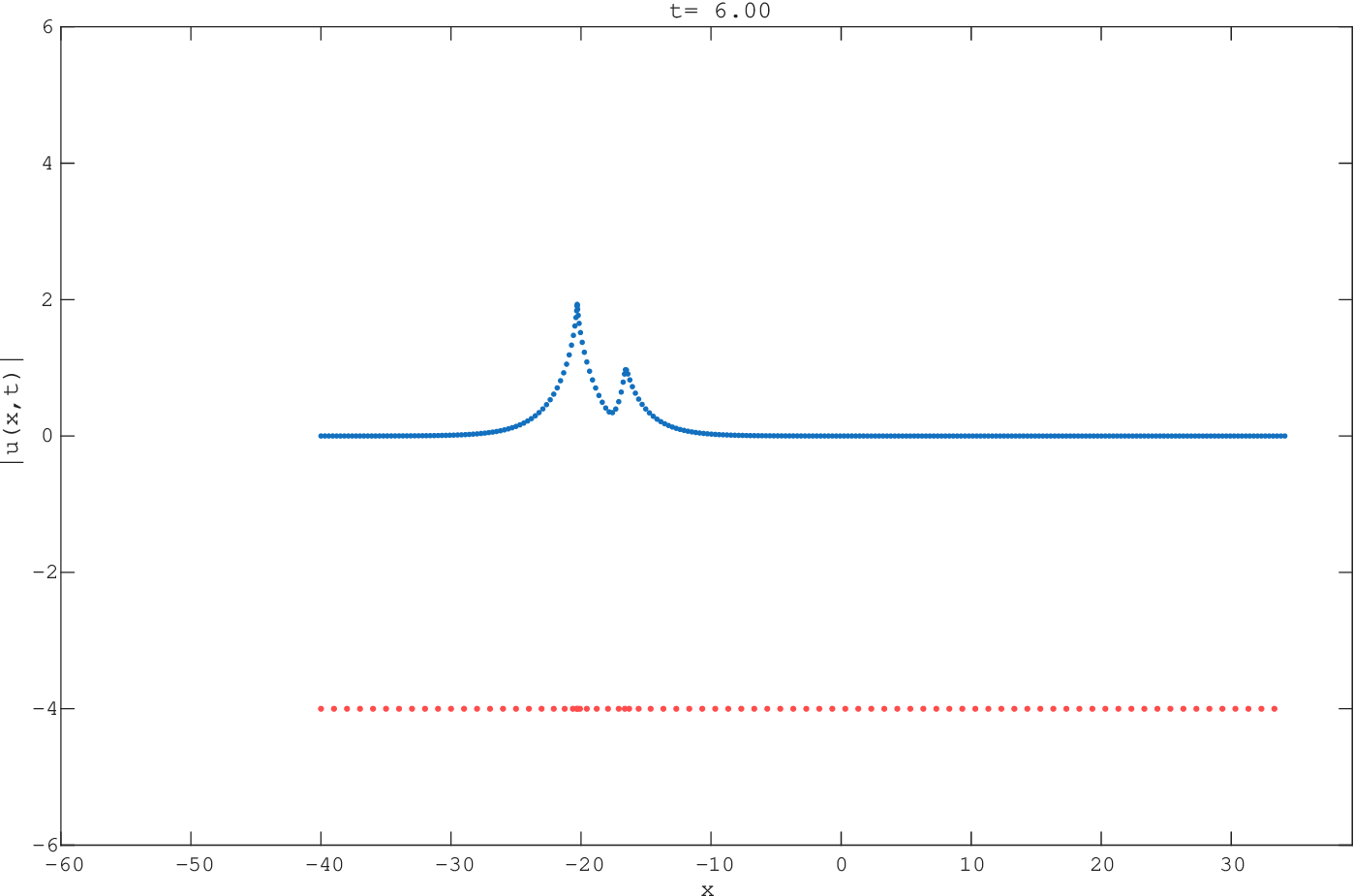}
      \end{minipage}\\

      \begin{minipage}[t]{0.47\hsize}
        \centering
        \includegraphics[width=\linewidth]{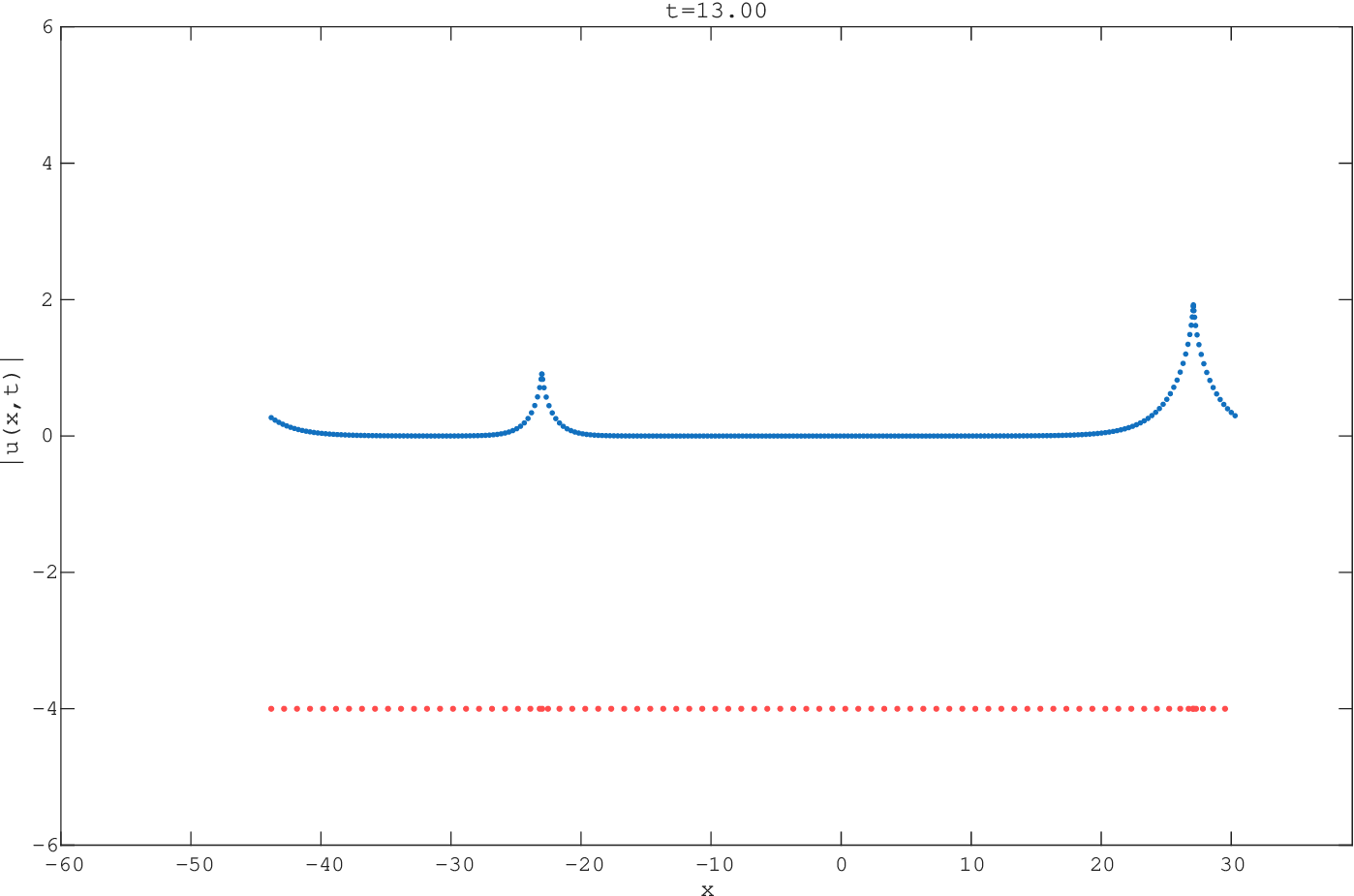}
      \end{minipage} &
      \begin{minipage}[t]{0.47\hsize}
        \centering
        \includegraphics[width=\linewidth]{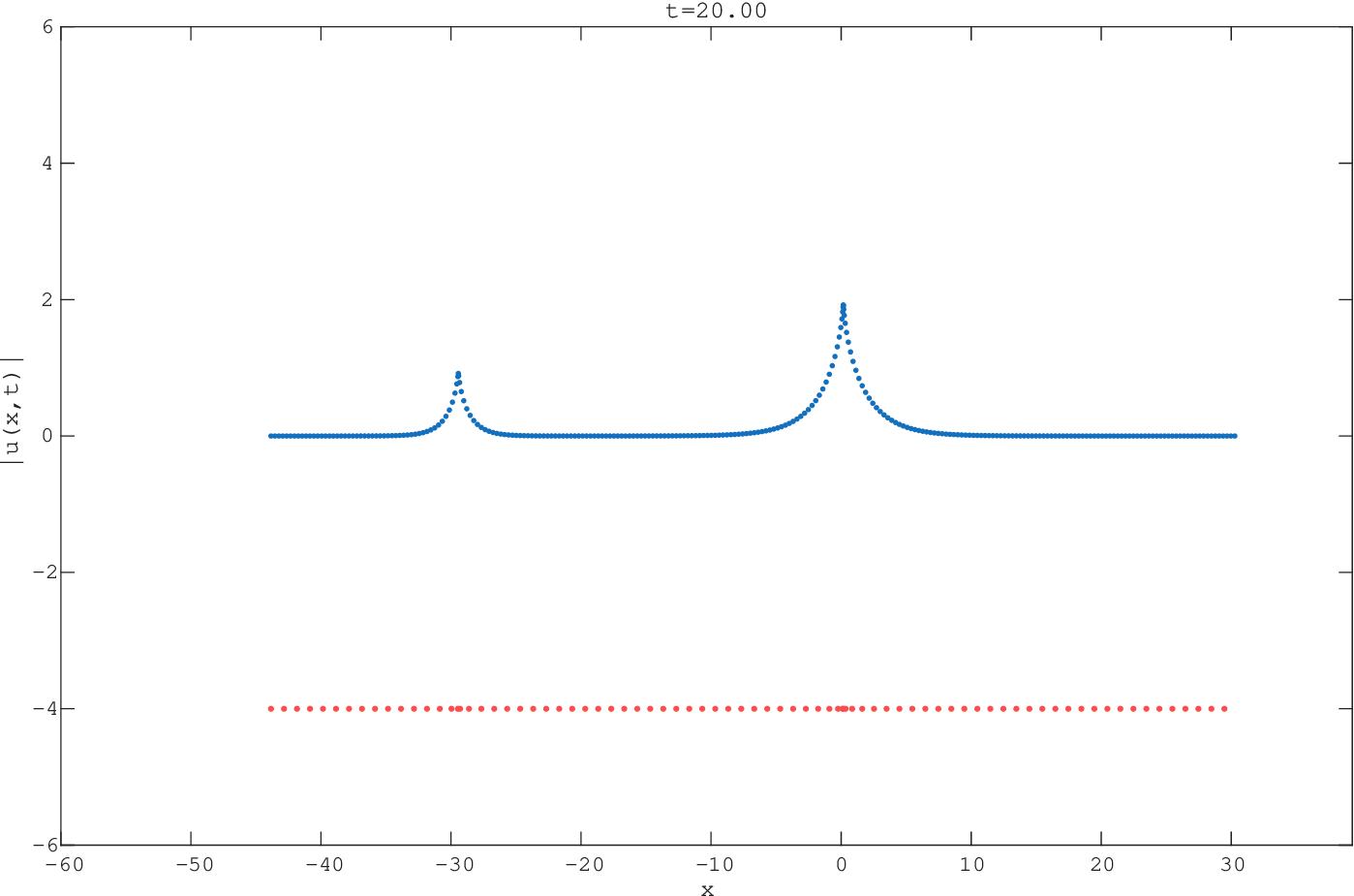}
      \end{minipage}
       \end{tabular}
     \caption{The numerical simulation of the $|u|$-profile of the two-soliton solution for the CmSP equation. maxerr=1.06$\times 10^{-4}$}
              \label{compmSP_abs}
  \end{figure}

\begin{figure}[h]
\centering
 \begin{tabular}{cc}
      \begin{minipage}[t]{0.47\hsize}
       \centering
        \includegraphics[width=\linewidth]{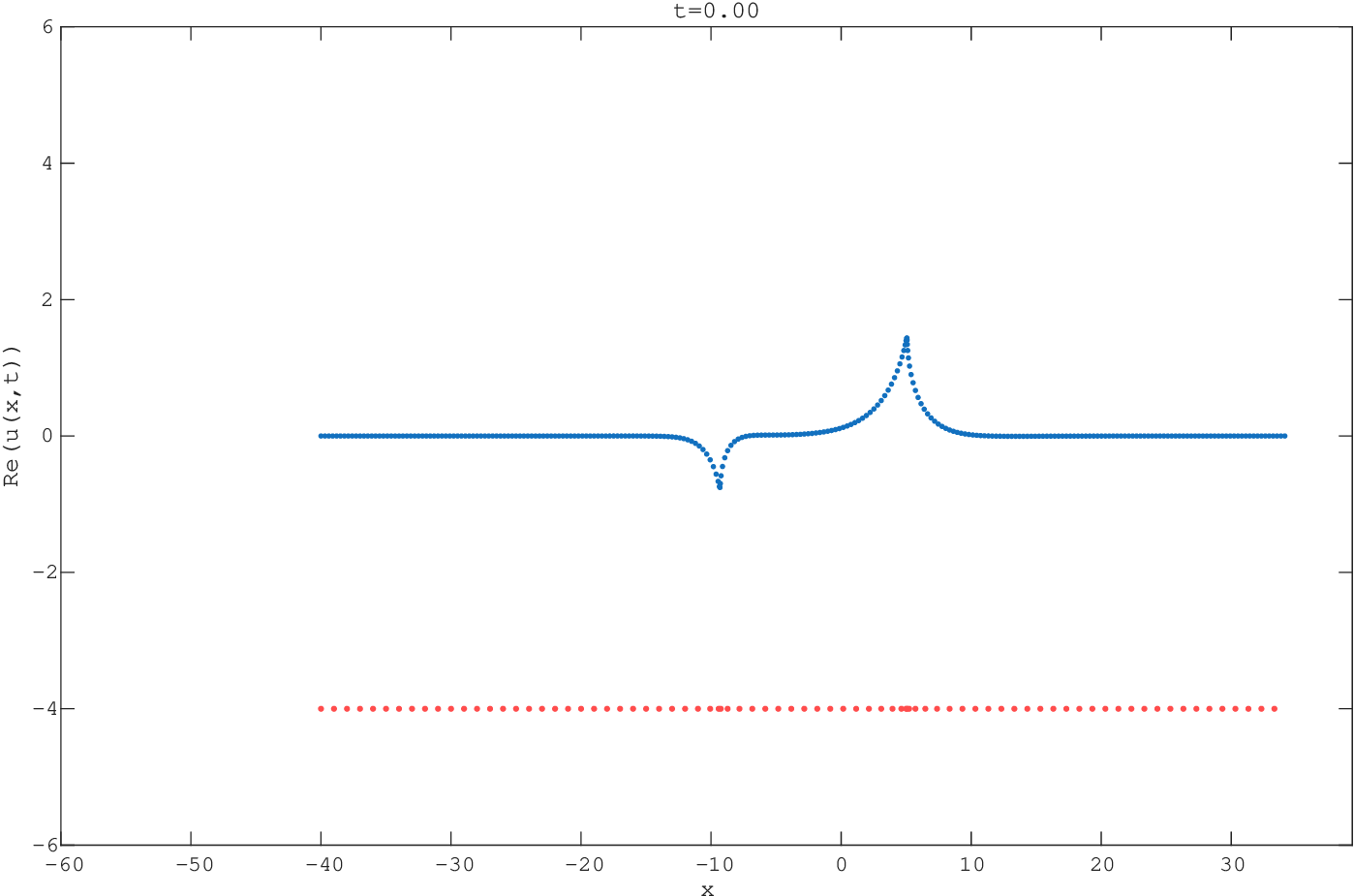}
      \end{minipage} &
      \begin{minipage}[t]{0.47\hsize}
        \centering
        \includegraphics[width=\linewidth]{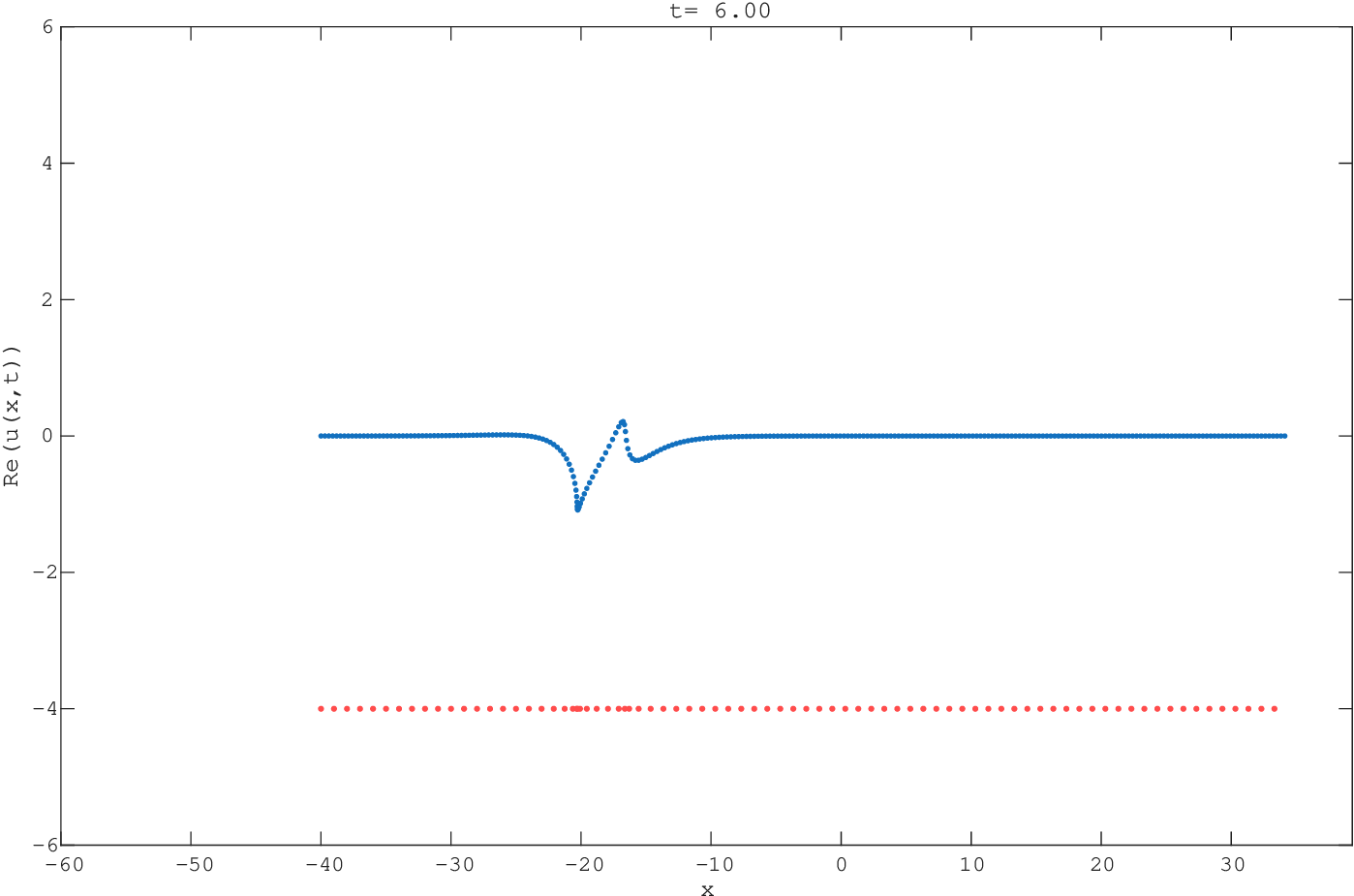}
      \end{minipage}\\

      \begin{minipage}[t]{0.47\hsize}
        \centering
        \includegraphics[width=\linewidth]{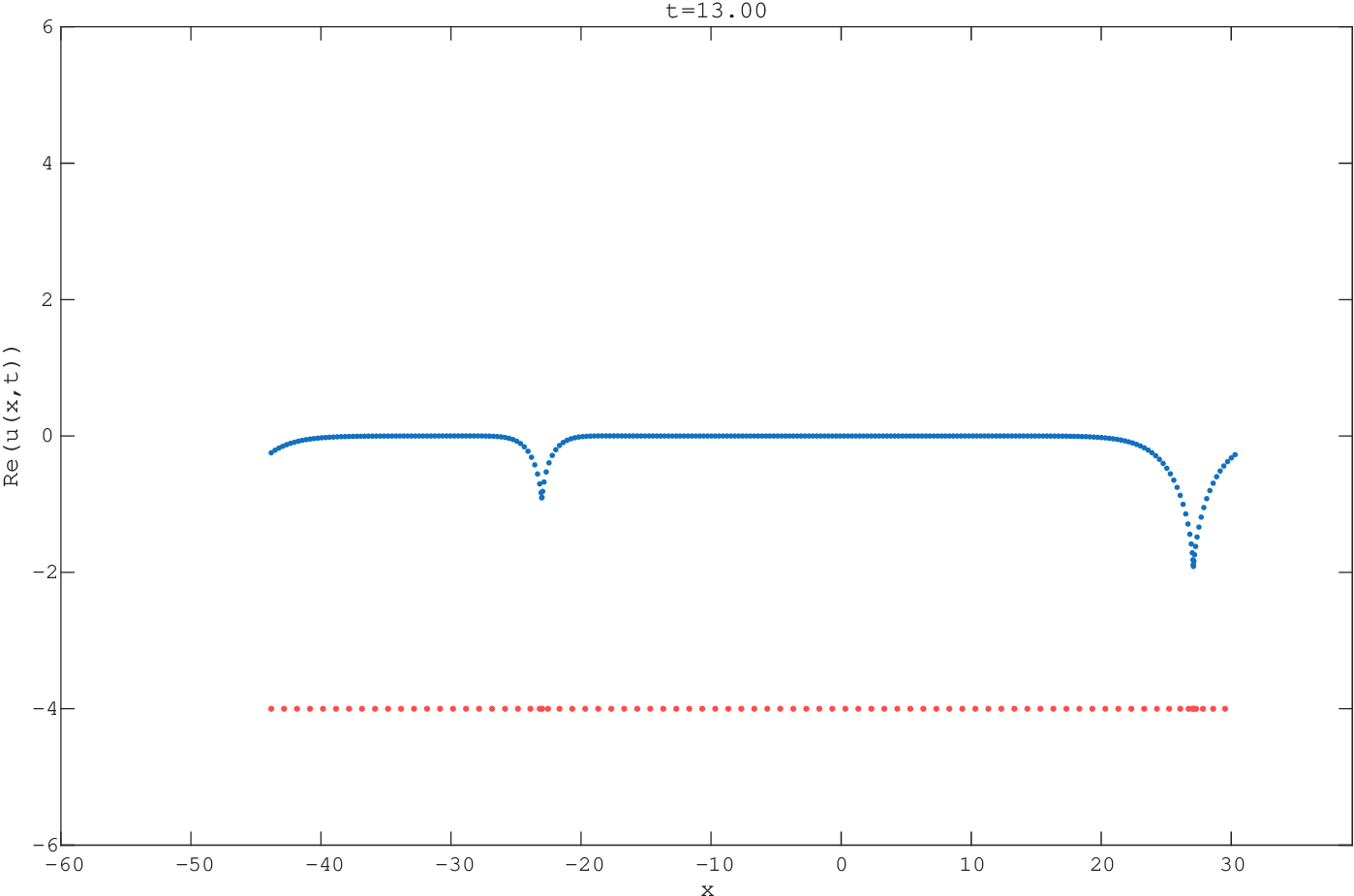}
      \end{minipage} &
      \begin{minipage}[t]{0.47\hsize}
        \centering
        \includegraphics[width=\linewidth]{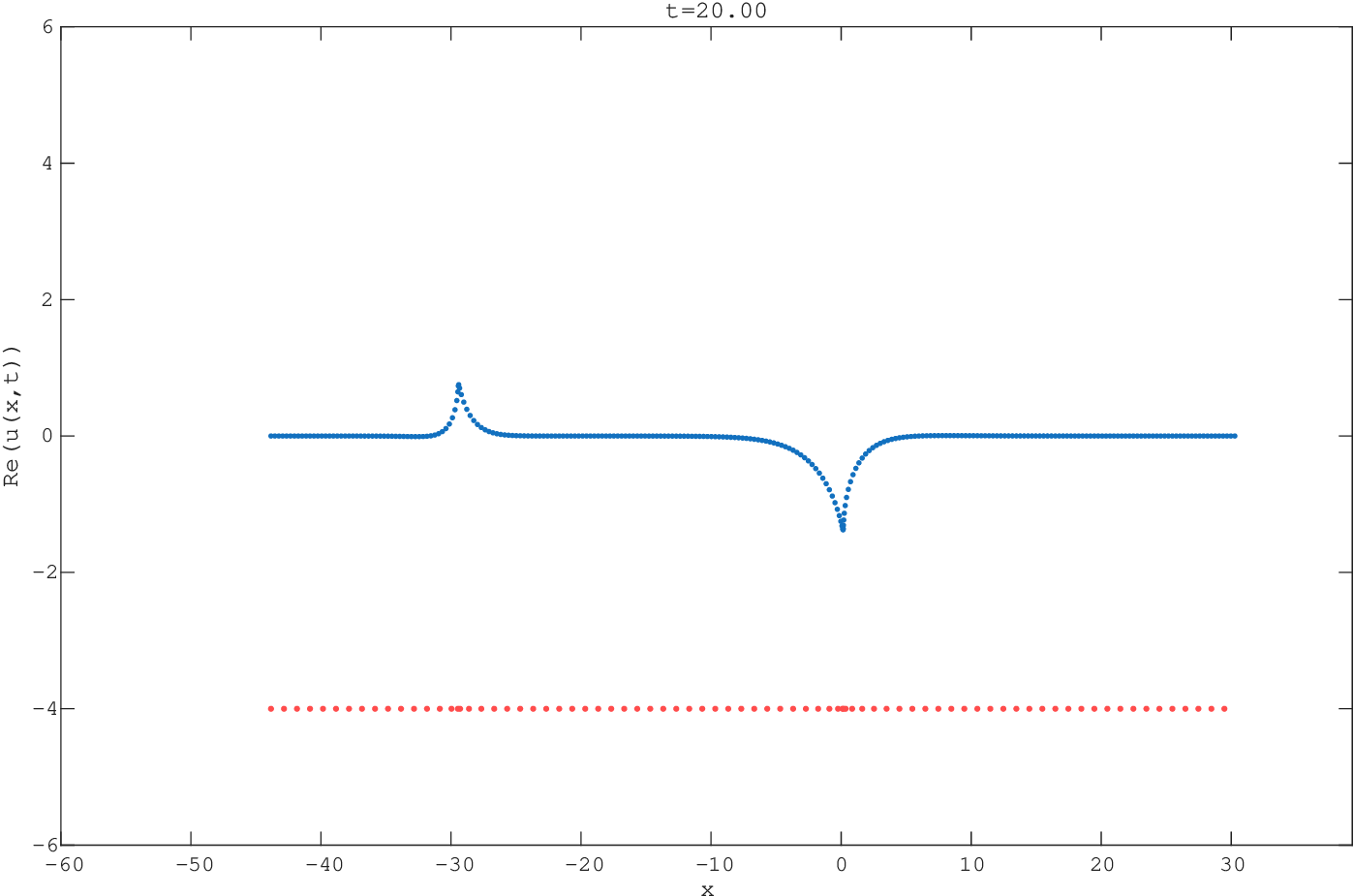}
      \end{minipage}
       \end{tabular}
     \caption{The numerical simulation of the Re($u$)-profile of the two-soliton solution for the CmSP equation. maxerr=1.84$\times 10^{-4}$}
              \label{compmSP_real}
  \end{figure}

\begin{figure}[h]
\centering
 \begin{tabular}{cc}
      \begin{minipage}[t]{0.47\hsize}
       \centering
        \includegraphics[width=\linewidth]{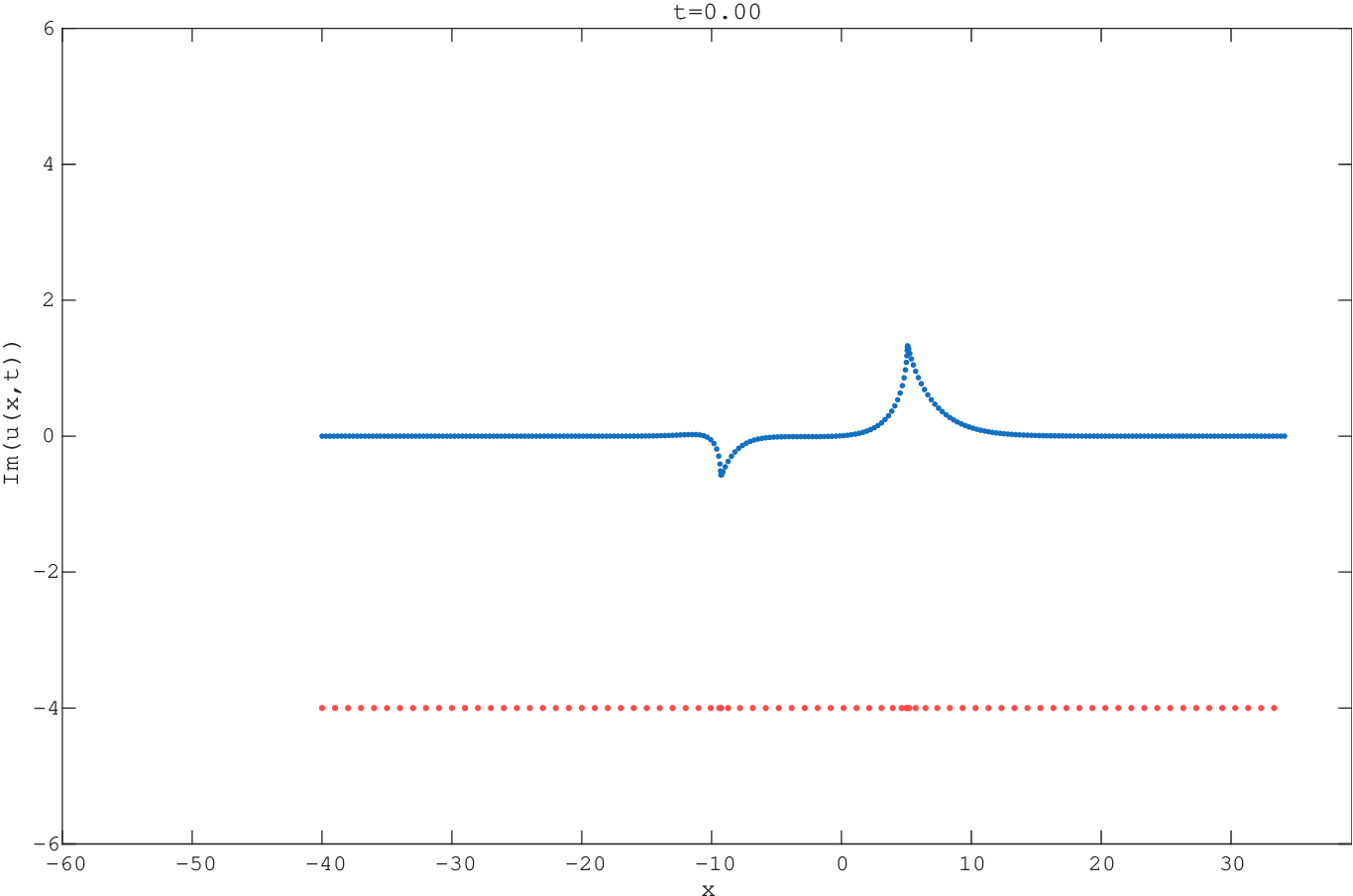}
      \end{minipage} &
      \begin{minipage}[t]{0.47\hsize}
        \centering
        \includegraphics[width=\linewidth]{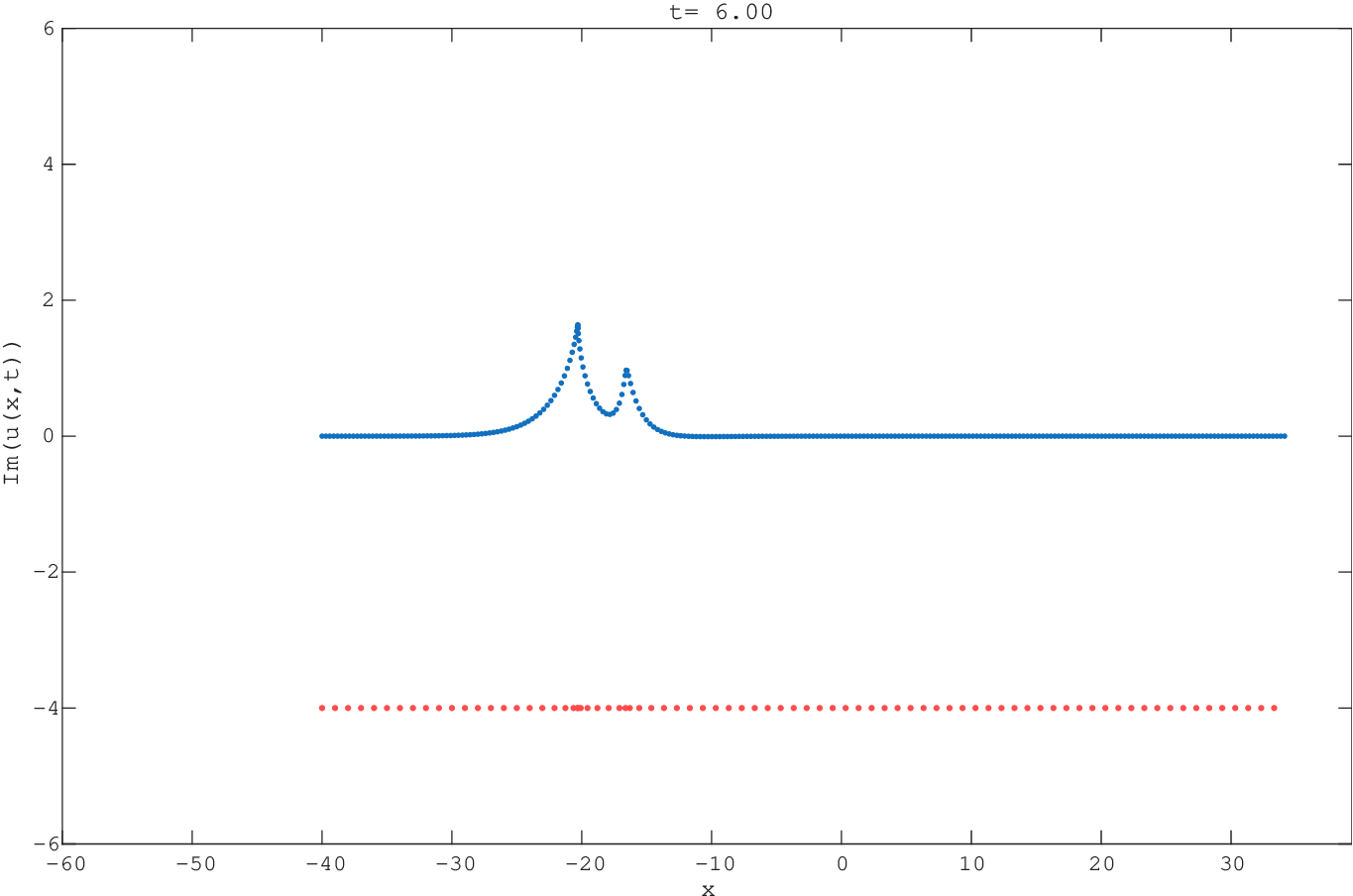}
      \end{minipage}\\

      \begin{minipage}[t]{0.47\hsize}
        \centering
        \includegraphics[width=\linewidth]{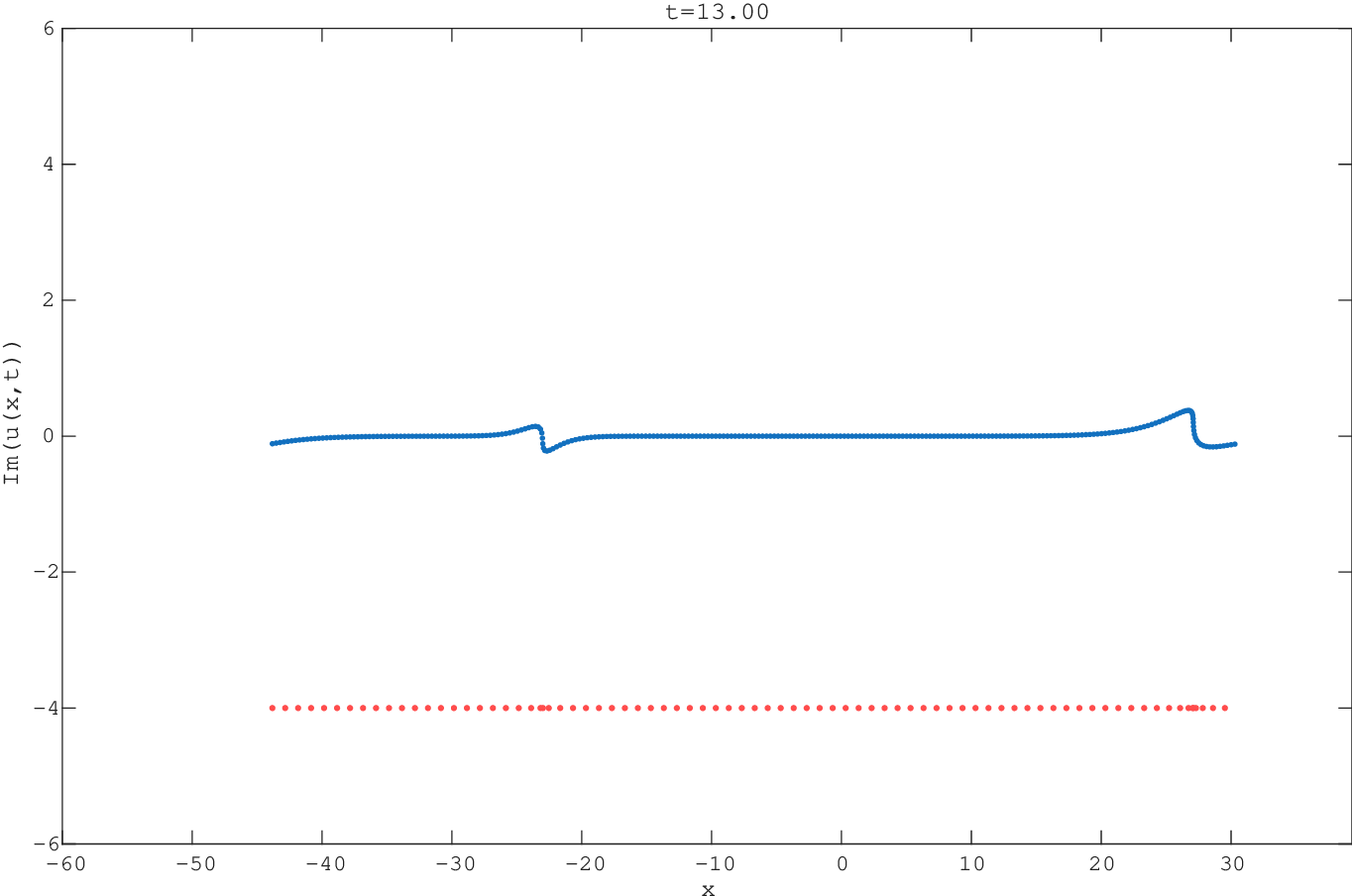}
      \end{minipage} &
      \begin{minipage}[t]{0.47\hsize}
        \centering
        \includegraphics[width=\linewidth]{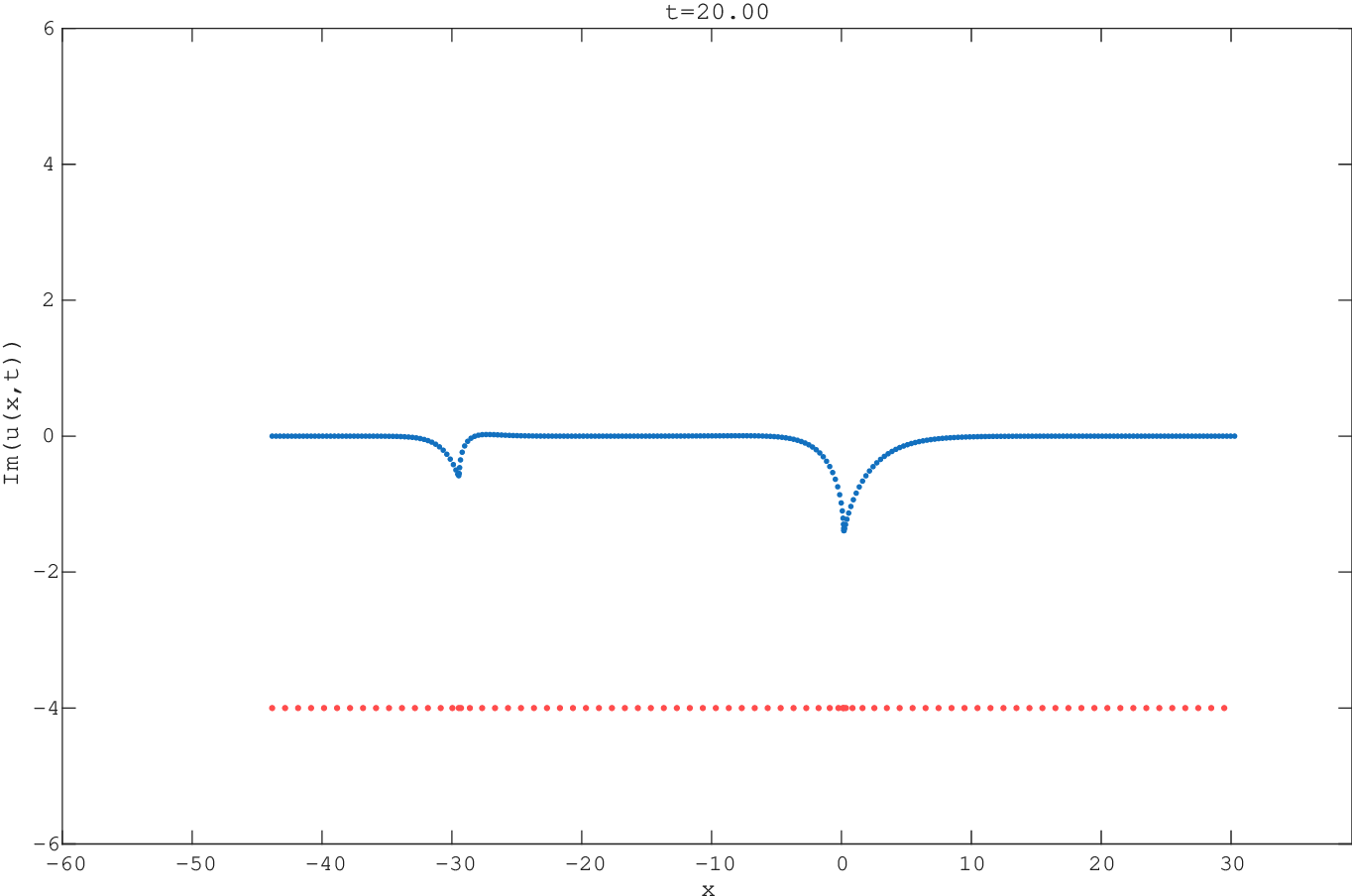}
      \end{minipage}
       \end{tabular}
     \caption{The numerical simulation of the Im($u$)-profile of the two-soliton solution for the CmSP equation. maxerr=1.70$\times 10^{-4}$}
              \label{compmSP_imag}
  \end{figure}

Next, we construct a self-adaptive moving mesh scheme for the CmSP equation. 
Setting $v = u^{*}$ in the 2-mSP equation (\ref{2mSP}) with complex-valued $u$ and $v$, we obtain the CmSP equation~\cite{MCmSP}
\begin{eqnarray}
u_{xt}=u+\frac{1}{2}u^{*}(u^{2})_{xx},\label{CmSP}
\end{eqnarray}
where $u^{*}$ denotes  the complex conjugate of $u$.
From (\ref{semidis2mSP}), we have a self-adaptive moving mesh scheme for the CmSP equation
\begin{eqnarray}
\left\{
\begin{array}{ll}
\displaystyle\frac{d }{d T}(u_{l+1}-u_{l})=(\delta_{l}-a)(u_{l+1}+u_{l}), \\
\\
\displaystyle\frac{d x_{l}}{d T}=-|u_{l}|^{2}.
\end{array}
\right.
\label{semidiscomplexmSP}
\end{eqnarray}

For the CmSP equation, the relative errors of the absolute value, real part, and imaginary part of the solution are evaluated separately using the same definitions as in (\ref{err}) and (\ref{maxerr}), with $u$ replaced by $|u|$, Re$(u)$, and Im$(u)$, respectively. 

Here we show examples of numerical simulations of the CmSP equation  (\ref{CmSP}) with a self-adaptive moving mesh scheme (\ref{semidiscomplexmSP}) under periodic boundary conditions. The initial conditions are given by (\ref{incon_mSP_1}) and (\ref{incon_mSP_2}).
Figures \ref{compmSP_abs_1}, \ref{compmSP_real_1} and \ref{compmSP_imag_1} show numerical simulations of the $|u|$, Re$(u)$, and Im$(u)$-profiles of the one-soliton solution for $p_{1}=0.5+0.1\mathrm{i}, p_{2}=0.5-0.1\mathrm{i}, a_{1}= a_{2}={\rm exp}(-6)$ and $\xi_{1}^{\prime}=\xi_{2}^{\prime}=5$.
Figures \ref{compmSP_abs}, \ref{compmSP_real} and \ref{compmSP_imag} show numerical simulations of the $|u|$, Re$(u)$, and Im$(u)$-profiles of the two-soliton solution for $p_{1}=0.5+0.1\mathrm{i}, p_{2}=1.0+0.3\mathrm{i}, p_{3}=0.5-0.1\mathrm{i}, p_{4}=1.0-0.3\mathrm{i}, a_{1}=a_{3}={\rm exp}(-6), a_{2}=a_{4}={\rm exp}(4)$ and $\xi_{1}^{\prime}=\xi_{2}^{\prime}=\xi_{3}^{\prime}=\xi_{4}^{\prime}=5$.
The solitons travel in a leftward direction with respect to the $x$-axis.

\subsection{Numerical simulations of the defocusing CmSP equation under one-end Dirichlet boundary conditions}
\label{subsec:defocusingCmSPDirichlet}

In this subsection, we first describe the treatment of one-end Dirichlet boundary conditions for the semi-discrete MCmSP system, and then apply it to numerical simulations of the defocusing CmSP equation.

\medskip
\noindent\textit{One-end Dirichlet boundary treatment.}

Here, we consider the semi-discrete system on a non-periodic lattice $0\leq l\leq L$. We impose the following Dirichlet boundary condition at the right endpoint $X=X_L$:
\begin{eqnarray}
 u_L^{(i)}(T)=b^{(i)}(T),\qquad i=1,\ldots,n,
 \label{eq:MCmSP_right_Dirichlet}
\end{eqnarray}
where $b^{(i)}(T)$ are prescribed functions. Although $X_L$ is fixed, the corresponding physical position $x_L(T)$ evolves according to 
 \begin{eqnarray}
 \frac{dx_L}{dT}=-F_L,\qquad
 F_L=\sum_{1\le j<k\le n}c_{jk}b^{(j)}(T)b^{(k)}(T).
 \label{eq:MCmSP_right_edge}
\end{eqnarray}
The first equation of (\ref{selfMCmSP}) can be written as
\begin{eqnarray}
 \frac{du_l^{(i)}}{dT}
 =\frac{du_{l+1}^{(i)}}{dT}
 -(\delta_l-a)\bigl(u_{l+1}^{(i)}+u_l^{(i)}\bigr),
 \qquad l=L-1,\ldots,0.
 \label{eq:MCmSP_right_recursion}
\end{eqnarray}
Since $du_L^{(i)}/dT=db^{(i)}/dT$ is known, (\ref{eq:MCmSP_right_recursion}) determines the time derivatives at all the other mesh points $l=L-1,\ldots,0$ successively from right to left. Equivalently,
\begin{eqnarray}
 \frac{du_l^{(i)}}{dT}
 =\frac{db^{(i)}}{dT}
 -\sum_{k=l}^{L-1}(\delta_k-a)
 \bigl(u_{k+1}^{(i)}+u_k^{(i)}\bigr),
 \qquad l=L-1,\ldots,0.
 \label{eq:MCmSP_right_cumulative}
\end{eqnarray}
Combining this equation with the mesh-point evolution equation in (\ref{selfMCmSP}), we obtain a system of ordinary differential equations for $u_l^{(i)}$ and $x_l$.

\medskip
\noindent\textit{Application to the defocusing CmSP equation.}
\nopagebreak

We now apply one-end Dirichlet boundary conditions to the defocusing CmSP equation.
Taking $n=2$, $u^{(1)}=u$, $u^{(2)}=u^*$, and $c_{12}=-1$ in the MCmSP equation (\ref{MCmSP}), we obtain the defocusing CmSP equation~\cite{MCmSP_def},
\begin{eqnarray}
 u_{xt}=u-u^*(uu_x)_x,
 \label{defocusingCmSP}
\end{eqnarray}
where $u^*$ denotes the complex conjugate of $u$. The corresponding self-adaptive moving mesh scheme is
\begin{eqnarray}
\left\{
\begin{array}{ll}
\displaystyle\frac{d}{dT}(u_{l+1}-u_l)=(\delta_l-a)(u_{l+1}+u_l),\\
\\
\displaystyle\frac{dx_l}{dT}=|u_l|^2.
\end{array}
\right.
\label{semidis_defocusingCmSP}
\end{eqnarray}
Equation (\ref{semidis_defocusingCmSP}) is obtained by reducing the self-adaptive moving mesh scheme for the MCmSP equation derived in Section~\ref{sec_disMCmSP}. 

To test the scheme under one-end Dirichlet boundary conditions with nonzero boundary values, we perform numerical simulations using the exact one- and two-dark-soliton solutions of (\ref{defocusingCmSP}) and compare the numerical solutions with the exact solutions. These exact solutions were obtained in~\cite{MCmSP_def} and are expressed as
\begin{eqnarray}
 \begin{aligned}
 u_e=\frac{s}{2}\frac{g}{f}\exp\!\left[\mathrm{i}\left(\frac{X}{s}-sT\right)\right],\qquad
 x_e=X+\frac{s^2}{4}T-\bigl(\log f\bigr)_T,
 \qquad t=T.
 \end{aligned}
 \label{defocusingCmSPexact}
\end{eqnarray}
For the one-dark-soliton solution,
\begin{eqnarray}
 f=1+z_1,\qquad g=1+\alpha_1z_1,
 \label{defocusingCmSPone}
\end{eqnarray}
whereas for the two-dark-soliton solution,
\begin{eqnarray}
 \begin{aligned}
 f&=1+z_1+z_2+A_{12}z_1z_2,\\
 g&=1+\alpha_1z_1+\alpha_2z_2+A_{12}\alpha_1\alpha_2z_1z_2,
 \end{aligned}
 \label{defocusingCmSPtwo}
\end{eqnarray}
with
\begin{eqnarray}
 \begin{aligned}
 &z_j=\exp\bigl[\beta_j(X-X_{j,0})+\omega_jT\bigr],\qquad
 \alpha_j=e^{-2\mathrm{i}\theta_j},\qquad
 \beta_j=-\frac{\sin\theta_j}{s(1+\cos\theta_j)},\\
 &\omega_j=-s\sin\theta_j,\qquad
 A_{12}=\frac{\sin^2\!\left((\theta_2-\theta_1)/2\right)}{\sin^2\!\left((\theta_2+\theta_1)/2\right)}.
  \end{aligned}
 \label{defocusingCmSPparameters}
\end{eqnarray}
Here, $s>0$, $\theta_j$ and $X_{j,0}$ are real parameters, so that $f$ is real and $g$ is complex.

We impose the Dirichlet boundary condition (\ref{eq:MCmSP_right_Dirichlet}) with
$u_L(T)=u_e(X_L,T)$, and take the initial values $u_l(0)=u_e(X_l,0)$ and
$x_l(0)=x_e(X_l,0)$ for $l=0,\ldots,L$. The time derivatives $du_l/dT$ are then
computed by (\ref{eq:MCmSP_right_cumulative}), and the mesh points evolve according to the second equation of (\ref{semidis_defocusingCmSP}). We impose the boundary condition at the right endpoint because small disturbances on the background propagate to the left. If the condition were imposed at the left endpoint, numerical errors would be amplified enormously.

As a time marching method, we use the improved Euler method.
The number of mesh intervals is $L=24000$, the computational domain width is $D=240$, and the time step is $\Delta t=0.0001$. The uniform lattice spacing in the $X$ variable is chosen as $2a=D/L$, and the computation is carried out up to $t=40$. As in Subsection~\ref{subsec:CCmSPperiodic}, the edge point is taken at $X_0=-D/2$, and the computational grid is $X_l=X_0+2al$ for $l=0,\ldots, L$. The background-normalized amplitude error at each time level  $t_m$ is defined by
\begin{eqnarray}
 \mathrm{err}_m=\max_{0\le l\le L}
 \frac{\bigl||u_{e,l}^m|-|u_{n,l}^m|\bigr|}{u_{\mathrm{bg}}},
 \label{defocusingCmSPerror}
\end{eqnarray}
where $u_{\mathrm{bg}}$ is the background amplitude, $u_{\mathrm{bg}}=s/2$ for the
defocusing CmSP equation, and $u_{e,l}^m$ and $u_{n,l}^m$ are the exact and numerical solutions, respectively, evaluated at the same lattice point. Finally, we define the maximum error 
\begin{eqnarray}
\mathrm{maxerr}=\max_{1\leq m\leq M} \mathrm{err}_m.
\label{defocusingCmSPmaxerr}
\end{eqnarray}


\begin{figure}[h]
\centering
 \begin{tabular}{cc}
      \begin{minipage}[t]{0.47\hsize}
       \centering
        \includegraphics[width=\linewidth]{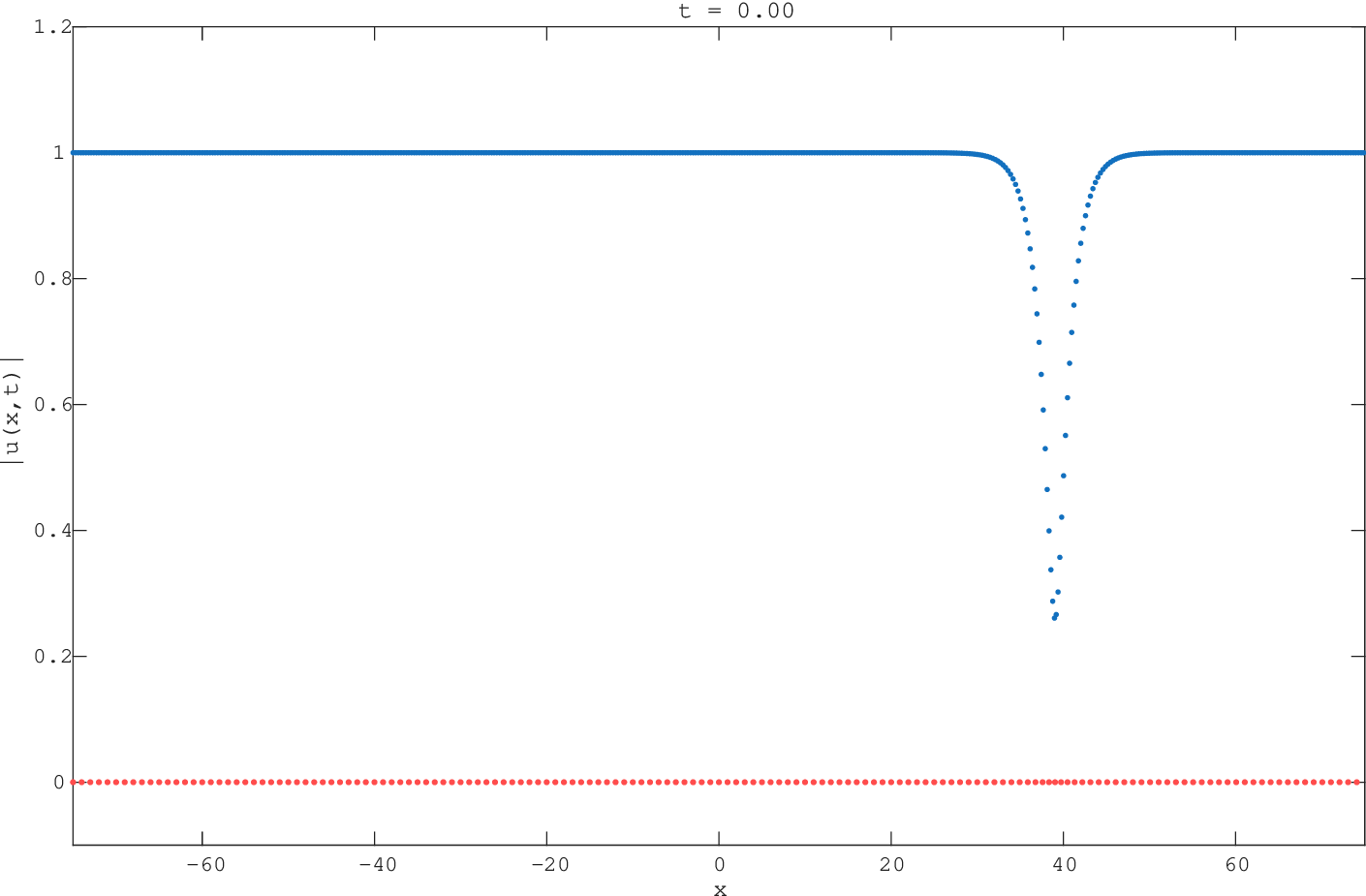}
      \end{minipage} &
      \begin{minipage}[t]{0.47\hsize}
        \centering
        \includegraphics[width=\linewidth]{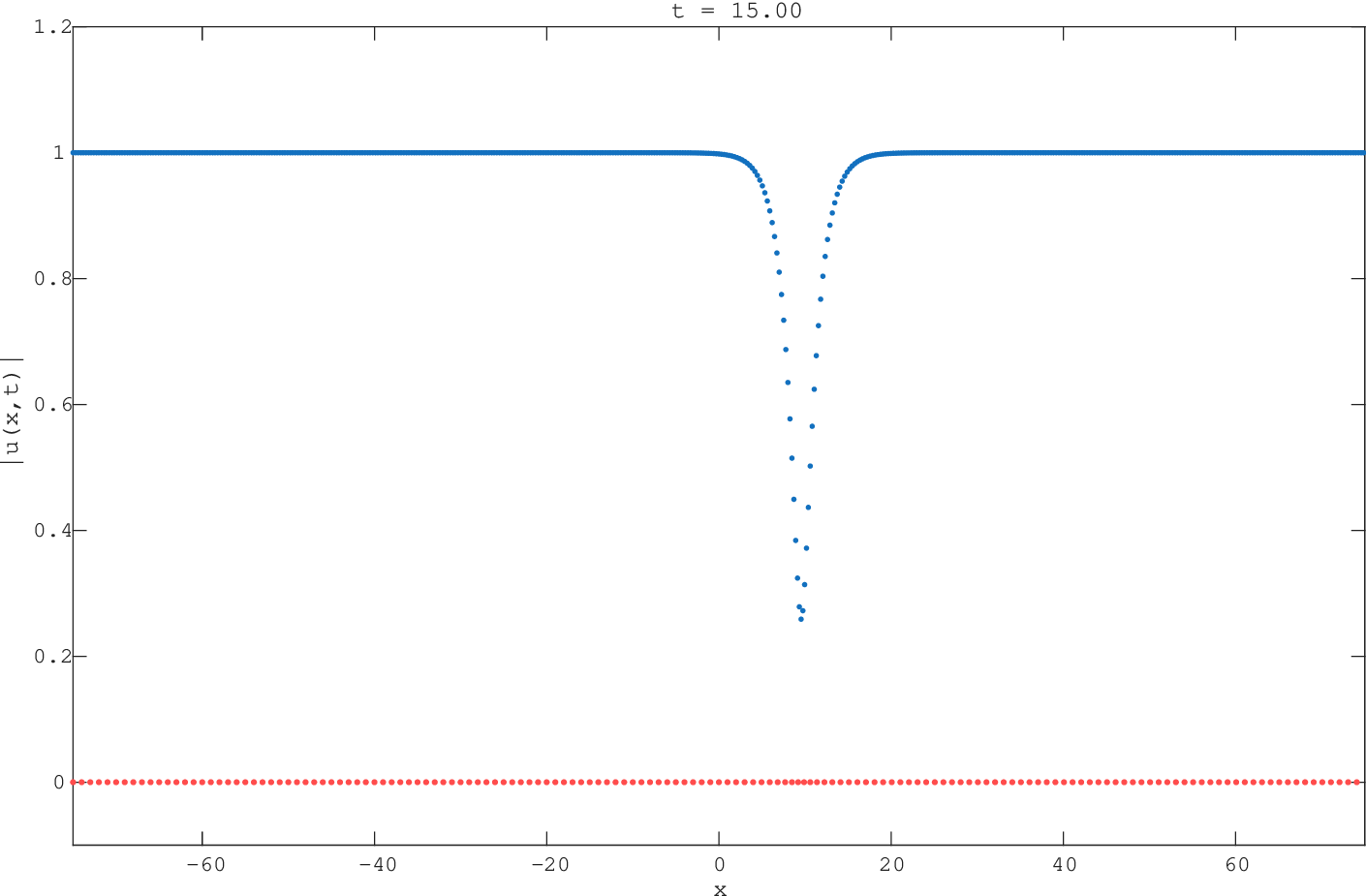}
      \end{minipage}\\

      \begin{minipage}[t]{0.47\hsize}
        \centering
        \includegraphics[width=\linewidth]{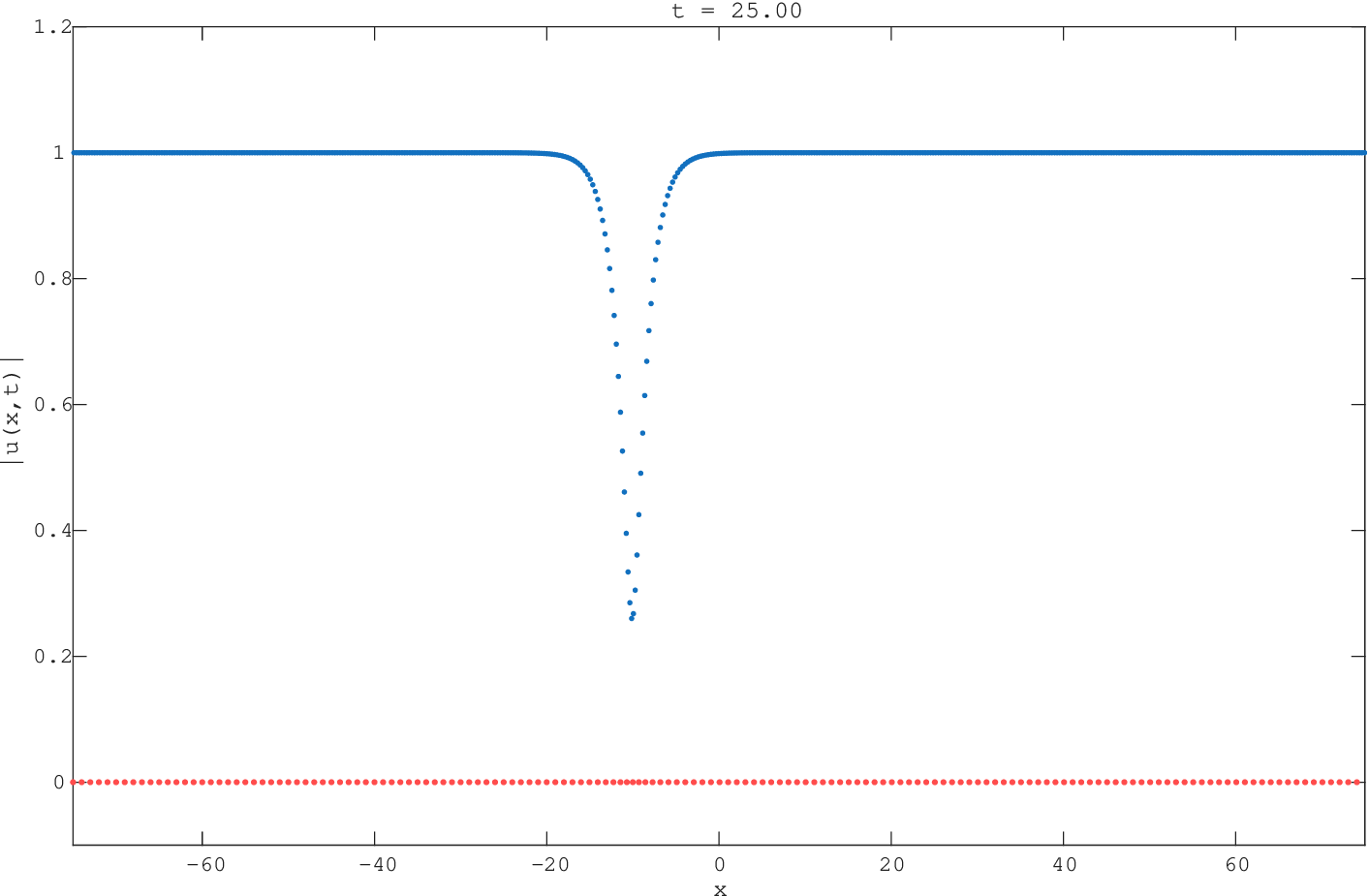}
      \end{minipage} &
      \begin{minipage}[t]{0.47\hsize}
        \centering
        \includegraphics[width=\linewidth]{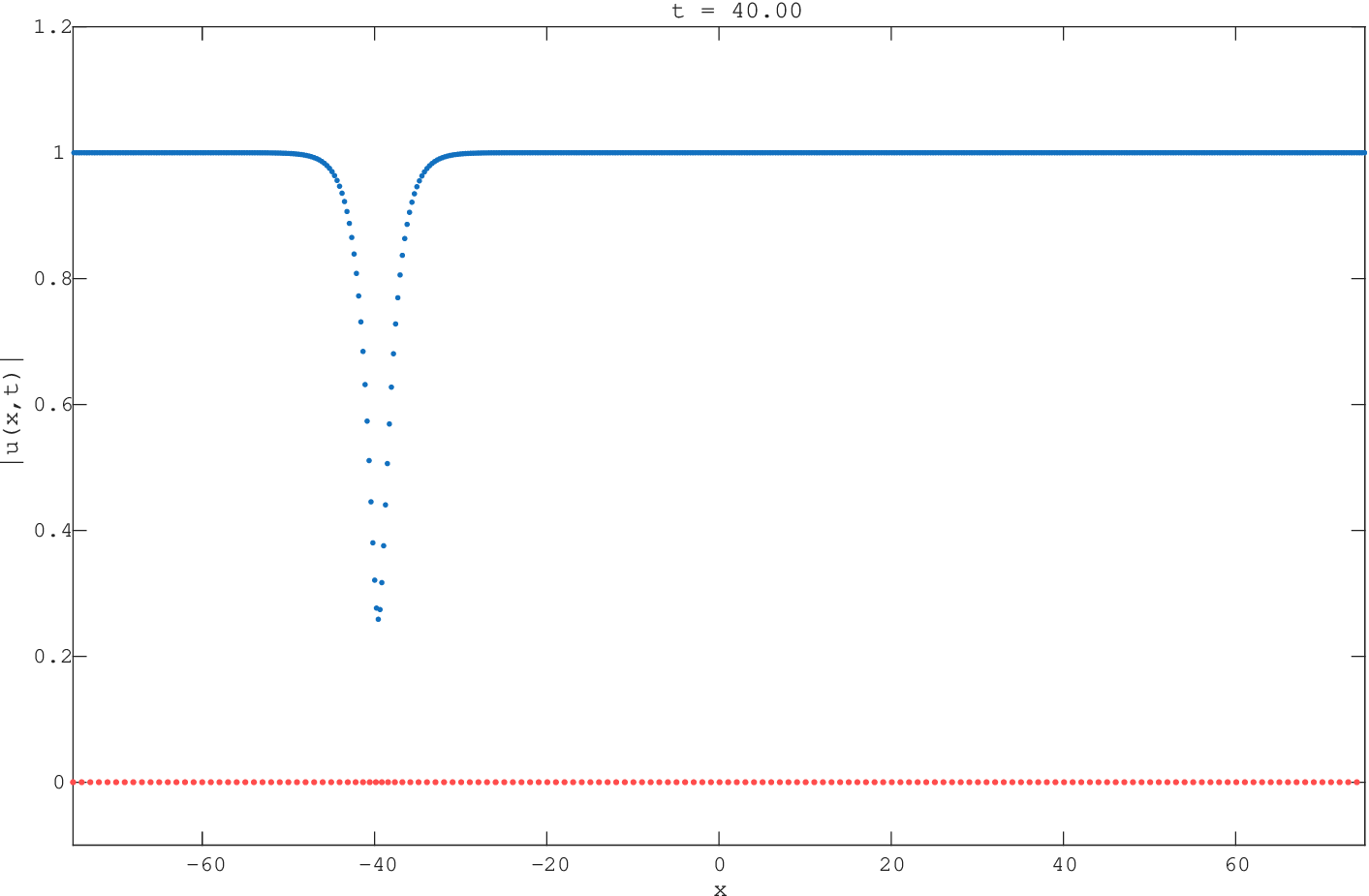}
      \end{minipage}
       \end{tabular}
     \caption{The numerical simulation of the $|u|$-profile of the one-dark-soliton solution for the defocusing CmSP equation. maxerr=1.95$\times 10^{-5}$}
              \label{dmSP_1}
  \end{figure}


\begin{figure}[h]
\centering
 \begin{tabular}{cc}
      \begin{minipage}[t]{0.47\hsize}
       \centering
        \includegraphics[width=\linewidth]{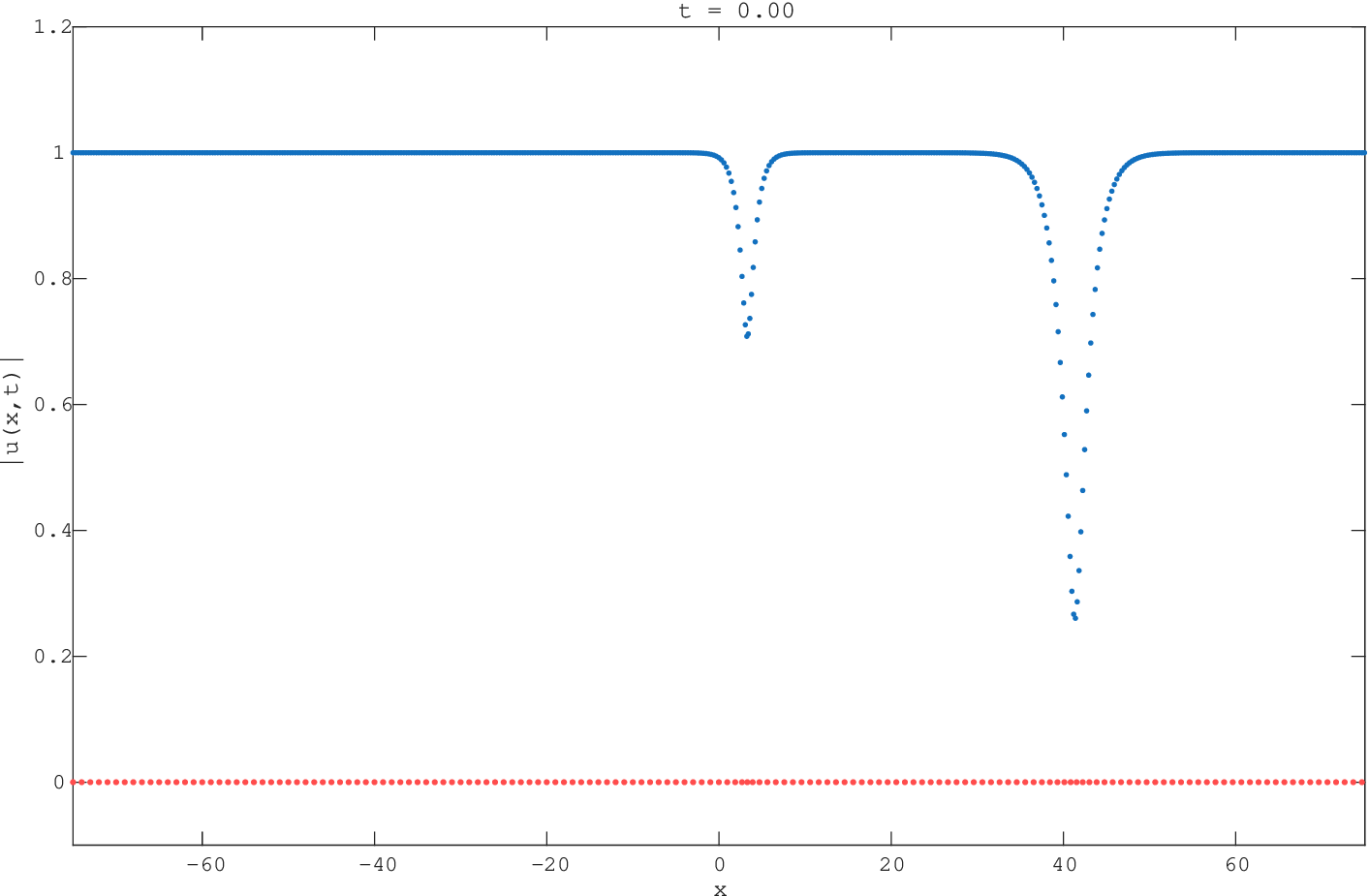}
      \end{minipage} &
      \begin{minipage}[t]{0.47\hsize}
        \centering
        \includegraphics[width=\linewidth]{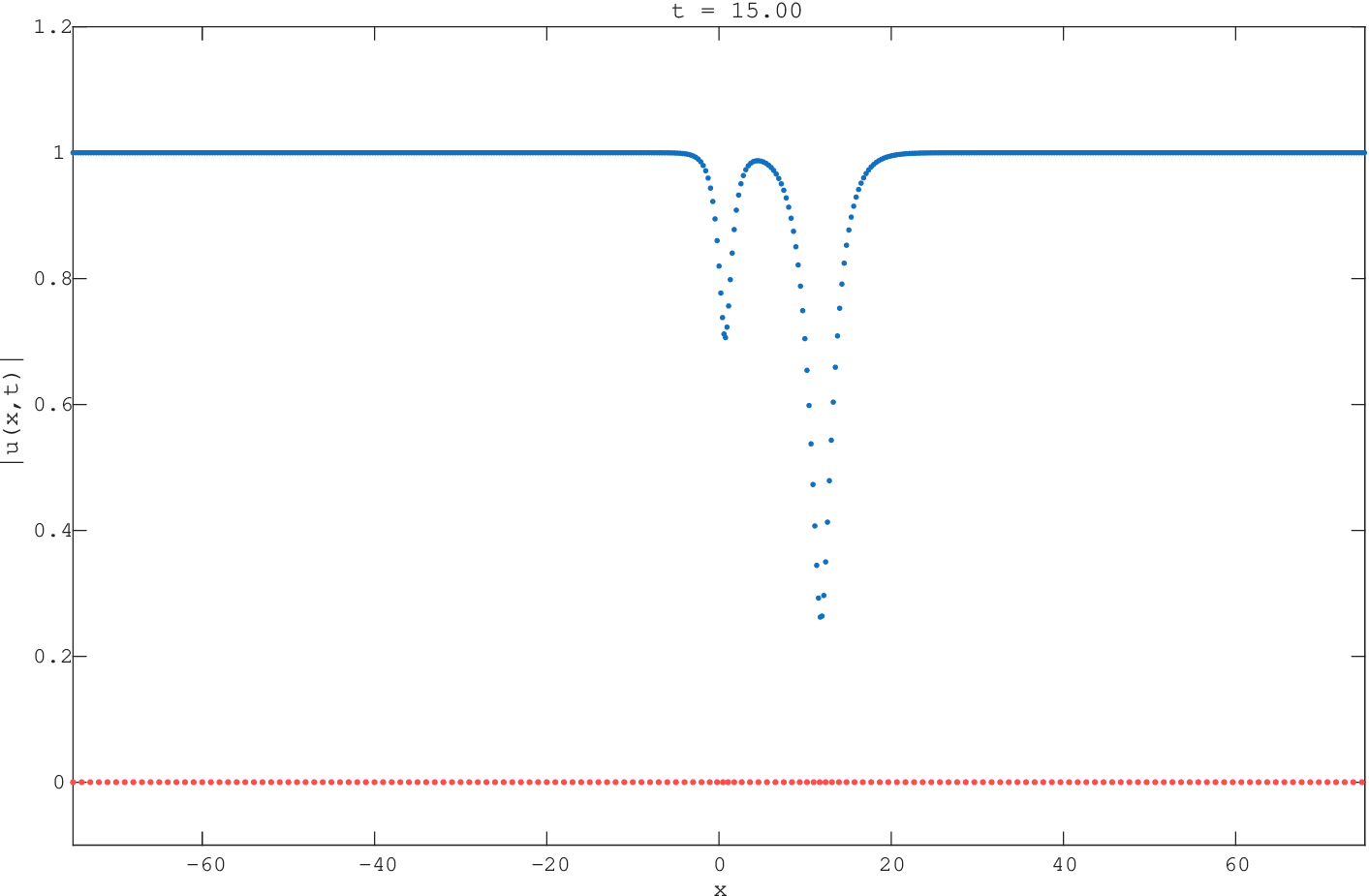}
      \end{minipage}\\

      \begin{minipage}[t]{0.47\hsize}
        \centering
        \includegraphics[width=\linewidth]{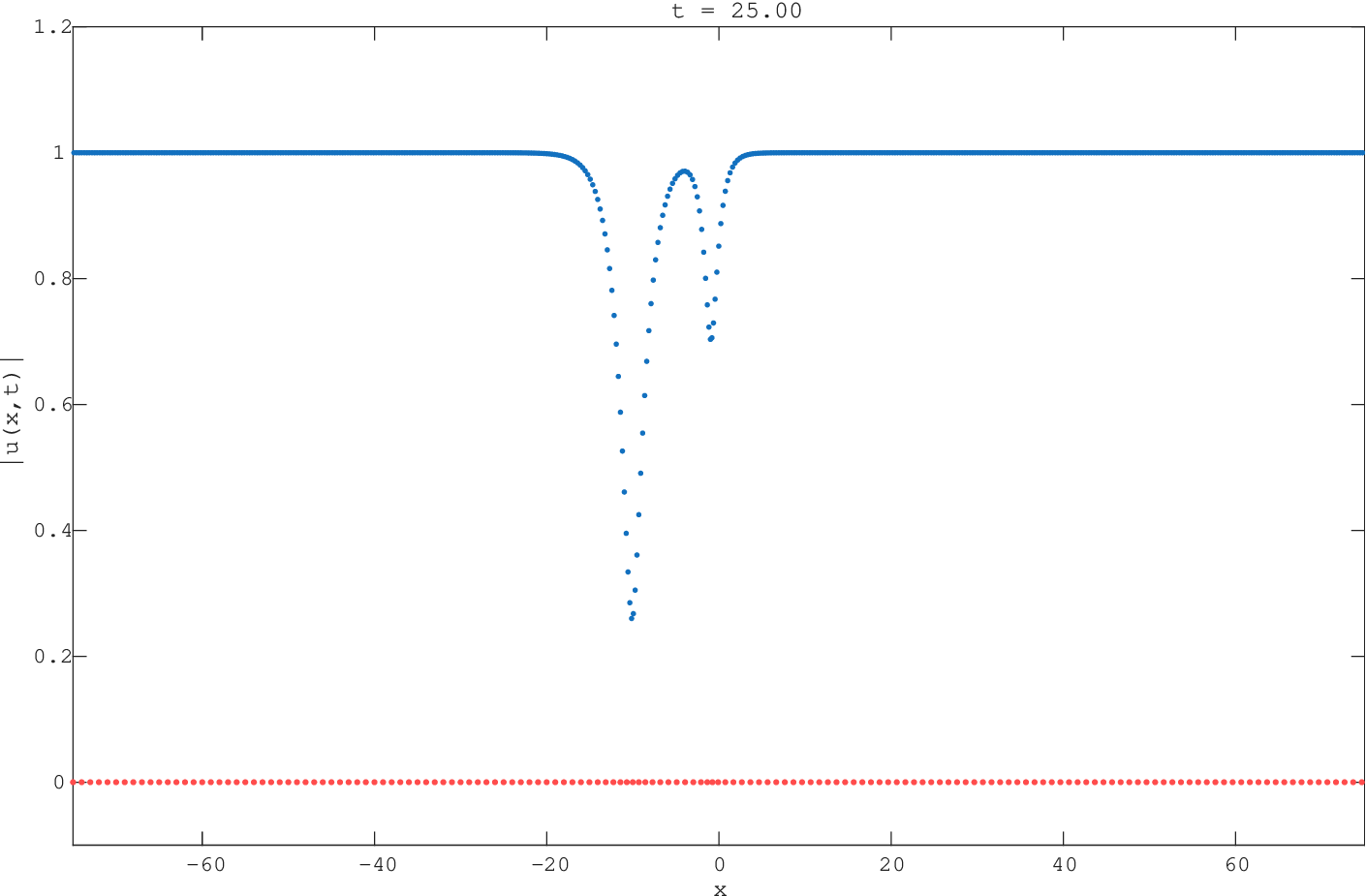}
      \end{minipage} &
      \begin{minipage}[t]{0.47\hsize}
        \centering
        \includegraphics[width=\linewidth]{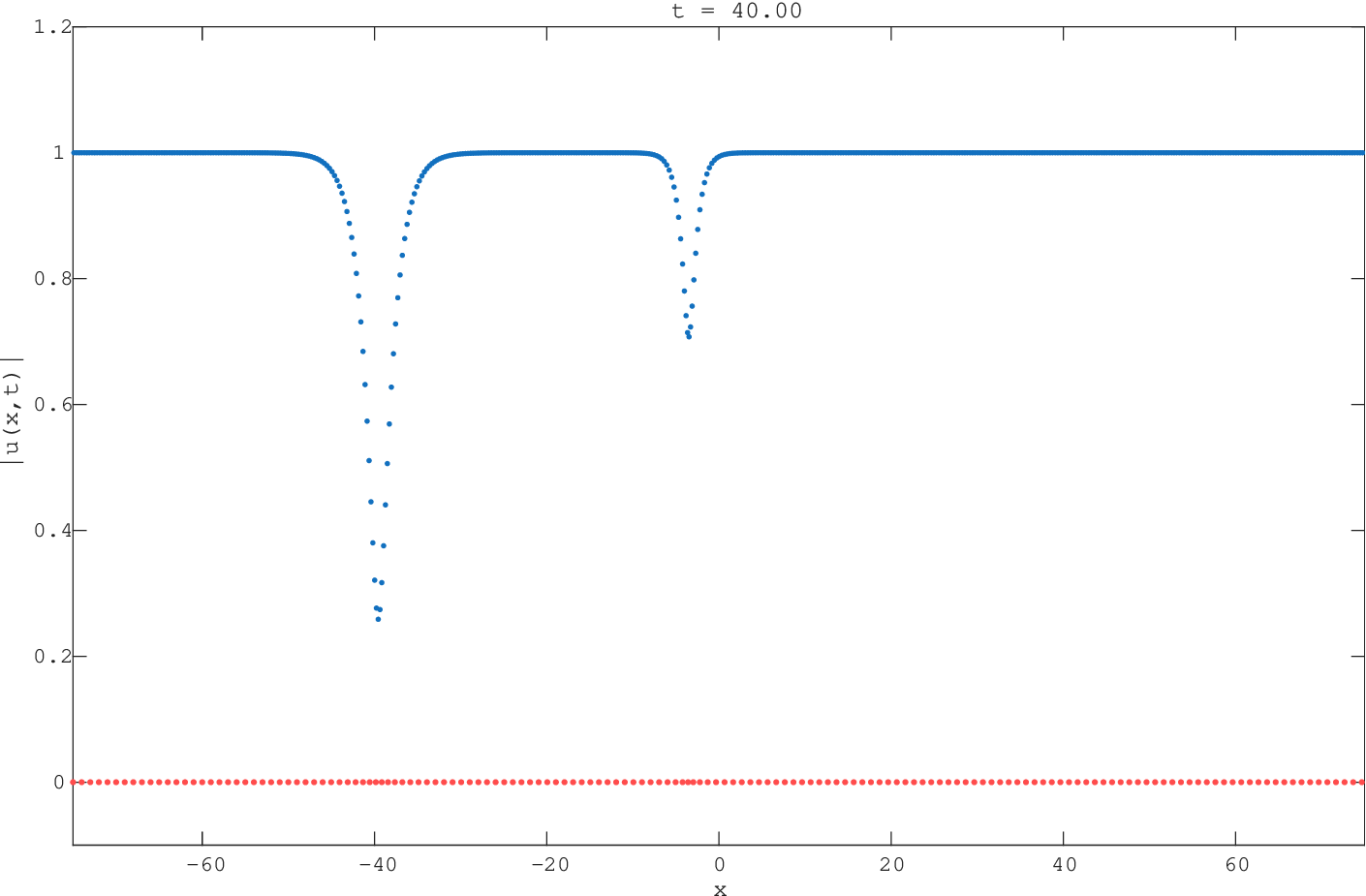}
      \end{minipage}
       \end{tabular}
     \caption{The numerical simulation of the $|u|$-profile of the two-dark-soliton solution for the defocusing CmSP equation. maxerr=1.95$\times 10^{-5}$}
              \label{dmSP_2}
  \end{figure}

Figure \ref{dmSP_1} shows a numerical simulation of the $|u|$-profile of the one-dark-soliton solution for $s=2, \theta_{1}=-7\pi/12$, and $X_{1,0}=40$. Figure \ref{dmSP_2} shows a numerical simulation of the $|u|$-profile of the two-dark-soliton solution for $s=2, \theta_{1}=-7\pi/12, \theta_{2}=-3\pi/4, X_{1,0}=40$, and $X_{2,0}=4$. The background plane wave is linearly stable for any $s>0$.

In both Subsections~\ref{subsec:CCmSPperiodic} and \ref{subsec:defocusingCmSPDirichlet}, the maxerr for the two-soliton solution, which involves soliton interaction, does not increase significantly compared with that for the one-soliton solution. This fact suggests that the self-adaptive moving mesh schemes (\ref{semidis2mSP}), (\ref{semidiscomplexmSP}) and (\ref{semidis_defocusingCmSP}) are effective and accurate even in the presence of soliton interaction.

\end{section}


\begin{section}{Bilinear equations and $N$-soliton solution for the MCSP equation}
\label{sec_MCSP}
In this section, we describe the bilinear form and the hodograph transformation of the MCSP equation (\ref{MCSP}) and construct its $N$-soliton solution following the previous study by Feng et al.~\cite{semidis-MCSP}.
First, we rewrite (\ref{MCSP}) in the conservation law form
\begin{eqnarray}
\left(\frac{1}{\rho}\right)_{t}-\left(\frac{F}{2\rho}\right)_{x}=0,\label{conservationlaw}
\end{eqnarray}
with
\begin{eqnarray}
\rho=\left(1+\sum_{1\leq j<k\leq n}c_{jk}u^{(j)}_{x}u^{(k)}_{x}\right)^{-\frac{1}{2}},\qquad F=\sum_{1\leq j<k\leq n}c_{jk}u^{(j)}u^{(k)},\label{conservationlaw2}
\end{eqnarray}
which determines the hodograph transformation
\begin{eqnarray}
dX=\frac{1}{\rho} dx+\frac{F}{2\rho}dt,\qquad dT=dt.\label{hodograph}
\end{eqnarray} 
The conservation law (\ref{conservationlaw}) is the closedness condition for this one-form. From the  hodograph transformation (\ref{hodograph}), we obtain the derivative law
\begin{eqnarray}
\frac{\partial}{\partial X}=\rho \frac{\partial}{\partial x},\qquad\frac{\partial}{\partial T}=\frac{\partial }{\partial t}-\frac{F}{2}\frac{\partial }{\partial x}.\label{derivativelaw}
\end{eqnarray}
Next, transforming the MCSP equation into
\begin{eqnarray}
\partial_{x}\left(\partial_{t}-\frac{1}{2}\sum_{1\leq j<k \leq n}c_{jk}u^{(j)}u^{(k)}\partial_{x}\right)u^{(i)}=u^{(i)},\label{6}
\end{eqnarray}
 and applying the derivative law (\ref{derivativelaw}) and (\ref{conservationlaw2}) to (\ref{6}), we obtain
 \begin{eqnarray}
 u^{(i)}_{XT}=u^{(i)}\rho.
 \end{eqnarray}
Furthermore, applying the derivative law (\ref{derivativelaw})  to the conservation law (\ref{conservationlaw}), we then have
 \begin{eqnarray}
\rho_{T}+\frac{1}{2}\left(\sum_{1\leq j<k\leq n}c_{jk}u^{(j)}u^{(k)}\right)_{X}=0.\label{conservationlawXT}
 \end{eqnarray}
The equation (\ref{conservationlawXT}) is the conservation law for $X$ and $T$, where $\rho$ is the conserved density and $\frac{1}{2}\sum_{1\leq j<k\leq n}c_{jk}u^{(j)}u^{(k)}$ is the flux.

Therefore, the MCSP equation (\ref{MCSP}) is transformed into the MCCID equations
\begin{eqnarray}
\left\{
\begin{array}{ll}
 u^{(i)}_{XT}=u^{(i)}\rho,\\
   \rho_{T}+\displaystyle\frac{1}{2}\left(\displaystyle\sum_{1\leq j<k\leq n}c_{jk}u^{(j)}u^{(k)}\right)_{X}=0,
 \end{array}
\right.
\label{dispersionlessMCSP}
\end{eqnarray}
 by the hodograph transformation (\ref{hodograph}).

Introducing the dependent variable transformation
\begin{eqnarray}
u^{(i)}=\frac{g^{(i)}}{f}\quad(i=1,2,\cdots,n),\qquad\rho=1-2({\rm log}f)_{XT},\label{dependenttransformation}
\end{eqnarray}
 the MCCID equations
 (\ref{dispersionlessMCSP}) lead to
\begin{eqnarray}
\left\{
\begin{array}{ll}
D_{X}D_{T}f\cdot g^{(i)}=fg^{(i)},\\
  D_{T}^{2}f\cdot f=\displaystyle\frac{1}{2}\displaystyle\sum_{1\leq j<k\leq n}c_{jk}g^{(j)}g^{(k)}.
\end{array}
\right.
\label{bilinearMCSP}
\end{eqnarray}

\begin{lem}
(Feng-Maruno-Ohta~\cite{semidis-MCSP}). The bilinear equations (\ref{bilinearMCSP}) have the following Pfaffian solution:
\begin{eqnarray}
f={\rm Pf}(a_{1},\cdots,a_{2N},b_{1},\cdots,b_{2N}),\qquad g^{(i)}={\rm Pf}(d_{0},B_{i},a_{1},\cdots,a_{2N},b_{1},\cdots,b_{2N}),
\end{eqnarray}
where $i=1,2,\cdots,n$ and the elements of the Pfaffians are defined as
\begin{eqnarray}
{\rm Pf}(a_{j},a_{k})=\frac{p_{j}-p_{k}}{p_{j}+p_{k}}e^{\xi_{j}+\xi_{k}},\qquad
{\rm Pf}(a_{j},b_{k})=\delta_{j,k},\label{pfMCSP1}
\end{eqnarray}
\begin{eqnarray}
{\rm Pf}(b_{j},b_{k})=\frac{1}{4}\frac{c_{\mu\nu}}{p_{j}^{-2}-p_{k}^{-2}}\quad(b_{j}\in B_{\mu},b_{k}\in B_{\nu} ),
\end{eqnarray}
\begin{eqnarray}
{\rm Pf}(d_{l},a_{k})=p_{k}^{l}e^{\xi_{k}},
\qquad
{\rm Pf}(b_{j},B_{\mu})=
\left\{
\begin{array}{ll}
1 & (b_{j}\in B_{\mu}) \\
0 & (b_{j}\notin B_{\mu})
\end{array}
\right..
\label{pfMCSP2}
\end{eqnarray}
Here, $j,k =1,2,\cdots,2N$,\quad$\mu,\nu=1,2,\cdots,n,\quad\xi_{j}=p_{j}X+p_{j}^{-1}T+\xi_{j0}$, $p_j$ and $\xi_{j0}$ are arbitrary constants, $\delta_{j,k}$ denotes the Kronecker delta, and $l$ is an integer. As in Lemma \ref{lemma1}, formal symbols $d
_{l}$ with different integer values of $l$ are introduced to differentiate Pfaffian expressions, although $g^{(i)}$ contains only $d_{0}$. A class of set $B_{\mu}(\mu=1,2,\cdots,n) $ satisfies the following conditions
\begin{eqnarray}
B_{\mu}\cap B_{\nu}=\varnothing\,\,{\rm if}\,\, \mu\neq\nu,\qquad\cup_{\mu=1}^{n}B_{\mu}=\{b_{1},b_{2},\cdots,b_{2N}\}.
\end{eqnarray}
All other entries of the Pfaffians, except those obtained from the above by antisymmetry, are zero.
\end{lem}

\begin{proof}
See~\cite{semidis-MCSP}.
\end{proof}
\end{section}

\begin{example}\label{ex2}
\leavevmode
\\
1-soliton: For $N=1, n=2, B_{1}=\{b_{1}\}, B_{2}=\{b_{2}\}, c_{12}=1$, we obtain the $\tau$-functions $f$, $g^{(1)}$ and $g^{(2)}$ as follows:
\begin{align}
&f={\rm Pf}(a_{1},\cdots,a_{2N},b_{1},\cdots,b_{2N})=\frac{1}{4}\frac{p_{1}-p_{2}}{p_{1}+p_{2}}\frac{1}{p_{1}^{-2}-p_{2}^{-2}}e^{\xi_{1}+\xi_{2}}-1=-1-\frac{b_{12}}{4}e^{\xi_{1}+\xi_{2}},\nonumber\\
&g^{(1)}={\rm Pf}(d_{0},B_{1},a_{1},\cdots,a_{2N},b_{1},\cdots,b_{2N})=-e^{\xi_{1}},\nonumber\\
&g^{(2)}={\rm Pf}(d_{0},B_{2},a_{1},\cdots,a_{2N},b_{1},\cdots,b_{2N})=-e^{\xi_{2}},\nonumber
\end{align}
where $b_{jk}=\left(\frac{p_{j}p_{k}}{p_{j}+p_{k}}\right)^{2}$ and $\xi_{j}=p_{j}X+\frac{1}{p_{j}}T+\xi_{j0}$.

Letting $\xi_{j0}=\xi_{j0}^{\prime}+\log{a_{j}}$, the $\tau$-functions can be rewritten as
\begin{eqnarray}
f=-1-\frac{a_{1}a_{2}b_{12}}{4}e^{\eta_{1}+\eta_{2}},\qquad
g^{(1)}=-a_{1}e^{\eta_{1}},\qquad g^{(2)}=-a_{2}e^{\eta_{2}},\label{gg21}
\end{eqnarray}
where $b_{jk}=\left(\frac{p_{j}p_{k}}{p_{j}+p_{k}}\right)^{2}$ and $\eta_{j} =p_{j}X+\frac{1}{p_{j}}T+\xi_{j0}^{\prime}$.\\
\\
2-soliton: For $N=2, n=2, B_{1}=\{b_{1},b_{2}\}, B_{2}=\{b_{3}, b_{4}\}, c_{12}=1$. By the expansion rule, we obtain the $\tau$-functions $f$, $g^{(1)}$ and $g^{(2)}$ as follows:
\begin{align}
f&={\rm Pf}(a_{1},\cdots,a_{2N},b_{1},\cdots,b_{2N})\nonumber\\
&=1+\frac{b_{13}}{4}e^{\xi_{1}+\xi_{3}}
  +\frac{b_{23}}{4}e^{\xi_{2}+\xi_{3}}
  +\frac{b_{14}}{4}e^{\xi_{1}+\xi_{4}}
  +\frac{b_{24}}{4}e^{\xi_{2}+\xi_{4}}\nonumber\\
&\quad +(p_{1}-p_{2})^{2}(p_{3}-p_{4})^{2}
  \frac{b_{13}b_{23}b_{14}b_{24}}{16p_{1}^{2}p_{2}^{2}p_{3}^{2}p_{4}^{2}}
  e^{\xi_{1}+\xi_{2}+\xi_{3}+\xi_{4}},\nonumber\\
g^{(1)}&={\rm Pf}(d_{0},B_{1},a_{1},\cdots,a_{2N},b_{1},\cdots,b_{2N})\nonumber\\
&=e^{\xi_{1}}+e^{\xi_{2}}
 +\frac{(p_{1}-p_{2})^{2}p_{3}^{4}}
        {4(p_{1}+p_{3})^{2}(p_{2}+p_{3})^{2}}
  e^{\xi_{1}+\xi_{2}+\xi_{3}}\nonumber\\
&\quad+\frac{(p_{1}-p_{2})^{2}p_{4}^{4}}
        {4(p_{1}+p_{4})^{2}(p_{2}+p_{4})^{2}}
  e^{\xi_{1}+\xi_{2}+\xi_{4}},\nonumber\\
g^{(2)}&={\rm Pf}(d_{0},B_{2},a_{1},\cdots,a_{2N},b_{1},\cdots,b_{2N})\nonumber\\
&=e^{\xi_{3}}+e^{\xi_{4}}
 +\frac{(p_{3}-p_{4})^{2}p_{2}^{4}}
        {4(p_{2}+p_{3})^{2}(p_{2}+p_{4})^{2}}
  e^{\xi_{2}+\xi_{3}+\xi_{4}}\nonumber\\
&\quad+\frac{(p_{3}-p_{4})^{2}p_{1}^{4}}
        {4(p_{1}+p_{3})^{2}(p_{1}+p_{4})^{2}}
  e^{\xi_{1}+\xi_{3}+\xi_{4}},\nonumber
\end{align}
where $\xi_{j}=p_{j}X+\frac{1}{p_{j}}T+\xi_{j0}$.

Letting $\xi_{j0}=\xi_{j0}^{\prime}+\log{a_{j}}$, the $\tau$-functions can be rewritten as
\begin{align}
f&=1+\frac{a_{1}a_{3}b_{13}}{4}e^{\eta_{1}+\eta_{3}}
  +\frac{a_{2}a_{3}b_{23}}{4}e^{\eta_{2}+\eta_{3}}
  +\frac{a_{1}a_{4}b_{14}}{4}e^{\eta_{1}+\eta_{4}}\nonumber\\
&\quad+\frac{a_{2}a_{4}b_{24}}{4}e^{\eta_{2}+\eta_{4}}
  +a_{1}a_{2}a_{3}a_{4}(p_{1}-p_{2})^{2}(p_{3}-p_{4})^{2}
   \frac{b_{13}b_{23}b_{14}b_{24}}{16p_{1}^{2}p_{2}^{2}p_{3}^{2}p_{4}^{2}}
   e^{\eta_{1}+\eta_{2}+\eta_{3}+\eta_{4}},\label{ff2}\\
g^{(1)}&=a_{1}e^{\eta_{1}}+a_{2}e^{\eta_{2}}
 +\frac{a_{1}a_{2}a_{3}(p_{1}-p_{2})^{2}p_{3}^{4}}
        {4(p_{1}+p_{3})^{2}(p_{2}+p_{3})^{2}}
  e^{\eta_{1}+\eta_{2}+\eta_{3}}\nonumber\\
&\quad+\frac{a_{1}a_{2}a_{4}(p_{1}-p_{2})^{2}p_{4}^{4}}
        {4(p_{1}+p_{4})^{2}(p_{2}+p_{4})^{2}}
  e^{\eta_{1}+\eta_{2}+\eta_{4}},\label{gg12}\\
g^{(2)}&=a_{3}e^{\eta_{3}}+a_{4}e^{\eta_{4}}
 +\frac{a_{2}a_{3}a_{4}(p_{3}-p_{4})^{2}p_{2}^{4}}
        {4(p_{2}+p_{3})^{2}(p_{2}+p_{4})^{2}}
  e^{\eta_{2}+\eta_{3}+\eta_{4}}\nonumber\\
&\quad+\frac{a_{1}a_{3}a_{4}(p_{3}-p_{4})^{2}p_{1}^{4}}
        {4(p_{1}+p_{3})^{2}(p_{1}+p_{4})^{2}}
  e^{\eta_{1}+\eta_{3}+\eta_{4}},\label{gg22}
\end{align}
where $b_{jk}=\left(\frac{p_{j}p_{k}}{p_{j}+p_{k}}\right)^{2}$ and $\eta_{j} =p_{j}X+\frac{1}{p_{j}}T+\xi_{j0}^{\prime}$.

These solutions  (\ref{gg21})--(\ref{gg22}) can also be obtained from Hirota's direct method.
\end{example}


\begin{section}{A self-adaptive moving mesh scheme for the MCSP equation}
\label{sec_disMCSP}
In~\cite{semidis-MCSP}, we proposed a semi-discrete analogue of the bilinear equations (\ref{bilinearMCSP})
\begin{eqnarray}
\left\{
\begin{array}{ll}
  \displaystyle\frac{1}{a}D_{T}(f_{l+1}\cdot g_{l}^{(i)}-f_{l}\cdot g_{l+1}^{(i)})=g_{l+1}^{(i)}f_{l}+g_{l}^{(i)}f_{l+1},\quad i=1,2,\cdots,n,\\
  D_{T}^{2}f_{l}\cdot f_{l}=\displaystyle\frac{1}{2}\displaystyle\sum_{1\leq j<k\leq n}c_{jk}g_{l}^{(j)}g_{l}^{(k)}.
\end{array}
\right.
\label{disbilinearMCSP}
\end{eqnarray}
Here, $2a$ is the parameter related to a discrete interval.


\begin{lem}
(Feng-Maruno-Ohta~\cite{semidis-MCSP}). The bilinear equations (\ref{disbilinearMCSP}) have the following Pfaffian solution:
\begin{eqnarray}
f_{l}={\rm Pf}(a_{1},\cdots,a_{2N},b_{1},\cdots,b_{2N})_{l},\qquad g_{l}^{(i)}={\rm Pf}(d_{0},B_{i},a_{1},\cdots,a_{2N},b_{1},\cdots,b_{2N})_{l},
\label{dis_g_MCSP}
\label{dis_f_MCSP}
\end{eqnarray}
where $i=1,2,\cdots,n$ and the elements of the Pfaffians are defined as
\begin{eqnarray}
{\rm Pf}(a_{j},a_{k})_{l}=\frac{p_{j}-p_{k}}{p_{j}+p_{k}}\phi_{j}^{(0)}(l)\phi_{k}^{(0)}(l),\qquad{\rm Pf}(a_{j},b_{k})_{l}=\delta_{j,k},
\label{dpfMCSP1}
\end{eqnarray}
\begin{eqnarray}
{\rm Pf}(b_{j},b_{k})_{l}=\frac{1}{4}\frac{c_{\mu\nu}}{p_{j}^{-2}-p_{k}^{-2}}\quad(b_{j}\in B_{\mu},b_{k}\in B_{\nu} ),
\end{eqnarray}
\begin{eqnarray}
{\rm Pf}(d_{m},a_{j})_{l}=\phi_{j}^{(m)}(l),\qquad{\rm Pf}(d^{l},a_{j})_{l}=\phi_{j}^{(0)}(l+1),
\end{eqnarray}
\begin{eqnarray}
{\rm Pf}(b_{k},B_{\mu})_{l}=
\left\{
\begin{array}{ll}
 1 & (b_{k}\in B_{\mu}),\\
 0 & (b_{k}\notin B_{\mu}) ,
\end{array}
\right.
\end{eqnarray}
\begin{eqnarray}
{\rm Pf}(d_{0},d^{l})_{l}=1,\qquad{\rm Pf}(d_{-1},d^{l})_{l}=-a,
\label{dpfMCSP2}
\end{eqnarray}
where
\begin{eqnarray}
\phi_{j}^{(n)}(l)=p_{j}^{n}\left(\frac{1+ap_{j}}{1-ap_{j}}\right)^{l}e^{p_{j}^{-1}T+\xi_{j0}},
\end{eqnarray}
and it satisfies the following equation
\begin{eqnarray}
\frac{\phi_{j}^{(n)}(l+1)-\phi_{j}^{(n)}(l)}{a}=\phi_{j}^{(n+1)}(l+1)+\phi_{j}^{(n+1)}(l).
\end{eqnarray}
All other entries of the Pfaffians, except those obtained from the above by antisymmetry, are zero.
\end{lem}
\begin{proof}
See~\cite{semidis-MCSP}.
\end{proof}


\begin{lem}
A semi-discrete analogue of the MCCID equations (\ref{dispersionlessMCSP})
\begin{eqnarray}
\left\{
\begin{array}{ll}
  \displaystyle\frac{d}{dT}(u_{l+1}^{(i)}-u_{l}^{(i)})=a\rho_{l}(u_{l+1}^{(i)}+u_{l}^{(i)}),\\
  \displaystyle\frac{d}{dT}\rho_{l}=-\displaystyle\frac{1}{2a}\left(\displaystyle\frac{1}{2}\sum_{1\leq j<k\leq n}c_{jk}(u_{l+1}^{(j)}u_{l+1}^{(k)}-u_{l}^{(j)}u_{l}^{(k)})\right),
\end{array}
\right.
\label{semidisdispersionlessMCSP}
\end{eqnarray}
is obtained from the bilinear equations (\ref{disbilinearMCSP}) through the dependent variable transformation
\begin{eqnarray}
u_{l}^{(i)}=\frac{g_{l}^{(i)}}{f_{l}},\qquad
\rho_{l}=1-\frac{1}{a}\left(\log{\frac{f_{l+1}}{f_{l}}}\right)_{T}.
\label{disdependenttransformationMCSP}
\end{eqnarray}

\end{lem}

\begin{proof}
First, dividing both sides of the first equation of the bilinear equations (\ref{disbilinearMCSP}) by $f_{l+1}f_{l}$, we obtain
\begin{eqnarray}
\left(\frac{g_{l+1}^{(i)}}{f_{l+1}}-\frac{g_{l}^{(i)}}{f_{l}}\right)_{T}=\left(a-\left(\log{\frac{f_{l+1}}{f_{l}}}\right)_{T}\right)\left(\frac{g_{l+1}^{(i)}}{f_{l+1}}+\frac{g_{l}^{(i)}}{f_{l}}\right).
\end{eqnarray}
Hence, we obtain
\begin{eqnarray}
\frac{d}{dT}(u_{l+1}^{(i)}-u_{l}^{(i)})=a\rho_{l}(u_{l+1}^{(i)}+u_{l}^{(i)}),
\end{eqnarray}
by the dependent variable transformation (\ref{disdependenttransformationMCSP}).

Next, from the dependent variable transformation (\ref{disdependenttransformationMCSP}), the second equation of the bilinear equations (\ref{disbilinearMCSP}) is rewritten as
\begin{eqnarray}
(\log{f_{l}})_{TT}=\frac{1}{4}\sum_{1\leq j<k\leq n}c_{jk}u_{l}^{(j)}u_{l}^{(k)}.\label{aux1MCSP}
\end{eqnarray}
By (\ref{aux1MCSP}) and the dependent variable transformation (\ref{disdependenttransformationMCSP}), we obtain
\begin{eqnarray}
\frac{d}{dT}\rho_{l}=-\frac{1}{2a}\left(\frac{1}{2}\sum_{1\leq j<k\leq n}c_{jk}(u_{l+1}^{(j)}u_{l+1}^{(k)}-u_{l}^{(j)}u_{l}^{(k)})\right).
\end{eqnarray}
\end{proof}

\begin{lem}
A semi-discrete analogue of the MCSP equation is of the form:
\begin{eqnarray}
\left\{
\begin{array}{ll}
  \displaystyle\frac{d} {dT}(u_{l+1}^{(i)}-u_{l}^{(i)})= \displaystyle\frac{\delta_{l}}{2}(u_{l+1}^{(i)}+u_{l}^{(i)}),\\
  \\
 \displaystyle\frac{d\delta_{l}}{dT}=-\displaystyle\frac{1}{2}\sum_{1\leq j<k \leq n}c_{jk}(u_{l+1}^{(j)}u_{l+1}^{(k)}-u_{l}^{(j)}u_{l}^{(k)}),\\
 x_{l}=x_{0}+ \displaystyle\sum_{m=0}^{l-1}\delta_{m},\quad t=T.
\end{array}
\right.
\label{semidisMCSP}
\end{eqnarray}
Here, $u_{l}^{(i)}$ is the value at $x_{l}$.
\end{lem}

\begin{proof}
From the derivative law (\ref{derivativelaw}), we obtain
\begin{eqnarray}
\frac{\partial x}{\partial X}=\rho.\label{3.19}
\end{eqnarray}
Integrating (\ref{3.19}) with respect to $X$, we have the integral form of the hodograph transformation (\ref{hodograph})
\begin{eqnarray}
x=\int\rho(X,T)dX=x_{0}+\int_{0}^{X}\rho(\bar{X},T)d\bar{X},
\label{9}
\end{eqnarray}
where $x_{0}$ is the integration constant with respect to $X$, determined by the value at the left-hand edge $(x=x_{0})$; it may depend on $T$ when the boundary flux is nonzero.
Discretizing (\ref{9}), we then have the discrete hodograph transformation
\begin{eqnarray}
x_{l}=x_{0}+\sum_{m=0}^{l-1}2a\rho_{m},
\label{dishodographMCSP}
\end{eqnarray}
where $X_l=2al$ is the uniform lattice coordinate in the $X$ variable and $ \rho_{l}\equiv\rho(X_{l},T)$.

Introducing a mesh interval
\begin{eqnarray}
\delta_{l}:=x_{l+1}-x_{l},
\end{eqnarray}
we obtain
\begin{eqnarray}
\delta_{l}=2a\rho_{l},\label{meshintMCSP}
\end{eqnarray}
by substituting the discrete hodograph transformation (\ref{dishodographMCSP}).
Substituting  (\ref{meshintMCSP}) into the semi-discrete MCCID equations (\ref{semidisdispersionlessMCSP}), we have a semi-discrete analogue of the MCSP equation (\ref{semidisMCSP}).
\end{proof}

\begin{remark}
Equation (\ref{dishodographMCSP}) is the summation form of the discrete hodograph transformation. At fixed $T$, its local form is
\begin{eqnarray}
\Delta x_{l}=\rho_{l}\Delta X_{l},\qquad \Delta X_l=2a.
\label{ddishodographMCSP}
\end{eqnarray}
Here $\Delta x_l=x_{l+1}-x_l$. Dividing (\ref{ddishodographMCSP}) by $\Delta X_l$ gives
\begin{eqnarray}
\frac{\Delta x_{l}}{\Delta X_{l}}=\rho_{l},
\end{eqnarray}
and summing this relation from $0$ to $l-1$ recovers (\ref{dishodographMCSP}). The second equation of (\ref{semidisdispersionlessMCSP}), which is the discrete version of the conservation law (\ref{conservationlawXT}), gives the compatible time evolution of the mesh interval,
\begin{eqnarray}
\frac{d}{dT}\Delta x_l
=-\frac{1}{2}\sum_{1\leq j<k\leq n}c_{jk}
(u_{l+1}^{(j)}u_{l+1}^{(k)}-u_l^{(j)}u_l^{(k)}).
\end{eqnarray}
Since the hodograph transformation (\ref{dishodographMCSP}) determines the mesh points only up to the additive term $x_0(T)$, the evolution of the mesh points cannot be determined until the evolution of the edge point $x_0$ is specified. Theorem~\ref{thm:MCSP} specifies the evolution of the edge point $x_0$, thereby determining the evolution of the mesh points.
\end{remark}

\begin{theorem}
\label{thm:MCSP}
Assume that the edge point is allowed to move and satisfies the consistency condition with the hodograph transformation, namely $dx_{0}/dT=-(1/2)\sum_{1\leq j<k\leq n}c_{jk}u_{0}^{(j)}u_{0}^{(k)}$. Then the semi-discrete system (\ref{semidisMCSP}) can be rewritten as the following self-adaptive moving mesh scheme for the MCSP equation:
\begin{eqnarray}
\left\{
\begin{array}{ll}
\displaystyle\frac{d}{dT}(u_{l+1}^{(i)}-u_{l}^{(i)})=\displaystyle\frac{\delta_{l}}{2}(u_{l+1}^{(i)}+u_{l}^{(i)}), \\
\\
\displaystyle\frac{d x_{l}}{d T}=-\displaystyle\frac{1}{2}\sum_{1\leq j<k \leq n}c_{jk}u_{l}^{(j)}u_{l}^{(k)}.
\end{array}
\right.
\label{3.24}
\end{eqnarray}
Here, $u_{l}^{(i)}$ is the value at $x_{l}$. 
\end{theorem}

\begin{proof}
As in the MCmSP case, we first recall the continuous calculation only to identify the corresponding moving-edge consistency condition.
From the derivative law (\ref{derivativelaw}), we have the evolution equation for $x$
\begin{eqnarray}
\frac{\partial x}{\partial T}=-\frac{1}{2}\displaystyle{\sum_{1\leq j < k \leq n}c_{jk}u^{(j)}u^{(k)}}.
\label{SP_evol}
\end{eqnarray}
Equation (\ref{SP_evol}) means that the edge point $x_0(T)=x(0,T)$ is fixed only when the boundary flux $\sum_{1\leq j<k\leq n}c_{jk}u^{(j)}(0,T)u^{(k)}(0,T)/2$ vanishes. When the boundary flux is nonzero, $x_{0}$ must evolve with $T$, and (\ref{9}) can be rewritten as
\begin{eqnarray}
x=x_{0}(T)+\int_{0}^{X}\rho(\bar{X},T)d\bar{X}.
\label{9999}
\end{eqnarray}
The following calculation derives the same edge-point condition from the integral representation.
Differentiating the integral form of the hodograph transformation (\ref{9999}) with respect to $T$, we have
\begin{align}
\frac{\partial x}{\partial T}&=\frac{\partial x_{0}(T)}{\partial T}+\frac{\partial}{\partial T}\int_{0}^{X}\rho(\bar{X},T)d\bar{X}\nonumber\\
&=\frac{\partial x_{0}(T)}{\partial T}-\int_{0}^{X}\frac{\partial}{\partial \bar{X}}\left(\frac{1}{2}\sum_{1\leq j<k\leq n}c_{jk}u^{(j)}(\bar{X},T)u^{(k)}(\bar{X},T)\right)d\bar{X}\nonumber\\
&=\frac{\partial x_{0}(T)}{\partial T}
 -\frac{1}{2}\displaystyle{\sum_{1\leq j <k \leq n}}c_{jk}u^{(j)}(X,T)u^{(k)}(X,T)\nonumber\\
&\quad
 +\frac{1}{2}\displaystyle{\sum_{1\leq j<k \leq n}}c_{jk}u^{(j)}(0,T)u^{(k)}(0,T).
\label{0.26}
\end{align}
The calculation from the first to the second line uses the conservation law (\ref{conservationlawXT}).
Comparing (\ref{SP_evol}) with (\ref{0.26}), the $X$-dependent flux terms agree, and the remaining boundary term gives the evolution equation for the edge point $x_{0}$:
\begin{eqnarray}
\frac{\partial x_{0}(T)}{\partial T}=-\frac{1}{2}\displaystyle{\sum_{1\leq j < k \leq n}}c_{jk}u^{(j)}(0,T)u^{(k)}(0,T).
\label{SP_const}
\end{eqnarray}
We call (\ref{SP_const}) the consistency condition with the hodograph transformation.

We now turn to the semi-discrete case, which gives the actual proof of the theorem. Differentiating the summation form of the discrete hodograph transformation (\ref{dishodographMCSP}) with respect to $T$ and using the second equation of (\ref{semidisdispersionlessMCSP}), we obtain
\begin{align}
\frac{d x_{l}}{d T}&=\frac{d x_{0}(T)}{d T}+\frac{d}{d T}\sum_{m=0}^{l-1}2a\rho_{m}=\frac{d x_{0}(T)}{d T}+\sum_{m=0}^{l-1}2a\frac{d \rho_{m}}{d T}\nonumber\\
&=\frac{d x_{0}(T)}{d T}-\sum_{m=0}^{l-1}\frac{1}{2}\sum_{1\leq j<k \leq n}c_{jk}(u_{m+1}^{(j)}u_{m+1}^{(k)}-u_{m}^{(j)}u_{m}^{(k)})\nonumber\\
&=\frac{d x_{0}(T)}{d T}+\frac{1}{2}\sum_{1\leq j<k \leq n}c_{jk}(-u_{l}^{(j)}u_{l}^{(k)}+u_{0}^{(j)}u_{0}^{(k)}).
\label{23}
\end{align}
Equation (\ref{23}) corresponds to (\ref{0.26}) in the continuous case. As in the continuous case, we obtain the consistency condition, i.e., the evolution equation for the edge point $x_{0}$:
\begin{eqnarray}
\frac{dx_{0}(T)}{dT}=-\frac{1}{2}\displaystyle{\sum_{1\leq j < k \leq n}}c_{jk}u_{0}^{(j)}u_{0}^{(k)}.
\label{disSP_const}
\end{eqnarray}
Substituting (\ref{disSP_const}) into (\ref{23}) yields
\begin{eqnarray}
\frac{d x_{l}}{d T}=-\frac{1}{2}\sum_{1\leq j<k \leq n}c_{jk}u_{l}^{(j)}u_{l}^{(k)}.
\end{eqnarray}
\end{proof}

\begin{remark}
 The integral form of the hodograph transformation (\ref{9}) and the summation form of the discrete hodograph transformation (\ref{dishodographMCSP}) can be expressed by the $\tau$-function, with an appropriate choice of $x_{0}(T)$, as
 \begin{eqnarray}
 x=X-2(\log{f})_{T},
 \end{eqnarray}
 \begin{eqnarray}
 x_{l}=2al-2(\log{f_{l}})_{T}.
  \end{eqnarray}
See Appendix~\ref{secA2} for details.
\end{remark}

Dividing the first equation of (\ref{3.24}) by $\delta_{l}$, one arrives at
\begin{eqnarray}
\frac{(u_{l+1}^{(i)}-u_{l}^{(i)})_{T}}{\delta_{l}}=\frac{u_{l+1}^{(i)}+u_{l}^{(i)}}{2}.\label{3.38}
\end{eqnarray}
As in the MCmSP case, $u_{l+1}^{(i)}-u_l^{(i)}=\delta_l u_x^{(i)}+O(\delta_l^2)$ and
$u_{l+1}^{(i)}+u_l^{(i)}=2u^{(i)}+O(\delta_l)$. Hence, before transforming back to the Eulerian variables, the continuous limit $a\rightarrow0$ and $\delta_l\rightarrow0$ of (\ref{3.38}) is
\begin{eqnarray}
u^{(i)}_{Tx}=u^{(i)}.
\label{136}
\end{eqnarray}
Finally, the derivative with respect to $T$ at fixed lattice label is transformed back to
the Eulerian variables as
\begin{eqnarray}
\frac{d}{dT}=\frac{dx_{l}}{dT}\frac{d}{dx_{l}}+\frac{d}{dt}\rightarrow-\frac{1}{2}\sum_{1\leq j < k \leq n}c_{jk}u^{(j)}u^{(k)} \partial_{x}+\partial_{t}\quad\text{for}\ \quad a\rightarrow 0.
\label{137}
\end{eqnarray}
Substituting (\ref{137}) into (\ref{136}), equation (\ref{3.38}) converges to
\begin{eqnarray}
\partial_{x}\left(\partial_{t}-\frac{1}{2}\sum_{1\leq j<k \leq n}c_{jk}u^{(j)}u^{(k)}\partial_{x}\right)u^{(i)}=u^{(i)},
\end{eqnarray}
which is exactly the MCSP equation.
\end{section}

\begin{section}{Numerical simulations of the 2-SP and the CSP equations}
\label{sec_numexpMCSP}
\subsection{Numerical simulations of the 2-SP and the CSP equations under periodic boundary conditions}
\label{subsec:CCSPperiodic}
In this subsection, we perform numerical simulations of the 2-SP and the CSP equations under periodic boundary conditions.

The 2-SP equation
\begin{eqnarray}
\left\{
\begin{array}{ll}
u_{xt}=u+\displaystyle\frac{1}{2}(uvu_{x})_{x}, \\
\\
v_{xt}=v+\displaystyle\frac{1}{2}(uvv_{x})_{x},
\end{array}
\right.
\label{2SP}
\end{eqnarray}
is a special case of the MCSP equation (\ref{MCSP}) with $n=2, u^{(1)}=u, u^{(2)}=v$ and $c_{12}=1$.

Based on the results in the previous section, we have a self-adaptive moving mesh scheme for the 2-SP equation
\begin{eqnarray}
\left\{
\begin{array}{ll}
\displaystyle\frac{d }{d T}(u_{l+1}-u_{l})=\displaystyle\frac{\delta_{l}}{2}(u_{l+1}+u_{l}), \\
\\
\displaystyle\frac{d}{d T}(v_{l+1}-v_{l})=\displaystyle\frac{\delta_{l}}{2}(v_{l+1}+v_{l}),\\
\\
\displaystyle\frac{d x_{l}}{d T}=-\displaystyle\frac{1}{2}u_{l}v_{l}.
\end{array}
\right.
\label{semidis2SP}
\end{eqnarray}
Here, $u_{l}$ and $v_{l}$ are values at $x_{l}$.

Equation (\ref{semidis2SP}) admits the following $N$-soliton solution
\begin{eqnarray}
u_{l}=\frac{g_{l}^{(1)}}{f_{l}},\qquad v_{l}=\frac{g_{l}^{(2)}}{f_{l}},\qquad x_{l}=2al-2(\log{f_{l}})_{T},\qquad t=T,
\end{eqnarray}
with
\begin{eqnarray}
f_{l}&=&{\rm Pf}(a_{1},\cdots,a_{2N},b_{1},\cdots,b_{2N})_{l},\nonumber\\
g_{l}^{(i)}&=&{\rm Pf}(d_{0},B_{i},a_{1},\cdots,a_{2N},b_{1},\cdots,b_{2N})_{l},
\end{eqnarray}
where $i=1,2$ and the elements of the Pfaffians are defined as (\ref{dpfMCSP1})--(\ref{dpfMCSP2}).

 Here we show examples of numerical simulations of the 2-SP equation  (\ref{2SP}) with the self-adaptive moving mesh scheme (\ref{semidis2SP}) under periodic boundary conditions. The initial conditions are given by
\begin{align}
u_e&=\frac{g^{(1)}}{f},\quad v_e=\frac{g^{(2)}}{f},\qquad x_e=X-2(\log{f})_{T},\quad t=T,
\nonumber\\
f&=-1-\frac{a_{1}a_{2}b_{12}}{4}e^{\eta_{1}+\eta_{2}},\qquad g^{(1)}=-a_{1}e^{\eta_{1}},\qquad g^{(2)}=-a_{2}e^{\eta_{2}},\nonumber\\
\eta_{i}&=p_{i}X+\frac{1}{p_{i}}T+\xi_{i}^{\prime},\quad b_{ij}=\left(\frac{p_{i}p_{j}}{p_{i}+p_{j}}\right)^{2},\quad i,j=1,2,\label{incon_SP_1}
\end{align}
and
\begin{align}
&u_e=\frac{g^{(1)}}{f},\quad v_e=\frac{g^{(2)}}{f},\quad x_e=X-2(\log{f})_{T},\quad t=T,\nonumber\\
&f=1+\frac{a_{1}a_{3}b_{13}}{4}e^{\eta_{1}+\eta_{3}}
  +\frac{a_{2}a_{3}b_{23}}{4}e^{\eta_{2}+\eta_{3}}
  +\frac{a_{1}a_{4}b_{14}}{4}e^{\eta_{1}+\eta_{4}}\nonumber\\
&\quad+\frac{a_{2}a_{4}b_{24}}{4}e^{\eta_{2}+\eta_{4}}
  +a_{1}a_{2}a_{3}a_{4}(p_{1}-p_{2})^{2}(p_{3}-p_{4})^{2}
   \frac{b_{13}b_{23}b_{14}b_{24}}{16p_{1}^{2}p_{2}^{2}p_{3}^{2}p_{4}^{2}}
   e^{\eta_{1}+\eta_{2}+\eta_{3}+\eta_{4}},\nonumber\\
&g^{(1)}=a_{1}e^{\eta_{1}}+a_{2}e^{\eta_{2}}
 +\frac{a_{1}a_{2}a_{3}(p_{1}-p_{2})^{2}p_{3}^{4}}
        {4(p_{1}+p_{3})^{2}(p_{2}+p_{3})^{2}}
  e^{\eta_{1}+\eta_{2}+\eta_{3}}\nonumber\\
&\quad+\frac{a_{1}a_{2}a_{4}(p_{1}-p_{2})^{2}p_{4}^{4}}
        {4(p_{1}+p_{4})^{2}(p_{2}+p_{4})^{2}}
  e^{\eta_{1}+\eta_{2}+\eta_{4}},\nonumber\\
&g^{(2)}=a_{3}e^{\eta_{3}}+a_{4}e^{\eta_{4}}
 +\frac{a_{2}a_{3}a_{4}(p_{3}-p_{4})^{2}p_{2}^{4}}
        {4(p_{2}+p_{3})^{2}(p_{2}+p_{4})^{2}}
  e^{\eta_{2}+\eta_{3}+\eta_{4}}\nonumber\\
&\quad+\frac{a_{1}a_{3}a_{4}(p_{3}-p_{4})^{2}p_{1}^{4}}
        {4(p_{1}+p_{3})^{2}(p_{1}+p_{4})^{2}}
  e^{\eta_{1}+\eta_{3}+\eta_{4}},\nonumber\\
&\eta_{i}=p_{i}X+\frac{1}{p_{i}}T+\xi_{i}^{\prime},\quad b_{ij}=\left(\frac{p_{i}p_{j}}{p_{i}+p_{j}}\right)^{2},\quad i,j=1,2,3,4,\label{incon_SP_2}
\end{align}
which are exact one- and two-soliton solutions of the 2-SP equation obtained in {\it Example \ref{ex2}}.

As in Subsection~\ref{subsec:CCmSPperiodic}, we use the improved Euler method with the same parameters $L$, $D$, $\Delta t$ and lattice spacing $2a=D/L$, take the edge point at $X_0=-D/2$, impose the periodic boundary conditions (\ref{eq:periodicBC}), and define maxerr in the same manner. The only difference is that, in the scheme, $\delta_l-a$ is replaced by $\delta_l/2$
and $F_l$ by $F_l/2$; in particular, the constraint (\ref{eq:constraint}) becomes
$\sum_{l=0}^{L-1}\frac{\delta_l}{2}\bigl(u^{(i)}_{l+1}+u^{(i)}_{l}\bigr)=0$, which is
a trapezoidal-rule approximation of the zero-mean condition $\int u^{(i)}\,dx=0$ for periodic solutions (cf.~the related constraint on the real line in~\cite{Horikis}), and $du_0^{(i)}/dT$ is obtained from its time derivative in the same manner. In contrast to the MCmSP case, the denominators in the counterparts of (\ref{du0_dt}) and (\ref{projection}) become $\sum_{l=0}^{L-1}\delta_l=P_x>0$, so that no additional condition such as $P_x\neq aL$ is required.

Figure \ref{2SP_u_1} and Figure \ref{2SP_v_1} show numerical simulations of the {\it u}-profile and the {\it v}-profile of the one-soliton solution for $p_{1}=0.95,p_{2}=1.1, a_{1}=0.5, a_{2}=20$ and $\xi_{1}^{\prime}=\xi_{2}^{\prime}=25$.  Figure \ref{2SP_u} and Figure \ref{2SP_v} show numerical simulations of the {\it u}-profile and the {\it v}-profile of the two-soliton solution for $p_{1}=0.95,p_{2}=1.0, p_{3}=1.1,p_{4}=1.2, a_{1}=0.5, a_{2}=1, a_{3}=20, a_{4}=40$ and $\xi_{1}^{\prime}=\xi_{2}^{\prime}=\xi_{3}^{\prime}=\xi_{4}^{\prime}=25$.

The blue dotted line represents the numerical solution and the red dotted line represents the mesh distribution. The solitons travel in a leftward direction with respect to the $x$-axis. 

Next, we construct a self-adaptive moving mesh scheme for the CSP equation. As in the case of the CmSP equation, setting $v = u^{*}$ in (\ref{2SP}), we obtain the CSP equation~\cite{Feng2015CSP},
\begin{eqnarray}
u_{xt}=u+\frac{1}{2}(|u|^{2}u_{x})_{x}, \label{CSP}
\end{eqnarray}
where $u^{*}$ denotes the complex conjugate of $u$.
From (\ref{semidis2SP}), we have a self-adaptive moving mesh scheme for the CSP equation
\begin{eqnarray}
\left\{
\begin{array}{ll}
\displaystyle\frac{d }{d T}(u_{l+1}-u_{l})=\displaystyle\frac{\delta_{l}}{2}(u_{l+1}+u_{l}), \\
\\
\displaystyle\frac{d x_{l}}{d T}=-\displaystyle\frac{1}{2}|u_{l}|^{2}.
\end{array}
\right.
\label{semidiscomplexSP}
\end{eqnarray}

Here we show examples of numerical simulations of the CSP equation  (\ref{CSP}) with the self-adaptive moving mesh scheme (\ref{semidiscomplexSP}) under periodic boundary conditions. The initial conditions are given by (\ref{incon_SP_1}) and (\ref{incon_SP_2}).
Figures \ref{compSP_abs_1}, \ref{compSP_real_1} and \ref{compSP_imag_1} show numerical simulations of the $|u|$, Re$(u)$, and Im$(u)$-profiles of the one-soliton solution for $p_{1}=0.5+0.1\mathrm{i},p_{2}=0.5-0.1\mathrm{i}, a_{1}= a_{2}={\rm exp}(-6)$ and $\xi_{1}^{\prime}=\xi_{2}^{\prime}=5$.
Figures \ref{compSP_abs},  \ref{compSP_real} and \ref{compSP_imag} show numerical simulations of the $|u|$, Re$(u)$, and Im$(u)$-profiles of the two-soliton solution for $p_{1}=0.5+0.1\mathrm{i}, p_{2}=1.0+0.3\mathrm{i}, p_{3}=0.5-0.1\mathrm{i}, p_{4}=1.0-0.3\mathrm{i}, a_{1}=a_{3}={\rm exp}(-6), a_{2}=a_{4}={\rm exp}(4)$ and $ \xi_{1}^{\prime}=\xi_{2}^{\prime}=\xi_{3}^{\prime}=\xi_{4}^{\prime}=5$.
The solitons travel in a leftward direction with respect to the $x$-axis.

\begin{figure}[h]
 \centering
 \begin{tabular}{cc}
      \begin{minipage}[t]{0.47\hsize}
       \centering
        \includegraphics[width=\linewidth]{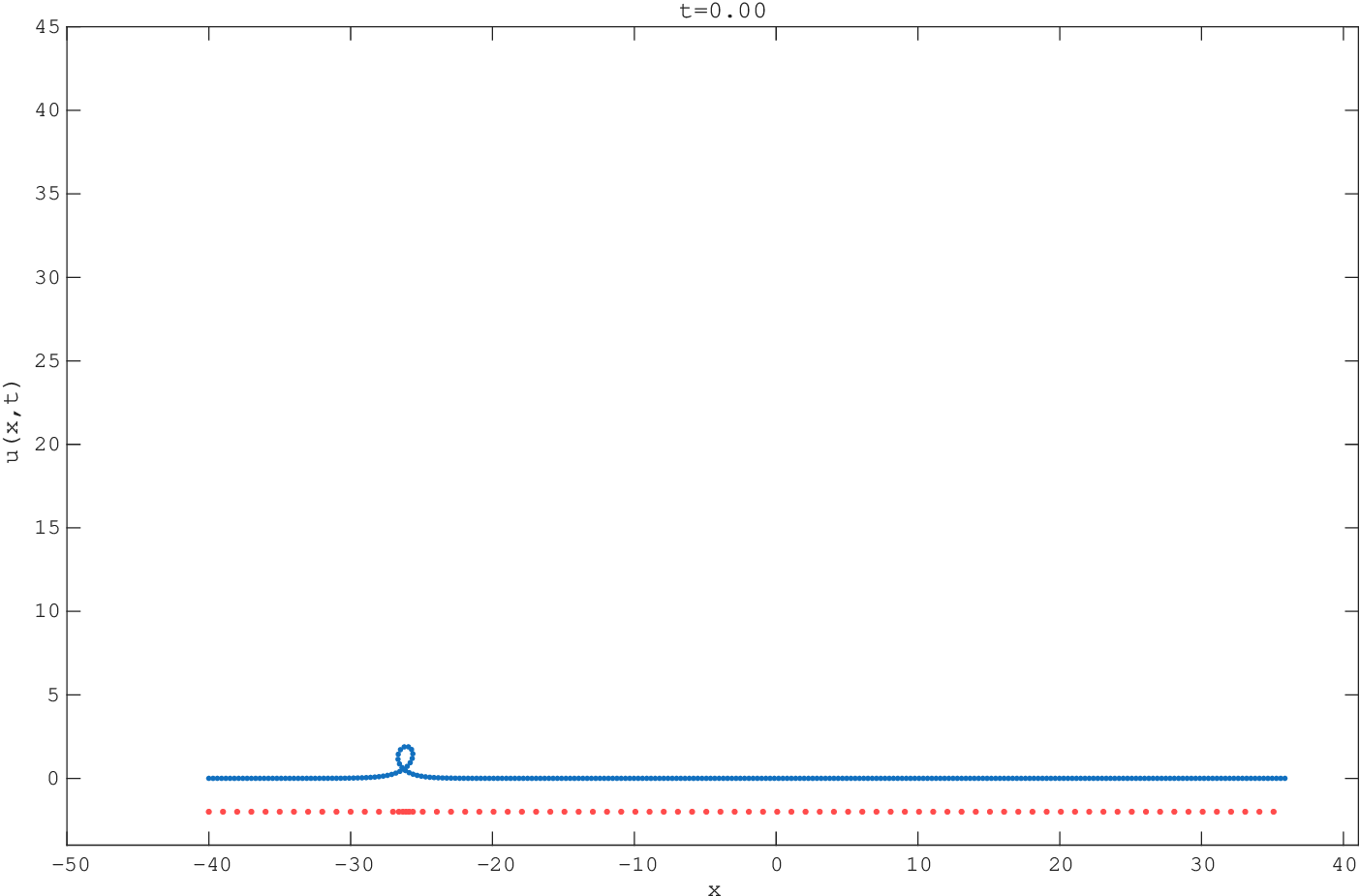}
      \end{minipage}&
      \begin{minipage}[t]{0.47\hsize}
        \centering
        \includegraphics[width=\linewidth]{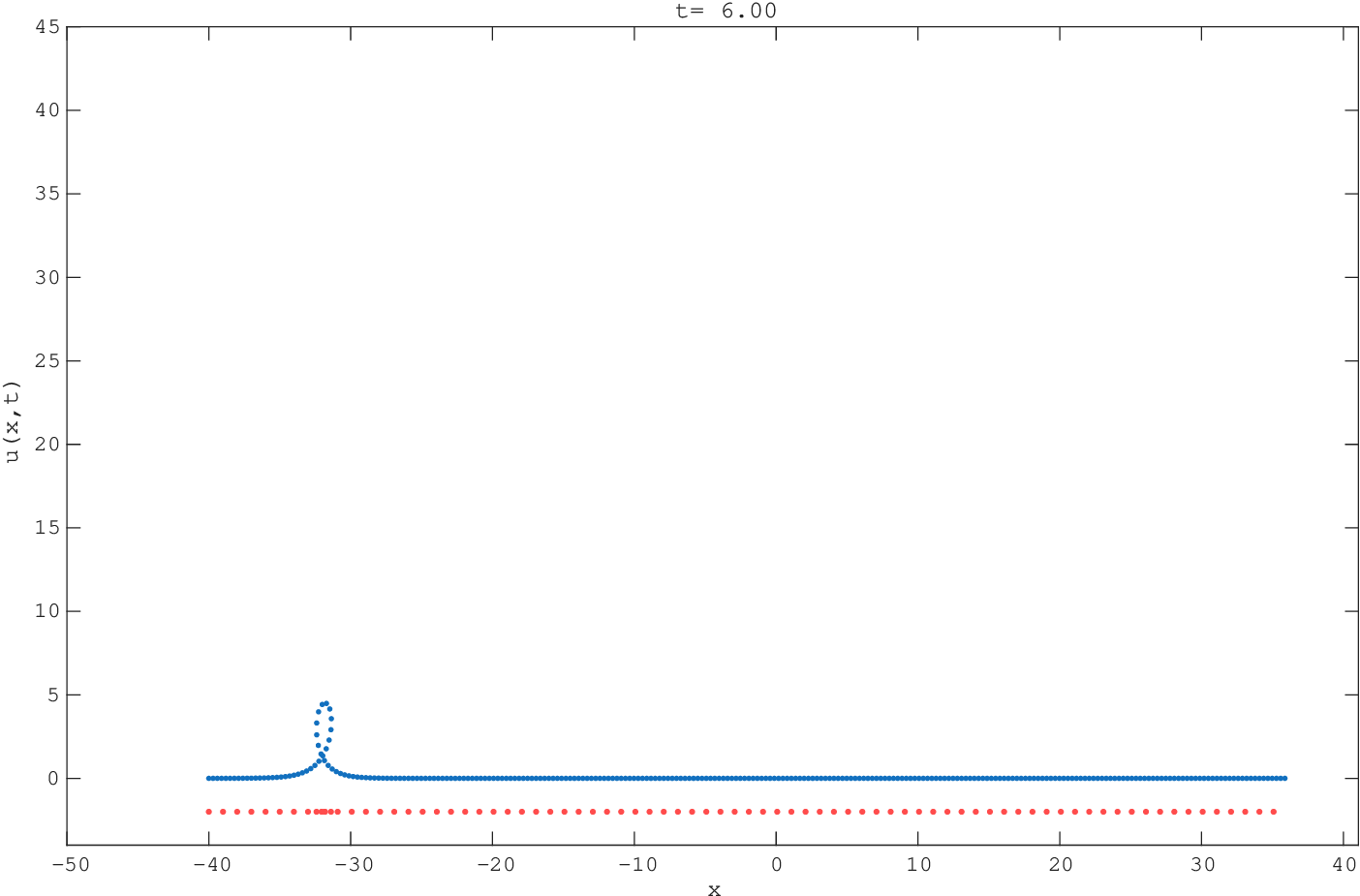}
      \end{minipage}\\

      \begin{minipage}[t]{0.47\hsize}
        \centering
        \includegraphics[width=\linewidth]{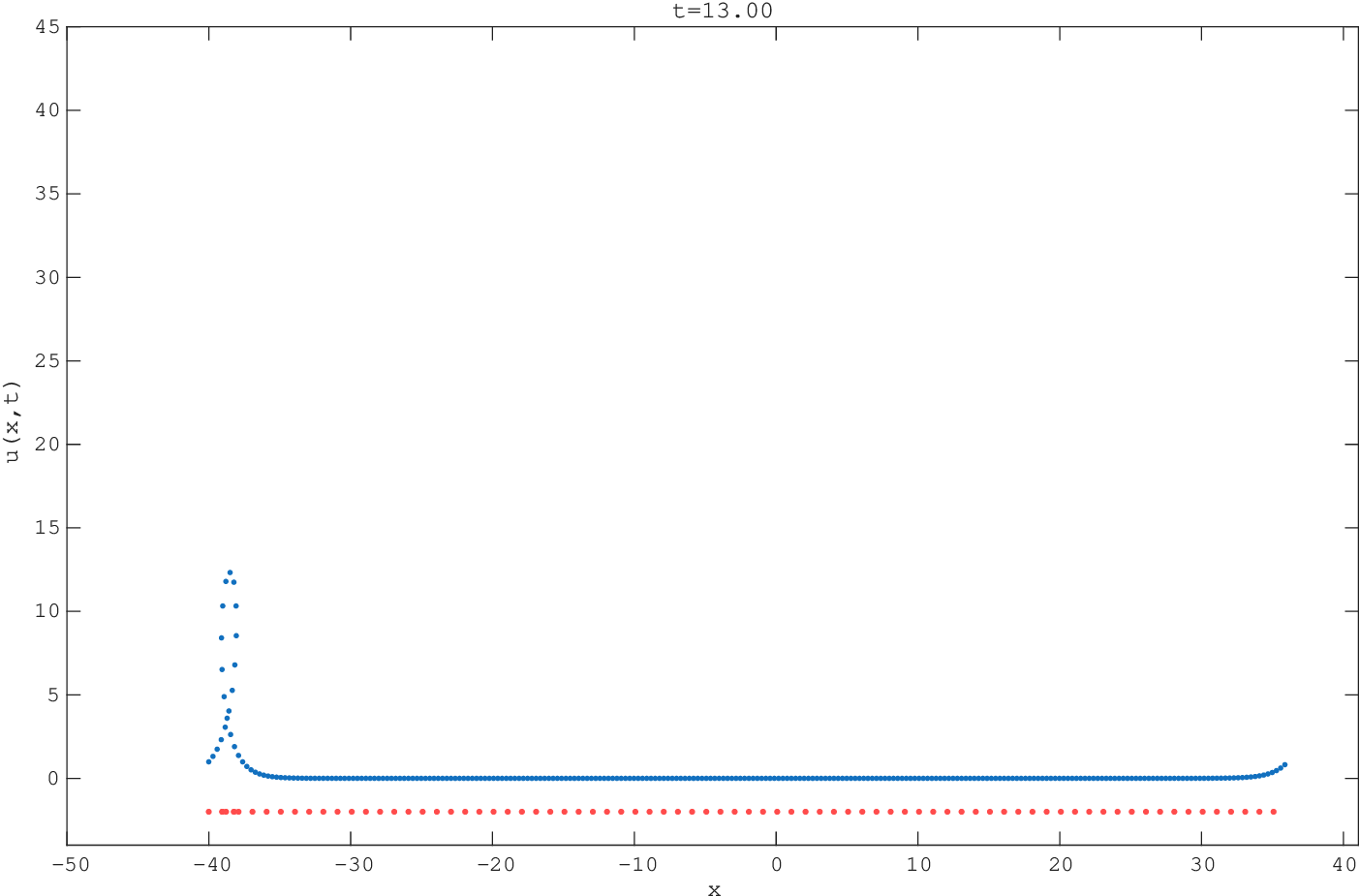}
      \end{minipage} &
      \begin{minipage}[t]{0.47\hsize}
        \centering
        \includegraphics[width=\linewidth]{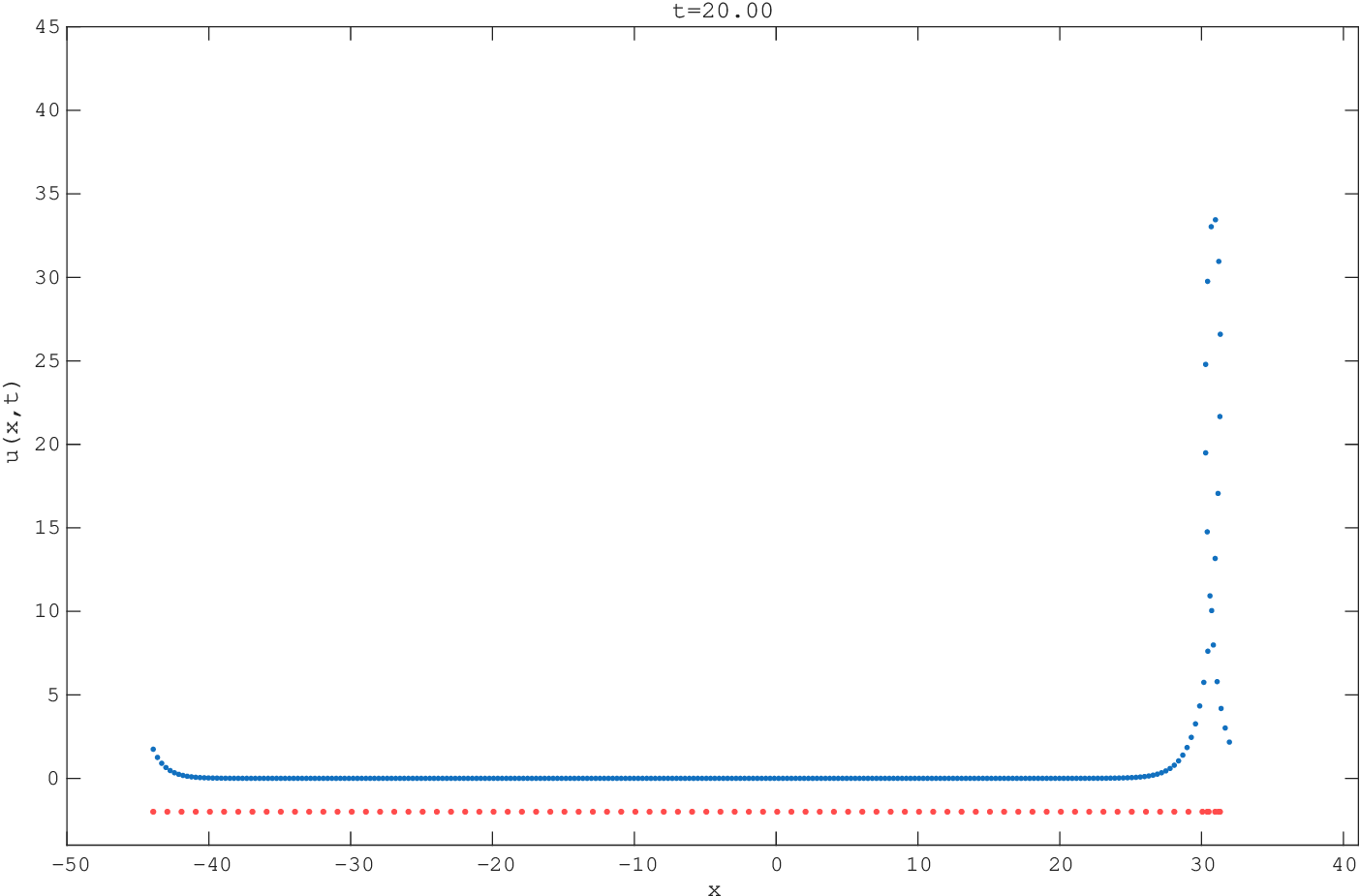}
      \end{minipage}
       \end{tabular}
     \caption{The numerical simulation of the {\it u}-profile of the one-soliton solution for the 2-SP equation. maxerr=2.33$\times 10^{-4}$}
              \label{2SP_u_1}
  \end{figure}

 \begin{figure}[h]
  \centering
 \begin{tabular}{cc}
      \begin{minipage}[t]{0.47\hsize}
       \centering
        \includegraphics[width=\linewidth]{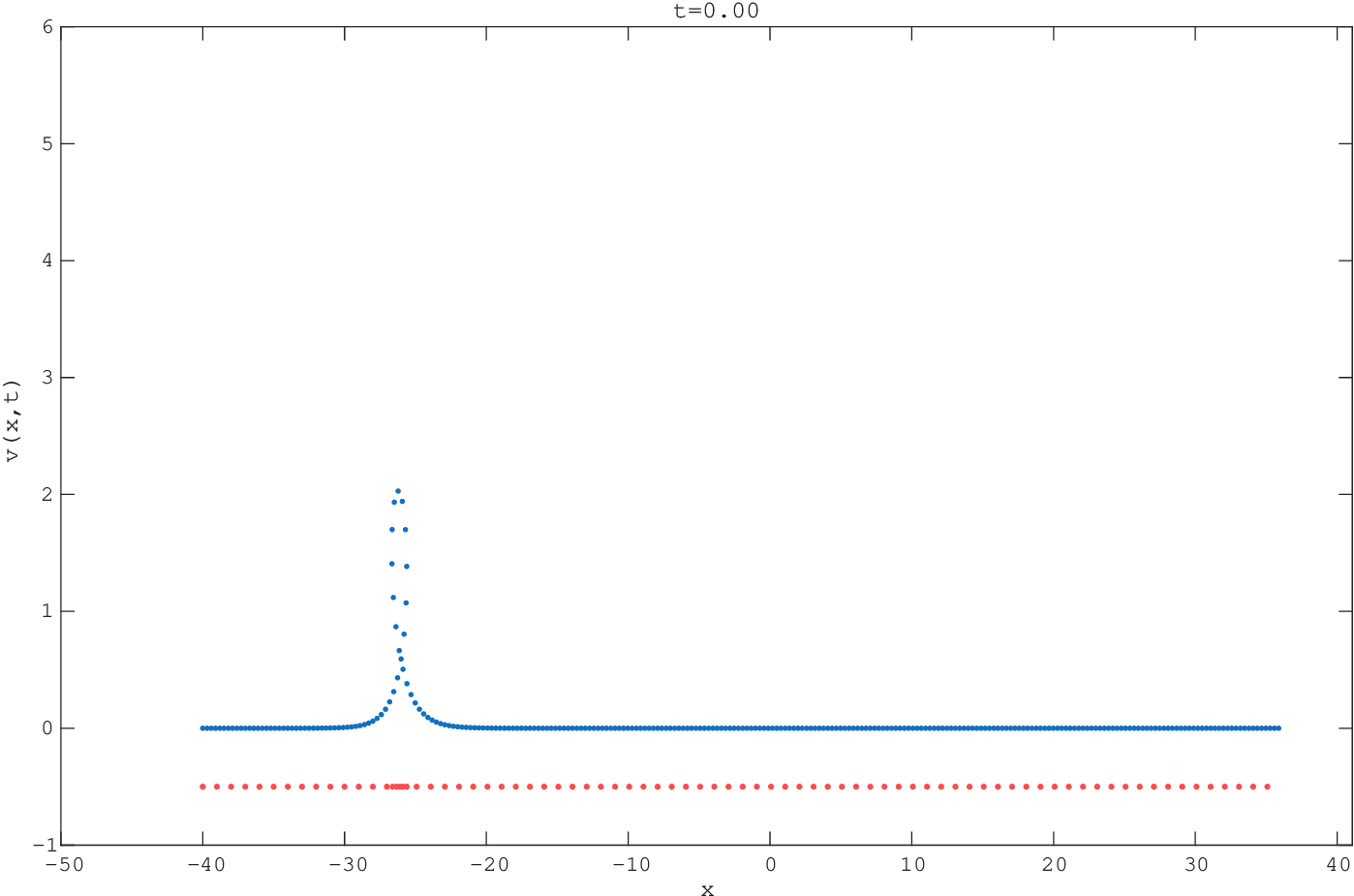}
      \end{minipage} &
      \begin{minipage}[t]{0.47\hsize}
        \centering
        \includegraphics[width=\linewidth]{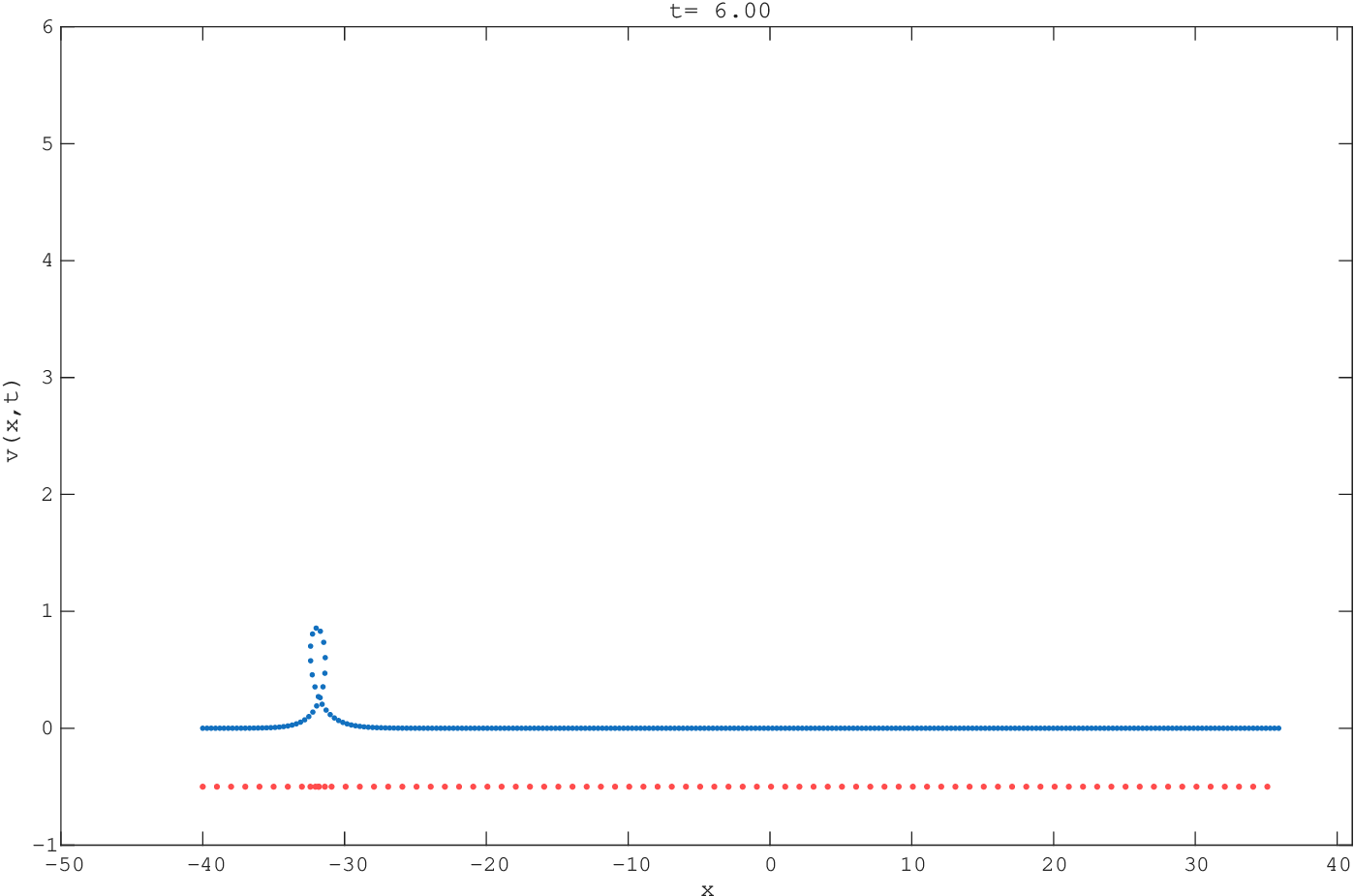}
      \end{minipage}\\

      \begin{minipage}[t]{0.47\hsize}
        \centering
        \includegraphics[width=\linewidth]{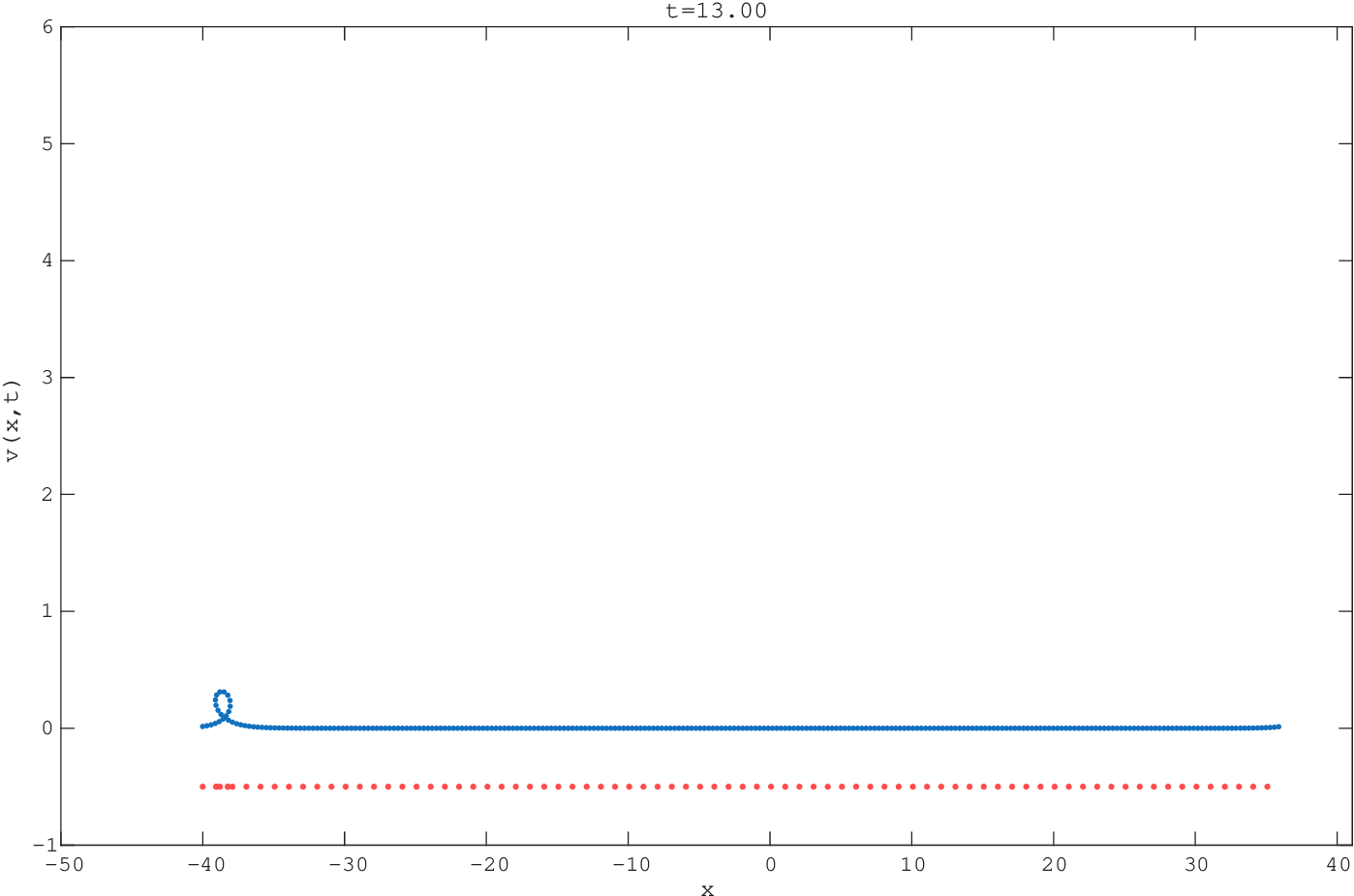}
      \end{minipage} &
      \begin{minipage}[t]{0.47\hsize}
        \centering
        \includegraphics[width=\linewidth]{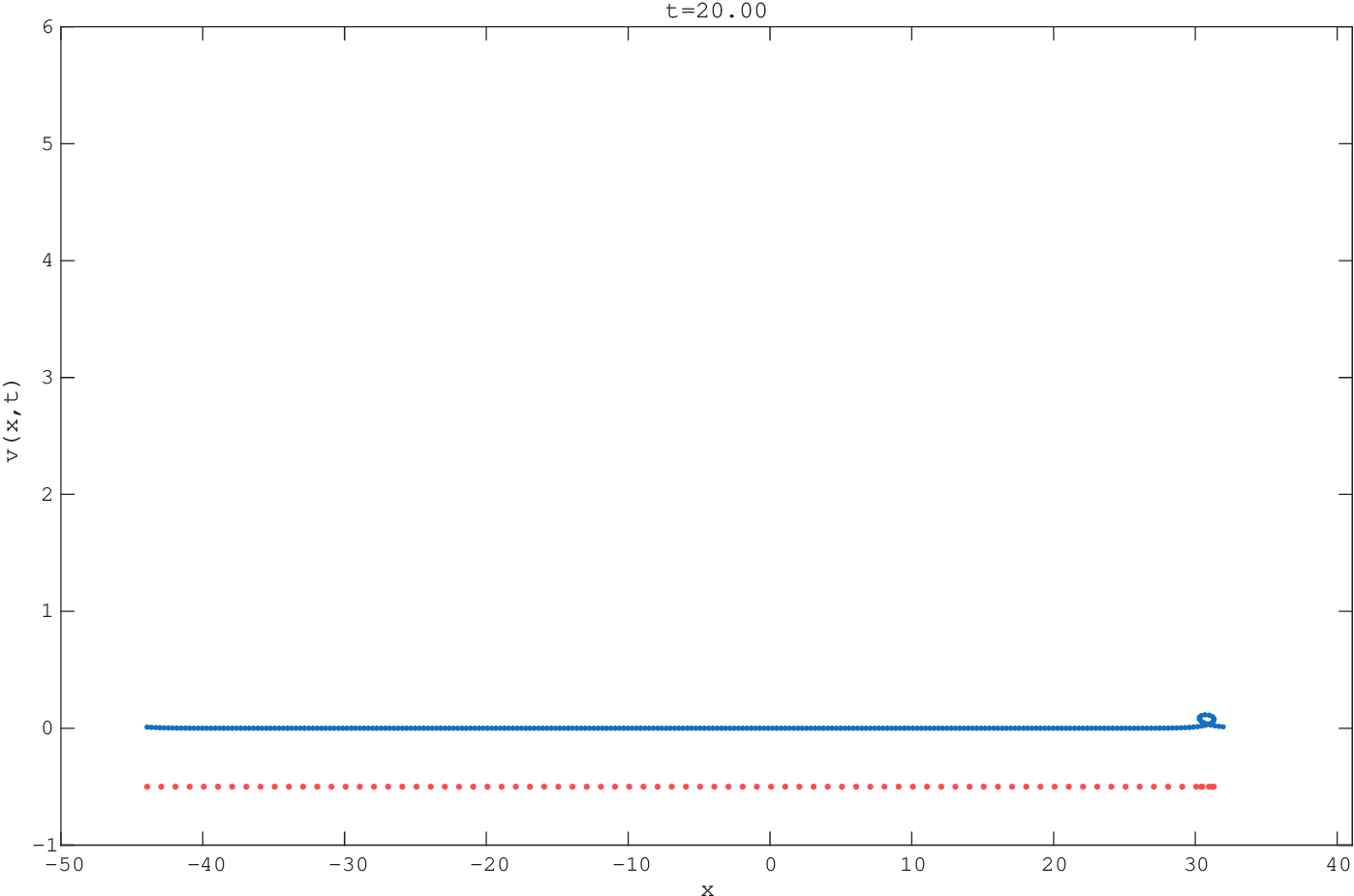}
      \end{minipage}
       \end{tabular}
     \caption{The numerical simulation of the {\it v}-profile of the one-soliton solution for the 2-SP equation. maxerr=2.65$\times 10^{-4}$}
              \label{2SP_v_1}
  \end{figure}


\begin{figure}[h]
 \centering
 \begin{tabular}{cc}
      \begin{minipage}[t]{0.47\hsize}
       \centering
        \includegraphics[width=\linewidth]{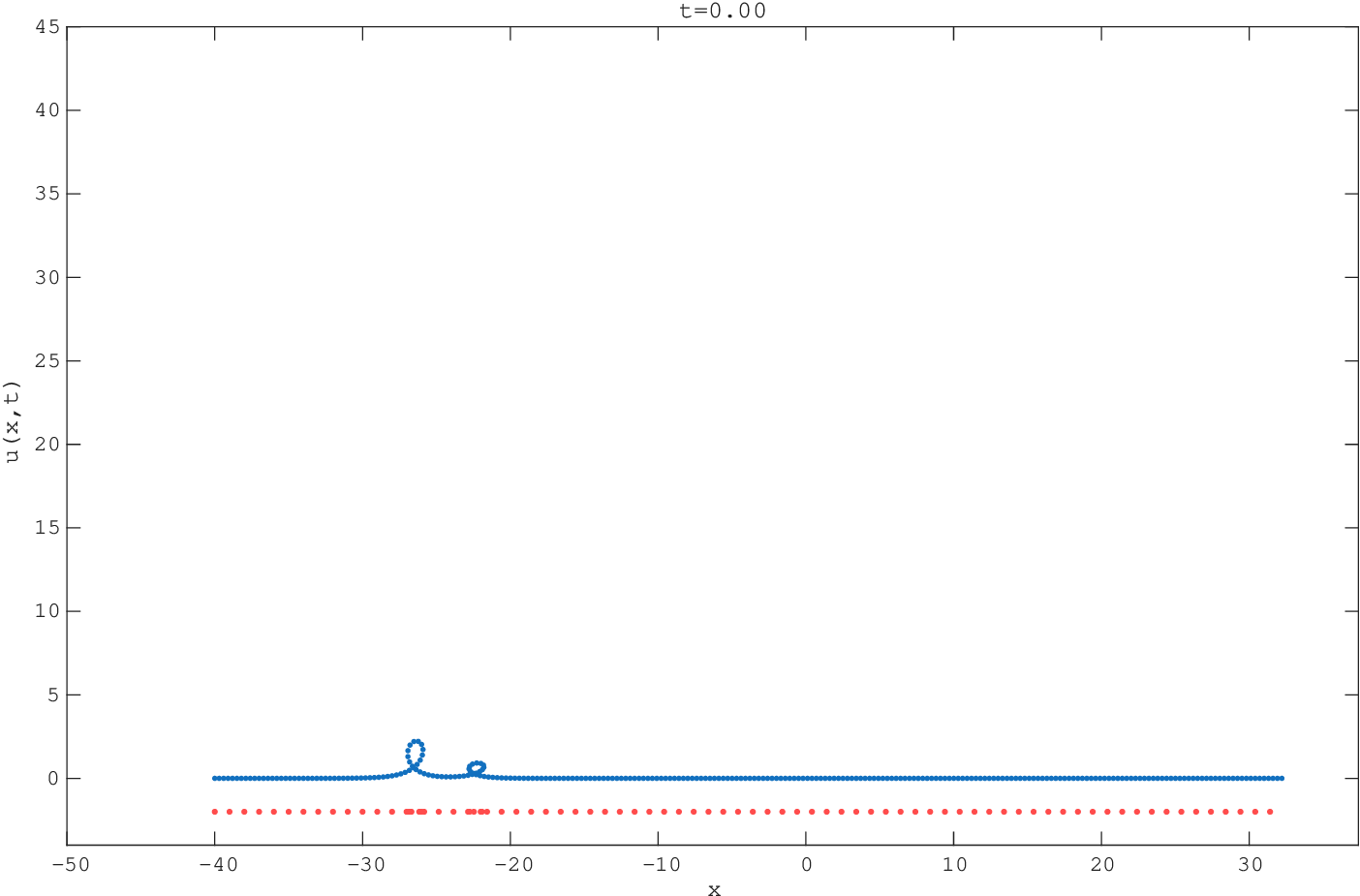}
      \end{minipage} &
      \begin{minipage}[t]{0.47\hsize}
        \centering
        \includegraphics[width=\linewidth]{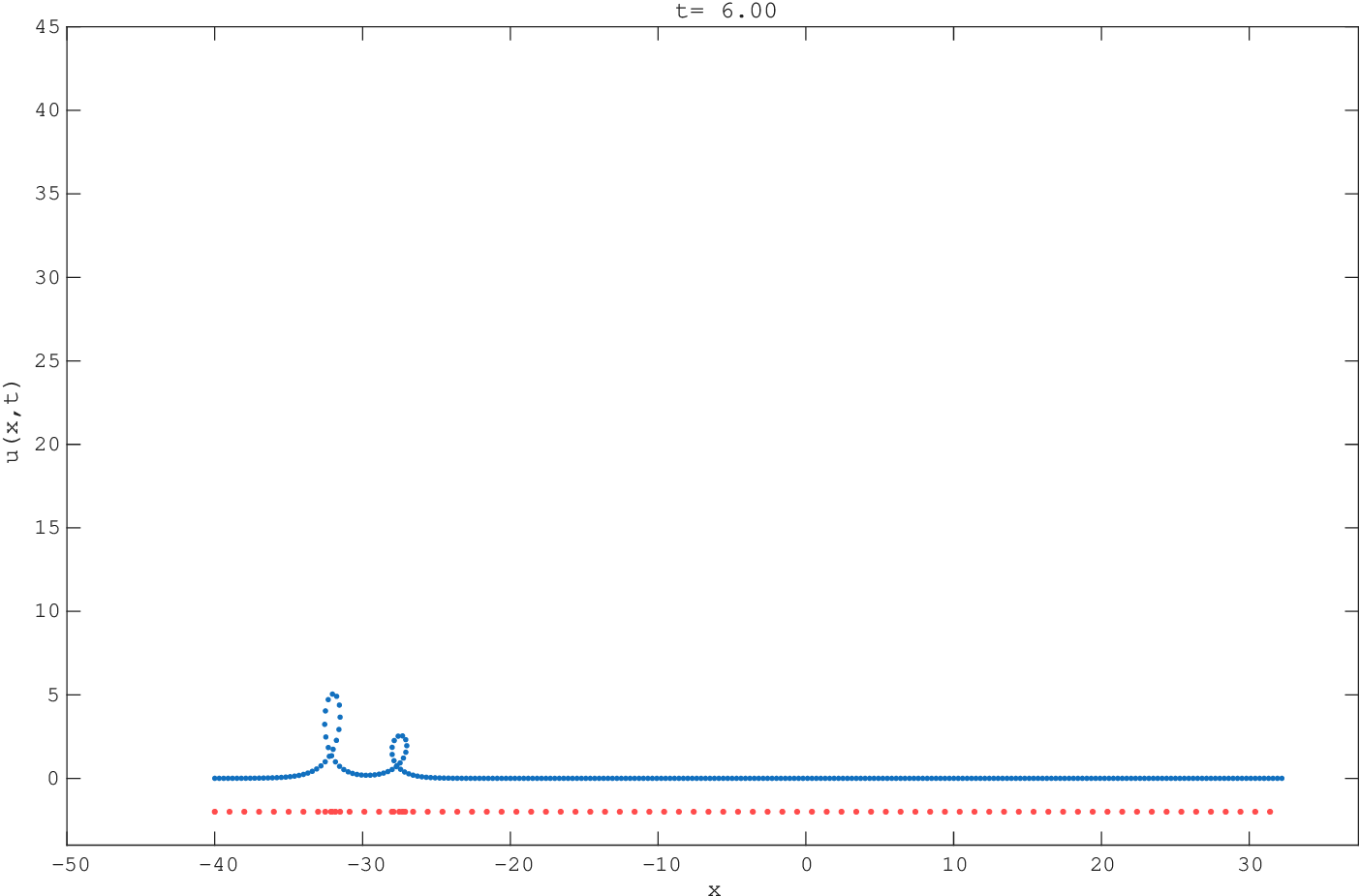}
      \end{minipage}\\

      \begin{minipage}[t]{0.47\hsize}
        \centering
        \includegraphics[width=\linewidth]{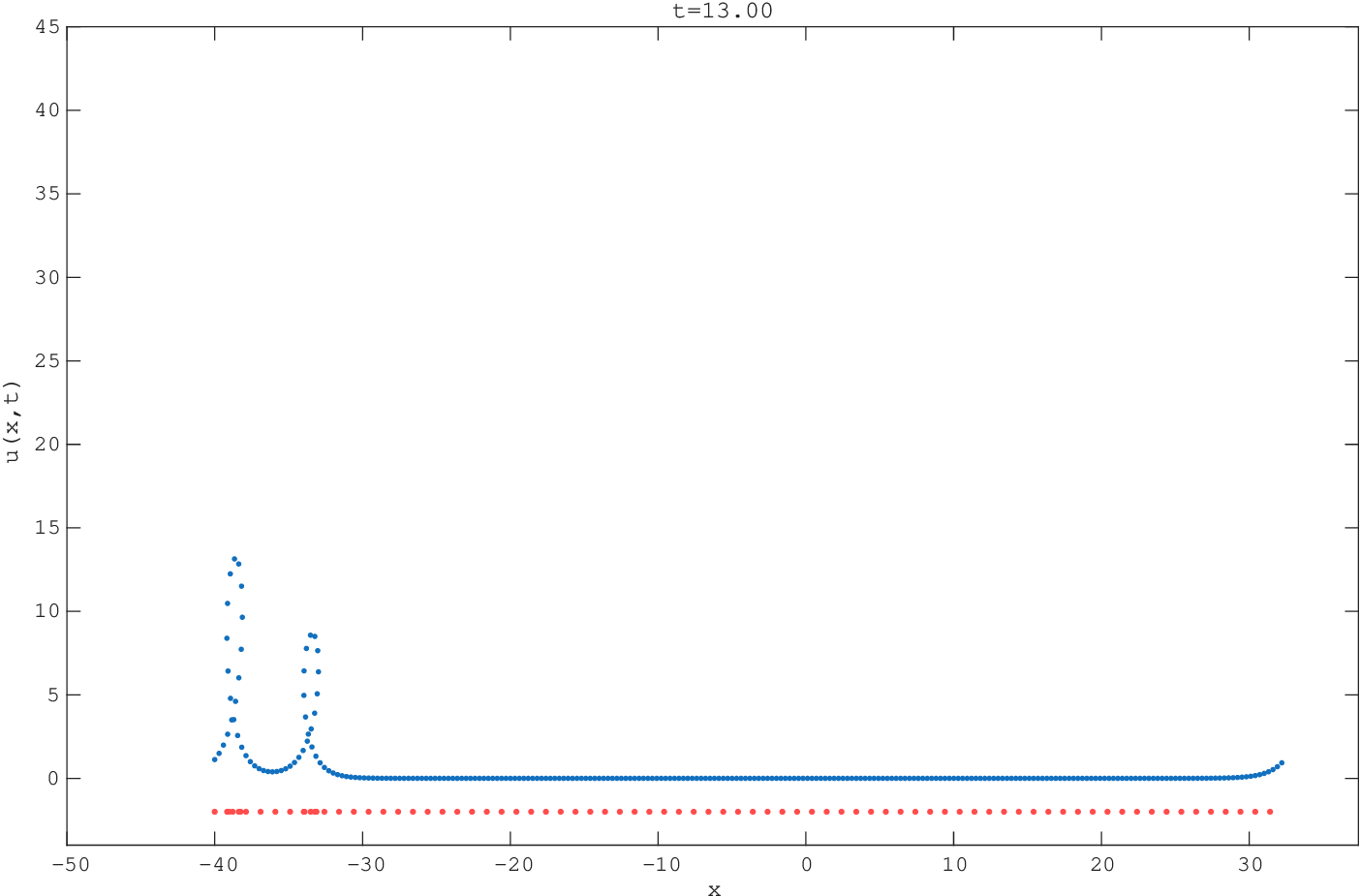}
      \end{minipage} &
      \begin{minipage}[t]{0.47\hsize}
        \centering
        \includegraphics[width=\linewidth]{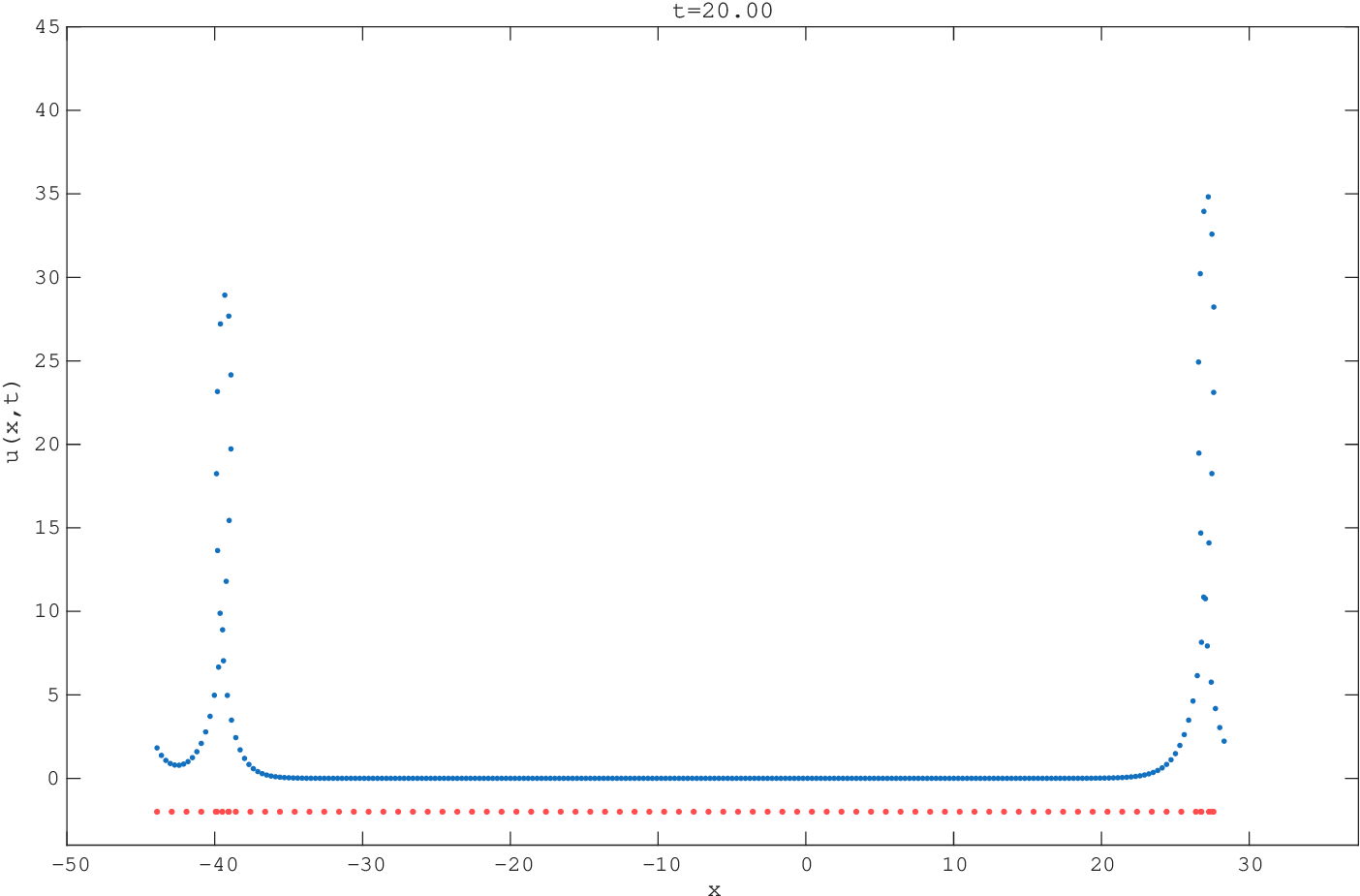}
      \end{minipage}
       \end{tabular}
     \caption{The numerical simulation of the {\it u}-profile of the two-soliton solution for the 2-SP equation. maxerr=2.75$\times 10^{-4}$}
              \label{2SP_u}
  \end{figure}

 \begin{figure}[h]
  \centering
 \begin{tabular}{cc}
      \begin{minipage}[t]{0.47\hsize}
       \centering
        \includegraphics[width=\linewidth]{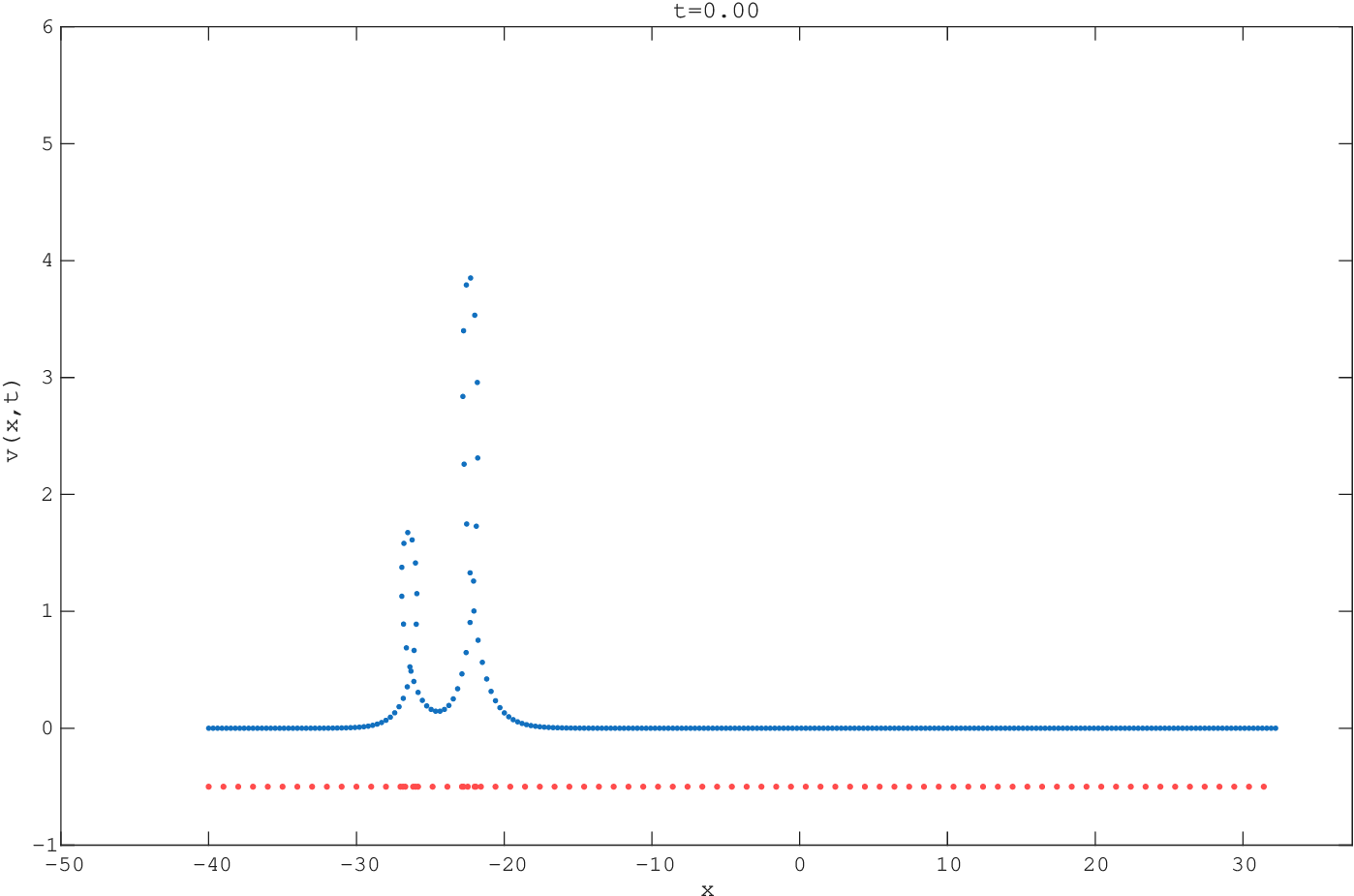}
      \end{minipage} &
      \begin{minipage}[t]{0.47\hsize}
        \centering
        \includegraphics[width=\linewidth]{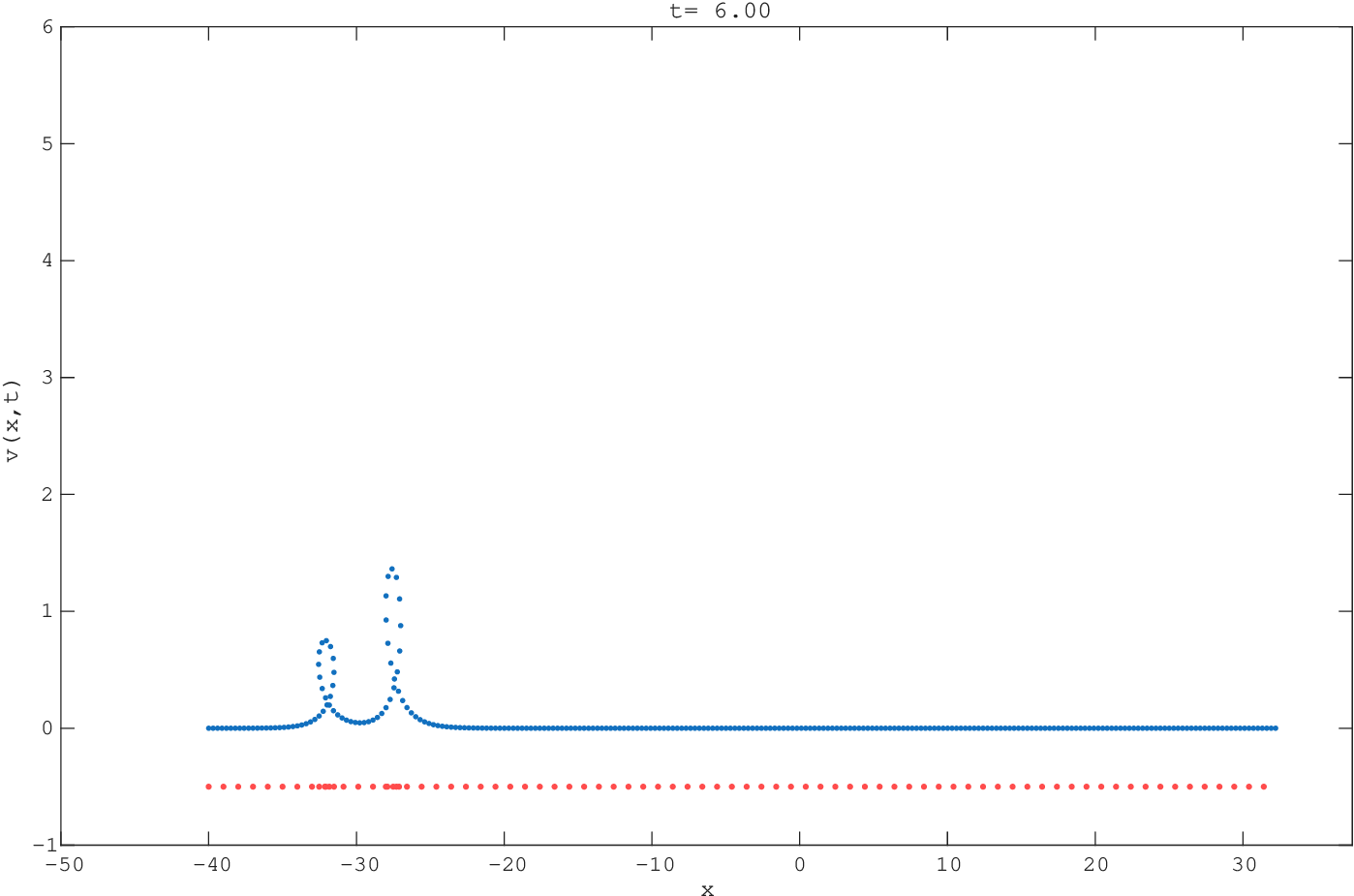}
      \end{minipage}\\

      \begin{minipage}[t]{0.47\hsize}
        \centering
        \includegraphics[width=\linewidth]{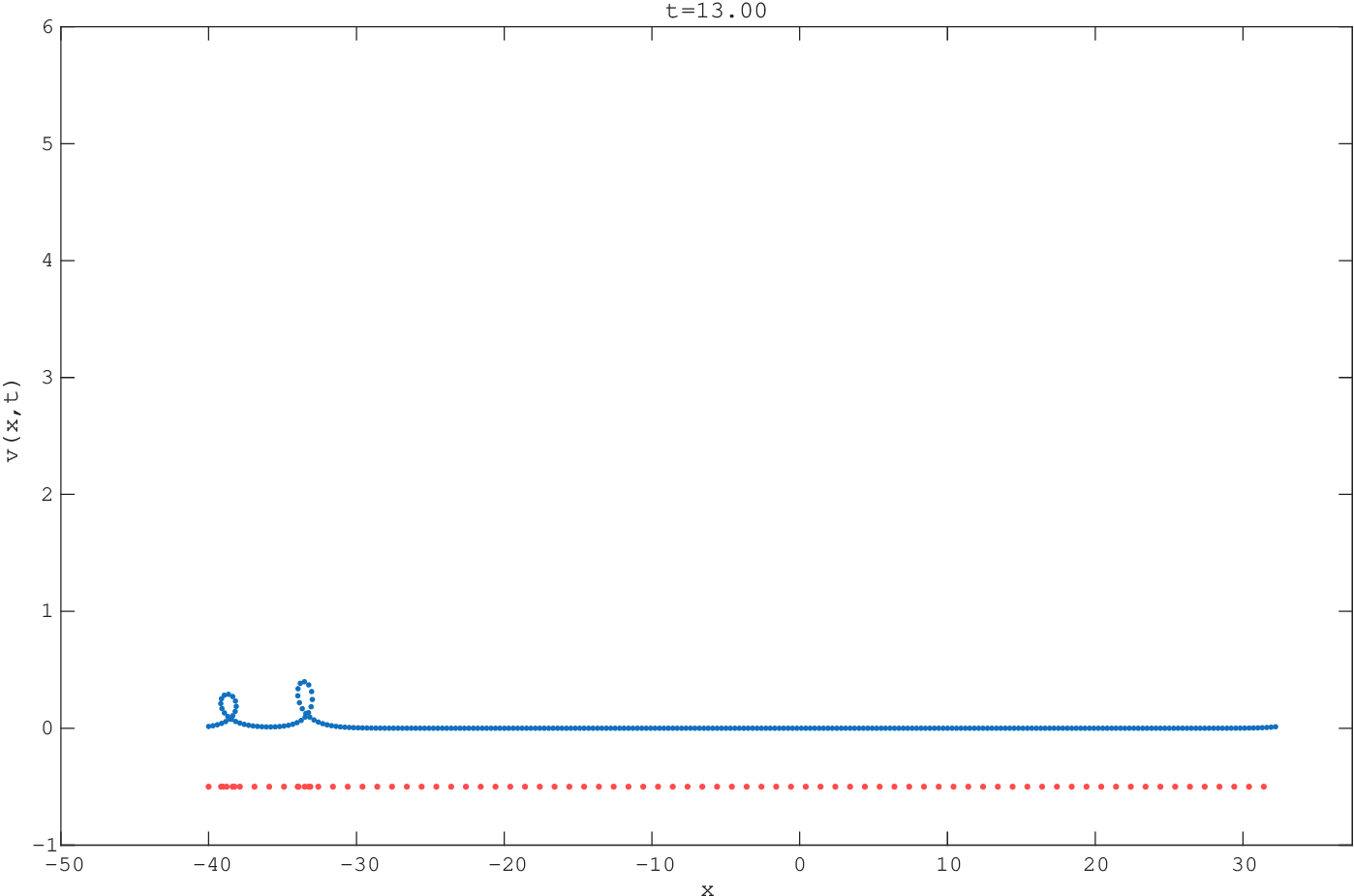}
      \end{minipage} &
      \begin{minipage}[t]{0.47\hsize}
        \centering
        \includegraphics[width=\linewidth]{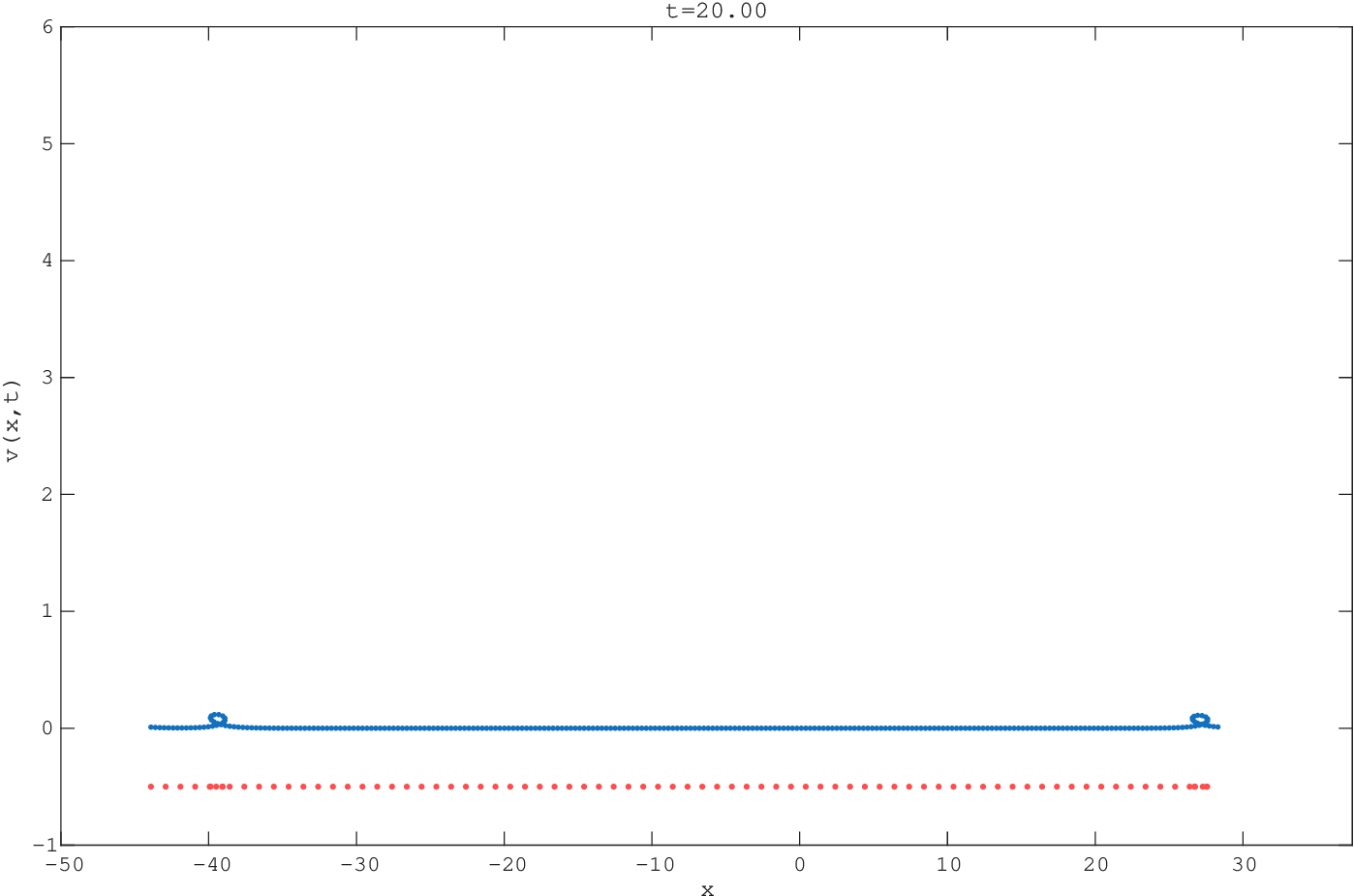}
      \end{minipage}
       \end{tabular}
     \caption{The numerical simulation of the {\it v}-profile of the two-soliton solution for the 2-SP equation. maxerr=2.99$\times 10^{-4}$}
              \label{2SP_v}
  \end{figure}

\begin{figure}[h]
 \centering
 \begin{tabular}{cc}
      \begin{minipage}[t]{0.47\hsize}
       \centering
        \includegraphics[width=\linewidth]{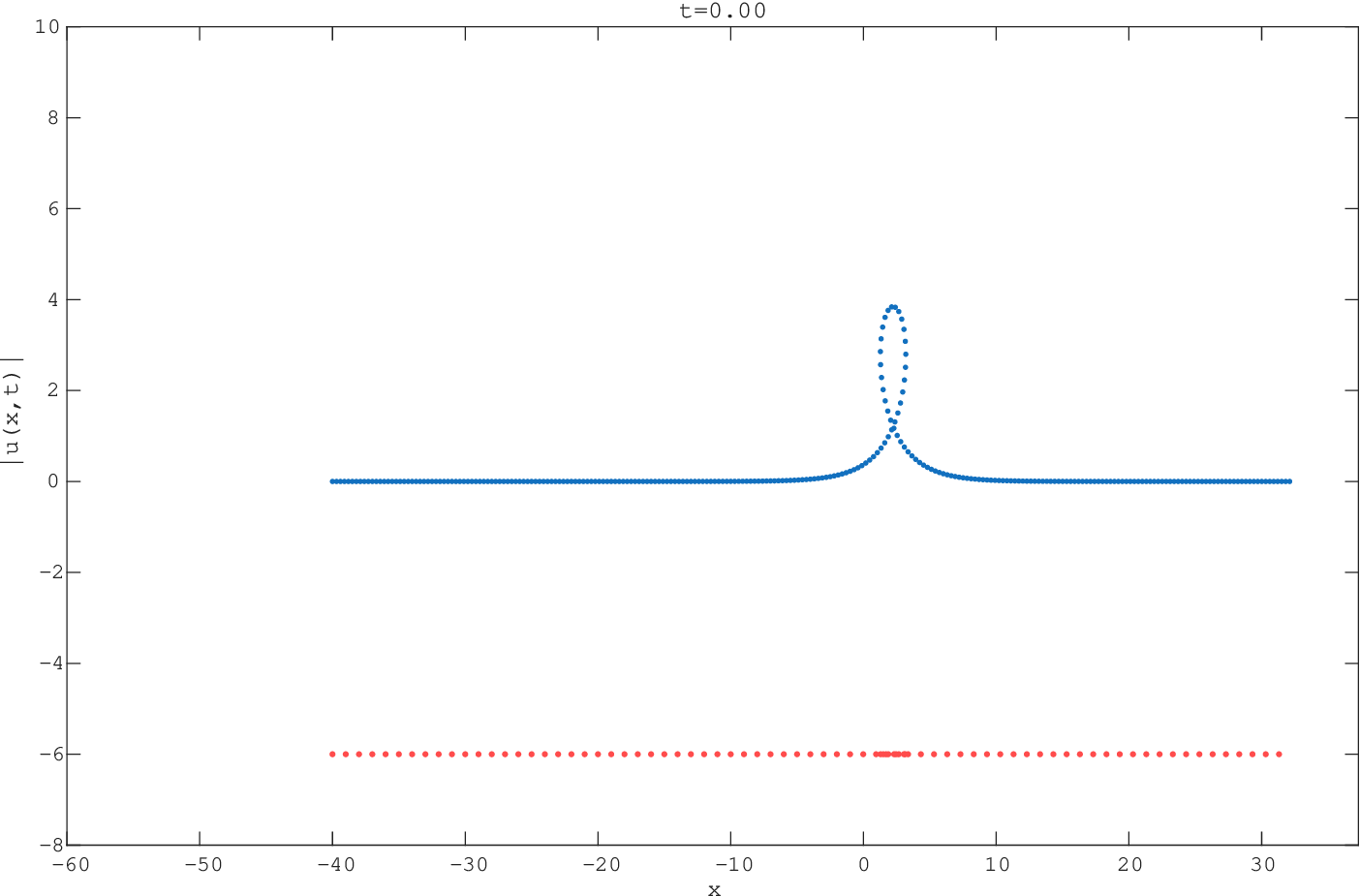}
      \end{minipage} &
      \begin{minipage}[t]{0.47\hsize}
        \centering
        \includegraphics[width=\linewidth]{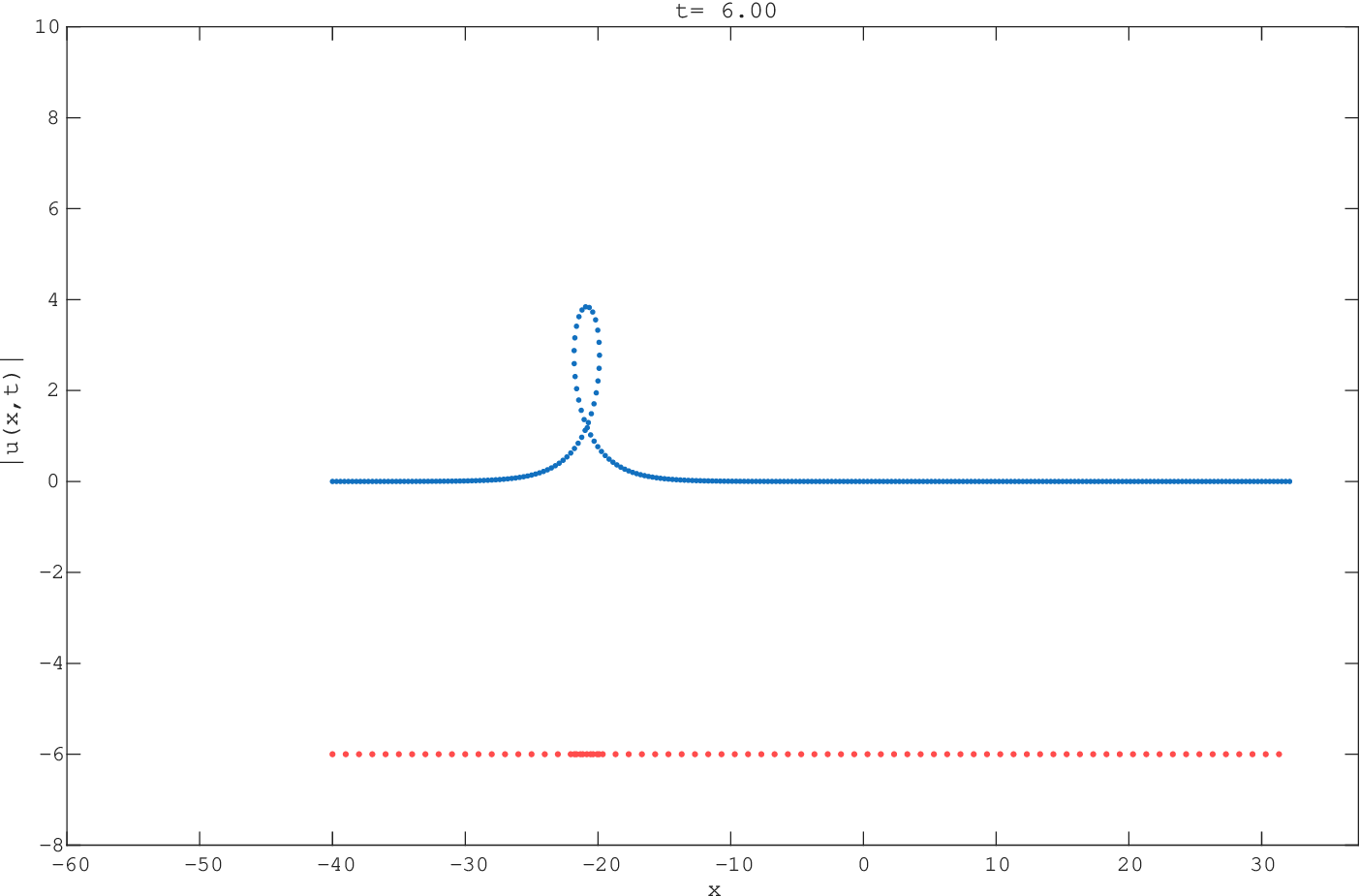}
      \end{minipage}\\

      \begin{minipage}[t]{0.47\hsize}
        \centering
        \includegraphics[width=\linewidth]{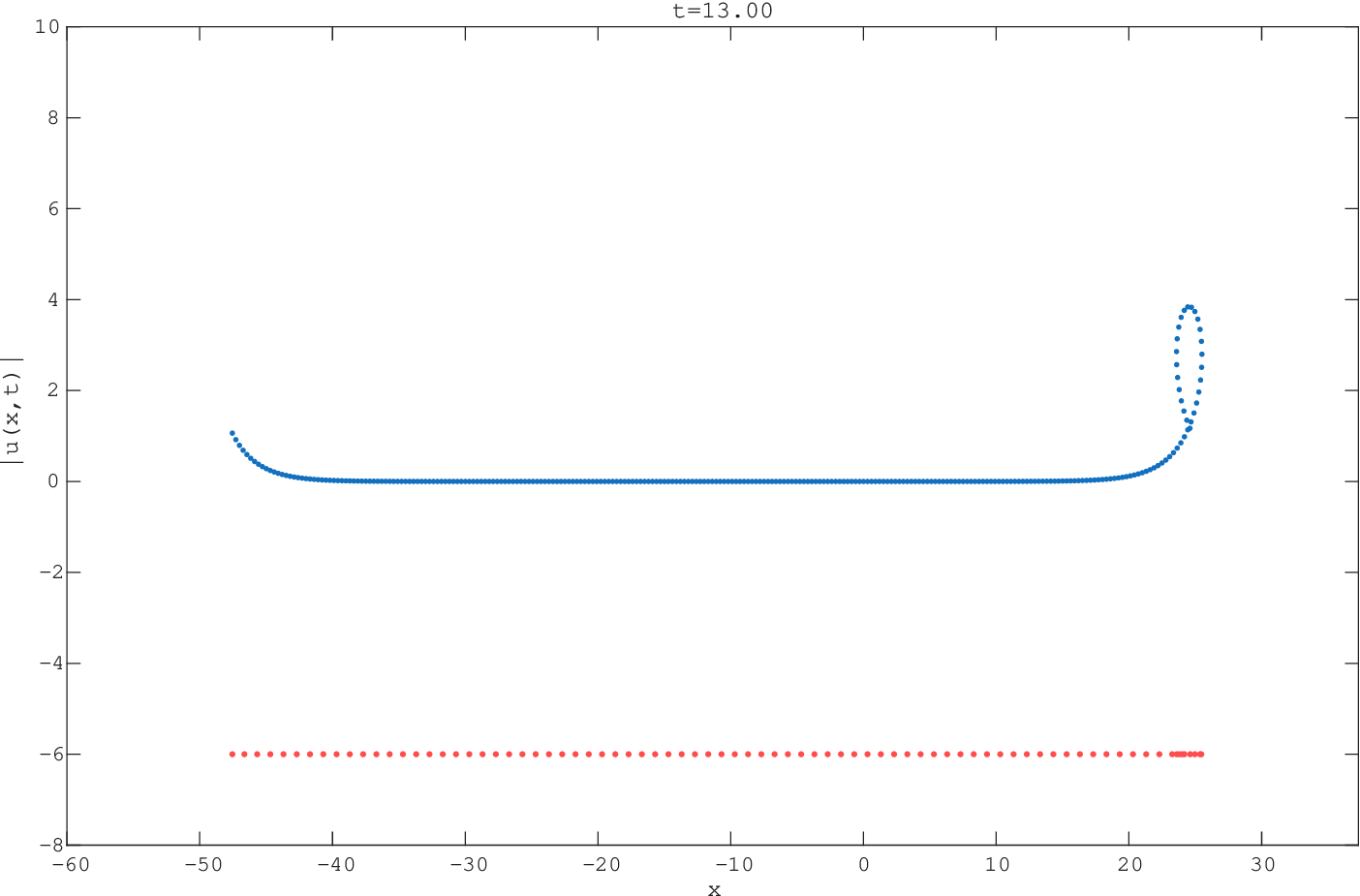}
      \end{minipage} &
      \begin{minipage}[t]{0.47\hsize}
        \centering
        \includegraphics[width=\linewidth]{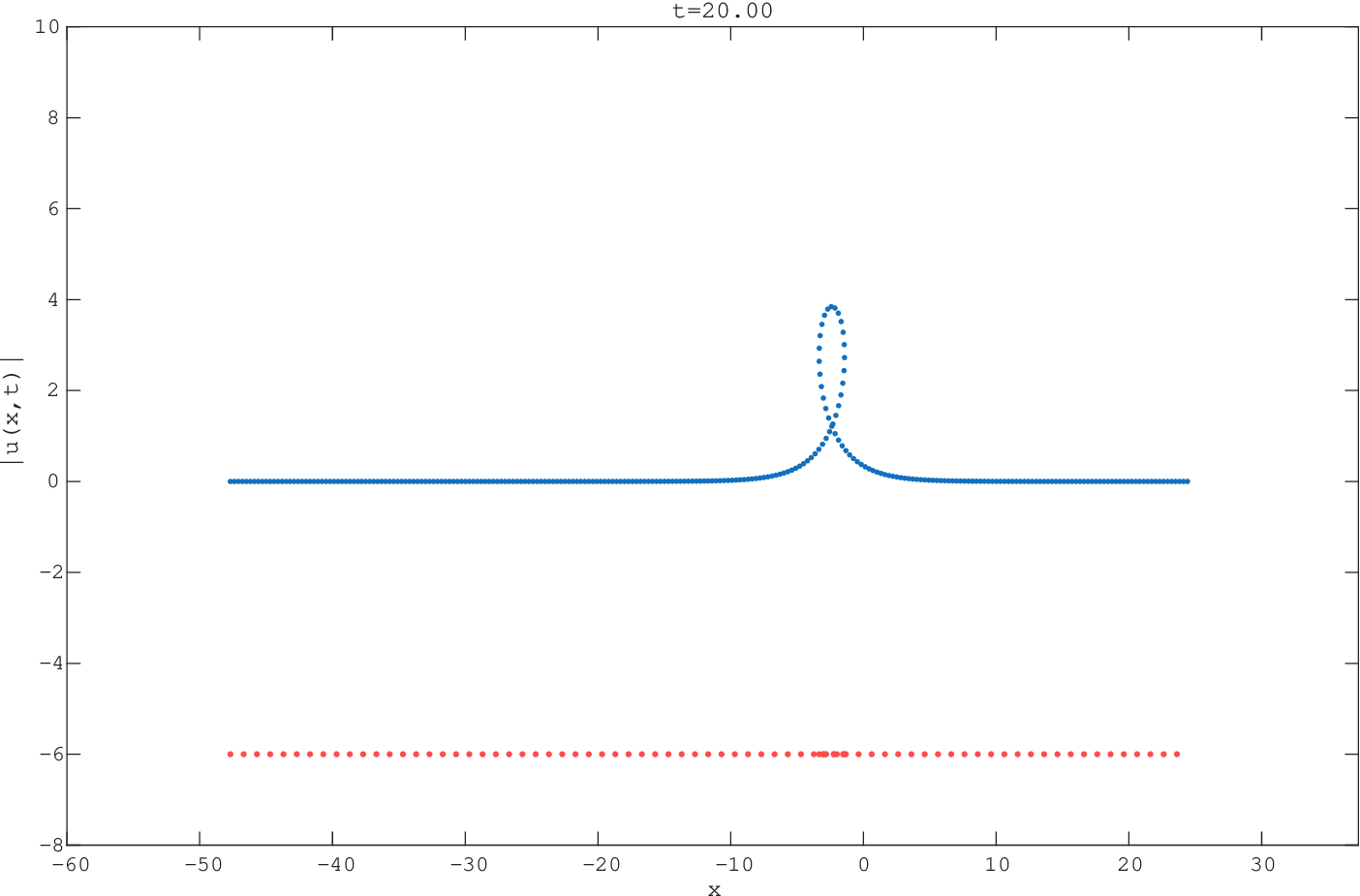}
      \end{minipage}
       \end{tabular}
     \caption{The numerical simulation of the $|u|$-profile of the one-soliton solution for the CSP equation. maxerr=1.03$\times 10^{-4}$}
              \label{compSP_abs_1}
  \end{figure}
\begin{figure}[h]
 \centering
 \begin{tabular}{cc}
      \begin{minipage}[t]{0.47\hsize}
       \centering
        \includegraphics[width=\linewidth]{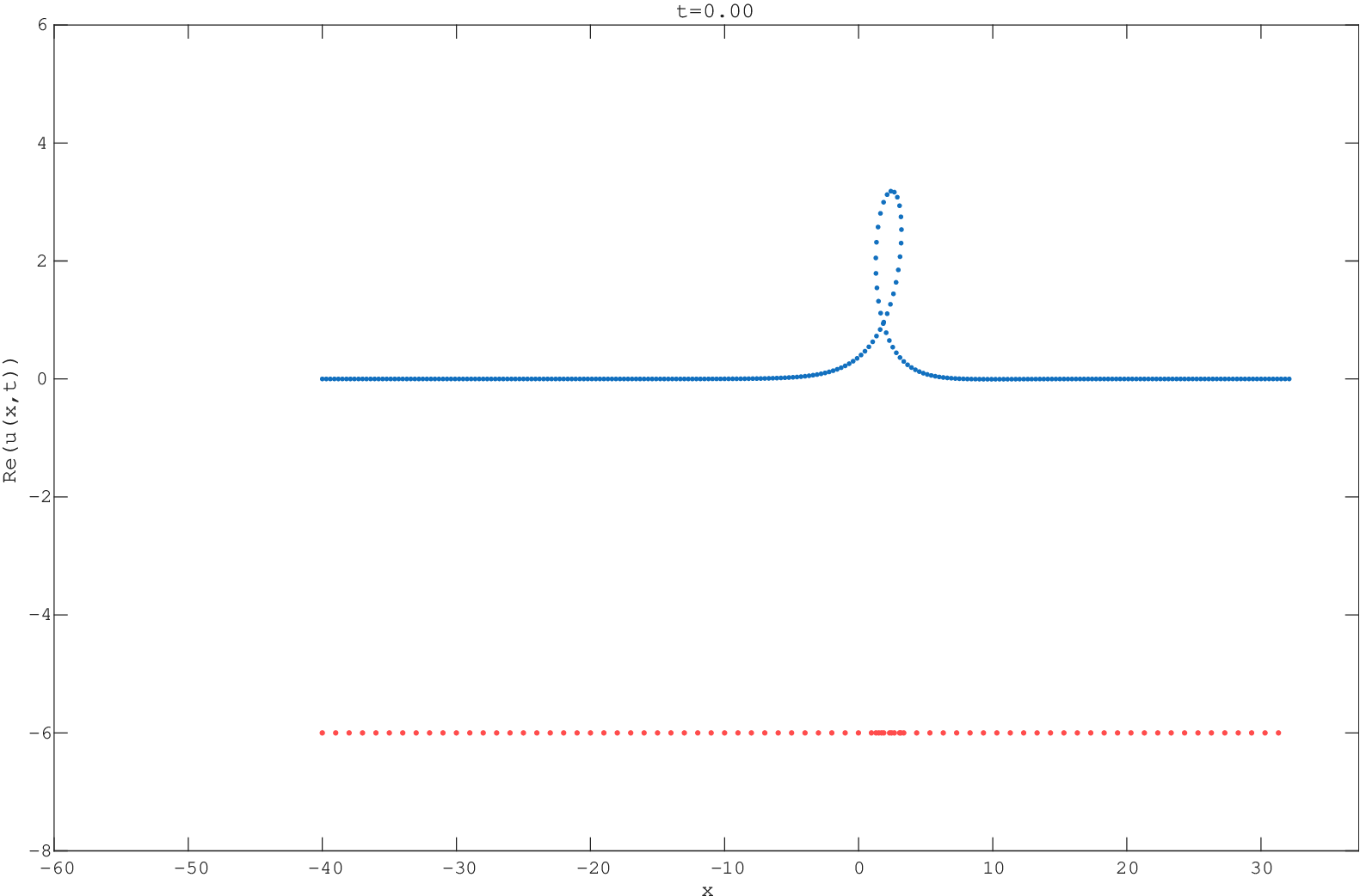}
      \end{minipage} &
      \begin{minipage}[t]{0.47\hsize}
        \centering
        \includegraphics[width=\linewidth]{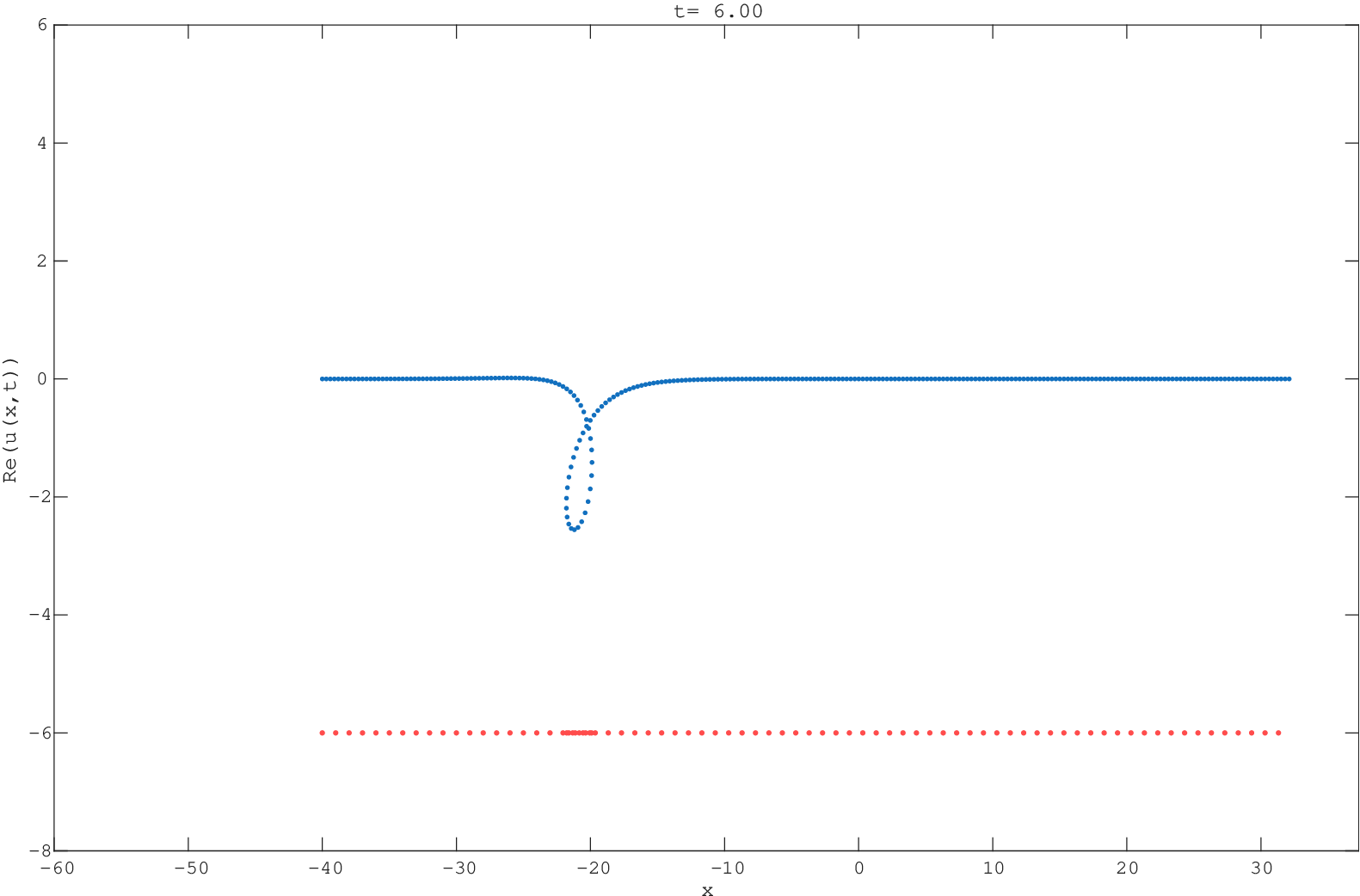}
      \end{minipage}\\

      \begin{minipage}[t]{0.47\hsize}
        \centering
        \includegraphics[width=\linewidth]{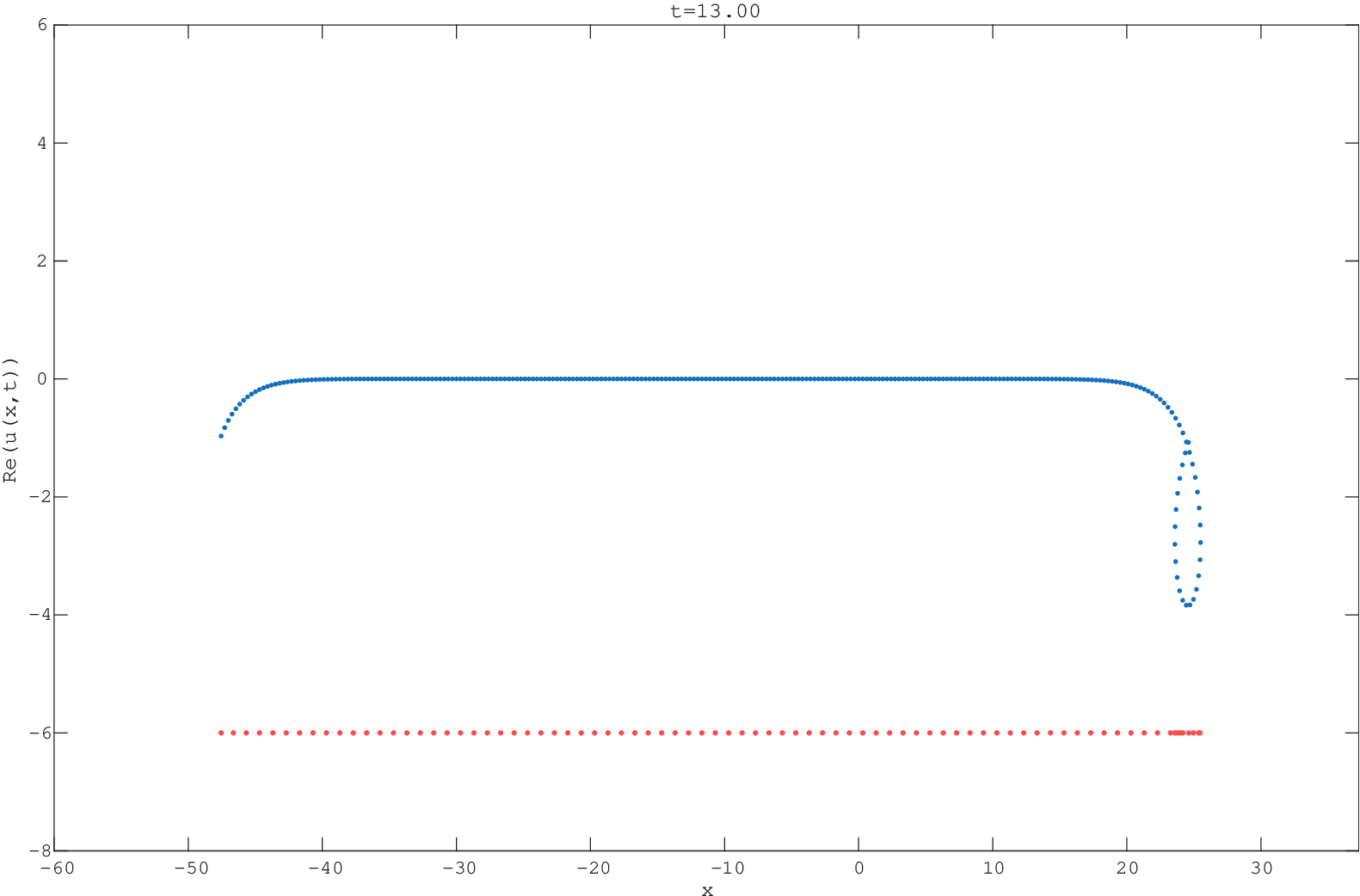}
      \end{minipage} &
      \begin{minipage}[t]{0.47\hsize}
        \centering
        \includegraphics[width=\linewidth]{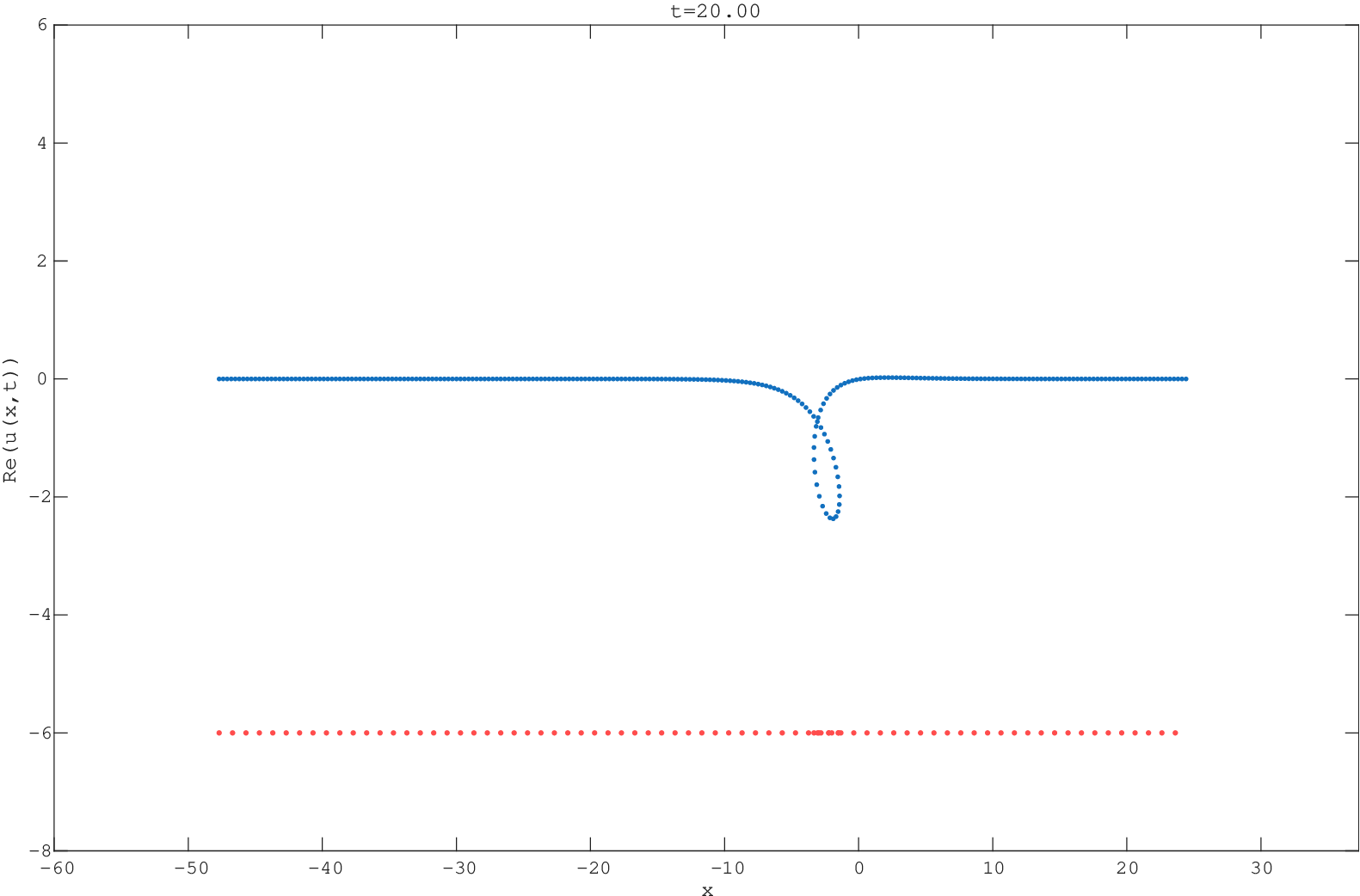}
      \end{minipage}
       \end{tabular}
     \caption{The numerical simulation of the Re($u$)-profile of the one-soliton solution for the CSP equation. maxerr=1.37$\times 10^{-4}$}
              \label{compSP_real_1}
  \end{figure}
\begin{figure}[h]
 \centering
 \begin{tabular}{cc}
      \begin{minipage}[t]{0.47\hsize}
       \centering
        \includegraphics[width=\linewidth]{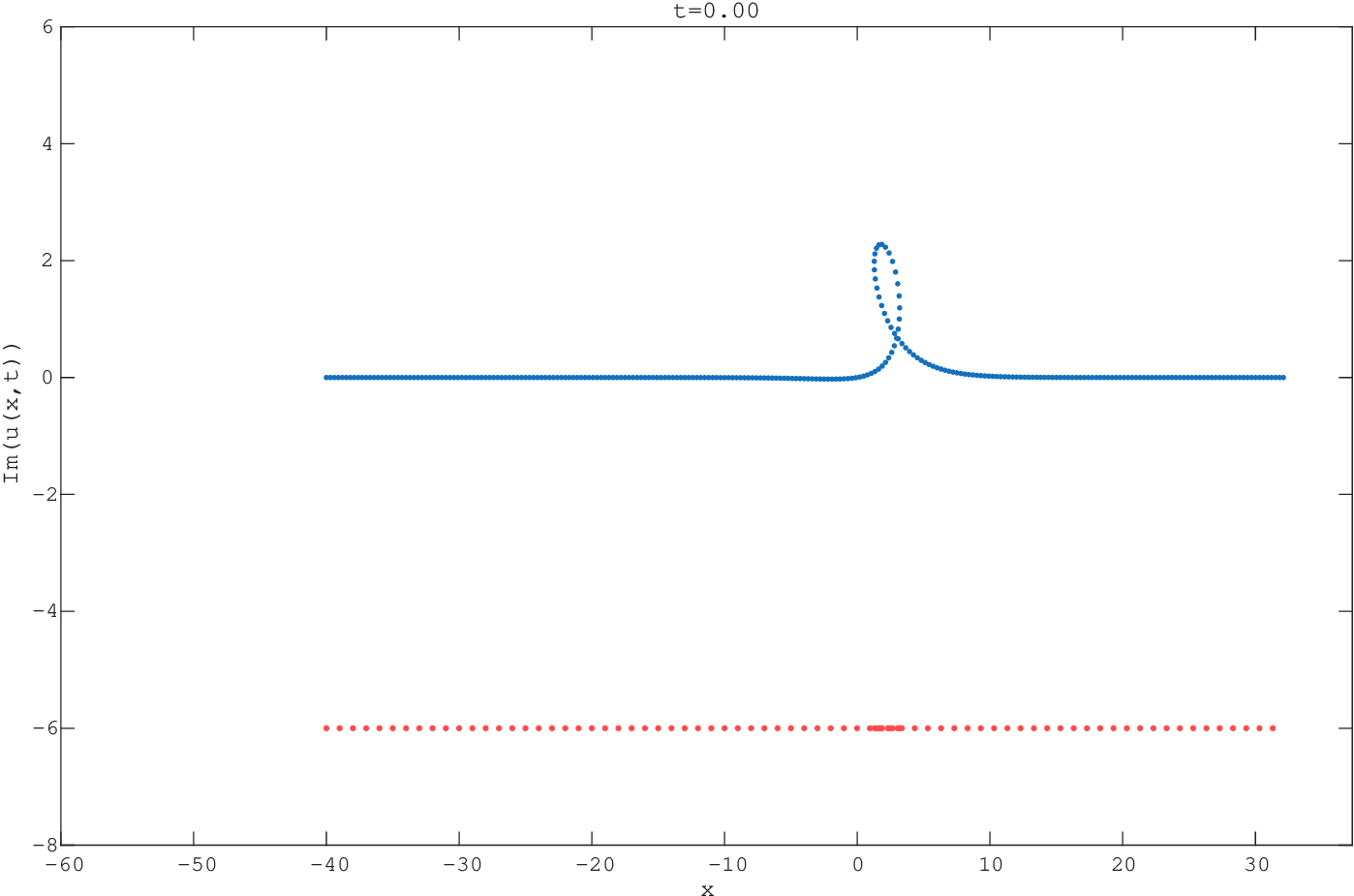}
      \end{minipage} &
      \begin{minipage}[t]{0.47\hsize}
        \centering
        \includegraphics[width=\linewidth]{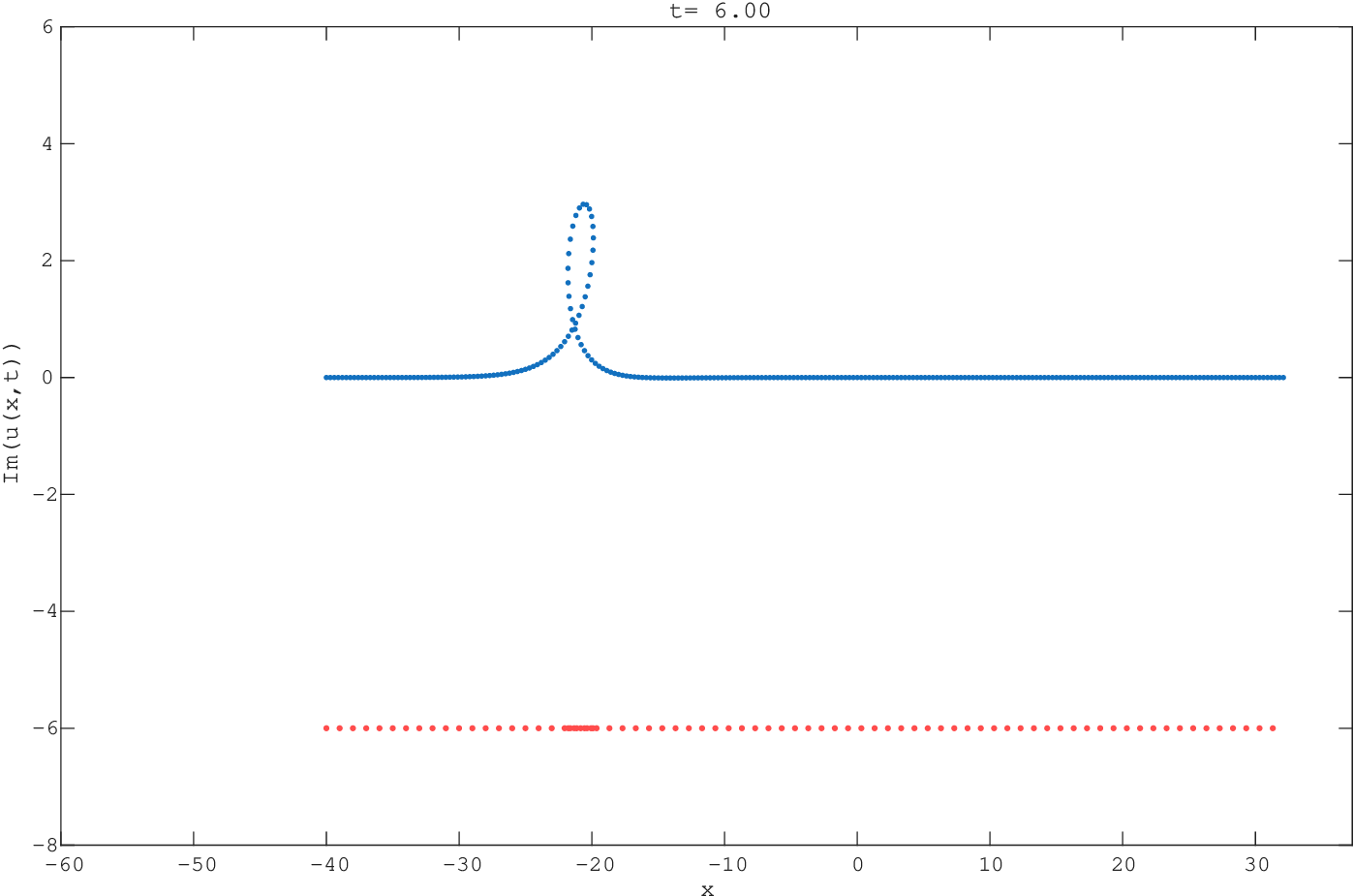}
      \end{minipage}\\

      \begin{minipage}[t]{0.47\hsize}
        \centering
        \includegraphics[width=\linewidth]{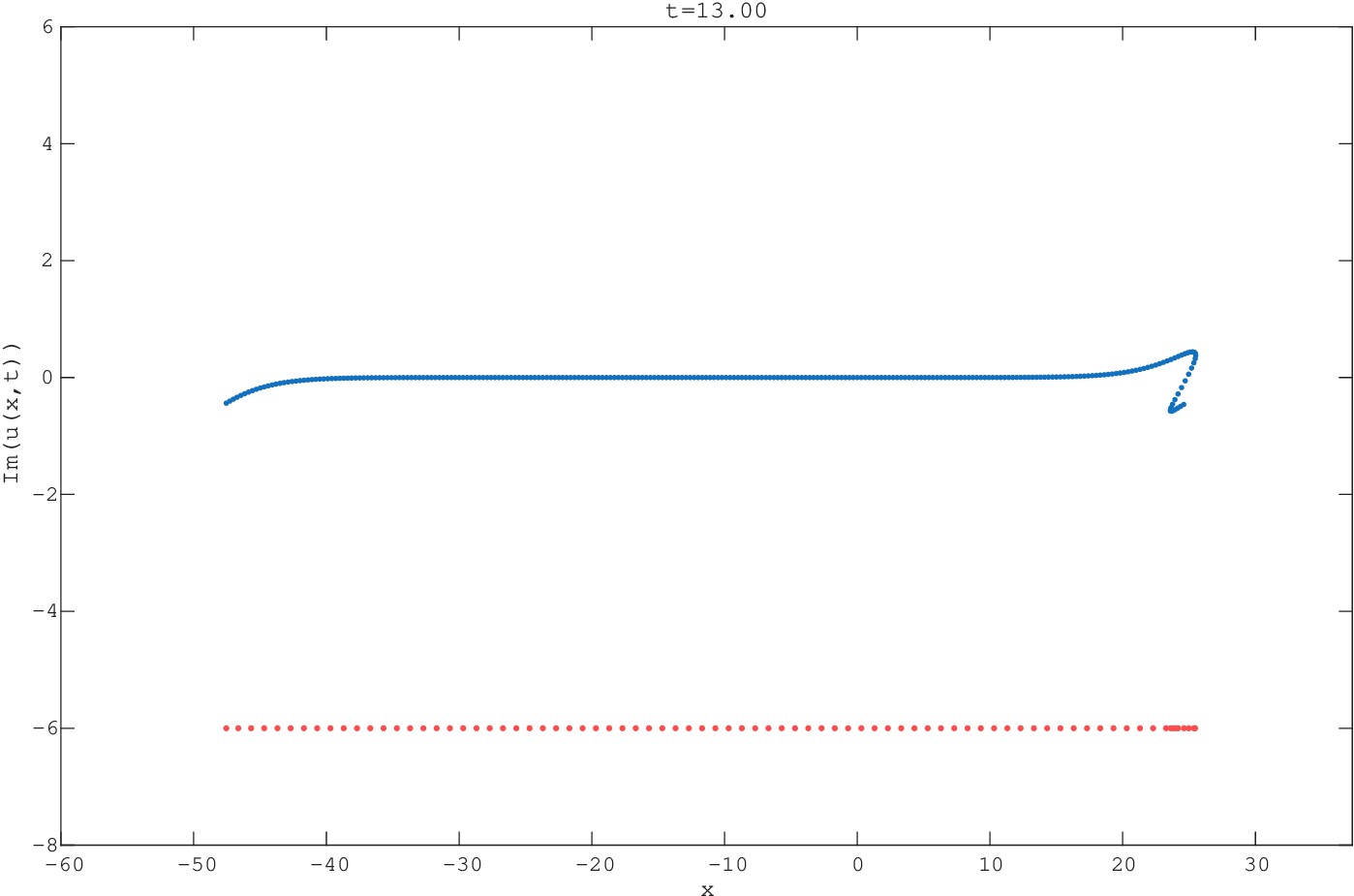}
      \end{minipage} &
      \begin{minipage}[t]{0.47\hsize}
        \centering
        \includegraphics[width=\linewidth]{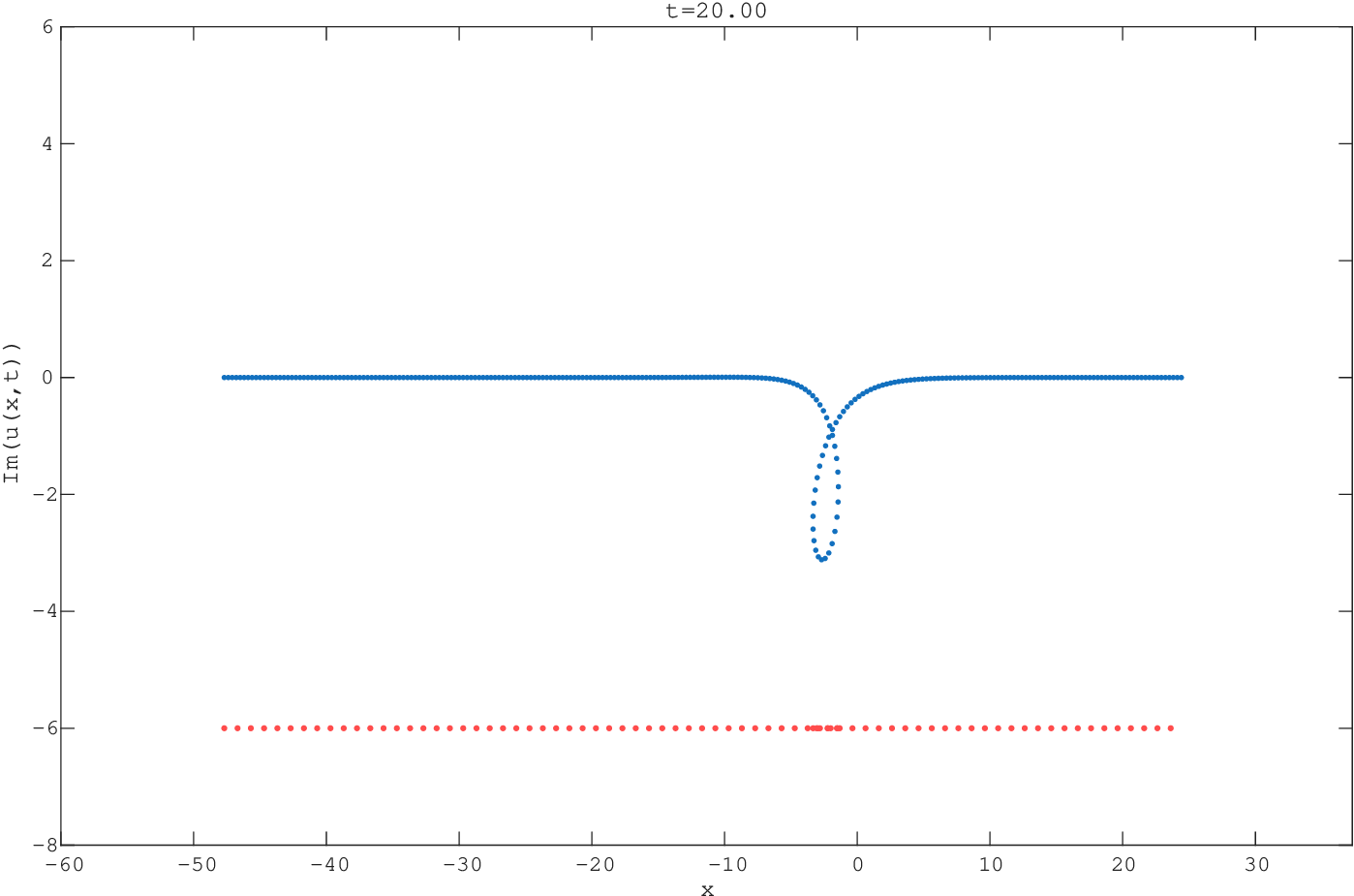}
      \end{minipage}
       \end{tabular}
     \caption{The numerical simulation of the Im($u$)-profile of the one-soliton solution for the CSP equation. maxerr=1.28$\times 10^{-4}$}
              \label{compSP_imag_1}
  \end{figure}
\begin{figure}[h]
 \centering
 \begin{tabular}{cc}
      \begin{minipage}[t]{0.47\hsize}
       \centering
        \includegraphics[width=\linewidth]{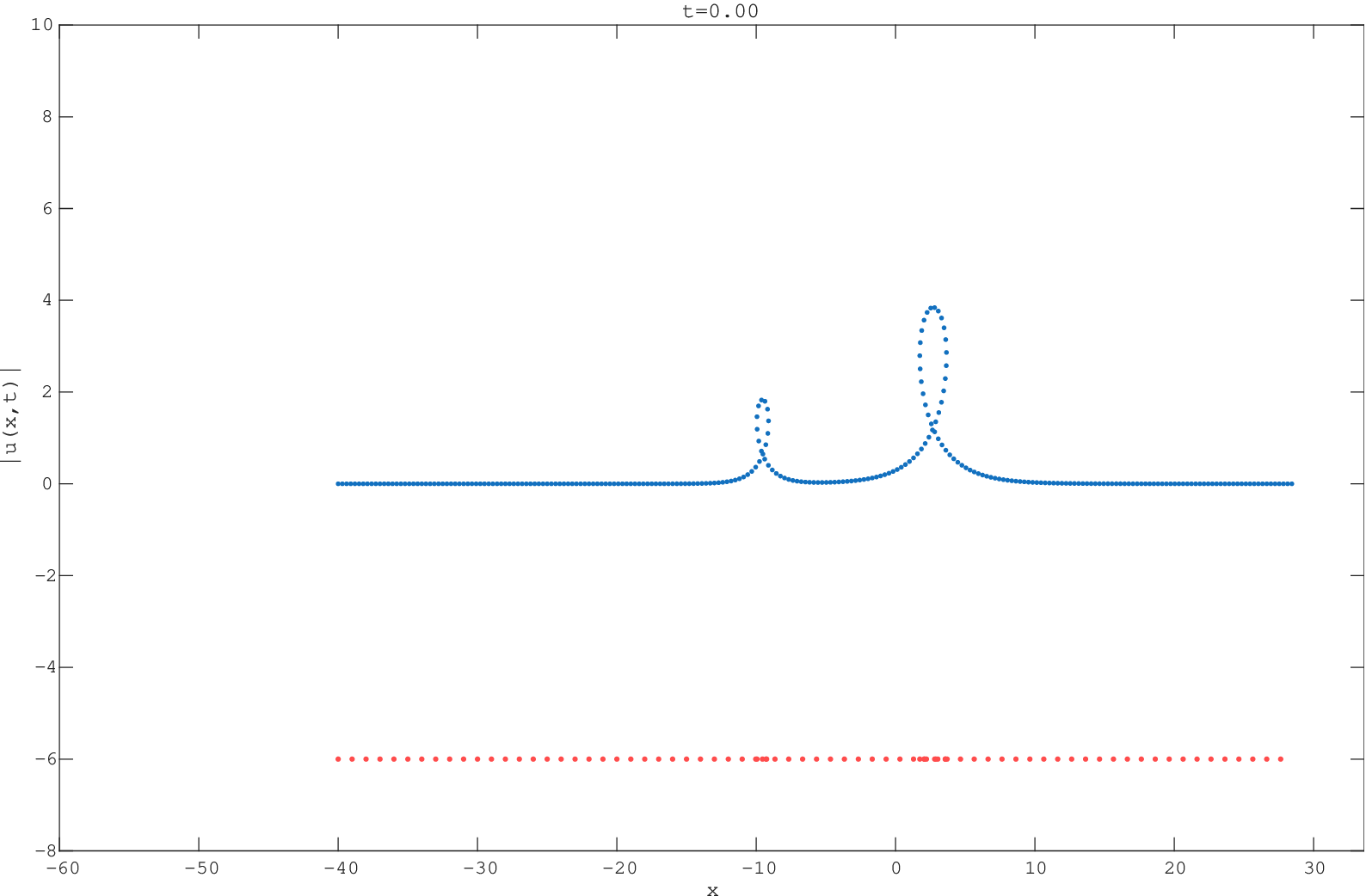}
      \end{minipage} &
      \begin{minipage}[t]{0.47\hsize}
        \centering
        \includegraphics[width=\linewidth]{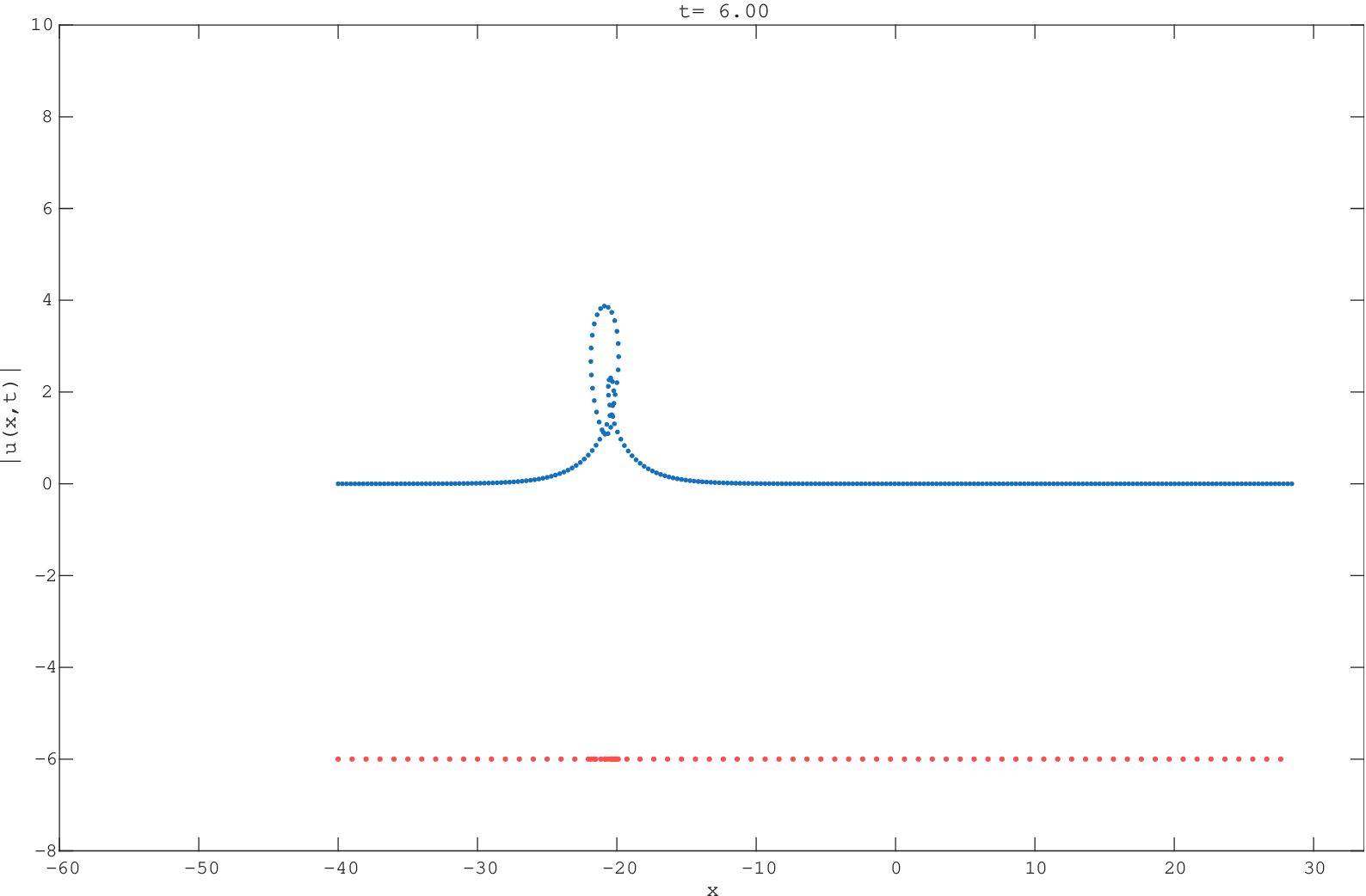}
      \end{minipage}\\

      \begin{minipage}[t]{0.47\hsize}
        \centering
        \includegraphics[width=\linewidth]{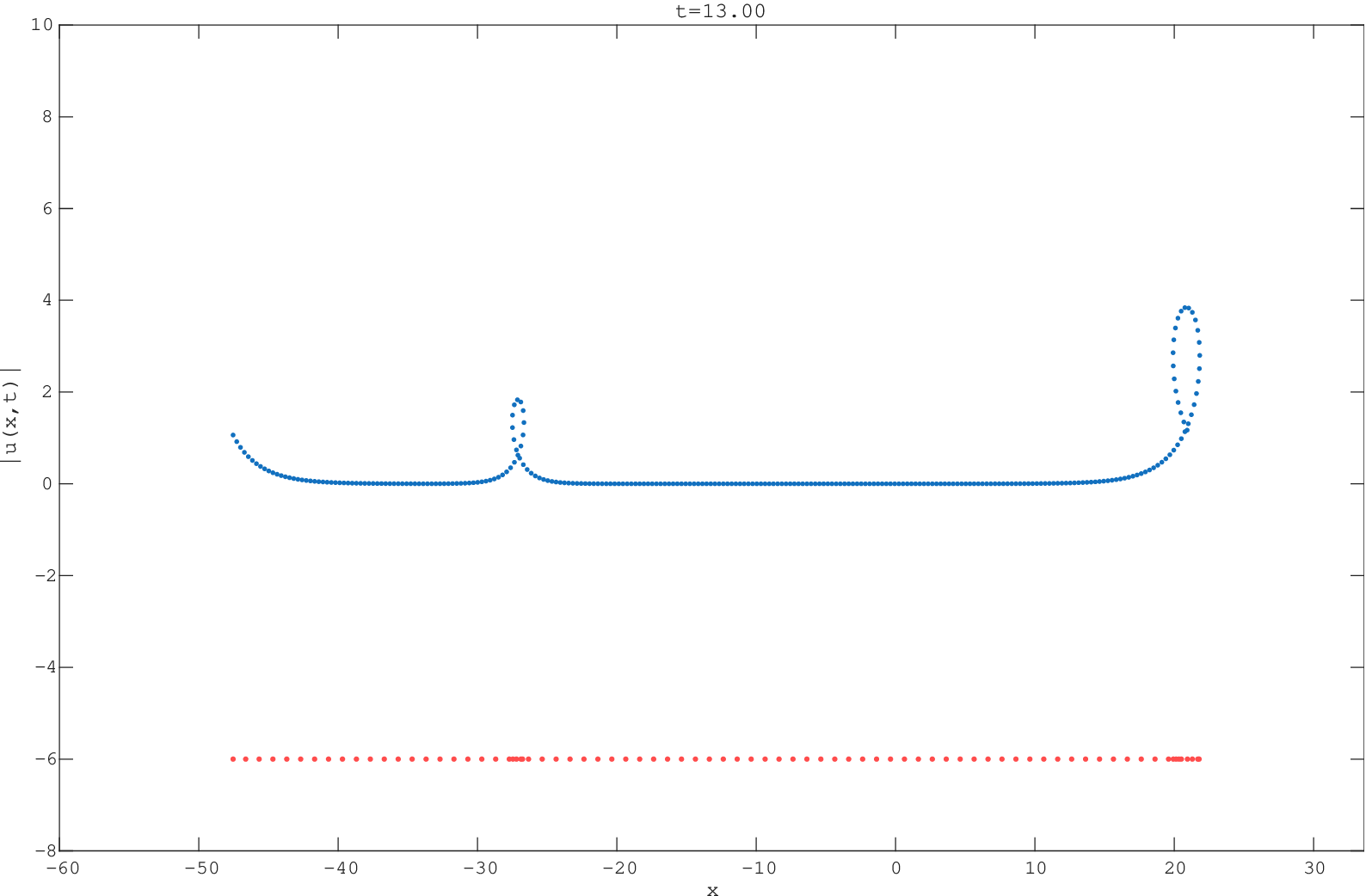}
      \end{minipage} &
      \begin{minipage}[t]{0.47\hsize}
        \centering
        \includegraphics[width=\linewidth]{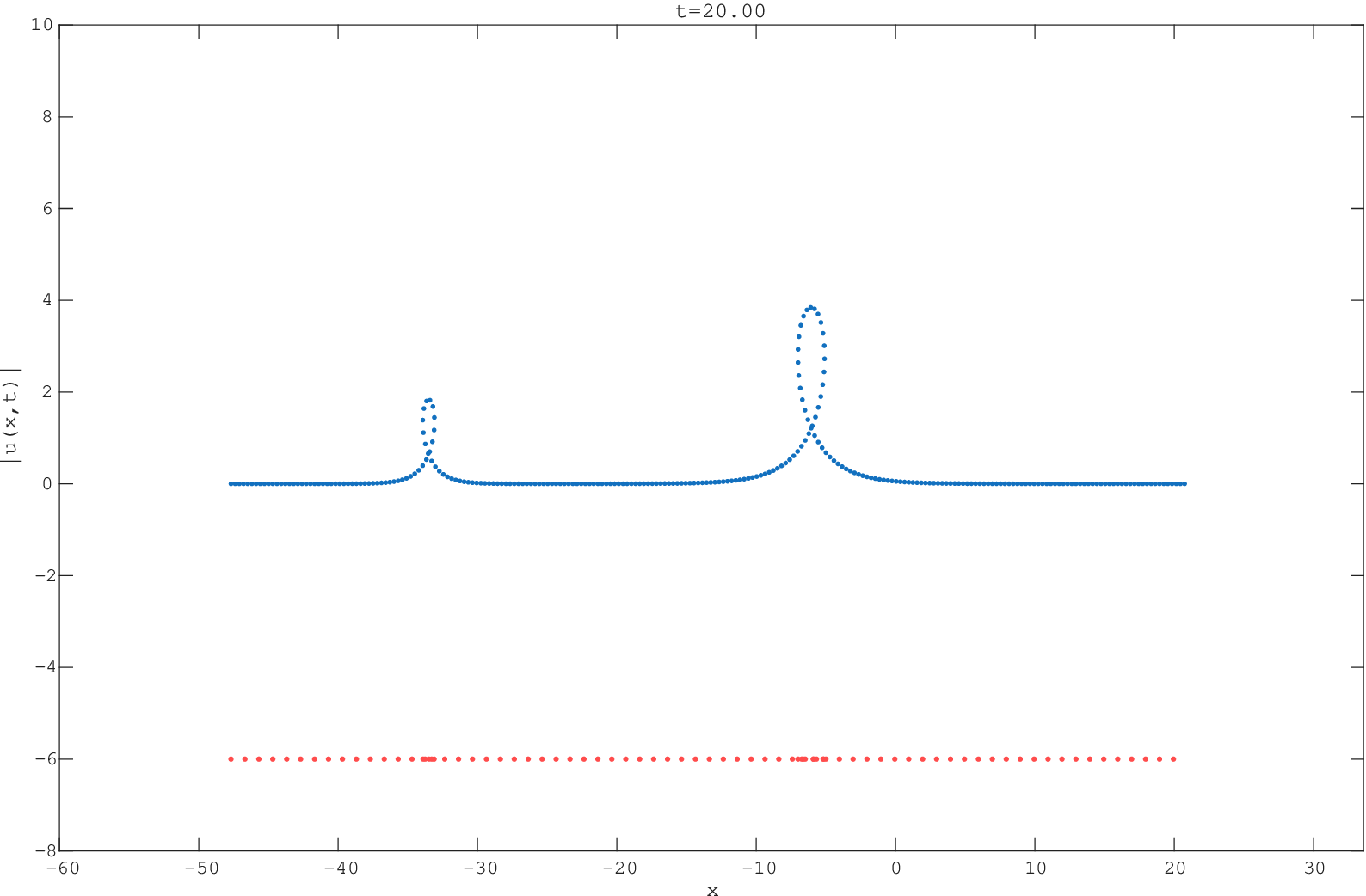}
      \end{minipage}
       \end{tabular}
     \caption{The numerical simulation of the $|u|$-profile of the two-soliton solution for the CSP equation. maxerr=1.06$\times 10^{-4}$}
              \label{compSP_abs}
  \end{figure}
\begin{figure}[h]
 \centering
 \begin{tabular}{cc}
      \begin{minipage}[t]{0.47\hsize}
       \centering
        \includegraphics[width=\linewidth]{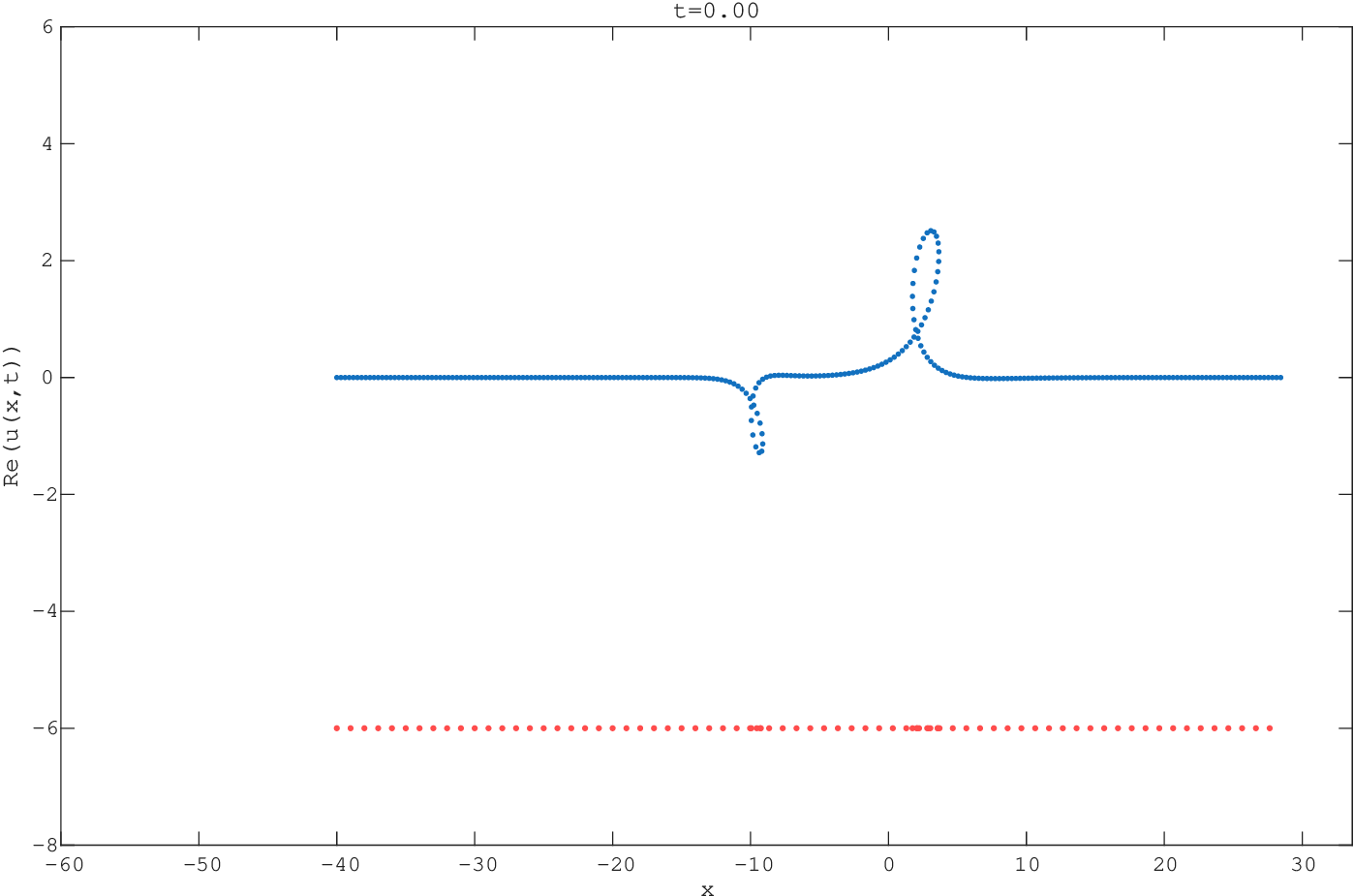}
      \end{minipage} &
      \begin{minipage}[t]{0.47\hsize}
        \centering
        \includegraphics[width=\linewidth]{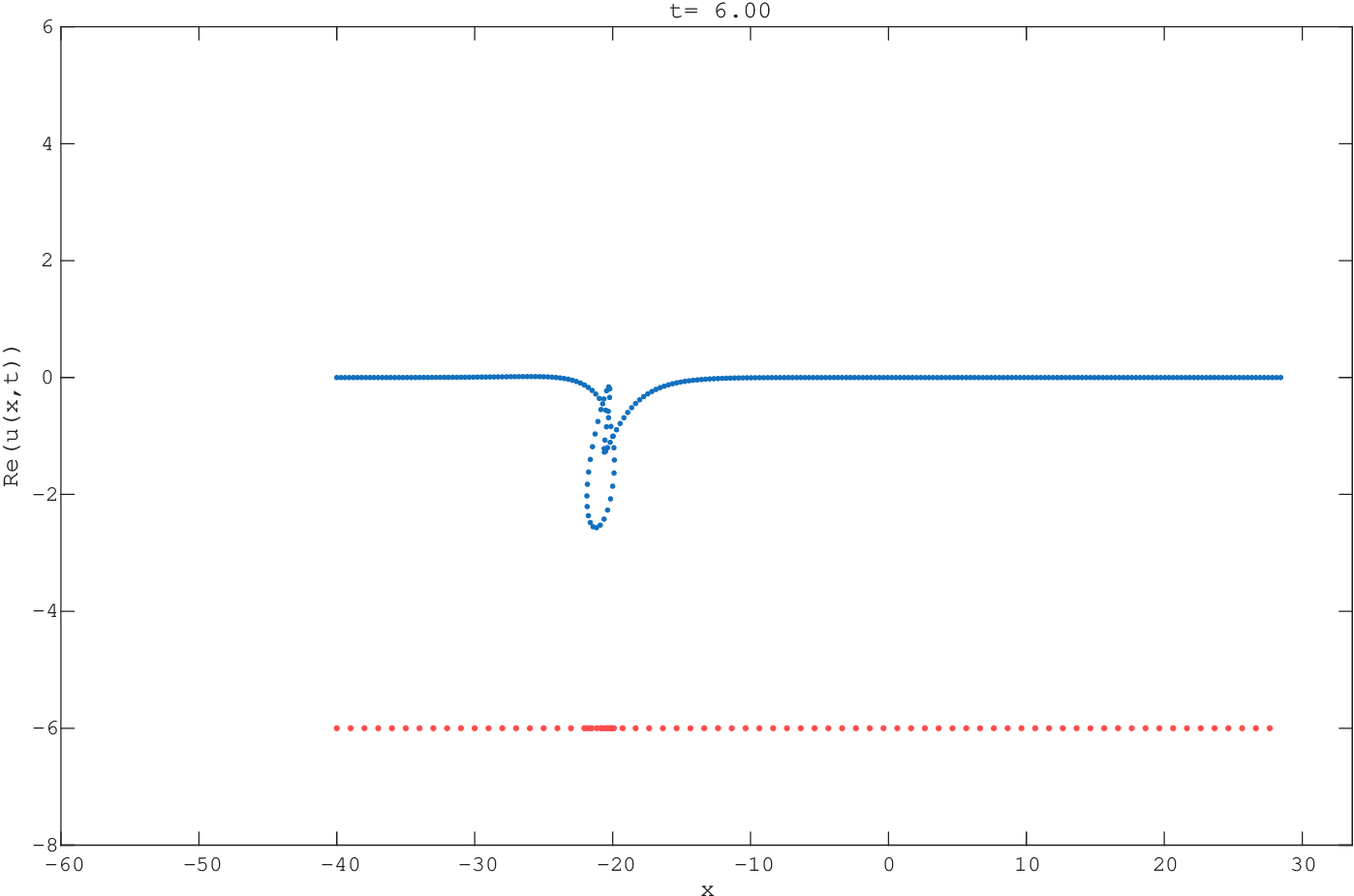}
      \end{minipage}\\

      \begin{minipage}[t]{0.47\hsize}
        \centering
        \includegraphics[width=\linewidth]{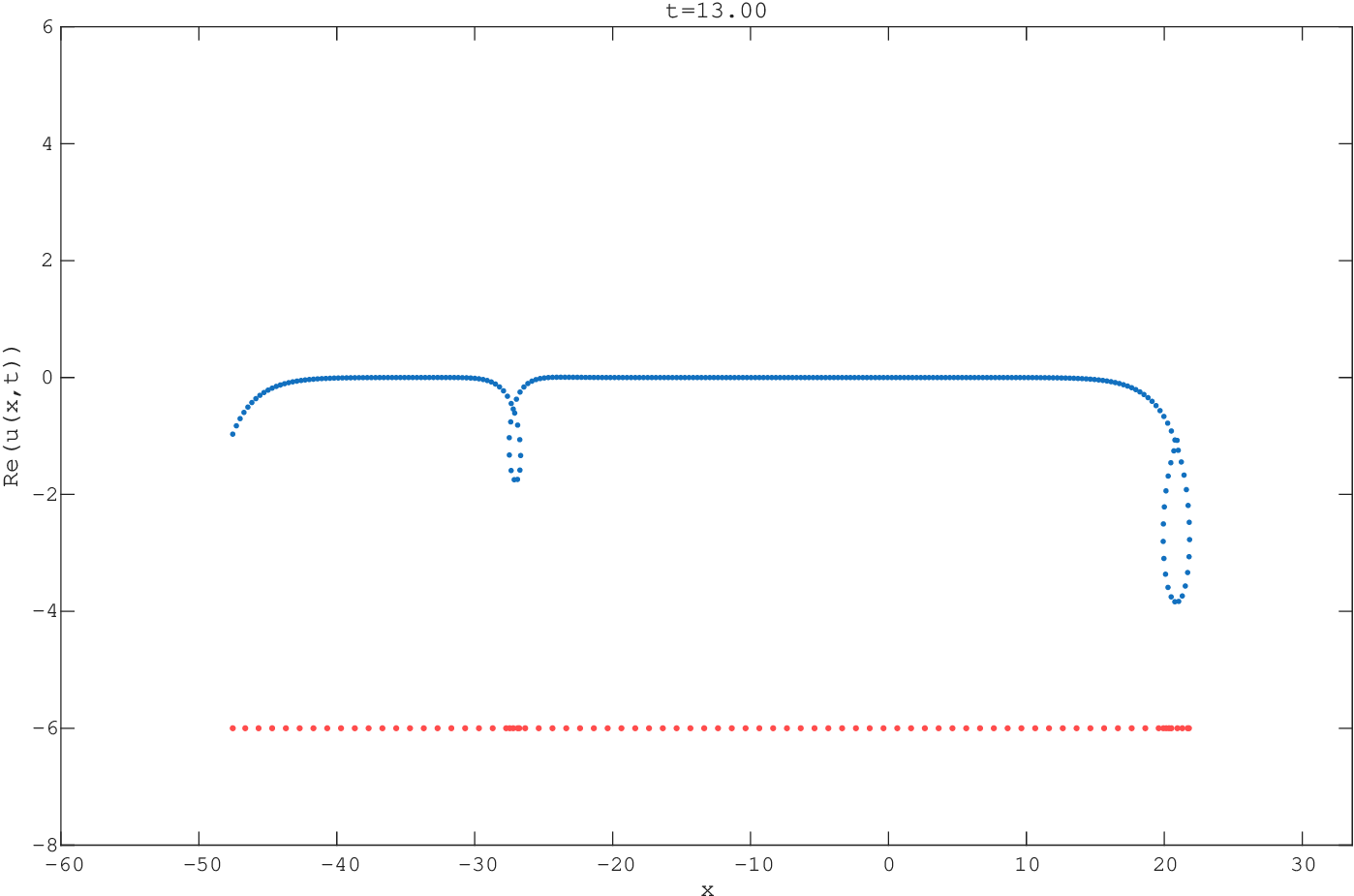}
      \end{minipage} &
      \begin{minipage}[t]{0.47\hsize}
        \centering
        \includegraphics[width=\linewidth]{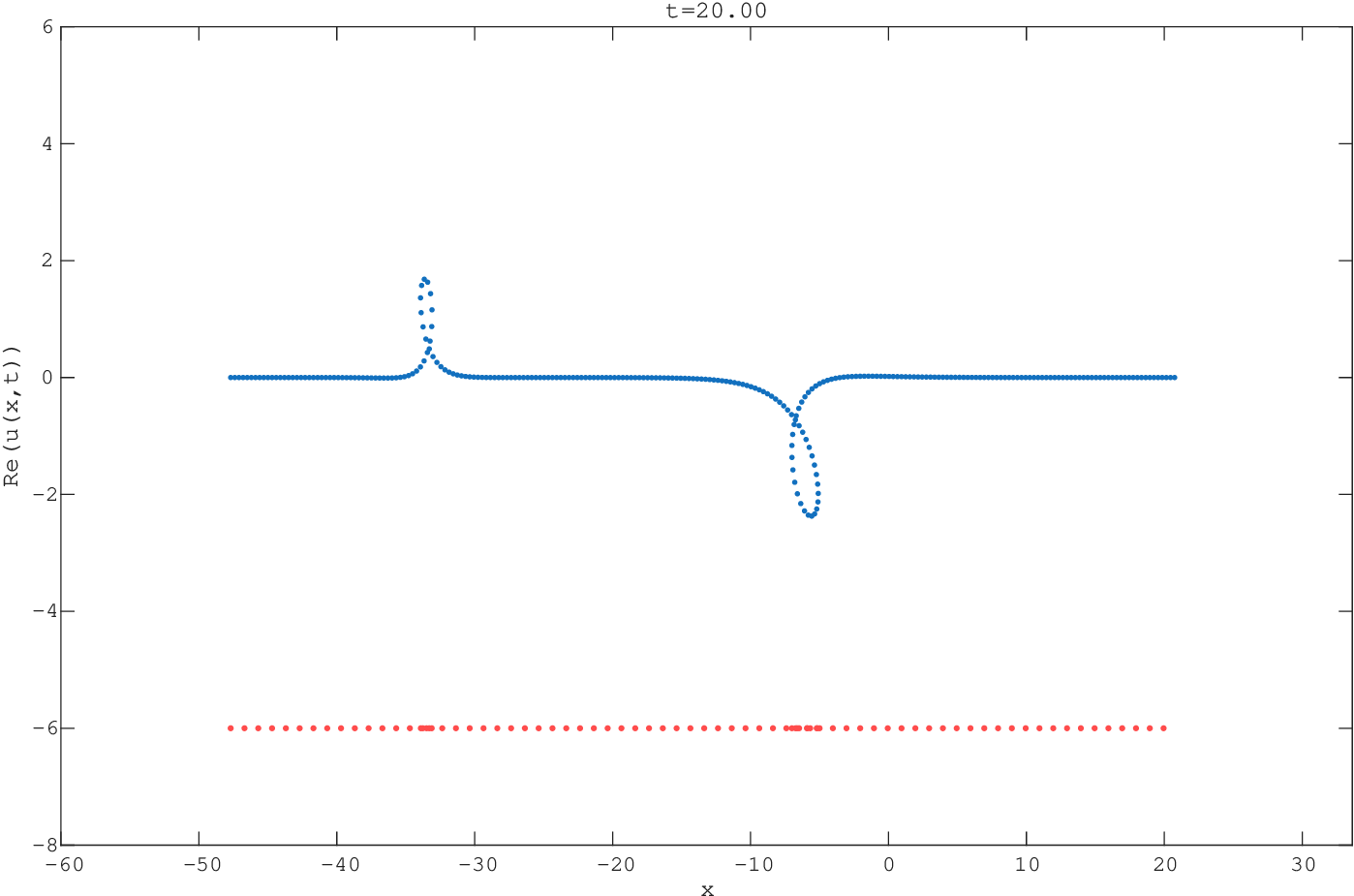}
      \end{minipage}
       \end{tabular}
     \caption{The numerical simulation of the Re($u$)-profile of the two-soliton solution for the CSP equation. maxerr=1.86$\times 10^{-4}$}
              \label{compSP_real}
  \end{figure}
\begin{figure}[h]
 \centering
 \begin{tabular}{cc}
      \begin{minipage}[t]{0.47\hsize}
       \centering
        \includegraphics[width=\linewidth]{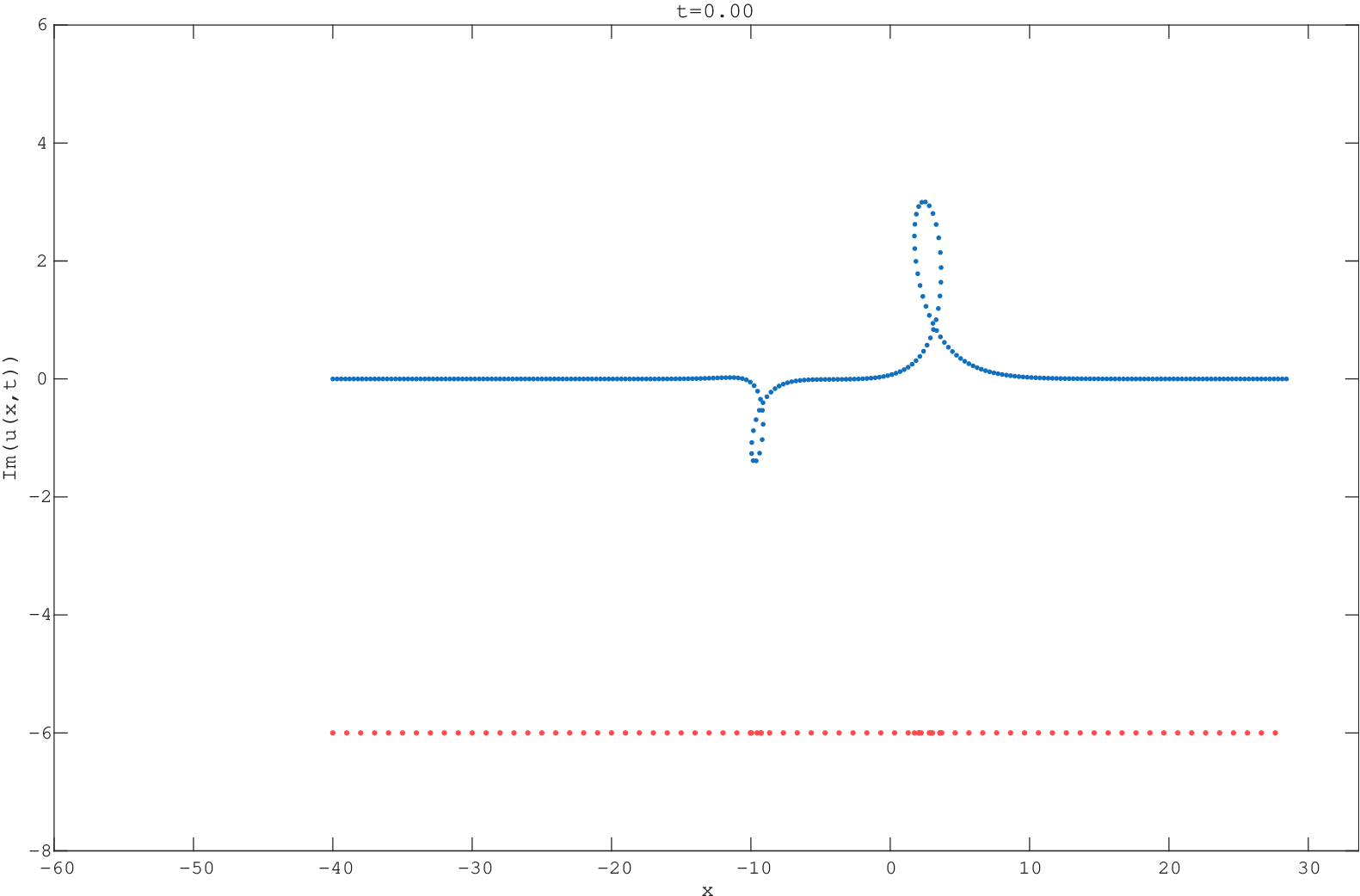}
      \end{minipage} &
      \begin{minipage}[t]{0.47\hsize}
        \centering
        \includegraphics[width=\linewidth]{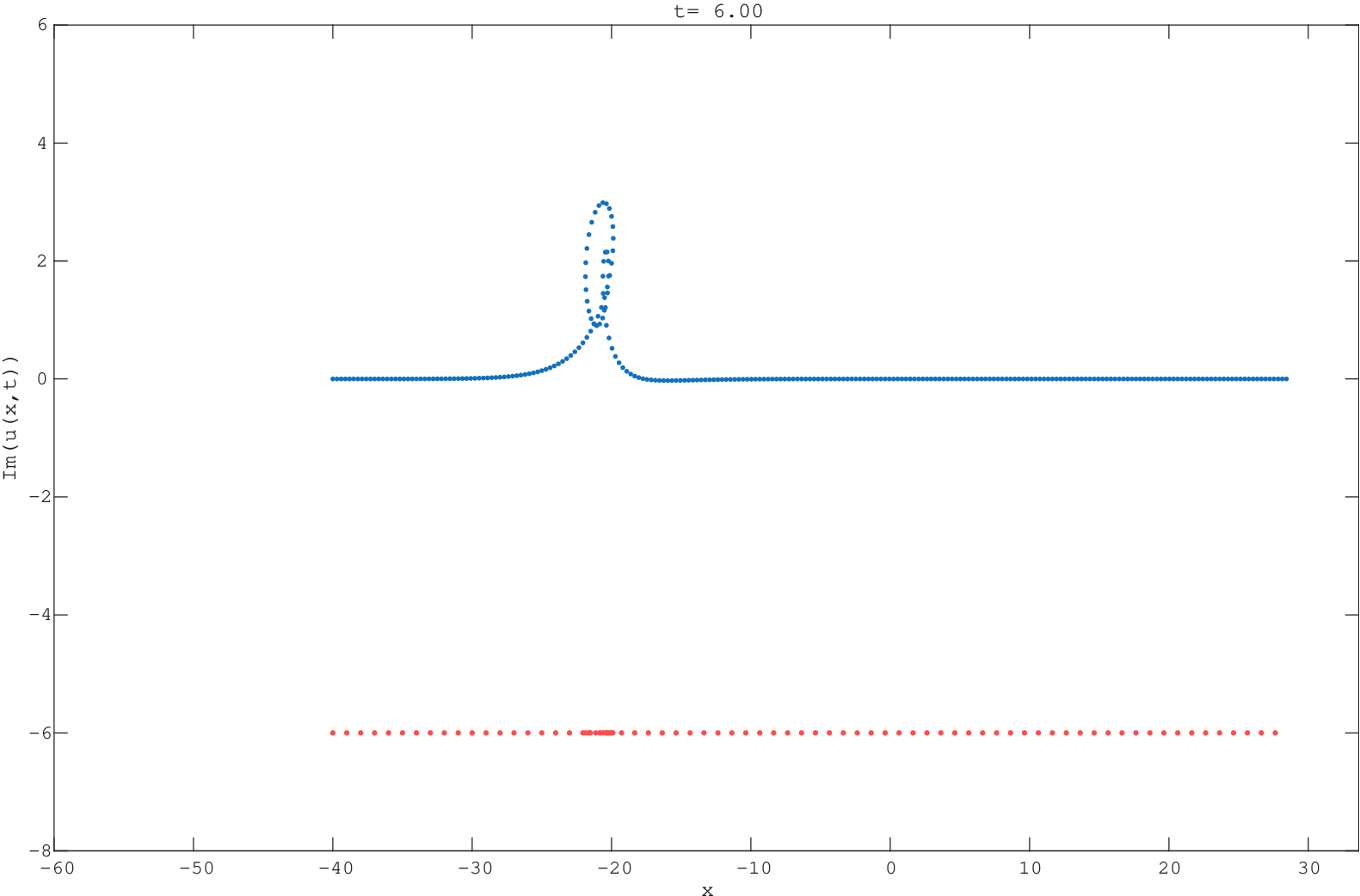}
      \end{minipage}\\

      \begin{minipage}[t]{0.47\hsize}
        \centering
        \includegraphics[width=\linewidth]{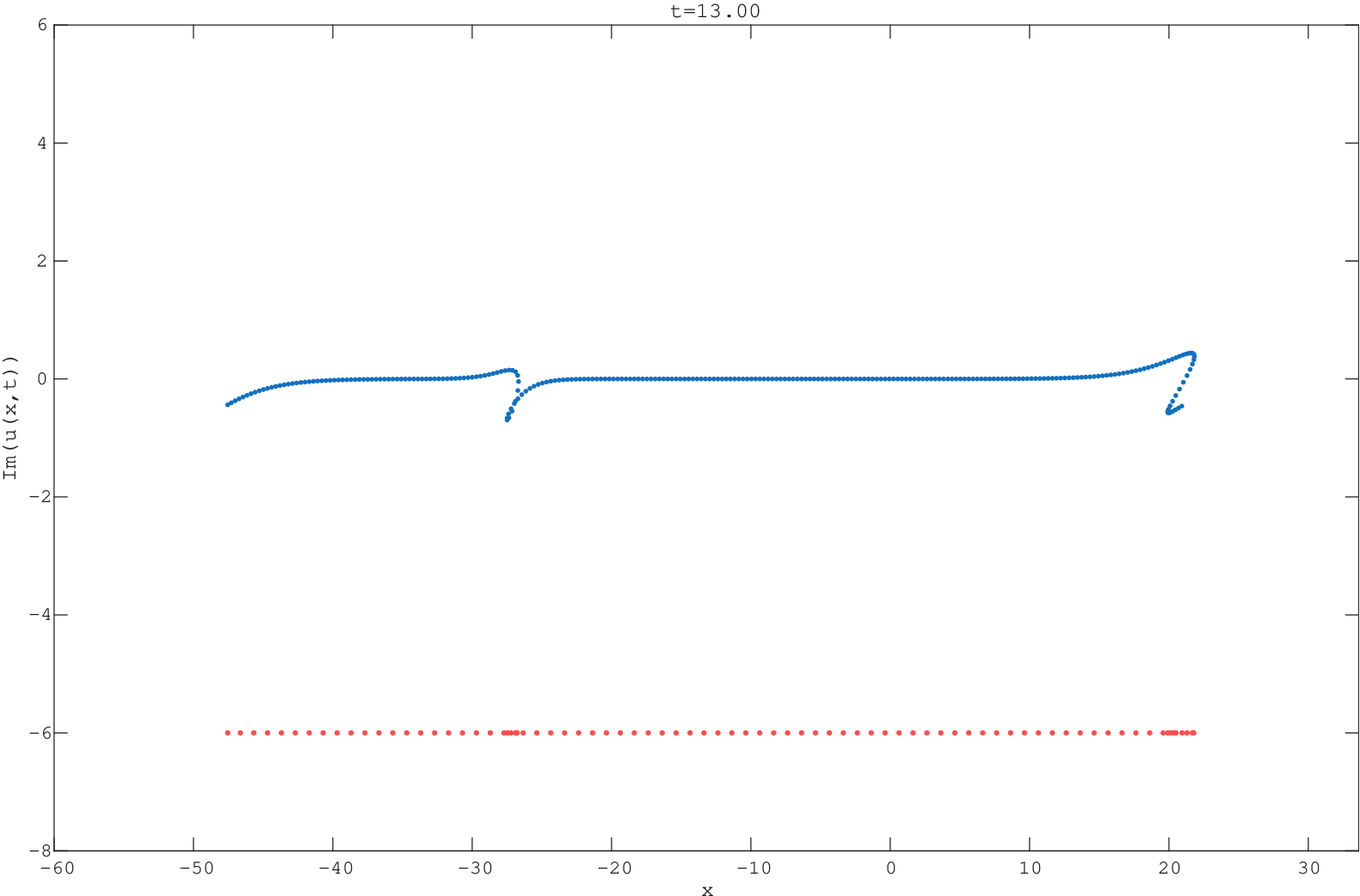}
      \end{minipage} &
      \begin{minipage}[t]{0.47\hsize}
        \centering
        \includegraphics[width=\linewidth]{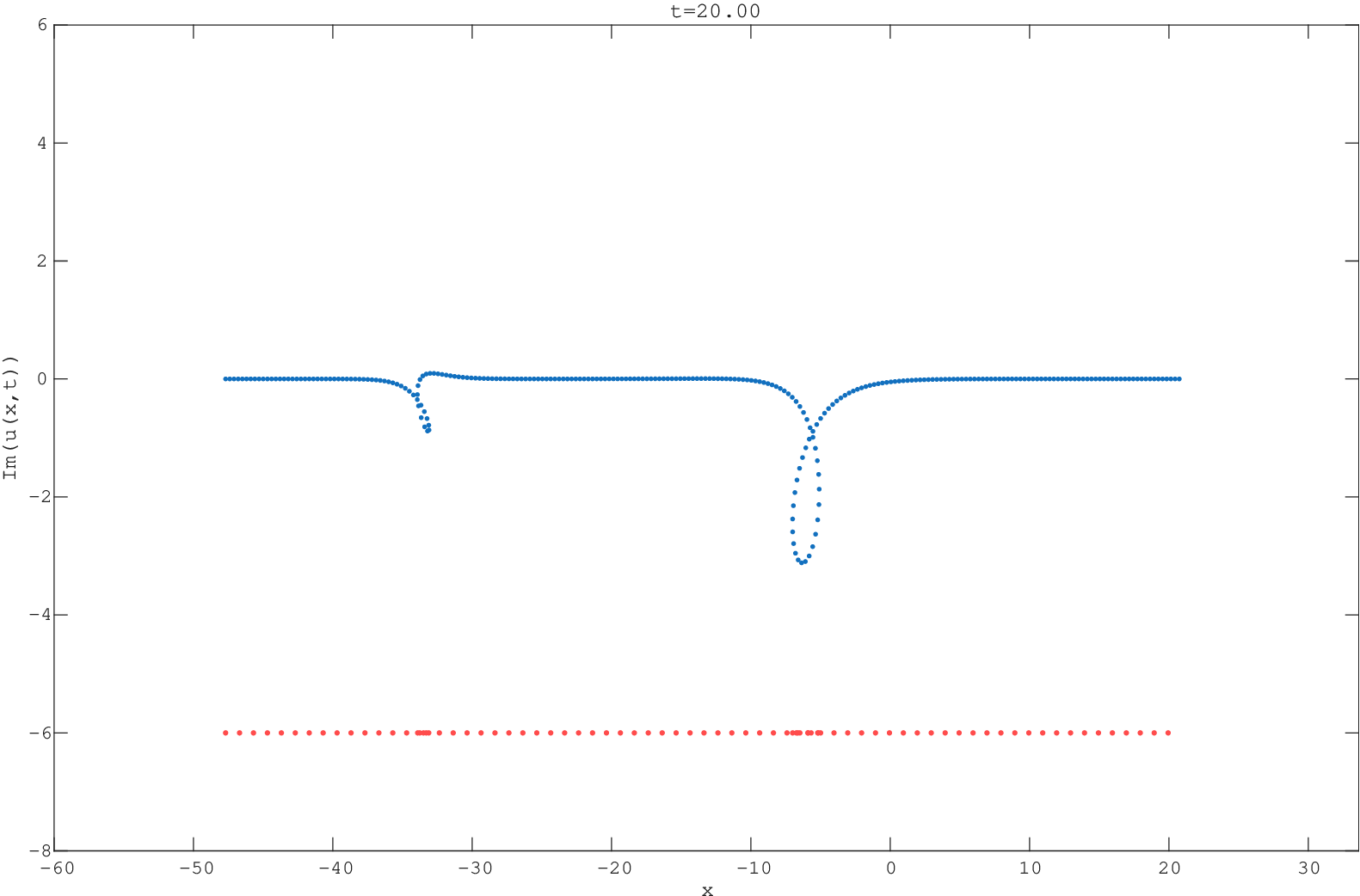}
      \end{minipage}
       \end{tabular}
     \caption{The numerical simulation of the Im($u$)-profile of the two-soliton solution for the CSP equation. maxerr=1.73$\times 10^{-4}$}
              \label{compSP_imag}
  \end{figure}

\subsection{Numerical simulations of the defocusing CSP equation under one-end Dirichlet boundary conditions}
\label{subsec:defocusing_CSP_Dirichlet}

In this subsection, we perform numerical simulations of the defocusing CSP equation under one-end Dirichlet boundary conditions. Taking $n=2$, $u^{(1)}=u$, $u^{(2)}=u^*$, and $c_{12}=-1$ in the MCSP equation (\ref{MCSP}), we obtain the defocusing CSP equation~\cite{FengLingZhu2016, MCSP_def},
\begin{eqnarray}
u_{xt}=u-\frac{1}{2}(|u|^2u_x)_x,
\label{defocusingCSP}
\end{eqnarray}
where $u^*$ denotes the complex conjugate of $u$. Based on the results in Section~\ref{sec_disMCSP}, the corresponding self-adaptive moving mesh scheme is
\begin{eqnarray}
\left\{\begin{array}{ll}
\displaystyle\frac{d}{dT}(u_{l+1}-u_l)=\frac{\delta_l}{2}(u_{l+1}+u_l),\\
\\
\displaystyle\frac{dx_l}{dT}=\frac{1}{2}|u_l|^2.
\end{array}\right.
\label{semidisdefocusingCSP}
\end{eqnarray}
The resulting mesh intervals satisfy
\begin{eqnarray}
\frac{d\delta_l}{dT}=\frac{1}{2}(|u_{l+1}|^2-|u_l|^2).
\label{defocusingCSPdelta}
\end{eqnarray}
Equation (\ref{semidisdefocusingCSP}) is obtained by reducing the self-adaptive moving mesh scheme for the MCSP equation derived in Section~\ref{sec_disMCSP}. The first equation of (\ref{semidisdefocusingCSP}), together with (\ref{defocusingCSPdelta}), coincides with the known integrable semi-discrete defocusing CSP system in~\cite{MCSP_def_semi}.

As in Subsection~\ref{subsec:defocusingCmSPDirichlet}, we perform numerical simulations
of the one- and two-dark-soliton solutions of (\ref{defocusingCSP}) and compare the
numerical solutions with the exact solutions. These exact solutions were obtained
in~\cite{MCSP_def} and are expressed as
\begin{eqnarray}
 \begin{aligned}
 u_e=\frac{\beta}{2}\frac{g}{f}\exp\!\left[\mathrm{i}\left(X+\frac{\gamma}{2}T\right)\right],\qquad
 x_e=-\frac{\gamma}{2}X+\frac{\beta^2}{8}T-2\bigl(\log f\bigr)_T,
 \qquad t=T.
 \end{aligned}
 \label{defocusingCSPexact}
\end{eqnarray}
For the one-dark-soliton solution,
\begin{eqnarray}
 f=1+z_1,\qquad g=1+\alpha_1z_1,
 \label{defocusingCSP1sol}
\end{eqnarray}
whereas for the two-dark-soliton solution,
\begin{eqnarray}
 \begin{aligned}
 f&=1+z_1+z_2+A_{12}z_1z_2,\\
 g&=1+\alpha_1z_1+\alpha_2z_2+A_{12}\alpha_1\alpha_2z_1z_2,
 \end{aligned}
 \label{defocusingCSP2sol}
\end{eqnarray}
with
\begin{eqnarray}
 \begin{aligned}
 &z_j=\exp\bigl[\kappa_j(X-X_{j,0})+\omega_jT\bigr],\qquad
 \alpha_j=e^{-2\mathrm{i}\theta_j},\qquad
 \kappa_j=-\frac{\beta\sin\theta_j}{\beta\cos\theta_j-\gamma},\\
 &\omega_j=-\frac{\beta}{2}\sin\theta_j,\qquad
 A_{12}=\frac{\sin^2\!\left((\theta_2-\theta_1)/2\right)}{\sin^2\!\left((\theta_2+\theta_1)/2\right)}.
 \end{aligned}
 \label{defocusingCSPparameters}
\end{eqnarray}
Here, $\beta>0$, $\gamma<0$, $\theta_j$ and $X_{j,0}$ are real parameters, so that $f$
is real and $g$ is complex.

As in Subsection~\ref{subsec:defocusingCmSPDirichlet}, we use the improved Euler method with the same parameters $L$, $D$, $\Delta t$ and lattice spacing $2a=D/L$, carry out the computation up to $t=40$, take the edge point at $X_0=-D/2$, impose the Dirichlet boundary condition, and take the initial values from the exact solution (\ref{defocusingCSPexact}). The only
differences are that the time derivatives $du_l/dT$ are computed by
(\ref{eq:MCmSP_right_cumulative}) with $\delta_k-a$ replaced by $\delta_k/2$, and that
the mesh points evolve according to the second equation of
(\ref{semidisdefocusingCSP}).

The background-normalized amplitude error $\mathrm{err}_m$ and the quantity
$\mathrm{maxerr}$ are defined in the same manner as in
(\ref{defocusingCmSPerror}) and (\ref{defocusingCmSPmaxerr}), with the background
amplitude $u_{\mathrm{bg}}=\beta/2$.


\begin{figure}[h]
\centering
 \begin{tabular}{cc}
      \begin{minipage}[t]{0.47\hsize}
       \centering
        \includegraphics[width=\linewidth]{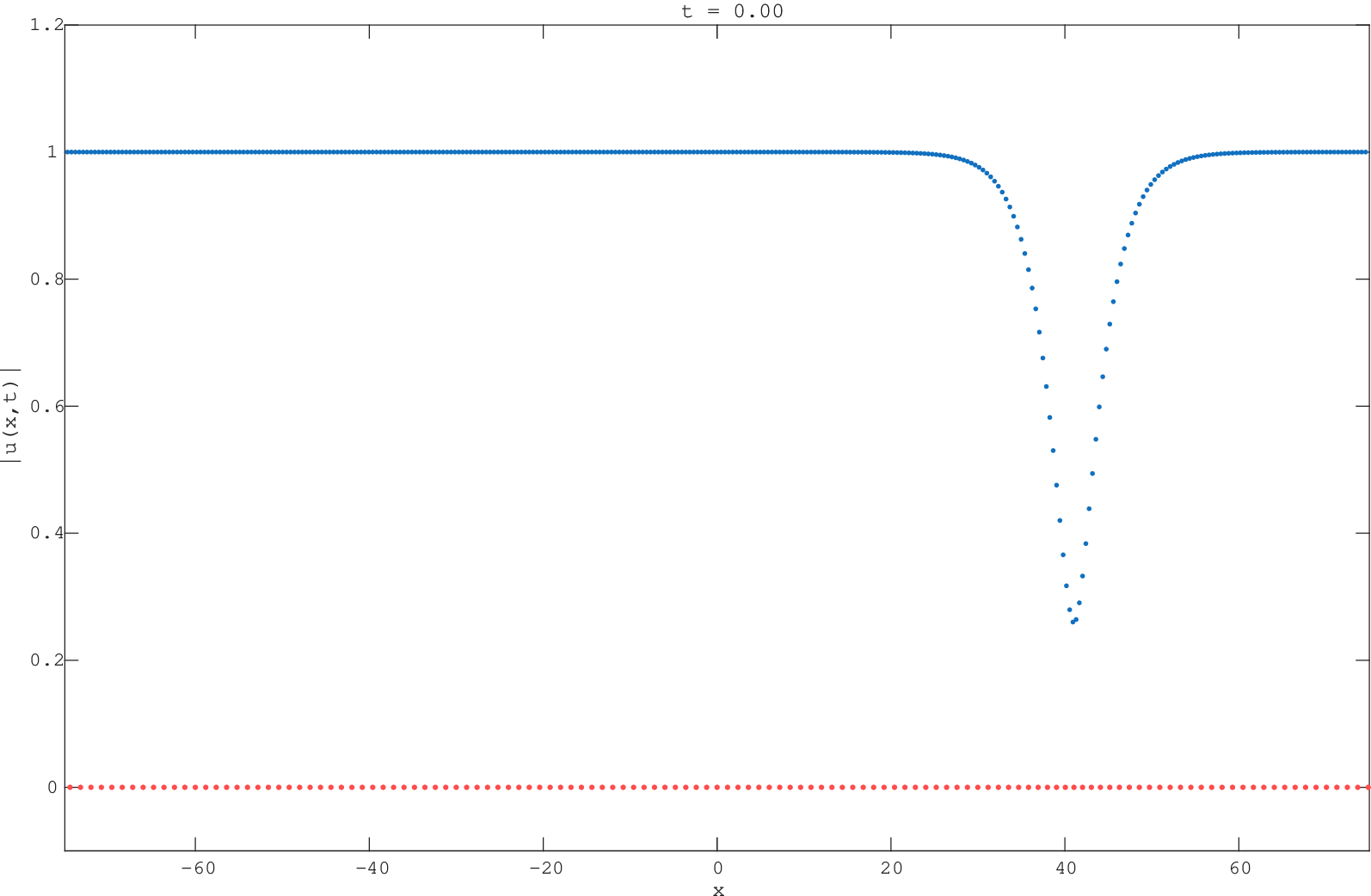}
      \end{minipage} &
      \begin{minipage}[t]{0.47\hsize}
        \centering
        \includegraphics[width=\linewidth]{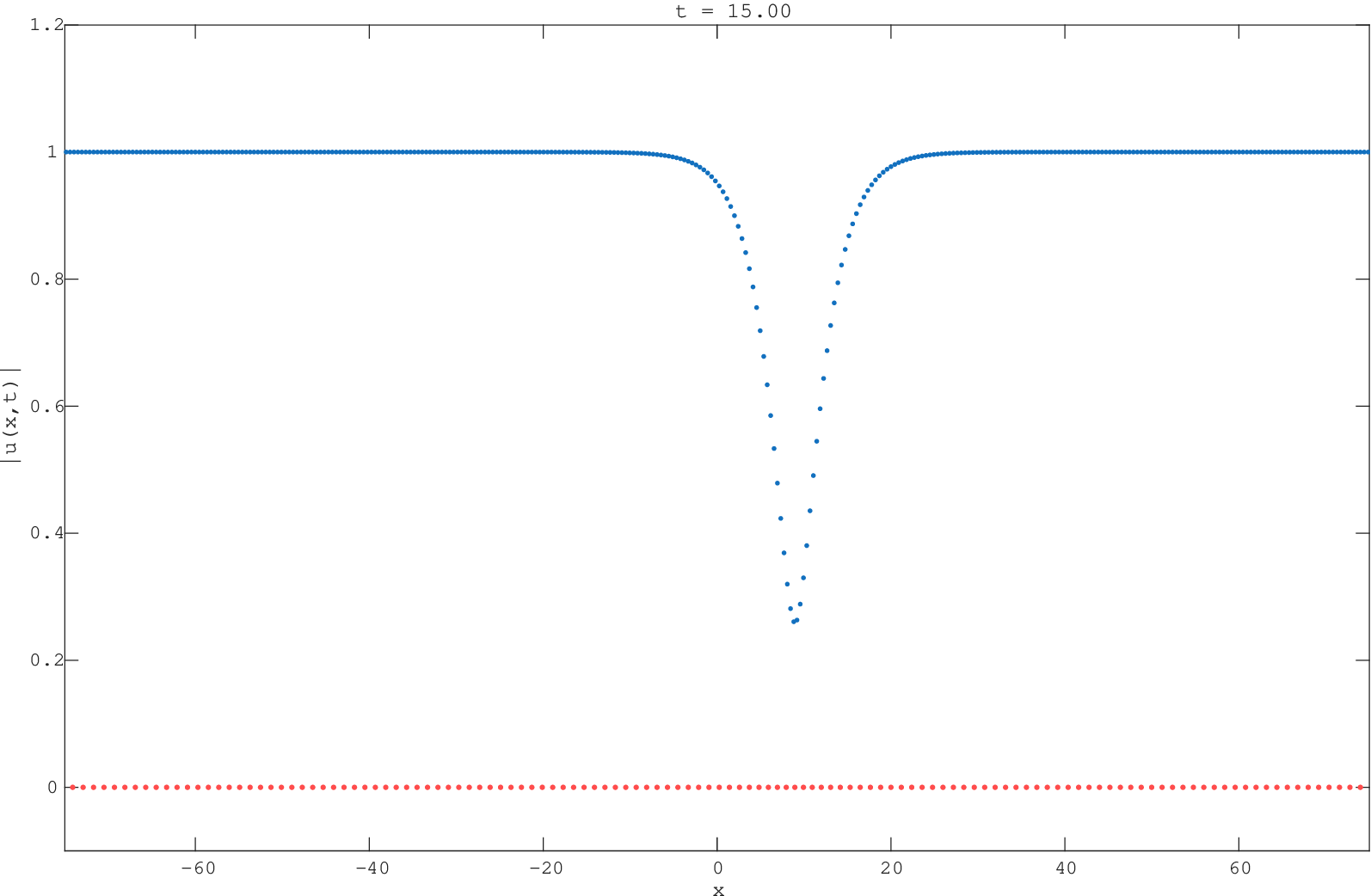}
      \end{minipage}\\

      \begin{minipage}[t]{0.47\hsize}
        \centering
        \includegraphics[width=\linewidth]{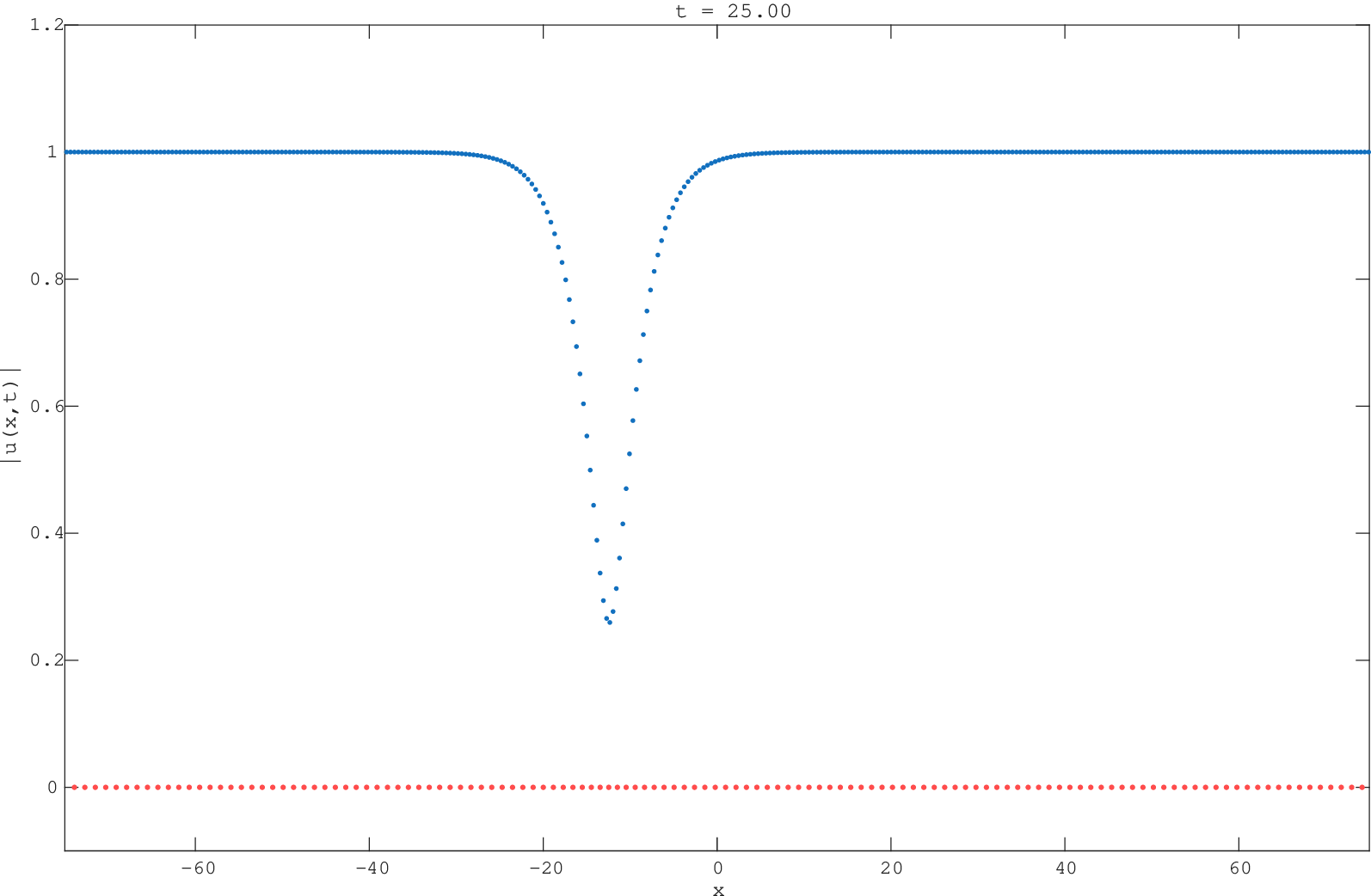}
      \end{minipage} &
      \begin{minipage}[t]{0.47\hsize}
        \centering
        \includegraphics[width=\linewidth]{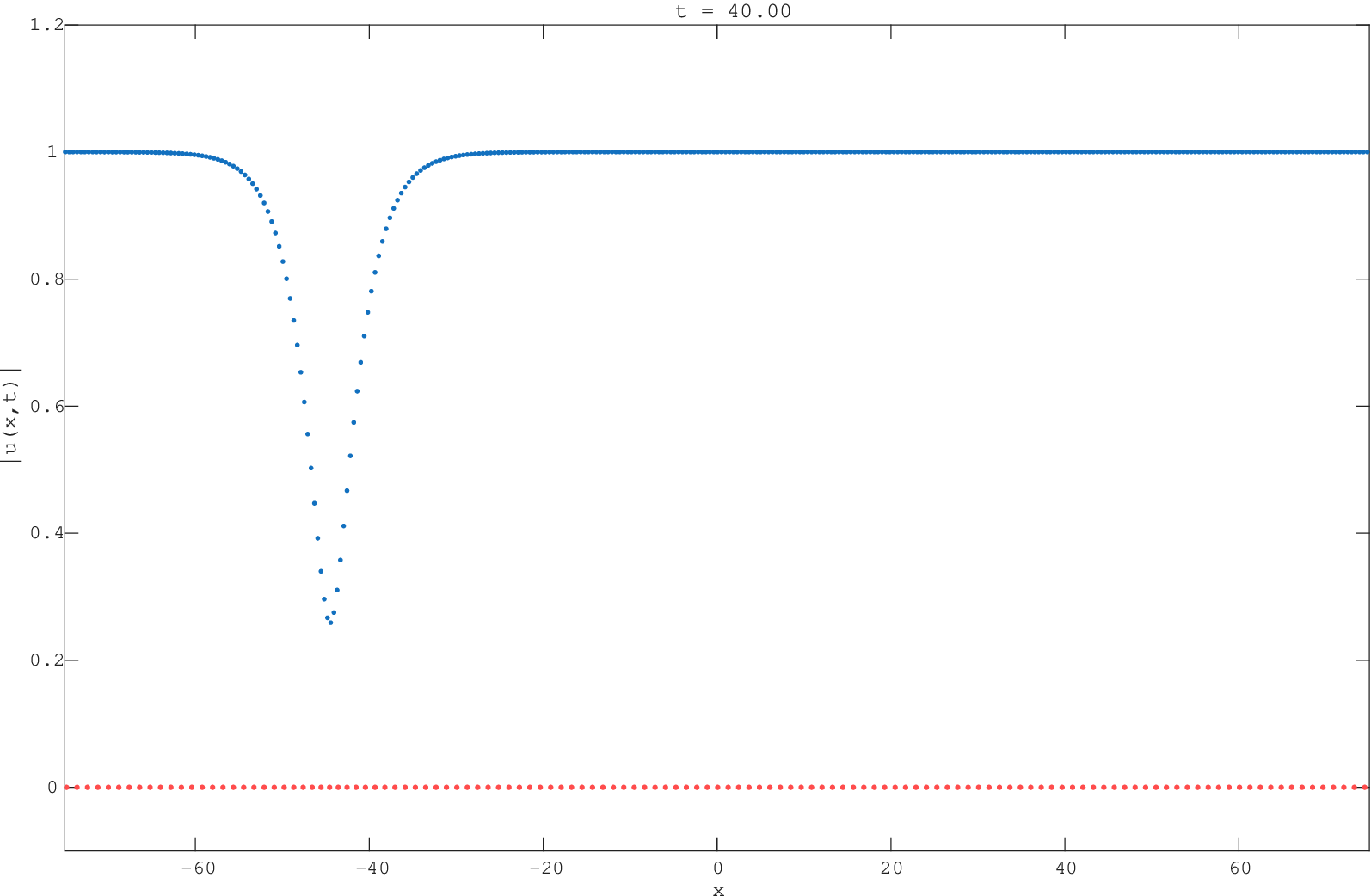}
      \end{minipage}
       \end{tabular}
     \caption{The numerical simulation of the $|u|$-profile of the one-dark-soliton solution for the defocusing CSP equation. maxerr=4.49$\times 10^{-5}$}
              \label{dSP_1}
  \end{figure}


\begin{figure}[h]
\centering
 \begin{tabular}{cc}
      \begin{minipage}[t]{0.47\hsize}
       \centering
        \includegraphics[width=\linewidth]{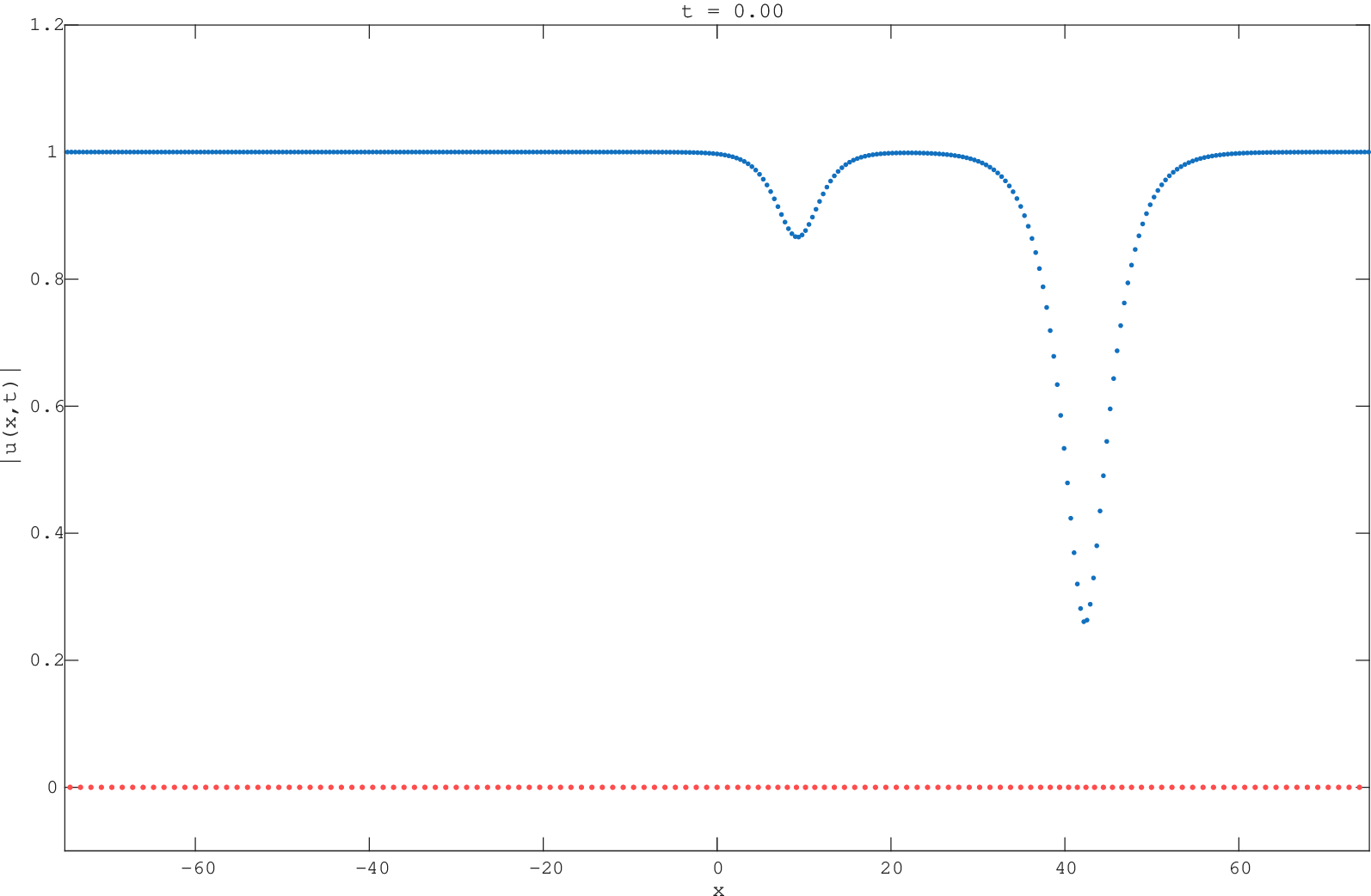}
      \end{minipage} &
      \begin{minipage}[t]{0.47\hsize}
        \centering
        \includegraphics[width=\linewidth]{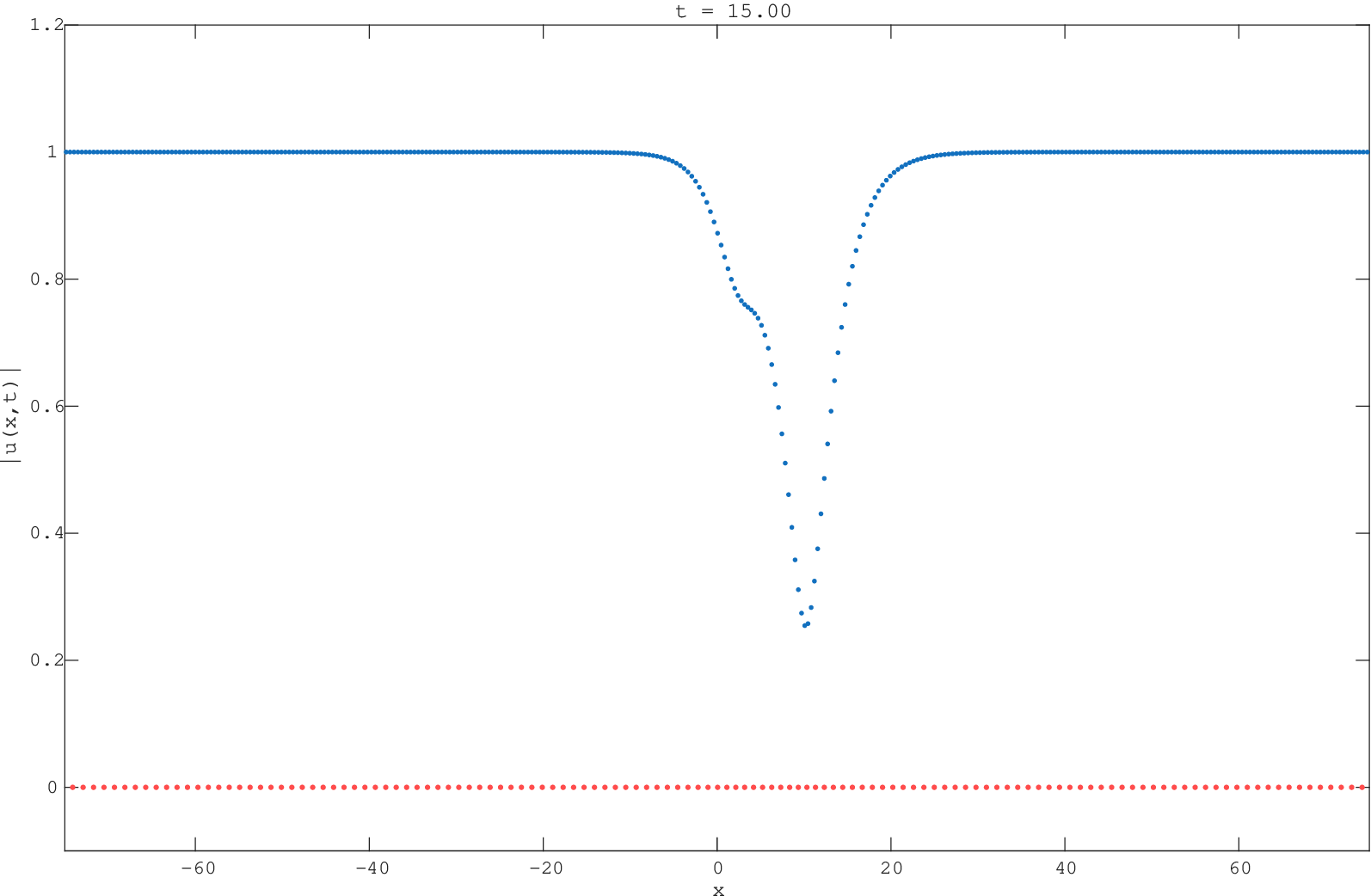}
      \end{minipage}\\

      \begin{minipage}[t]{0.47\hsize}
        \centering
        \includegraphics[width=\linewidth]{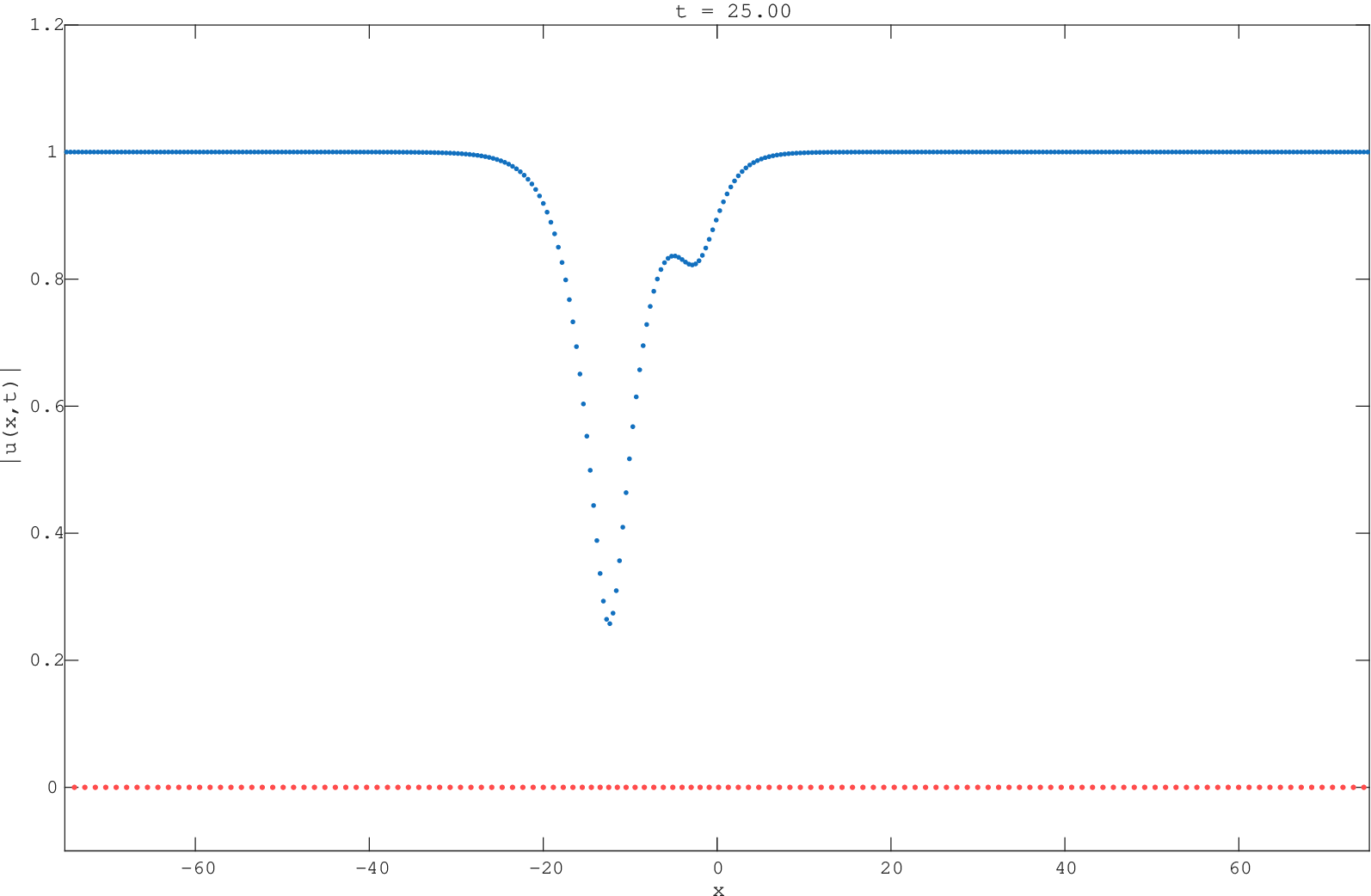}
      \end{minipage} &
      \begin{minipage}[t]{0.47\hsize}
        \centering
        \includegraphics[width=\linewidth]{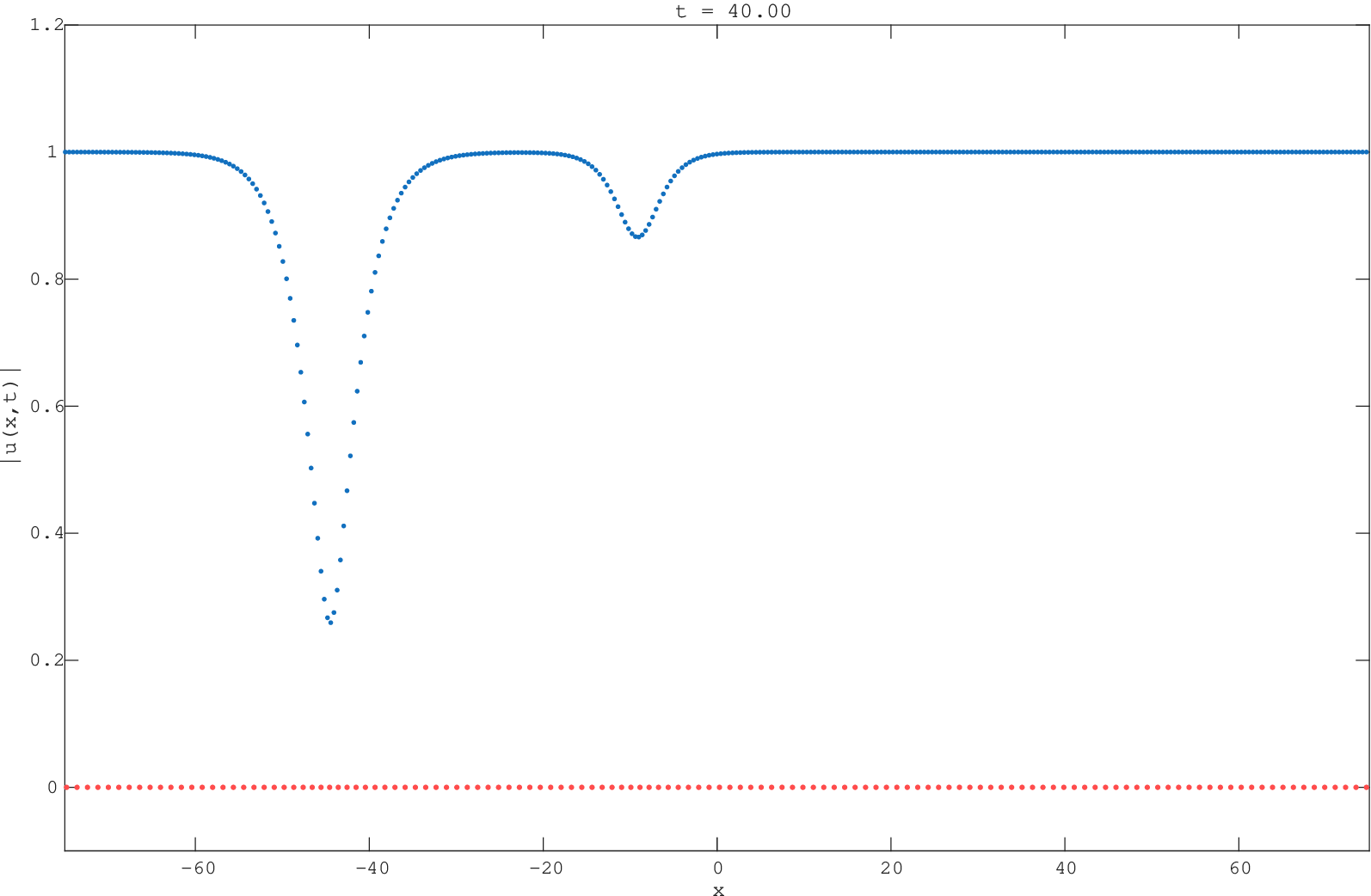}
      \end{minipage}
       \end{tabular}
     \caption{The numerical simulation of the $|u|$-profile of the two-dark-soliton solution for the defocusing CSP equation. maxerr=4.38$\times 10^{-5}$}
              \label{dSP_2}
  \end{figure}

Figure \ref{dSP_1} shows a numerical simulation of the $|u|$-profile of the one-dark-soliton solution for $\beta=2,\gamma=-3,\theta_1=-5\pi/12$, and $X_{1,0}=28$. Figure \ref{dSP_2} shows a numerical simulation of the $|u|$-profile of the two-dark-soliton solution for $\beta=2,\gamma=-3,\theta_1=-5\pi/12, \theta_2=-5\pi/6, X_{1,0}=28$, and $X_{2,0}=6.5$. We chose these parameters so that the background plane wave is linearly stable, which is the case for $\beta\le|\gamma|$.

In both Subsections~\ref{subsec:CCSPperiodic} and \ref{subsec:defocusing_CSP_Dirichlet}, the maxerr for the two-soliton solution, which involves soliton interaction, does not increase significantly compared with that for the one-soliton solution. This fact suggests that the self-adaptive moving mesh schemes (\ref{semidis2SP}), (\ref{semidiscomplexSP}) and (\ref{semidisdefocusingCSP}) are effective and accurate even in the presence of soliton interaction.

\end{section}

\begin{section}{Conclusion}
\label{sec_conc}
In this paper, we proposed integrable self-adaptive moving mesh schemes for the MCmSP and the MCSP equations under periodic and one-end Dirichlet boundary conditions with nonzero boundary values. We also constructed multi-soliton solutions in Pfaffian form for the proposed semi-discrete schemes. The schemes obtained in our previous work~\cite{maruno0, maruno12} were fixed-edge schemes and therefore did not include the evolution of the edge point $x_0$. However, $x_0$ evolves in time when the boundary flux at the edge point is nonzero. The main theoretical contribution of this paper is to derive the evolution equation for $x_0$ from the consistency condition with the hodograph transformation and to incorporate it into the integrable self-adaptive moving mesh scheme for multi-component systems, so that the resulting schemes still admit multi-soliton solutions in Pfaffian form. 

The construction in Sections~\ref{sec_MCmSP} and \ref{sec_disMCmSP} for the MCmSP equation and in Sections~\ref{sec_MCSP} and \ref{sec_disMCSP} for the MCSP equation proceeds as follows.

\begin{enumerate}[(1)]
\item
From the conservation law, the hodograph transformation and the derivative law are derived. Applying them to the MCmSP and MCSP equations, which belong to the WKI-type class, yields the multi-component coupled integrable dispersionless equations, which belong to the AKNS-type class.

\item
Applying the dependent variable transformation to the multi-component coupled integrable dispersionless equations yields bilinear equations together with their Pfaffian solutions.

\item
An integrable spatial discretization of the bilinear equations  produces semi-discrete bilinear equations together with their Pfaffian solutions.

\item
Applying the dependent variable transformation to the semi-discrete bilinear equations yields a semi-discrete analogue of the multi-component coupled integrable dispersionless equations.

\item
Applying the discrete hodograph transformation to the semi-discrete multi-component coupled integrable dispersionless equations yields the evolution equations for $u_{l+1}^{(i)}-u_{l}^{(i)}$ and the mesh intervals $\delta_{l}$.

\item
When the boundary flux determined by the boundary values is nonzero, the consistency condition with the hodograph transformation implies that the edge point $x_{0}$ must also evolve. We therefore treat $x_{0}$ as a function of $T$ and derive its evolution equation. Combining this edge-point equation with the evolution equations for $u_{l+1}^{(i)}-u_l^{(i)}$ and the mesh intervals yields a self-adaptive moving mesh scheme with nonzero boundary values.
\end{enumerate}

In Sections~\ref{sec_numexpMCmSP} and \ref{sec_numexpMCSP}, we performed numerical simulations for the 2-mSP, CmSP, 2-SP, and CSP equations under periodic boundary conditions, and for the defocusing CmSP and defocusing CSP equations under one-end Dirichlet boundary conditions, using the self-adaptive moving mesh schemes constructed in this paper. Under periodic boundary conditions, the semi-discrete systems involve algebraic constraints, whereas under one-end Dirichlet boundary conditions, they form systems of ordinary differential equations. In all the simulations, the maxerr for the two-soliton solutions, measured by the relative error for the bright solitons and by the background-normalized amplitude error for the dark solitons, did not increase significantly compared with that for the one-soliton solutions. This indicates that the proposed schemes are effective and accurate. Evolving the edge point together with the mesh points enables accurate numerical computation under periodic and one-end Dirichlet boundary conditions with nonzero boundary values.

Future work includes applying the proposed schemes to other problems with nonzero
boundary values, including those admitting rogue wave solutions. In addition, the dark-soliton simulations under one-end Dirichlet boundary conditions in this paper were intended mainly to verify the proposed schemes against exact solutions. A more detailed analysis of these simulations, including the long-time behavior of the numerical solutions, is left for future work. In this paper, the proposed semi-discrete systems were integrated in time by the improved Euler method. An integrable full discretization of the MCSP equation admitting the coupled complex short pulse
equation as a reduction has recently been reported in~\cite{HoriEtAl2026}. An
important direction for future research is the development of integrable fully
discrete schemes for both the MCmSP and MCSP equations with nonzero boundary values,
including the consistent treatment of boundary data at the fully discrete level.
\end{section}

\section*{Declarations}

\subsection*{Funding}
This work was partially supported by JSPS KAKENHI Grant Numbers JP22K03441,
JP23K22407, JP26K06919 and Waseda University Grants for Special Research Projects.

\subsection*{Conflict of interest}
The authors declare no conflict of interest.

\subsection*{Ethics approval and consent to participate}
Not applicable.

\subsection*{Consent for publication}
Not applicable.

\subsection*{Data availability}
The data generated and analyzed during the current study are available from the corresponding author upon reasonable request.

\subsection*{Materials availability}
Not applicable.

\subsection*{Code availability}
The code used in this study is available from the corresponding author upon reasonable request.

\subsection*{Author contribution}
All authors contributed equally to this work.

\backmatter


\begin{appendices}
\numberwithin{equation}{section}

\section{}
\label{secA1}


From the second equation of the dependent variable transformation (\ref{m-dependenttransformation}), we have
\begin{align}
x&=\int \rho(X,T)dX=x_{0}(T)+\int_{0}^{X}\rho(\bar{X},T)d\bar{X}\nonumber\\
&=x_{0}(T)+\int_{0}^{X}\left(1-(\log{f(\bar{X},T)})_{\bar{X}T}\right)d\bar{X}\nonumber\\
&=x_{0}(T)+X-(\log{f(X,T)})_{T}+(\log{f(0,T)})_{T}.
\label{B1}
\end{align}
Differentiating (\ref{B1}) with respect to $T$, one obtains
\begin{eqnarray}
\frac{\partial x}{\partial T}=\frac{\partial x_{0}(T)}{\partial T}-(\log{f(X,T)})_{TT}+(\log{f(0,T)})_{TT}.
\end{eqnarray}
The second bilinear equation in (\ref{MCmSPbilinear}) gives
\begin{eqnarray}
F(0,T)=\sum_{1\leq j<k\leq n}c_{jk}u^{(j)}(0,T)u^{(k)}(0,T)=(\log f(0,T))_{TT}.
\end{eqnarray}
Therefore, the consistency condition $x_{0,T}=-F(0,T)$ is written as
\begin{eqnarray}
\frac{\partial x_{0}}{\partial T}=-(\log{f(0,T)})_{TT}.
\end{eqnarray}
Choosing $x_{0}(T)=-(\log{f(0,T)})_{T}$, we obtain $x=X-(\log{f(X,T)})_{T}$.

Next, we consider the discrete case in the same way.  From the second equation of the dependent variable transformation (\ref{disdependenttransformation}), we obtain
\begin{align}
x_{l}&=x_{0}(T)+\sum_{m=0}^{l-1}2a\rho_{m}=x_{0}(T)+\sum_{m=0}^{l-1}2a\left(1-\frac{1}{2a}\left(\log{\frac{f_{m+1}}{f_{m}}}\right)_{T}\right)\nonumber\\
&=x_{0}(T)+2al-(\log{f_{l}})_{T}+(\log{f_{0}})_{T}.
\label{B5}
\end{align}
Differentiating (\ref{B5}) with respect to $T$, one obtains
\begin{eqnarray}
\frac{dx_{l}}{dT}=\frac{dx_{0}(T)}{dT}-(\log{f_{l}})_{TT}+(\log{f_{0}})_{TT}.
\label{B6}
\end{eqnarray}
The second bilinear equation in (\ref{disbilinear}) gives
$\sum_{1\leq j<k\leq n}c_{jk}u_{0}^{(j)}u_{0}^{(k)}=(\log f_{0})_{TT}$.
Therefore, the discrete consistency condition (\ref{mSP_disconst}) is written as
\begin{eqnarray}
\frac{dx_{0}}{dT}=-(\log{f_{0}})_{TT}.
\end{eqnarray}
Choosing $x_{0}(T)=-(\log{f_{0}})_{T},$ we obtain $x_{l}=2al-(\log{f_{l}})_{T}$.

\section{}
\label{secA2}
From the second equation of the dependent variable transformation (\ref{dependenttransformation}), we have
\begin{align}
x&=\int \rho(X,T)dX=x_{0}(T)+\int_{0}^{X}\rho(\bar{X},T)d\bar{X}\nonumber\\
&=x_{0}(T)+\int_{0}^{X}\left(1-2(\log{f(\bar{X},T)})_{\bar{X}T}\right)d\bar{X}\nonumber\\
&=x_{0}(T)+X-2(\log{f(X,T)})_{T}+2(\log{f(0,T)})_{T}.
\label{A1}
\end{align}
Differentiating (\ref{A1}) with respect to $T$, one obtains
\begin{eqnarray}
\frac{\partial x}{\partial T}=\frac{\partial x_{0}(T)}{\partial T}-2(\log{f(X,T)})_{TT}+2(\log{f(0,T)})_{TT}.
\end{eqnarray}
The second bilinear equation in (\ref{bilinearMCSP}) gives
\begin{eqnarray}
\frac{1}{2}F(0,T)=\frac{1}{2}\sum_{1\leq j<k\leq n}c_{jk}u^{(j)}(0,T)u^{(k)}(0,T)=2(\log f(0,T))_{TT}.
\end{eqnarray}
Therefore, the consistency condition $x_{0,T}=-F(0,T)/2$ is written as
\begin{eqnarray}
\frac{\partial x_{0}}{\partial T}=-2(\log{f(0,T)})_{TT}.
\end{eqnarray}
Choosing $x_{0}(T)=-2(\log{f(0,T)})_{T}$, we have $x=X-2(\log{f(X,T)})_{T}$.

Next, we consider the discrete case in the same way.  From the second equation of the dependent variable transformation (\ref{disdependenttransformationMCSP}), we obtain
\begin{align}
x_{l}&=x_{0}(T)+\sum_{m=0}^{l-1}2a\rho_{m}=x_{0}(T)+\sum_{m=0}^{l-1}2a\left(1-\frac{1}{a}\left(\log{\frac{f_{m+1}}{f_{m}}}\right)_{T}\right)\nonumber\\
&=x_{0}(T)+2al-2(\log{f_{l}})_{T}+2(\log{f_{0}})_{T}.
\label{A5}
\end{align}
Differentiating (\ref{A5}) with respect to $T$, one obtains
\begin{eqnarray}
\frac{dx_{l}}{dT}=\frac{dx_{0}(T)}{dT}-2(\log{f_{l}})_{TT}+2(\log{f_{0}})_{TT}.
\label{A6}
\end{eqnarray}
The second bilinear equation in (\ref{disbilinearMCSP}) gives
$\frac{1}{2}\sum_{1\leq j<k\leq n}c_{jk}u_{0}^{(j)}u_{0}^{(k)}=2(\log f_{0})_{TT}$.
Therefore, the discrete consistency condition (\ref{disSP_const}) is written as
\begin{eqnarray}
\frac{dx_{0}}{dT}=-2(\log{f_{0}})_{TT}.
\end{eqnarray}
Choosing $x_{0}(T)=-2(\log{f_{0}})_{T}$, we obtain $x_{l}=2al-2(\log{f_{l}})_{T}$.




\end{appendices}


\bibliography{sn-bibliography}

\end{document}